%% file: networks.tex
\documentclass[12pt, oneside]{book}

\input{preamble.tex}

\makeindex  

\begin{document}

\frontmatter

\title{Economic Networks: \\
    Theory and Computation\\ 
    \vspace{1em}
\textsc{\normalsize QuantEcon Book I} \\ }

\author{John Stachurski and Thomas J. Sargent}

\maketitle

\tableofcontents

\input{ch_preface.tex}

\input{common_symbols.tex}

\mainmatter

\input{ch_intro.tex}

\input{ch_production.tex}

\input{ch_opt.tex}

\input{ch_mcs.tex}

\input{ch_fpms.tex}

\part{Appendices}

\input{appendix.tex}

\backmatter

\printindex

\listofauthors

\bibliographystyle{apalike}
\bibliography{qe_bib}

\end{document}

%% file: preamble.tex
\usepackage{amsmath, amssymb, amsthm, amsbsy}
\usepackage[mathscr]{euscript} 
\usepackage{stmaryrd} 
\usepackage{bbm}
\usepackage[usenames,dvipsnames]{xcolor}
\usepackage{graphicx}
\usepackage{epigraph}
\usepackage{graphviz} 
\usepackage{multicol} 
\usepackage{caption} 
\usepackage{subcaption} 
\usepackage[left=1.2in, right=1.2in, top=1.0in, bottom=0.8in, includehead, includefoot]{geometry}
\usepackage{setspace}

\usepackage[xcharter]{newtxmath}
\usepackage{XCharter}

\usepackage{makeidx}

\usepackage{authorindex} 

\usepackage[frozencache=true, cachedir=minted-cache]{minted}

\usepackage[lastexercise, answerdelayed]{exercise}

\usepackage{chngcntr}
\counterwithin{Exercise}{section}
\counterwithin{Answer}{section}

\usepackage{natbib}

\usepackage{enumitem}
\setlist[enumerate]{itemsep=2pt,topsep=3pt}
\setlist[itemize]{itemsep=2pt,topsep=3pt}
\setlist[enumerate,1]{label={\upshape (\roman*)}}

\usepackage[linesnumbered]{algorithm2e}
\SetAlCapSkip{1.2em}

\usepackage{tikz}
\usepackage{tikz-cd}
\usetikzlibrary{positioning}
\usetikzlibrary{arrows}
\usetikzlibrary{calc}
\usetikzlibrary{intersections}
\usetikzlibrary{matrix}
\usetikzlibrary{decorations}
\usepackage{pgf}
\usepackage{pgfplots}
\pgfplotsset{compat=1.16} 
\usetikzlibrary{shapes, fit}
\usetikzlibrary{arrows.meta} 
\usetikzlibrary{decorations.pathreplacing}  

\usepackage[colorlinks=true, citecolor=Blue, linkcolor=blue]{hyperref}
\newcommand\myshade{85}
\colorlet{mylinkcolor}{violet}
\colorlet{mycitecolor}{PineGreen}
\colorlet{myurlcolor}{Aquamarine}
\hypersetup{
  linkcolor  = mylinkcolor!\myshade!black,
  citecolor  = mycitecolor!\myshade!black,
  urlcolor   = myurlcolor!\myshade!black,
  colorlinks = true,
}

\newcommand{\navy}[1]{\textcolor{MidnightBlue}{\bf #1}}

\definecolor{pale}{RGB}{235, 235, 235}
\definecolor{pale2}{RGB}{255,255,255}
\definecolor{aquamarine}{RGB}{69,139,116}

\newcommand{\CC}{\mathbbm C}
\newcommand{\EE}{\mathbbm E}

\newcommand{\RR}{\mathbbm R}

\newcommand{\MM}{\mathbbm M}
\newcommand{\NN}{\mathbbm N}
\newcommand{\PP}{\mathbbm P}

\newcommand{\QQ}{\mathbbm Q}

\newcommand{\ZZ}{\mathbbm Z}

\newcommand{\Xsf}{\mathsf X}
\newcommand{\Ysf}{\mathsf Y}

\newcommand{\cC}{\mathscr C}
\newcommand{\dD}{\mathscr D}

\newcommand{\gG}{\mathscr G}

\newcommand{\iI}{\mathscr I}
\newcommand{\fF}{\mathscr F}

\newcommand{\pP}{\mathscr P}

\newcommand{\sS}{\mathscr S}

\newcommand{\mM}{\mathscr M}
\newcommand{\oO}{\mathscr O}

\renewcommand{\leq}{\leqslant}
\renewcommand{\geq}{\geqslant}

\renewcommand{\phi}{\varphi}
\renewcommand{\epsilon}{\varepsilon}

\providecommand{\inner}[1]{\left\langle{#1}\right\rangle}

\newcommand{\matset}[2]{ \MM^{ #1 \times #2 } }

\newcommand{\argmax}{\operatornamewithlimits{argmax}}
\newcommand{\argmin}{\operatornamewithlimits{argmin}}

\newcommand{\st}{\ensuremath{\ \mathrm{s.t.}\ }}
\newcommand{\setntn}[2]{ \{ #1 : #2 \} }
\newcommand{\natset}[1]{[ #1 ]}

\newcommand{\given}{\, | \,}
\newcommand{\cf}[1]{ \lstinline|#1| }
\newcommand{\fore}{\therefore \quad}
\newcommand{\1}{\mathbbm 1}
\newcommand{\me}{\mathrm{e}}               
\newcommand*\diff{\mathop{}\!\mathrm{d}}   

\newcommand{\la}{\langle}
\newcommand{\ra}{\rangle}

\newcommand{\eqdist}{\stackrel{d} {=} }
\newcommand{\iidsim}{\stackrel{\textrm{ {\sc iid }}} {\sim} }

\newcommand{\listofauthorsname}{List of Authors}%

\newcommand{\listofauthors}{%
\chapter*{\listofauthorsname}%
\phantomsection%
\addcontentsline{toc}{chapter}{\listofauthorsname}%
\noindent%
\printauthorindex%
}%

\DeclareMathOperator{\Exp}{Exp}  

\DeclareMathOperator{\trace}{trace}

\DeclareMathOperator{\Span}{span}
\DeclareMathOperator{\diag}{diag}

\DeclareMathOperator{\rank}{rank}
\DeclareMathOperator{\kernel}{null}

\DeclareMathOperator{\var}{Var}

\DeclareMathOperator{\range}{range}

\DeclareMathOperator{\epi}{epi}
\DeclareMathOperator{\vecop}{vec}

\DeclareMathOperator{\real}{Re}
\DeclareMathOperator{\imag}{Im}

\DeclareMathOperator{\csum}{colsum} 
\DeclareMathOperator{\rsum}{rowsum} 

\theoremstyle{plain}

\newtheorem{theorem}{Theorem}[section]
\newtheorem{corollary}[theorem]{Corollary}
\newtheorem{lemma}[theorem]{Lemma}
\newtheorem{proposition}[theorem]{Proposition}

\theoremstyle{definition}

\newtheorem{example}{Example}[section]
\newtheorem{remark}{Remark}[section]
\newtheorem{assumption}{Assumption}[section]


%% file: ch_preface.tex
\chapter{Preface}\label{c:preface}

The development and use of network science has grown exponentially since the
beginning of the 21st century.  The ideas and techniques found in this field
are already core tools for analyzing a vast range of phenomena, from epidemics and
disinformation campaigns to chemical reactions and brain function.  

In economics, network theory is typically taught as a specialized
subfield, available to students as one of many elective courses
towards the end of their program. However, we are rapidly approaching the
stage where every aspiring scientist---including social scientists and
economists---wants to know the foundations of this field.  It is
arguably the case that, just as every well-trained economist learns the basics
of convex optimization, maximum likelihood and linear regression, so too
should every graduate student in economics learn the fundamental ideas of network theory.

This textbook is an introduction to economic networks, intended for students
and researchers in the fields of economics and applied mathematics. The
textbook emphasizes quantitative modeling, with the main underlying tools
being graph theory, linear algebra, fixed point theory and programming. Most
mathematical tools are covered from first principles, with the two main
technical results---the Neumann series lemma and the Perron--Frobenius
theorem---playing a central role. 

The text is suitable for a one-semester course, taught either to
advanced undergraduate students who are comfortable with linear
algebra or to beginning graduate students.  (For example, although we define
eigenvalues, an ideal student would already know what eigenvalues and
eigenvectors are, so that concepts like ``eigenvector centrality'' or results
like the Neumann series lemma are readily absorbed.) The text will also suit
students from mathematics, engineering, computer science and other related
fields who wish to learn about connection between economics and networks. 

Several excellent textbooks on network theory in economics
and social science already exist, including \cite{jackson2010social},
\cite{easley2010networks}, and \cite{borgatti2018analyzing}, as well as the
handbook by \cite{bramoulle2016oxford}. These textbooks have broad scope and
treat many useful applications. In contrast, our book is narrower and
more technical.  It provides mathematical, computational and
graph-theoretic foundations that are required to understand and apply network
theory, along with a treatment of some of the most important network
applications in economics, finance and operations research.  It can be used as
a complementary resource, or as a preliminary course that facilitates
understanding of the alternative texts listed above, as well as research
papers in the area.  

The book contains a mix of Python and Julia code. The majority is in Python
because the libraries are somewhat more stable at the time of writing,
although Julia also has strong graph manipulation and optimization libraries.
Code for figures is available from the authors.  There are many solved
exercises, ranging from simple to quite hard.  At the end of each chapter we
provide notes, informal comments and references.

We are greatly indebted to Jim Savage and Schmidt Futures for generous
financial support, as well as to Shu Hu and Chien Yeh for their outstanding
research assistance.  QuantEcon research fellow Matthew McKay generously lent
us his time and remarkable expertise in data analysis, networks and
visualization.  QuantEcon research assistant Mark Dawkins turned a messy
collection of code files into an elegant companion Jupyter book.  For many
important fixes, comments and suggestions, we thank Quentin Batista, Rolf
Campos, Fernando Cirelli, Rebekah Dix, Saya Ikegawa, Fazeleh Kazemian, Dawie
van Lill, Simon Mishricky, Pietro Monticone, Flint O'Neil, Zejin Shi, Akshay
Shanker, Arnav Sood, Natasha Watkins, Chao Wei, and Zhuoying Ye.  Finally,
Chase Coleman, Alfred Galichon, Spencer Lyon, Daisuke Oyama and Jesse Perla
are collaborators at QuantEcon, and almost everything we write has benefited
from their input.  This text is no exception.

%% file: common_symbols.tex
\clearpage{\pagestyle{empty}\cleardoublepage}

\chapter*{Common Symbols}
\addcontentsline{toc}{chapter}{Common Symbols}
\label{c:cs}

{\setstretch{1.2}

\begin{tabular}{c | l}\label{c:coms}
    $P \implies Q$ & $P$ implies $Q$ \\
    $P \iff Q$ & $P \implies Q$ and $Q \implies P$ \\
    $\natset{n}$  & the set $\{1, \ldots, n\}$ \\
    $\alpha := 1$ & $\alpha$ is defined as equal to $1$ \\
    $f \equiv 1$ & function $f$ is everywhere equal to $1$ \\
    $\wp(A)$ & the power set of $A$; that is, the collection of all subsets of given set $A$ \\
    $\NN$, $\ZZ$ and $\RR$ & the natural numbers, integers and real numbers
        respectively  \\ 
    $\CC$ & the set of complex numbers (see \S\ref{sss:complex}) \\
    $\ZZ_+$, $\RR_+$, etc. & the nonnegative elements of $\ZZ$, $\RR$, etc. \\
    $\matset{n}{k}$ & all $n \times k$ matrices  \\
    $\diag(a_1, \ldots, a_n)$ & the diagonal matrix with $a_1, \ldots a_n$ on
        the principle diagonal \\
    $\delta_x$ & the probability distribution concentrated on point $x$\\
    $|x|$ & the absolute value of $x \in \RR$ \\
    $|B|$ & the cardinality of (number of elements in) set $B$ \\
    $f \colon A \to B$ & $f$ is a function from set $A$ to set $B$  \\
    $\RR^S$ & the set of all functions from $S$ to $\RR$  \\
    $\RR^n$ & all $n$-tuples of real numbers  \\
    $\| x \|_1$ & the $\ell_1$ norm $\sum_i |x_i|$  (see \S\ref{sss:norms}) \\
    $\| x \|_\infty$ & the $\ell_\infty$ norm $\max_i |x_i|$ (see \S\ref{sss:norms}) \\
    $\| A \|$ when $A \in \matset{n}{k}$ & the operator norm of $A$ (see \S\ref{sss:onorm}) \\
    $\la a, b \ra$ & the inner product of $a$ and $b$ \\
    $g \ll h$ & function (or vector) $g$ is everwhere strictly less than $h$ \\
    $\1$  & vector of ones or function everywhere equal to one \\
    $\1\{P\}$  & indicator, equal to 1 if statement $P$ is true and 0 otherwise \\
    {\sc iid} & independent and identically distributed  \\
\end{tabular}

\begin{tabular}{c || l}
    $i_d(v)$  & in-degree of node $v$ \\
    $o_d(v)$  & out-degree of node $v$ \\
    $\iI(v)$  & set of direct predecessors of node $v$ \\
    $\oO(v)$  & set of direct successors of node $v$ \\
    $u \to v$  & node $v$ is accessible from node $u$ \\
    $X \eqdist Y$ & $X$ and $Y$ have the same distribution \\
    $X \sim F$ & $X$ has distribution $F$ \\
    $\Pi(\phi, \psi)$ & the set of all couplings of $(\phi, \psi)$ \\
\end{tabular}

}

\clearpage{\pagestyle{empty}\cleardoublepage}

%% file: ch_intro.tex
\chapter{Introduction}\label{c:intro}

\epigraph{Relations are the fundamental fabric of reality.}{Michele Coscia}

\section{Motivation}\label{s:prelob}

Alongside the exponential growth of computer networks over the last few decades,
we have witnessed concurrent and equally rapid growth in a field called
\emph{network science}.  Once computer networks brought network
structure into clearer focus, scientists began to 
recognize networks almost everywhere, even in phenomena that had already
received centuries of attention using other methods, and to apply network theory to organize and
expand knowledge right throughout the sciences, in every field and discipline.

The set of possible examples is vast, and sources mentioning or
treating hundreds of different applications of network methods and graph
theory are listed in the reading notes at the end of the chapter.  In computer
science and machine learning alone, we see computational graphs, graphical
networks, neural networks and deep learning.  In operations research, network
analysis focuses on minimum cost flow, traveling salesman, shortest path, and
assignment problems.  In biology, networks are a standard way to represent
interactions between bioentities.

In this text, our interest lies in economic and social phenomena.
Here, too, networks are pervasive.  Important examples include financial
networks, production networks, trade networks, transport networks and social
networks.   For example, social and information networks affect trends in
sentiments and opinions, consumer decisions, and a range of peer effects.  The
topology of financial networks helps to determine relative fragility of the
financial system, while the structure of production networks affects trade,
innovation and the propagation of local shocks.

Figures~\ref{f:crude_oil_2019}--\ref{f:commercial_aircraft_2019_1} show
two examples of trade networks. Figure~\ref{f:crude_oil_2019} is called a
\emph{Sankey diagram}, which is a kind of figure used to represent flows.
Oil flows from left to right.
The countries on the left are the top 10 exporters of crude oil, while the countries on the
right are the top 20 consumers.  The figure relates to one of our core topics:
optimal (and equilibrium) flows across networks.  We treat optimal flows at
length in Chapter~\ref{c:ofd}.\footnote{This figure
    was constructed by QuantEcon research fellow Matthew McKay, using
    International Trade Data (SITC, Rev 2) collected by The Growth Lab at
    Harvard University.}

Figure~\ref{f:commercial_aircraft_2019_1} shows international trade in large
commercial aircraft in
2019.\footnote{This figure was also constructed by Matthew McKay, using data
          2019 International Trade Data SITC Revision 2, code 7924.  The data
               pertains to trade in commercial aircraft weighted at least
           15,000kg. It was sourced from CID Dataverse.} Node size is
           proportional to total exports and link width is proportional to
           exports to the target country.  The US, France and Germany are
           revealed as major export hubs.  

\begin{figure}
   \centering
   \scalebox{1.1}{\includegraphics[trim = 10mm 10mm 10mm 20mm, clip]{ 
       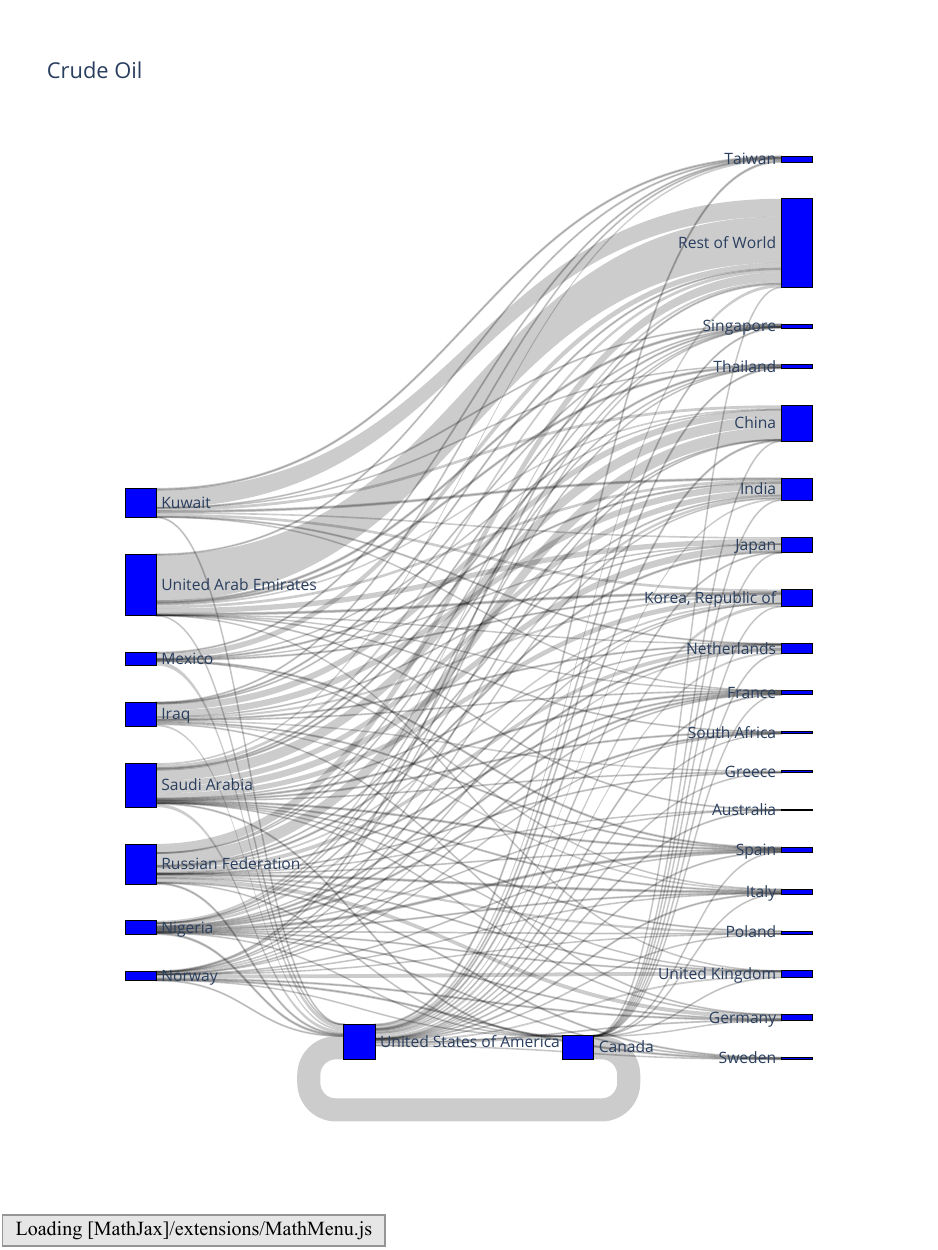}}
   \caption{\label{f:crude_oil_2019} International trade in crude oil 2019}
\end{figure}

\begin{figure}
   \centering
   \scalebox{0.65}{\includegraphics[trim = 0mm 30mm 0mm 30mm, clip]{ 
       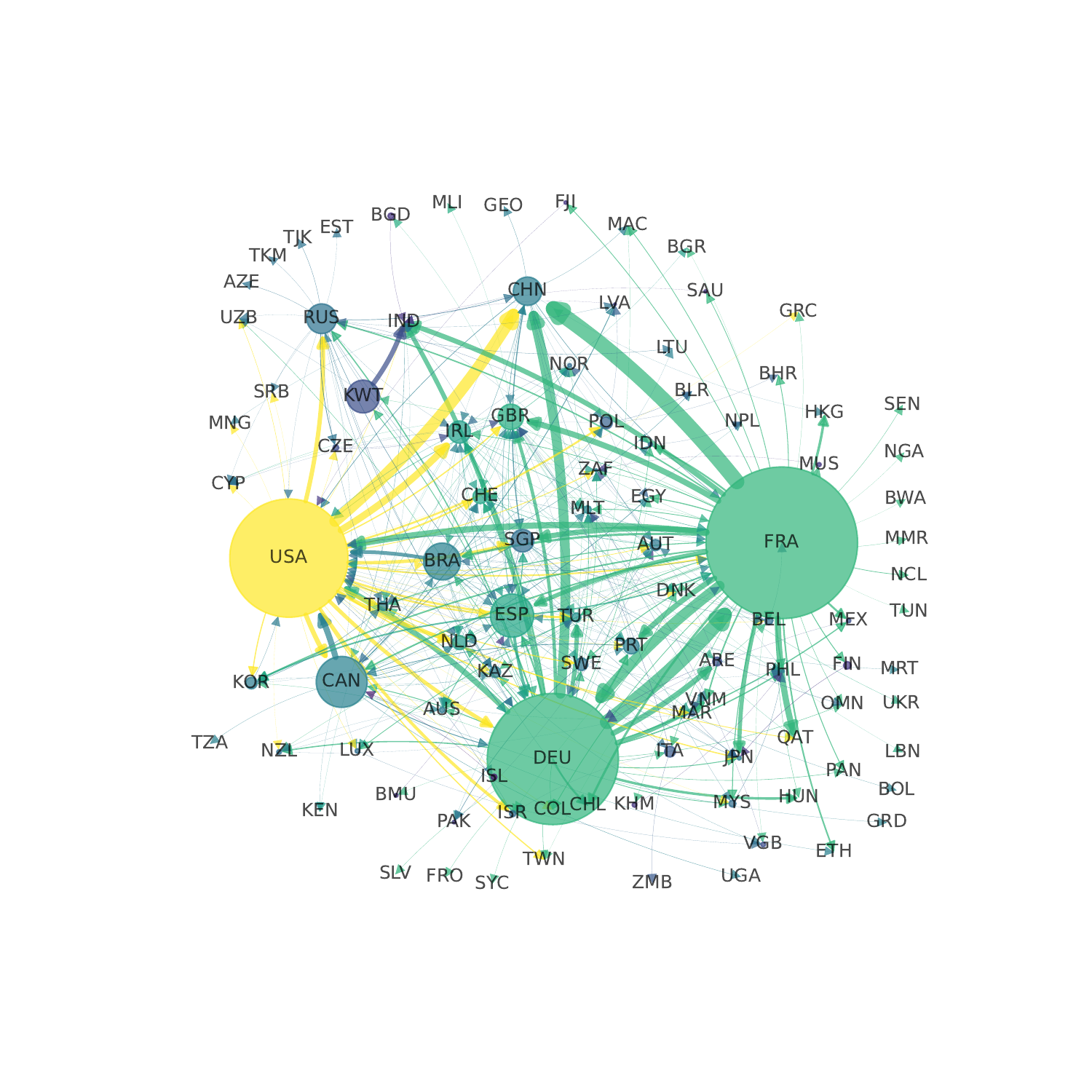}}
   \caption{\label{f:commercial_aircraft_2019_1} International trade in
   commercial aircraft during 2019}
\end{figure}

While some readers viewing
Figures~\ref{f:crude_oil_2019}--\ref{f:commercial_aircraft_2019_1} might at
first
suspect that the network perspective adds little more than an attractive
technique for visualizing data, it actually adds much more. For example, in
Figure~\ref{f:commercial_aircraft_2019_1}, node colors are based on a ranking
of ``importance'' in the network called \emph{eigenvector centrality}, which
we introduce in \S\ref{sss:eigcen}.  Such rankings and centrality measures are
an active area of research among network scientists. Eigenvector and other
forms of centrality feature throughout the text. For example, we will see that 
these concepts are closely connected to---and shed new light
on---fundamental ideas first developed many years ago by researchers in the
field of input-output economics. 

In addition, in production networks, it turns out that the nature of shock
propagation is heavily dependent on the underlying structure of the network.
For example, for a few highly connected nodes, shocks occurring
within one firm or sector can have an outsized influence on aggregate-level
fluctuations.  Economists are currently racing to understand these
relationships, their interactions with various centrality measures, and other
closely related phenomena.

To understand this line of work, as well as other applications of network
methods to economics and finance, some technical foundations are required. For
example, to define eigenvector centrality, we need to be familiar with
eigenvectors, spectral decompositions and the Perron--Frobenius theorem.  To
work with \emph{Katz centrality}, which also features regularly in network
science and economics, we require a sound understanding of the Neumann series
lemma. The Perron--Frobenius theorem and the Neumann series lemma form much of
the technical foundation of this textbook. We review them in detail in
\S\ref{s:introcs} and develop extensions throughout remaining chapters.  


One reason that analysis of networks is challenging is high-dimensionality. To
see why, consider implementing a model with $n$ economic agents.  This
requires $n$ times more data than one representative agent in a setting where
agents are atomistic or coordinated by a fixed number of prices.  For example,
\cite{carvalho2019large} model the dynamics of $n = 6 \times 10^6$ firms, all
of which need to be tracked when running a simulation. However, if we wish to
model interactions between each pair $i, j$ (supply linkages, liabilities,
etc.), then, absent sparsity conditions, the data processing requirement grows
like $O(n^2)$.\footnote{See \S\ref{ss:convergence} for a discussion of big O
notation.} In the \cite{carvalho2019large} example, $n^2$ is $3.6 \times
10^{13}$, which is very large even for modern computers.    One lesson is that
network models can be hard to solve, even with powerful computers, unless we
think carefully about algorithms.

In general, to obtain a good grasp on the workings of economic networks, we
will need computational skills plus a firm understanding of linear algebra,
probability and a field of discrete mathematics called graph theory.  The rest
of this chapter provides relevant background in these topics.  
Before tackling this background, we recommend that readers skim the list of
common symbols on page~\pageref{c:coms}, as well the
mathematical topics in the appendix, which start on page~\pageref{c:ap}.
(The appendix is not intended for sequential reading, but rather as a
source of definitions and fundamental results to be drawn on in what follows.)

\section{Spectral Theory}\label{s:introcs}

In this section we review some linear algebra needed for the study of
graphs and networks.  Highlights include the spectral decomposition of
diagonalizable matrices, the Neumann series lemma, and the fundamental theorem
of Perron and Frobenius.

\subsection{Eigendecompositions}\label{ss:eigen}

Our first task is to cover spectral decompositions and the spectral theorem.
We begin with a brief review of eigenvalues and their properties.
(If you are not familiar with eigenvalues and eigenvectors, please consult an
elementary treatment first.  See, for example, \cite{cohen2021linear}.)


\subsubsection{Eigenvalues}

Fix $A$ in $\matset{n}{n}$.  A scalar $\lambda \in \CC$ is
called an \navy{eigenvalue}\index{Eigenvalue} of $A$ if there exists a
nonzero $e \in \CC^n$ such that $Ae = \lambda e$.  A vector $e$
satisfying this equality is called an \navy{eigenvector}\index{Eigenvector}
corresponding to the eigenvalue $\lambda$.  (Notice that eigenvalues and
eigenvectors are allowed to be complex, even though we restrict elements of
$A$ to be real.) The set of all eigenvalues of $A$ is called the
\navy{spectrum}\index{Spectrum} of $A$ and written as $\sigma(A)$.  As we show
below, $A$ has at most $n$ distinct eigenvalues.

In Julia, we can check for the eigenvalues of a given square matrix $A$ via
\texttt{eigvals(A)}.  Here is one example
\begin{minted}{julia}
using LinearAlgebra
A = [0 -1;
     1  0]
eigenvals = eigvals(A) 
\end{minted}
Running this code in a Jupyter cell (with Julia kernel) or Julia REPL produces 
\begin{minted}{julia}
2-element Vector{ComplexF64}:
 0.0 - 1.0im
 0.0 + 1.0im
\end{minted}
Here \texttt{im} stands for $i$, the imaginary unit (i.e., $i^2=-1$). 

\begin{Exercise}
    Using pencil and paper, confirm that Julia's output is correct.  In
    particular,  show that
    \begin{equation*}
        A = 
        \begin{pmatrix}
            0 & -1 \\
            1 & 0
        \end{pmatrix}
        \quad \implies \quad
        \sigma(A) = \{i, -i\},
    \end{equation*}
    with corresponding eigenvectors $(-1, i)^\top$ and $(-1, -i)^\top$.
\end{Exercise}

\begin{Answer}
    For $\lambda_1= i$, we have
    \begin{equation*}
        A e_1 = 
        \begin{pmatrix}
            0 & -1 \\
            1 & 0
        \end{pmatrix} 
        \begin{pmatrix}
            -1 \\
            i
        \end{pmatrix} =
        \begin{pmatrix}
            -i \\
            -1
        \end{pmatrix} = \lambda_1 e_1.
    \end{equation*}

    Similarly for $\lambda_2= -i$, we have
    \begin{equation*}
        A e_2 = 
        \begin{pmatrix}
            0 & -1 \\
            1 & 0
        \end{pmatrix} 
        \begin{pmatrix}
            -1 \\
            -i
        \end{pmatrix} =
        \begin{pmatrix}
            i \\
            -1
        \end{pmatrix} = \lambda_2 e_2.
    \end{equation*}
\end{Answer}

If $\lambda \in \sigma(A)$ and $e$ is an eigenvector for
$\lambda$, then $(\lambda, e)$ is called an \navy{eigenpair}\index{Eigenpair}.  

\begin{Exercise}\label{ex:uniev}
    Prove: if $(\lambda, e)$ is an eigenpair of $A$ and $\alpha$ is a nonzero
    scalar, then $(\lambda, \alpha e)$ is also an eigenpair of $A$.  
\end{Exercise}

\begin{Answer}
    Fix an eigenpair $(\lambda, e)$ of $A$ and a nonzero scalar $\alpha$. We have
    \begin{equation*}
        A (\alpha e) = \alpha A e
                     = \lambda (\alpha e).
    \end{equation*}
    Hence $\alpha e$ is an eigenvector and $\lambda$ is an eigenvalue, as claimed.
\end{Answer}

\begin{lemma}\label{l:chlam}
    $\lambda \in \CC$ is an eigenvalue of $A$ if and only if $\det(A -
    \lambda I) = 0$. 
\end{lemma}

\begin{proof}
    If $\lambda \in \RR$, then Lemma~\ref{l:chlam} follows directly from
    Theorem~\ref{t:nncase} on page~\pageref{t:nncase}, since $\det(A - \lambda I)
    = 0$ is equivalent to existence of nonzero vector $e$ such that $(A - \lambda
    I) e = 0$, which in turn says that $\lambda$ is an eigenvalue of $A$.
    The same arguments extend to the case $\lambda \in \CC$ because
    the statements in Theorem~\ref{t:nncase} are also valid for complex-valued
    matrices (see, e.g., \cite{jan}).
\end{proof}

It can be shown that $p(\lambda) := \det(A
-\lambda I)$ is a polynomial of degree $n$.\footnote{See, for example,
\cite{jan}, Chapter 6.}  This polynomial is called the
\navy{characteristic polynomial}\index{characteristic polynomial} of $A$.  By
the Fundamental Theorem of Algebra, there are $n$ roots (i.e., solutions in
$\CC$ to the equation $p(\lambda)=0$), although some may be repeated 
as in the complete factorization of $p(\lambda)$.  By Lemma~\ref{l:chlam},
\begin{enumerate}
    \item each of these roots is an eigenvalue, and
    \item no other eigenvalues exist besides these $n$ roots.
\end{enumerate}

If $\lambda \in \sigma(A)$ appears $k$ times in the factorization of
the polynomial $p(\lambda)$, then $\lambda$ is said to have \navy{algebraic multiplicity} 
$k$.  An eigenvalue with algebraic multiplicity one is called
\navy{simple}\index{Simple eigenvalue}.  A simple eigenvalue $\lambda$ has the
property that its eigenvector is unique up to a scalar multiple, in the sense
of Exercise~\ref{ex:uniev}.  In other words, the linear span of 
$\setntn{e \in \CC^n}{(\lambda, e) \text{ is an eigenpair}}$ (called 
the \navy{eigenspace}\index{Eigenspace} of $\lambda$) is one-dimensional.

\begin{Exercise}\label{ex:tausigma}
    For $A \in \matset{n}{n}$, show that $\lambda \in \sigma(A)$ iff $\tau
    \lambda \in \sigma(\tau A)$ for all $\tau > 0$.
\end{Exercise}

\begin{Answer}
    Fix $A \in \matset{n}{n}$ and $\tau > 0$.  If
    $\lambda \in \sigma(A)$, then $\tau^n \det(A - \lambda I) = 0$, or
    $\det(\tau A - \tau \lambda I) = 0$.  Hence $\tau \lambda \in \sigma(\tau
    A)$.  To obtain the converse implication, multiply by $1/\tau$.
\end{Answer}

\begin{Exercise}\label{ex:dirip}
    A useful fact concerning eigenvectors is that
    if the characteristic polynomial $p(\lambda) := \det(A -
    \lambda I)$ has $n$ distinct roots, then the $n$ corresponding
    eigenvectors form a basis of $\CC^n$.  Prove this for the case where all
    eigenvectors are real---that is show that the $n$ (real) eigenvectors
    form a basis of $\RR^n$.  (Bases are defined in \S\ref{sss:bvd}.  Proving
    this for $n=2$ is also a good effort.)
\end{Exercise}

\begin{Answer}
    If $p(\lambda) := \det(A - \lambda I)$ has $n$ distinct roots, then
    $|\sigma(A)|=n$.  For each $\lambda_i \in \sigma(A)$, let $e_i$ be a
    corresponding eigenvector.  It suffices to show that $\{e_i\}_{i=1}^n$ is linearly
    independent.  To this end, let $k$ be the largest number such that $\{e_1,
    \ldots, e_k\}$ is independent.  Seeking a contradiction, suppose that $k < n$.
    Then $e_{k+1} = \sum_{i=1}^k \alpha_i e_i$ for suitable scalars $\{\alpha_i\}$.
    Hence, by $A e_{k+1} = \lambda_{k+1} e_{k+1}$, we have
    \begin{equation*}
        \sum_{i=1}^k \alpha_i \lambda_i e_i
        = \sum_{i=1}^k \alpha_i \lambda_{k+1} e_i
        \quad \iff \quad
        \sum_{i=1}^k \alpha_i (\lambda_i - \lambda_{k+1}) e_i = 0.
    \end{equation*}
    Since $\{e_1, \ldots, e_k\}$ is independent, we have 
    $\alpha_i (\lambda_i - \lambda_{k+1})  = 0$ for all $i$.  At least one
    $\alpha_i$ is nonzero, so $\lambda_i = \lambda_{k+1}$ for some $i \leq k$.
    Contradiction.
\end{Answer}

\subsubsection{The Eigendecomposition}\label{sss:eigde}

What are the easiest matrices to work with?  An obvious answer to this
question is: the diagonal matrices.  For example, when $D = \diag(\lambda_i)$ with $i \in \natset{n}$,
\begin{itemize}
    \item the linear system $Dx = b$ reduces to $n$ completely independent scalar equations,
    \item the $t$-th power $D^t$ is just $\diag(\lambda_i^t)$, and
    \item the inverse $D^{-1}$ is just $\diag(\lambda_i^{-1})$, assuming all
        $\lambda_i$'s are nonzero.
\end{itemize}

While most matrices are not diagonal,  there is a way that
``almost any'' matrix can be viewed as a diagonal matrix, after translation
of the usual coordinates in $\RR^n$ via an alternative basis.  This can be
extremely useful.  The key ideas are described below.

$A \in \matset{n}{n}$ is called
\navy{diagonalizable}\index{Diagonalizable} if 
\begin{equation*}
    A = P D P^{-1}
    \text{ for some } 
    D = \diag(\lambda_1, \ldots, \lambda_n)
    \text{ and nonsingular matrix } P.  
\end{equation*}
We allow both $D$ and $P$ to contain complex values.
The representation $P D P^{-1}$ is called the
\navy{eigendecomposition}\index{Eigendecomposition} or the \navy{spectral
decomposition} of $A$.  

One way to think about diagonalization is in terms of maps, as in
\begin{equation*}
    \begin{tikzcd}
        \RR^n \arrow{r}{A} \arrow[swap]{d}{P^{-1}} & \RR^n  \\
        \CC^n \arrow{r}{D}& \CC^n \arrow[swap]{u}{P}
    \end{tikzcd}
\end{equation*}
Either we can map directly with $A$ or, alternatively, we can shift to
$\CC^n$ via $P^{-1}$, apply the diagonal matrix $D$, and then shift back to $\RR^n$
via $P$.

The equality $A = P D P^{-1}$ can also be written as $A P =  P D$.  Decomposed
across column vectors, this equation says that each column of $P$ is an
eigenvector of $A$ and each element along the principal diagonal of $D$ is an
eigenvalue.  

\begin{Exercise}
    Confirm this.  Why are column vectors taken from $P$ nonzero, as
    required by the definition of eigenvalues?
\end{Exercise}

\begin{Answer}
    Suppose to the contrary that there is one zero column vector in $P$. 
    Then $P$ is not nonsingular. Contradiction.
\end{Answer}

\begin{Exercise}
    The trace of a matrix is equal to the sum of its eigenvalues, and the
    determinant is their product.  Prove this fact in the case where $A$ is
    diagonalizable. 
\end{Exercise}

\begin{Answer}
    Let $A$ be as stated, with $A = P D P^{-1}$.  Using elementary properties of
    the trace and determinant, we have
    \begin{equation*}
        \trace(A) 
        = \trace(P D P^{-1}) 
        = \trace(D P P^{-1}) 
        = \trace(D) 
        = \sum_i \lambda_i
    \end{equation*}
    and 
    \begin{equation*}
        \det(A) 
        = \det(P)\det(D) \det(P^{-1}) 
        = \det(P)\det(D) \det(P)^{-1} 
        = \det(D) 
        = \prod_i \lambda_i.
    \end{equation*}
\end{Answer}

\begin{Exercise}
    The asymptotic properties of $m \mapsto A^m$ are determined by the
    eigenvalues of $A$. This is clearest in the diagonalizable case, where 
    $A = P \diag(\lambda_i) P^{-1}$.  To illustrate, use induction to show that 
    \begin{equation}\label{eq:diagpow}
        A = P \diag(\lambda_i) P^{-1} 
        \; \implies \;
        A^m = P \diag (\lambda_i^m) P^{-1}
        \text{ for all } m \in \NN.
    \end{equation}
\end{Exercise}

When does diagonalizability hold?

While diagonalizability is not universal, the set of
matrices in $\matset{n}{n}$ that fail to be diagonalizable has ``Lebesgue
measure zero'' in $\matset{n}{n}$. (Loosely
speaking, only special or carefully constructed examples will fail to be
diagonalizable.) The next results provide conditions for the property.

\begin{theorem}\label{t:diagiff}
    A matrix $A \in \matset{n}{n}$ is diagonalizable if and only if its
    eigenvectors form a basis of $\CC^n$.
\end{theorem}

This result is intuitive: for $A = P D P^{-1}$ to hold we need $P$ to be
invertible, which requires that its $n$ columns are linearly independent.
Since $\CC^n$ is $n$-dimensional, this means that the columns form a basis of
$\CC^n$.  

\begin{corollary}\label{c:dimdi}
    If $A \in \matset{n}{n}$ has $n$ distinct eigenvalues, then $A$ is
    diagonalizable.
\end{corollary}

\begin{proof}
    See Exercise~\ref{ex:dirip}.
\end{proof}

\begin{Exercise}
    Give a counterexample to the statement that the condition in
    Corollary~\ref{c:dimdi} is necessary as well as sufficient.
\end{Exercise}

\begin{Answer}
    If $I$ is the identity and $Ie = \lambda e$ for some nonzero $e$, then 
    $e = \lambda e$ and hence $\lambda=1$.  Hence $\sigma(A) = \{1\}$.
    At the same time, $I$ is diagonalizable, since $I = I D I^{-1}$ when
    $D=I$.
\end{Answer}

There is another way that we can establish diagonalizability, based on
symmetry.  Symmetry also lends the diagonalization certain properties that
turn out to be very useful in applications.  We are referring to the following
celebrated theorem.

\begin{theorem}[Spectral theorem]\label{t:spec}
    If $A \in \matset{n}{n}$ is symmetric, then there exists a real orthonormal
    $n \times n$ matrix $U$ such that 
    \begin{equation*}
        A = U D U^\top
        \quad \text{with} \quad
        \lambda_i \in \RR_+ \text{ for all } i,
        \text{ where } 
        D = \diag(\lambda_1, \ldots, \lambda_n) .
    \end{equation*}
\end{theorem}

Since, for the orthonormal matrix $U$, we have $U^\top = U^{-1}$ (see
Lemma~\ref{l:porthms}), one consequence of the spectral theorem is that $A$ is
diagonalizable.  For obvious reasons, we often say that $A$ is
\navy{orthogonally diagonalizable}\index{Orthogonally diagonalizable}.

\subsubsection{Worker Dynamics}\label{sss:wdi}

Let's study a small application of the eigendecomposition.  Suppose that, each month, workers are hired at
rate $\alpha$ and fired at rate $\beta$.  Their two states are unemployment
(state 1) and employment (state 2).  Figure~\ref{f:worker_switching} shows the
transition probabilities for a given worker in each of these two states.

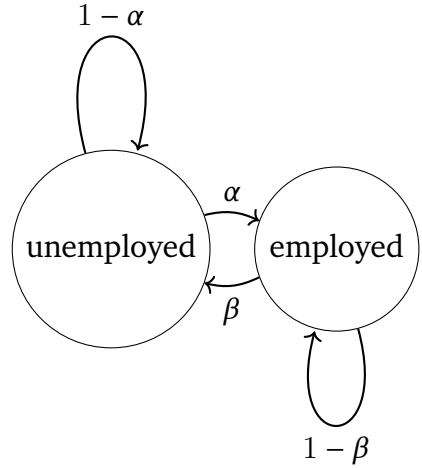
\begin{figure}
   \centering
   \input{tikz/worker_switching.tex}
   \caption{\label{f:worker_switching} Worker transition dynamics}
\end{figure}

We translate these dynamics into the matrix 
\begin{equation*}
    P_w = 
    \begin{pmatrix}
        1-\alpha & \alpha \\
        \beta & 1- \beta
    \end{pmatrix}
    \quad \text{where} \quad
    0 \leq \alpha, \beta \leq 1.
\end{equation*}

\begin{itemize}
    \item Row 1 of $P_w$ gives probabilities for unemployment and employment
        respectively when currently unemployed.
    \item Row 2 of $P_w$ gives probabilities for unemployment and employment
        respectively when currently employed.
\end{itemize}
\begin{Exercise}
    Using Lemma~\ref{l:chlam}, show that the two eigenvalues of $P_w$ are 
        $\lambda_1 := 1$
        and $\lambda_2 := 1-\alpha-\beta$.
    Show that, when $\min\{\alpha, \beta\} > 0$,
    \begin{equation*}
        e_1 := 
        \begin{pmatrix}
            1 \\
            1
        \end{pmatrix}
        \quad \text{and} \quad 
        e_2 :=
        \begin{pmatrix}
            -\alpha \\
            \beta
        \end{pmatrix}
    \end{equation*}
    are two corresponding eigenvectors, and that $\lambda_1$ and
    $\lambda_2$ are simple.
\end{Exercise}

\begin{Exercise}
    Show that, when $\alpha=\beta=0$, the eigenvalue $\lambda_1$ 
    is not simple.
\end{Exercise}

Below we demonstrate that the $m$-th power of $P_w$ provides $m$-step
transition probabilities for workers.  Anticipating this discussion, we now
seek an expression for $P^m_w$ at arbitrary $m \in \NN$.  This problem is
simplified if we use diagonalization.

\begin{Exercise}
    Assume that $\min\{\alpha, \beta\} > 0$.
    (When $\alpha=\beta=0$, computing the powers of $P_w$ is trivial.)
    Show that
    \begin{equation*}
        P_w = E D E^{-1}
        \quad \text{when} \quad 
        D = 
        \begin{pmatrix}
            1 & 0 \\
            0 & \lambda_2
        \end{pmatrix}
        \quad \text{and} \quad 
        E = 
        \begin{pmatrix}
            1 & -\alpha  \\
            1 & \beta
        \end{pmatrix}.
    \end{equation*}
    Using \eqref{eq:diagpow}, prove that
    \begin{equation}\label{eq:pwpk}
        P_w^m =
        \frac{1}{\alpha + \beta}
        \begin{pmatrix}
            \beta + \alpha(1-\alpha-\beta)^m & 
                \alpha(1 - (1-\alpha-\beta)^m) \\
            \beta(1-(1-\alpha-\beta)^m)  & 
                \alpha + \beta(1-\alpha-\beta)^m
        \end{pmatrix}
    \end{equation}
    for every $m \in \NN$.
\end{Exercise}

\subsubsection{Left Eigenvectors}\label{sss:lefteps}

A vector $\epsilon \in \CC^n$ is called a \navy{left eigenvector} of $A \in
\matset{n}{n}$ if $\epsilon$ is an eigenvector of $A^\top$.  In other words,
$\epsilon$ is nonzero and there exists a $\lambda \in \CC$ such that $A^\top
\epsilon = \lambda \epsilon$.  We can alternatively write the expression as
$\epsilon^\top A = \lambda \epsilon^\top$, which is where the name ``left''
eigenvector originates.

Left eigenvectors will play important roles in what follows, including that of
stochastic steady states for dynamic models under a Markov assumption.
To help distinguish between ordinary and left eigenvectors, we will at times
call (ordinary) eigenvectors of $A$ \navy{right eigenvectors} of $A$.  

If $A$ is diagonalizable, then so is $A^\top$.  To show this, let $A = P D
P^{-1}$ with $D = \diag(\lambda_i)$.  We know from earlier discussion that the
columns of $P$ are the (right) eigenvectors of $A$.

\begin{Exercise}\label{ex:lei1}
    Let $Q = (P^\top)^{-1}$.  Prove that $Q^\top P = I$ and $A^\top = Q D Q^{-1}$.
\end{Exercise}

The results of the last exercise show that, when $A = P D
P^{-1}$, the columns of $(P^\top)^{-1}$ coincide with the left eigenvectors of
$A$. (Why?)  Equivalently, $A = P D Q^\top$ where $Q = (\epsilon_1,
\ldots, \epsilon_n)$ is the $n\times n$ matrix with $i$-th column equal to the
$i$-th left eigenvector of $A$.

\begin{Exercise}\label{ex:lei2}
    Let $(e_i)_{i=1}^n$ be right eigenvectors of $A$ and let
    $(\epsilon_i)_{i=1}^n$ be the left eigenvectors.  Prove that
    \begin{equation}\label{eq:lei2}
        \inner{\epsilon_i, e_j} = \1\{i = j\}
        \qquad (i, j \in \natset{n}).
    \end{equation}
    (Hint: Use the results of Exercise~\ref{ex:lei1}.)
\end{Exercise}

\begin{Exercise}\label{ex:aspecrep}
    Continuing with the notation defined above and continuing to assume that
    $A$ is diagonalizable, prove that
    \begin{equation}\label{eq:aspecrep}
        A = \sum_{i=1}^n \lambda_i e_i \epsilon_i^\top
        \quad \text{and} \quad
        A^m = \sum_{i=1}^n \lambda_i^m e_i \epsilon_i^\top
    \end{equation}
    for all $m \in \NN$.  The expression for $A$ on the left hand side of
    \eqref{eq:aspecrep} is called the \navy{spectral representation} of $A$.
\end{Exercise}

\begin{Exercise}
    Prove that each $n \times n$ matrix $\lambda_i e_i \epsilon_i^\top$ in the sum
        $\sum_{i=1}^n \lambda_i e_i \epsilon_i^\top$ is rank 1.
\end{Exercise}

\subsubsection{Similar Matrices}\label{sss:sim}

Diagonalizability is a special case of a more general
concept: $A \in \matset{n}{n}$ is called \navy{similar}\index{Similar matrix} to 
$B \in \matset{n}{n}$ if there exists an invertible matrix $P$ such
that $A =  P B P^{-1}$.  In this terminology, $A$ is diagonalizable if and
only if it is similar to a diagonal matrix.

\begin{Exercise}
    Prove that similarity between matrices is an equivalence relation (see
    \S\ref{sss:eqclass}) on $\matset{n}{n}$.
\end{Exercise}

\begin{Exercise}
    The fact that similarity is an equivalence relation on $\matset{n}{n}$
    implies that this relation partitions $\matset{n}{n}$ into disjoint
    equivalence classes, elements of which are all similar.
    Prove that all matrices in each equivalence class share the same eigenvalues.
\end{Exercise}

\begin{Exercise}
   Prove: If $A$ is similar to $B$, then $A^m$ is similar to $B^m$.  In
   particular
    \begin{equation*}
        A =  P B P^{-1} \; \implies \; A^m =  P B^m P^{-1}
        \text{ for all } m \in \NN.
    \end{equation*}
\end{Exercise}

The last result is a generalization of~\eqref{eq:diagpow}.  When $A$ is large,
calculating the powers $A^k$ can be computationally expensive or infeasible.
If, however, $A$ is similar to some simpler matrix $B$, then we can take
powers of $B$ instead, and then transition back to $A$ using the similarity
relation.\footnote{The only concern with this shift process is that $P$ can be
ill-conditioned, implying that the inverse is numerically unstable.} 


\subsection{The Neumann Series Lemma}\label{ss:nsl}

Most high school students learn that, if $a$ is a number with $|a| < 1$, then
\begin{equation}\label{eq:sgeom}
    \sum_{i \geq 0} a^i = \frac{1}{1-a}. 
\end{equation}
This geometric series representation extends to matrices: If $A$ is a
matrix satisfying a certain condition, then \eqref{eq:sgeom} holds, in the
sense that $\sum_{i \geq 0} A^i = (I-A)^{-1}$.  (Here $I$ is the identity
matrix.)  But what is the ``certain condition'' that
we need to place on $A$, which generalizes the concept $|a|<1$ to matrices?
The answer to this question involves the ``spectral radius'' of a matrix,
which we now describe.

\subsubsection{Spectral Radii}

Fix $A \in \matset{n}{n}$.  With $|z|$ indicating the modulus of a complex
number $z$, the \navy{spectral radius}\index{Spectral radius} of $A$ is
defined as
\begin{equation}
    \label{eq:sr_fc}
    r(A) 
    := \max \setntn{|\lambda|}{ \lambda \text{ is an eigenvalue of } A}.
\end{equation}
Within economics, the spectral radius has important applications in
dynamics, asset pricing, and numerous other fields.  As we will see, the same
concept also plays a key role in network analysis.

\begin{remark}
    For any square matrix $A$, we have $r(A^{\top}) = r(A)$.
    This follows from the fact that $A$ and $A^\top$ always have the same
    eigenvalues.
\end{remark}

\begin{example}\label{eg:normdiag2}
    As usual, diagonal matrices supply the simplest example: If $D =
    \diag(d_i)$, then the spectrum $\sigma(D)$ is just $\{d_i\}_{i\in
    \natset{n}}$ and hence $r(D) = \max_i |d_i|$.
\end{example}

After executing
\begin{minted}{python}
import numpy as np
\end{minted}
The following Python code computes the spectral radius of a square matrix $M$:
\begin{@empty}\label{m:sr}  
\end{@empty}                
\begin{minted}{python}
def spec_rad(M):
    return np.max(np.abs(np.linalg.eigvals(M)))    
\end{minted}

\subsubsection{Geometric Series}

We can now return to the matrix extension of \eqref{eq:sgeom} and state a
formal result.

\index{Neumann series lemma}
\index{NSL}
\begin{theorem}[Neumann series lemma (NSL)]\label{t:nsl}
    If $A$ is in $\matset{n}{n}$ and $r(A) < 1$, then $I - A$ is nonsingular
    and
    \begin{equation}\label{eq:nsl}
         (I - A)^{-1} = \sum_{m=0}^{\infty} A^m.
    \end{equation}
\end{theorem}

The sum $\sum_{m=0}^{\infty} A^m$ is called the \navy{power series}
representation of $(I - A)^{-1}$.  Convergence of the matrix series is
understood as element-by-element convergence.  

A full proof of Theorem~\ref{t:nsl} can be found in \cite{cheney2013analysis}
and many other sources.  The core idea is simple: if $S = I + A + A^2 +
\cdots$ then $I + A S = S$.  Reorganizing gives $(I - A)S = I$, which is
equivalent to~\eqref{eq:nsl}.  The main technical issue is showing that the
power series converges.  The full proof shows that this always holds when
$r(A)<1$.

\begin{Exercise}
    Fix $A \in \matset{n}{n}$. Prove the following:  if $r(A) < 1$, then, for
    each $b \in \RR^n$ the linear system $x = A x + b$ has the unique solution
    $x^* \in \RR^n$ given by 
    \begin{equation}\label{eq:xstarnsl}
        x^* = \sum_{m=0}^{\infty} A^m  b .
    \end{equation}
\end{Exercise}

\begin{Answer}
    Fix $A \in \matset{n}{n}$ and $b \in \RR^n$ where $r(A) < 1$. We can write
    $x = A x + b$ as $(I-A)x = b$.  Since $r(A) < 1$, $I-A$ is invertible and
    hence the linear system $(I-A)x = b$ has unique solution $x^* = (I-A)^{-1}
    b$.  The expression $x^* = \sum_{m=0}^{\infty} A^m  b$ follows from the
    Neumann series lemma.
\end{Answer}

\subsection{The Perron--Frobenius Theorem}\label{ss:pft}

In this section we state and discuss a suprisingly far reaching theorem due to
Oskar Perron and Ferdinand Frobenius, which has applications in network
theory, machine learning, asset pricing, Markov dynamics, nonlinear dynamics,
input-output analysis and many other fields.  In essence, the theorem provides additional
information about eigenvalues and eigenvectors when the matrix in question is
positive in some sense.

\subsubsection{Order in Matrix Space}\label{ss:omat}

We require some definitions.  In what follows, for $A \in \matset{n}{k}$, we write
\begin{itemize}
    \item \navy{$A \geq 0$} if all elements of $A$ are nonnegative and
    \item \navy{$A \gg 0$} if all elements of $A$ are strictly positive. 
\end{itemize}

It's easy to imagine how nonnegativity and positivity are important notions
for matrices, just as they are for numbers.  However, strict positivity of
every element of a matrix is hard to satisfy, especially for a large
matrix.  As a result, mathematicians routinely use two notions of ``almost
everywhere strictly positive,'' which sometimes provide sufficient positivity
for the theorems that we need.

Regarding these two notions, for $A \in \matset{n}{n}$, we say that $A \geq 0$ is 
\begin{itemize}
    \item \navy{irreducible}\index{Irreducible} if $\sum_{m=0}^\infty A^m \gg 0$ and
    \item \navy{primitive}\index{Primitive matrix} if there exists an $m \in
        \NN$ such that $A^m \gg 0$.
\end{itemize}
Evidently, for $A \in \matset{n}{n}$ we have 
\begin{equation*}
    A \gg 0
    \; \implies \;
    A \text{ primitive }
    \; \implies \;
    A \text{ irreducible }
    \; \implies \;
    A \geq 0.
\end{equation*}
A nonnegative matrix is called \navy{reducible} if it fails to be irreducible.

\begin{Exercise}\label{ex:pwprop}
    By examining the expression for $P_w^m$ in~\eqref{eq:pwpk}, show that $P_w$ is
    \begin{enumerate}
        \item irreducible if and only if $0 < \alpha, \beta \leq 1$; and
        \item primitive if and only if $0 < \alpha, \beta \leq 1$ and
            $\min\{\alpha, \beta\} < 1$.
    \end{enumerate}
\end{Exercise}

\begin{Answer}
    ((i), ($\Rightarrow$)) Suppose that $P_w$ is irreducible and yet $\alpha=0$.
    Then, by the expression for $P^m_w$ in \eqref{eq:pwpk}, we have $P_w^m(1,2)=0$
    for all $m$.  This contradicts irreducibility, so $\alpha>0$ must hold.  A
    similar argument shows that $\beta > 0$.

    ((i), ($\Leftarrow$)) If $\alpha, \beta > 0$, then the diagonal elements are
    strictly positive whenever $m$ is even.  A small amount of algebra shows that
    the off-diagonal elements are strictly positive whenever $m$ is odd.  Hence
    $P_w + P_w^2 \gg 0$ and $P_w$ is irreducible.

    ((ii), ($\Rightarrow$)) Suppose that $P_w$ is primitive.  Then $P_w$ is
    irreducible, so it remains only to show that $\min\{\alpha, \beta\} < 1$.
    Suppose to the contrary that $\alpha=\beta=1$. Then 
    $P_w^m$ has zero diagonal elements when $m$ is even and zero off-diagonal
    elements when $m$ is odd.  This contradicts the primitive property, so 
    $\min\{\alpha, \beta\} < 1$ must hold.

    ((ii), ($\Leftarrow$)) Suppose that $0 < \alpha, \beta \leq 1$ and $\alpha<1$.
    Some algebra shows that $P^2_2 \gg 0$.  The same is true when $0 < \alpha,
    \beta \leq 1$ and $\beta<1$.  Hence $P_w$ is primitive.
\end{Answer}


%
In addition to the above notation, for $A, B \in \matset{n}{k}$, we also write
\begin{itemize}
    \item $A \geq B$ if $A - B \geq 0$ and $A \gg B$ if $A - B \gg 0$,
    \item $A \leq 0$ if $-A \geq 0$, etc.
\end{itemize}

\begin{Exercise}\label{ex:ppo}
    Show that $\leq$ is a partial order (see \S\ref{ss:posets}) on $\matset{n}{k}$.
\end{Exercise}

The partial order $\leq$ discussed in Exercise~\ref{ex:ppo} is usually called the
\navy{pointwise partial order}\index{Pointwise partial order} on 
$\matset{n}{k}$.  Analogous notation and terminology are used for vectors.  

The following exercise shows that nonnegative matrices are order-preserving
maps (see \S\ref{sss:mono}) on vector space under the pointwise partial
order---a fact we shall exploit many times.

\begin{Exercise}\label{ex:nnmatop}
    Show that the map $x \mapsto Ax$ is order-preserving (see
    \S\ref{sss:mono}) whenever $A \geq 0$ (i.e., $x \leq y$ implies $Ax \leq
    Ay$ for any conformable vectors $x, y$).
\end{Exercise}

\begin{Answer}
    Fix $A \in \matset{n}{k}$ with $A \geq 0$, along with $x, y \in \RR^k$.  From $x
    \leq y$ we have $y - x \geq 0$, so $A (y-x) \geq 0$.
    But then $Ay - Ax \geq 0$, or $Ax \leq Ay$.
\end{Answer}

\subsubsection{Statement of the Theorem}

Let $A$ be in $\matset{n}{n}$.  In general, $r(A)$ is not an eigenvalue of $A$. 
For example, 
\begin{equation*}
    A = \diag(-1, 0) 
    \; \implies \; \sigma(A) = \{-1, 0\} 
    \; \text{ while } r(A) = 1.     
\end{equation*}
But $r(A)$ is always an eigenvalue when $A \geq 0$.   This is just one
implication of the following famous theorem.

\begin{theorem}[Perron--Frobenius]\label{t:pf}
    If $A \geq 0$, then $r(A)$ is an eigenvalue of $A$ with
    nonnegative real right and left eigenvectors:
    \begin{equation}\label{eq:pfrl}
        \exists \,
        \text{ nonzero }
        e, \epsilon \in \RR^n_+
        \text{ such that }
        A e = r(A) e 
        \text{ and } 
        \epsilon^\top A = r(A) \epsilon^\top.
    \end{equation}
    If $A$ is irreducible, then, in addition,
    \begin{enumerate}
        \item $r(A)$ is strictly positive and a simple eigenvalue,
        \item the eigenvectors $e$ and $\epsilon$ are everywhere positive, and
        \item eigenvectors of $A$ associated with other eigenvalues fail to be nonnegative.
    \end{enumerate}
    If $A$ is primitive, then, in addition,
    \begin{enumerate}
        \item the inequality $|\lambda| \leq r(A)$ is strict for all eigenvalues $\lambda$ of $A$ distinct from $r(A)$, and
        \item with $e$ and $\epsilon$ normalized so that $\inner{\epsilon,
            e}=1$, we have
            \begin{equation}\label{eq:patocon}
                r(A)^{-m} A^m
                \to
                e \, \epsilon^\top 
                \qquad (m \to \infty).
            \end{equation}
    \end{enumerate}
\end{theorem}

The fact that $r(A)$ is simple under irreducibility means that 
its eigenvectors are unique up to scalar multiples.   We will exploit this
property in several important uniqueness proofs.

In the present context, $r(A)$ is called the \navy{dominant
eigenvalue} or \navy{Perron root}\index{Perron root} of $A$, while $\epsilon$
and $e$ are called the \navy{dominant left and right eigenvectors}\index{Dominant eigenvector}
of $A$, respectively.

Why do we use the word ``dominant'' here? To help illustrate, let us suppose
that $A \in \matset{n}{n}$ is primitive and fix any $x \in \RR^n$.
Consider what happens to the point $x_m := A^m x$ as $m$ grows.
By~\eqref{eq:patocon} we have $A^m x \approx r(A)^m c e$ for large $m$, where
$c = \epsilon^\top x$.   In other words, asymptotically, the
sequence $(A^m x)_{m \in \NN}$ is just scalar multiples of $e$, growing at
rate $\ln r(A)$.  Thus, $r(A)$ dominates other eigenvalues in controlling the
growth rate of $A^m x$, while $e$ dominates other eigenvectors in controlling
the direction of growth.

\begin{Exercise}
    The $n \times n$ matrix $P := e \, \epsilon^\top$ in \eqref{eq:patocon} is
    called the \navy{Perron projection}\index{Perron projection} of
    $A$.  Prove that $P^2 = P$ (a property that is often used to
    define projection matrices) and $\rank P = 1$.  Describe the one-dimensional
    space that $P$ projects all of $\RR^n$ into.
\end{Exercise}

\begin{example}
    Fix $A \geq 0$. If $r(A)=1$, then $I-A$ is not invertible.  To see this,
    observe that, by Theorem~\ref{t:pf}, since $r(A)$ is an eigenvalue of $A$,
    there exists a nonzero vector $e$ such that $(I-A)e=0$. The claim follows.
    (Why?)
\end{example}

\subsubsection{Worker Dynamics II}\label{sss:wdii}

We omit the full proof of Theorem~\ref{t:pf}, which is quite long and can be
found in \cite{meyer2000matrix}, \cite{seneta2006non} or
\cite{meyer2012banach}.\footnote{See also \cite{glynn2018probabilistic}, which
    provides a new proof of the main results, based on probabilistic
arguments, including extensions to infinite state spaces.}  Instead, to build
intuition, let us prove the theorem in a rather simple special case.  

The special case we will consider is the class of matrices
\begin{equation*}
    P_w = 
    \begin{pmatrix}
        1-\alpha & \alpha \\
        \beta & 1- \beta
    \end{pmatrix}
    \quad \text{with} \quad
    0 \leq \alpha, \beta \leq 1.
\end{equation*}
This example is drawn from the study of worker dynamics in \S\ref{sss:wdi}.

You might recall from $\S\ref{sss:wdi}$ that $\lambda_1 = 1$ and $\lambda_2 =
1-\alpha-\beta$.  Clearly $r(A)=1$, so $r(A)$ is an eigenvalue, as claimed by
the first part of the Perron--Frobenius theorem.

From now on we assume that $\min\{\alpha, \beta\} > 0$, which just means that
we are excluding the identity matrix in order to avoid some tedious qualifying
remarks. 

The two right eigenvectors ($e_1, e_2$) and two left eigenvectors
($\epsilon_1, \epsilon_2$) are, respectively,
\begin{equation*}
    e_1 := 
        \begin{pmatrix}
            1 \\
            1
        \end{pmatrix},
    \quad
    e_2 :=
        \begin{pmatrix}
            -\alpha \\
            \beta
        \end{pmatrix},
    \quad
    \epsilon_1 :=
        \frac{1}{\alpha + \beta}
        \begin{pmatrix}
            \alpha \\
            \beta 
        \end{pmatrix}
    \quad \text{and} \quad
    \epsilon_2 :=
        \begin{pmatrix}
            \alpha \\
            -\alpha 
        \end{pmatrix}.
\end{equation*}

\begin{Exercise}
    Verify these claims.  (The right eigenvectors were treated in~\S\ref{sss:wdi}.)
\end{Exercise}

\begin{Exercise}
    Recall from Exercise~\ref{ex:pwprop} that $P_w$ is irreducible if and only
    if both $\alpha$ and $\beta$ are strictly positive.  Show that all the claims about
    irreducible matrices in the Perron--Frobenius theorem are valid for $P_w$
    under this irreducibility condition.
\end{Exercise}

\begin{Exercise}
    Recall from Exercise~\ref{ex:pwprop} that $P_w$ is primitive if and only
    if $0 < \alpha, \beta \leq 1$ and $\min\{\alpha, \beta\} < 1$.  Verify the
    additional claim~\eqref{eq:patocon}
    for $P_w$ under these conditions.  In doing so, you can use the
    expression for $P_w^m$ in~\eqref{eq:pwpk}.
\end{Exercise}

\subsubsection{Bounding the Spectral Radius}\label{sss:someimp}

Using the Perron--Frobenius theorem, we can provide useful bounds on the
spectral radius of a nonnegative matrix.  In what follows, fix 
$A = (a_{ij}) \in \matset{n}{n}$ and set
\begin{itemize}
    \item $\rsum_i(A) := \sum_j a_{ij} = $ the $i$-th row sum of $A$ and
    \item $\csum_j(A) := \sum_i a_{ij} = $ the $j$-th column sum of $A$.
\end{itemize}

\vspace{0.2em}
\begin{lemma}\label{l:rscsbounds}
    If $A \geq 0$, then
    \begin{enumerate}
        \item $\min_i \rsum_i(A) \leq r(A) \leq \max_i \rsum_i(A)$ and 
        \item $ \min_j \csum_j(A) \leq r(A) \leq \max_j \csum_j(A)$.
    \end{enumerate}
\end{lemma}

\begin{proof}
    Let $A$ be as stated and let $e$ be the right eigenvector in
    \eqref{eq:pfrl}.  Since $e$ is nonnegative and nonzero, we can and do
    assume that $\sum_j e_j = 1$.  From $A e = r(A) e$ we have $\sum_j a_{ij}
    e_j = r(A) e_i$ for all $i$.  Summing with respect to $i$ gives $\sum_j
    \csum_j(A) e_j = r(A)$.  Since the elements of $e$ are nonnegative and sum
    to one, $r(A)$ is a weighted average of the column sums. Hence the second
    pair of bounds in Lemma~\ref{l:rscsbounds} holds.  The remaining proof is
    similar (use the left eigenvector).
\end{proof}

\section{Probability}\label{s:prob}

Next we review some elements of probability that will be required for
analysis of networks.

\subsection{Discrete Probability}\label{ss:spt}

We first introduce probability models on finite sets and then consider
sampling methods and stochastic matrices.

\subsubsection{Probability on Finite Sets}

Throughout this text, if $S$ is a finite set, then we set
\begin{equation*}
    \dD(S) := 
        \left\{
            \phi \in \RR^S_+ 
            \; : \:
            \sum_{x \in S} \phi(x) = 1
        \right\}
\end{equation*}
and call $\dD(S)$ the set of \navy{distributions}\index{Distributions} on $S$.
We say that a random variable $X$ taking values in $S$ has distribution $\phi
\in \dD(S)$ and write $X \eqdist \phi$ if
\begin{equation*}
    \PP \{X = x \} = \phi(x) \quad \text{for all } x \in S.
\end{equation*}

A distribution $\phi$ can also be understood as a vector $(\phi(x_i))_{i =1}^n
\in \RR^n$ (see Lemma~\ref{l:rxrn} in \S\ref{sss:functions}).    As a result,
$\dD(S)$ can be viewed as a subset of $\RR^n$.  Figure~\ref{f:simplex_1}
provides a visualization when $S = \{1, 2, 3\}$.  Each $\phi \in \dD(S)$ is
identified by the point $(\phi(1), \phi(2), \phi(3))$ in $\RR^3$.  

More generally, if $|S|=n$, then $\dD(S)$ can be identified with the
\navy{unit simplex}\index{Unit simplex} in $\RR^n$, which is the set of all
$n$-vectors that are nonnegative and sum to one.

\begin{figure}
   \centering
   \scalebox{0.45}{\includegraphics[trim = 0mm 15mm 0mm 0mm, clip]{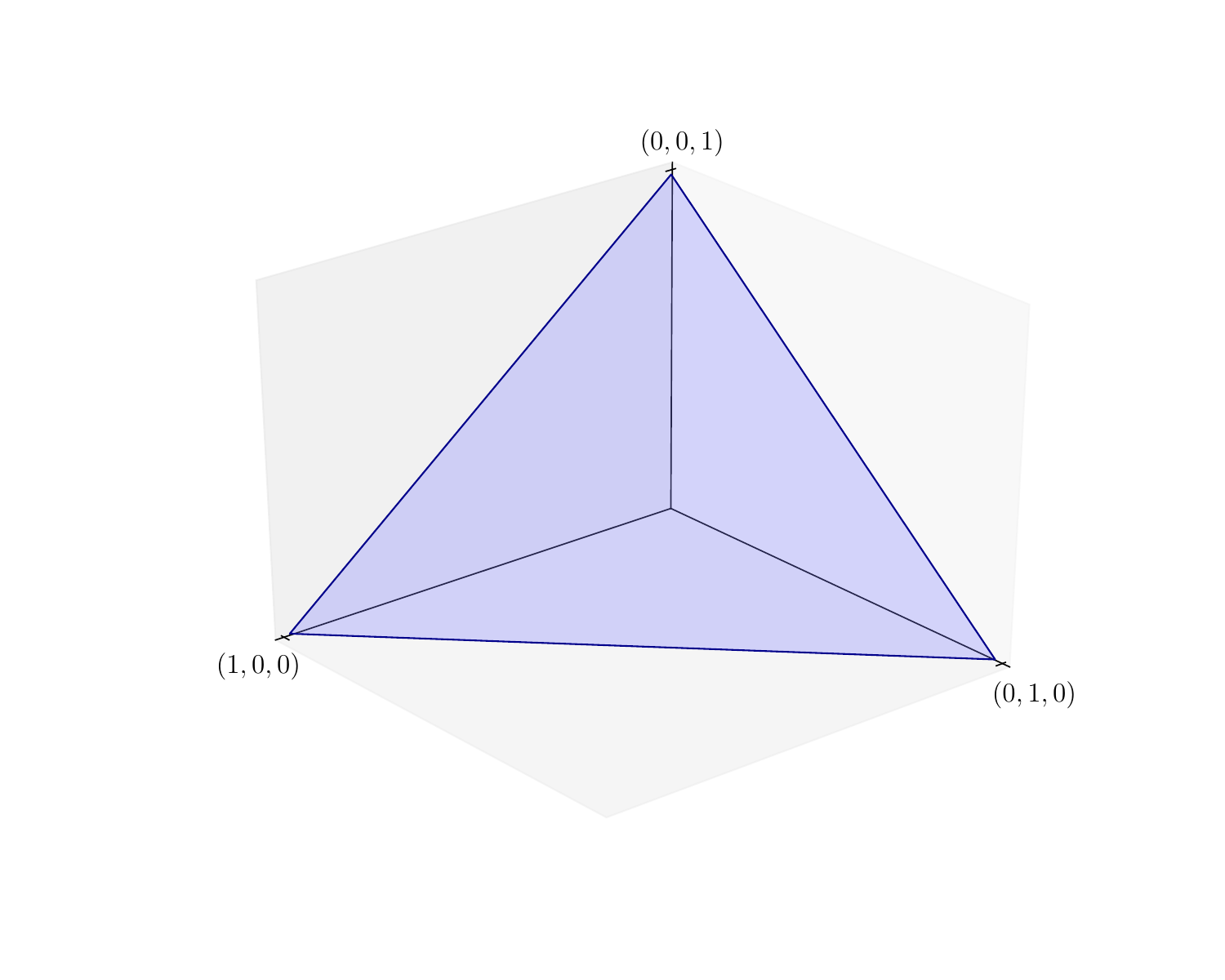}}
   \caption{\label{f:simplex_1} If $S = \{1, 2, 3\}$, then $\dD(S)$ is the unit simplex in $\RR^3$}
\end{figure}

Throughout, given $x \in S$, we use the symbol $\delta_x$ to represent the
element of $\dD(S)$ that puts all mass on $x$.  In other words,
$\delta_x(y) = \1\{y = x\}$ for all $y \in S$.  In
Figure~\ref{f:simplex_1}, each $\delta_x$ is a vertex of the unit simplex.

We frequently make use of the \navy{law of total probability}\index{Law of
    total probability}, which states that, for a random variable $X$ on $S$
    and arbitrary $A \subset S$,
\begin{equation}\label{eq:discreteltp}
    \PP\{X \in A\}
    = \sum_{i} \PP\{ X \in A \given X \in B_i\} \PP\{X \in B_i\}
\end{equation}
where $\{B_i\}$ is a partition of $S$ (i.e., finite collection of disjoint subsets of
$S$ such that their union equals $S$).

\begin{Exercise}
    Prove~\eqref{eq:discreteltp} assuming $\PP\{X \in B_i\} > 0$ for all $i$.
\end{Exercise}

\subsubsection{Inverse Transform Sampling}\label{sss:its}

Let $S$ be a finite set. Suppose we have the ability to generate random
variables that are uniformly distributed on $(0, 1]$.  We now want to generate
random draws from $S$ that are distributed according to arbitrary $\phi \in \dD(S)$.

Let $W$ be uniformly distributed on $(0,1]$, so that, for any $a \leq b \in
(0,1]$, we have $\PP\{a < W \leq b\} = b - a$, which is the length of the
interval $(a, b]$.\footnote{The probability is the same no matter whether inequalities
are weak or strict.}  Our problem will be solved if we can create a function
$z \mapsto \kappa(z )$ from $(0,1]$ to $S$ such that $\kappa(W)$ has
distribution $\phi$. One technique is as follows.  First we divide the unit
interval $(0,1]$ into disjoint subintervals, one for each $x \in S$.  The
interval corresponding to $x$ is denoted $I(x)$ and is chosen to have length
$\phi(x)$.  More specifically, when $S = \{x_1, \ldots, x_N\}$, we take
\begin{equation*}
    I(x_i ) 
    := (q_{i-1}, \, q_i] 
    \quad \text{where} \quad
    q_i := \phi(x_1) + \cdots + \phi(x_{i}) 
    \quad \text{and} \quad
    q_0 := 0.
\end{equation*}
You can easily confirm that the length of $I(x_i )$ is $\phi(x_i)$ for
all $i$.  

Now consider the function $z \mapsto \kappa(z)$ defined by
\begin{equation}
    \label{eq:fpdfu}
    \kappa (z) := \sum_{x \in S} x \1\{z \in I(x )\}
    \qquad (z \in (0,1])
\end{equation}
where $\1\{z \in I(x )\}$ is one when $z \in I(x )$ and zero
otherwise.  It turns out that $\kappa(W)$ has the distribution we desire.

\begin{Exercise}
    Prove: 
    \begin{enumerate}
        \item For all $x \in S$, we have $\kappa(z) = x$ if and only if $z \in I(x)$.
        \item The random variable $\kappa (W)$ has distribution $\phi$.
    \end{enumerate}
\end{Exercise}

\begin{Answer}
    For the first claim, fix $x \in S$ and $z \in (0, 1]$.  If $\kappa(z) =
    x$, then, since all elements of $S$ are distinct, the definition of
    $\kappa$ implies $z \in I(x)$.  Conversely, if $z \in I(x)$, then, since
    all intervals are disjoint, we have $\kappa(z) = x$.

    For the second claim, pick any $x \in S$, and observe that, by the first
    claim, the $\kappa (W) = x$ precisely when $W \in I(x)$. The probability
    of this event is the length of the interval $I(x)$, which, by
    construction, is $\phi(x)$. Hence $\PP\{\kappa (W) = x\} = \phi(x)$ for
    all $x \in S$ as claimed.  
\end{Answer}

\begin{Exercise}\label{ex:eber}
    Let $\phi$, $\kappa$ and $W$ be as defined above.  Prove that $\EE
    \1\{\kappa(W)=j\} = \phi(j)$ holds for all $j \in \natset{n}$.
\end{Exercise}
\begin{Answer}
    Fix $j \in \natset{n}$.   Observe that $Y := \1\{\kappa(W)=j\}$ is
    Bernoulli random variable.  The expectation of such a $Y$ equals $\PP\{Y =
    1\}$.  As $\kappa(W) \eqdist \phi$, this is $\phi(j)$.
\end{Answer}

\begin{Exercise}
    Using Julia or another language of your choice, implement the inverse
    transform sampling procedure described above when $S = \{1, 2, 3\}$ and
    $\phi = (0.2, 0.1, 0.7)$.  Generate $1,000,000$ (quasi) independent draws
    $(X_i)$ from $\phi$ and confirm that $(1/n) \sum_{i=1}^n \1\{X_i = j\}
    \approx \phi(j)$ for $j \in \{1, 2, 3\}$.
\end{Exercise}

The last exercise tells us that that the law of large numbers holds in this
setting, since, under this law, we expect that
\begin{equation*}
    \frac{1}{n} \sum_{i=1}^n \1\{X_i = j\} 
    \to \EE\1\{X_i = j\}
\end{equation*}
with probability one as $n \to \infty$.  In view of Exercise~\ref{ex:eber},
the right hand side equals $\phi(j)$.

\begin{Exercise}
    Suppose that, on a computer, you can generate only uniform random variables on
    $(0,1]$, and you wish to simulate a flip of a biased coin with heads
    probability $\delta \in (0,1)$.  Propose a method.
\end{Exercise}

\begin{Answer}
    Draw $U$ uniformly on $(0,1]$ and set the coin to heads if $U \leq \delta$.
    The probability of this outcome is $\PP\{U \leq \delta\}=\delta$.
\end{Answer}

\begin{Exercise}\label{ex:concond}
    Suppose that, on a computer, you are able to sample from distributions $\phi$
    and $\psi$ defined on some set $S$.  The set $S$ can be discrete or
    continuous and, in the latter case, the distributions are understood as
    densities. Propose a method to sample on a computer from the convex
    combination $f(s) = \delta \phi(s) + (1- \delta) \psi(s)$, where $\delta
    \in (0, 1)$.
\end{Exercise}

\begin{Answer}
    For the solution we assume that $S$ is finite, although the argument can
    easily be extended to densities. On the computer, we flip a biased coin $B \in
    \{0, 1\}$ with $\PP\{B = 0\} = \delta$ and then
    \begin{enumerate}
        \item draw $Y$ from $\phi$ if $B = 0$, or
        \item draw $Y$ from $\psi$ if $B = 1$.
    \end{enumerate}
    With this set up, by the law of total probability,
    \begin{equation*}
        \PP\{Y = s\}
        = \PP\{Y =s \given B = 0\}\PP\{B=0\}
            + \PP\{Y =s \given B = 1\}\PP\{B=1\}
            = \delta \phi(s) + (1-\delta) \psi(s).
    \end{equation*}
    In other words, $Y \eqdist f$.
\end{Answer}

\subsubsection{Stochastic Matrices}\label{sss:stochmat}

A matrix $P = (p_{ij}) \in \matset{n}{n}$ is called a \navy{stochastic matrix}
\index{Stochastic matrix} if 
\begin{equation*}
    P \geq 0 
    \quad \text{and} \quad
    P \1 = \1,
    \text{ where } \1 \in \RR^n \text{ is a column vector of ones}.
\end{equation*}
In other words, $P$ is nonnegative and has unit row sums.

We will see many applications of stochastic matrices in this text.  Often the
applications are probabilistic, where each row of $P$ is interpreted as a
distribution over a finite set.

\begin{Exercise}\label{ex:sm_sr1}
    Let $P, Q$ be $n \times n$ stochastic matrices. Prove the following facts.
    \begin{enumerate}
        \item $P Q$ is also stochastic.
        \item $r(P)=1$.
        \item There exists a row vector $\psi \in \RR^n_+$ such that $\psi \1 =
            1$ and $\psi P = \psi$.
    \end{enumerate}
\end{Exercise}

\begin{Answer}
    Let $P$ and $Q$ be as stated. Evidently $PQ \geq 0$.  Moreover, $PQ \1 = P
    \1 = \1$, so $PQ$ is stochastic. That $r(P)=1$ follows directly from
    Lemma~\ref{l:rscsbounds}. By the Perron--Frobenius theorem, there exists a
    nonzero, nonnegative row vector $\phi$ satisfying $\phi P = \phi$.  Rescaling $\phi$ to
    $\phi / (\phi \1)$ gives the desired vector $\psi$.
\end{Answer}

The vector $\psi$ in part (iii) of Exercise~\ref{ex:sm_sr1} is called the
\emph{PageRank vector} by some authors, due to its prominence in Google's
PageRank algorithm.  We will call it a \navy{stationary
    distribution}\index{Stationary distribution} instead.\footnote{Stationary
    distributions of stochastic matrices were intensively studied by many
    mathematicians well over a century before Larry Page and Sergey Brin
    patented the PageRank algorithm, so it seems unfair to allow them to
    appropriate the name.}  Stationary distributions play a key role in the
    theory of Markov chains, to be treated in \S\ref{s:amcs}.
Ranking methods are discussed again in \S\ref{ss:netcen}.  PageRank is treated
in more detail in \S\ref{sss:pagerank}.

\subsection{Power Laws}\label{ss:pow}

Next we discuss distributions on the (non-discrete) sets $\RR$ and $\RR_+$. We
are particularly interested in a certain class of distributions that are
apparently non-standard and yet appear with surprising regularity in
economics, social science, and the study of networks.  We refer to the
distributions that are said to obey a ``power law.''

In what follows, given a real-valued random variable $X$, the function 
\begin{equation*}
    F(t) := \PP\{X \leq t\}
    \qquad (t \in \RR)
\end{equation*}
is called the \navy{cumulative distribution function ({\sc cdf})}
\index{Cumulative distribution function}\index{cdf} of $X$. The \navy{counter
{\sc cdf} ({\sc ccdf})}\index{Counter CDF (CCDF)}  of $X$ is the function $G(t)
:= \PP\{X > t\} = 1 - F(t)$. 

A useful property that holds for any nonnegative random variable $X$
and $p \in \RR_+$ is the identity 
\begin{equation}
    \label{eq:uexpid}
    \EE \, X^p = \int_0^\infty p t^{p-1} \PP\{ X > t \} \diff t.
\end{equation}
See, for example, \cite{cinlar2011probability}, p.~63.

\subsubsection{Heavy Tails}

Recall that a random variable $X$ on $\RR$ is said to be \navy{normally
distributed}\index{Normal distribution} with mean $\mu$ and variance
$\sigma^2$, and we write $X \eqdist N(\mu, \sigma^2)$, if $X$ has density
\begin{equation*}
    \phi(t) := 
        \sqrt{\frac{1}{2 \pi \sigma^2}} 
        \exp \left( \frac{-(t-\mu)^2}{2\sigma^2} \right)
        \qquad (t \in \RR).
\end{equation*}
One notable feature of the normal density is that the tails of the density
approach zero quickly.  For example, $\phi(t)$ goes to zero like $\exp(-t^2)$
as $t \to \infty$, which is extremely fast.  

A random variable $X$ on $\RR_+$ is called \navy{exponentially
distributed} and we write \navy{$X \eqdist \Exp(\lambda)$} if, for some
$\lambda > 0$, $X$ has density 
\begin{equation*}
    p(t) = \lambda \me^{-\lambda t}  
    \qquad (t \geq 0).
\end{equation*}
The tails of the exponential density go to zero like $\exp(-t)$ as $t \to
\infty$, which is also relatively fast.  

When a distribution is relatively light-tailed, in the sense that its tails go
to zero quickly, draws rarely deviate more than a few standard deviations from
the mean.  In the case of a normal random variable, the probability of
observing a draw more than 3 standard deviations above the mean is around
0.0014.  For 6 standard deviations, the probability falls to $10^{-11}$. 

In contrast, for some distributions, ``extreme'' outcomes occur relatively
frequently.  The left panel of Figure~\ref{f:heavy_tailed_draws} helps to illustrate by
simulating 1,000 independent draws from Student's
t-distribution\index{t-distribution}, with $1.5$ degrees of freedom.  For
comparison, the right subfigure shows an equal number of independent draws
from the $N(0, 4)$ distribution.  The Student's t draws reveal tight
clustering around zero combined with a few large deviations.

\begin{figure}
    \centering
    \scalebox{0.68}{\includegraphics{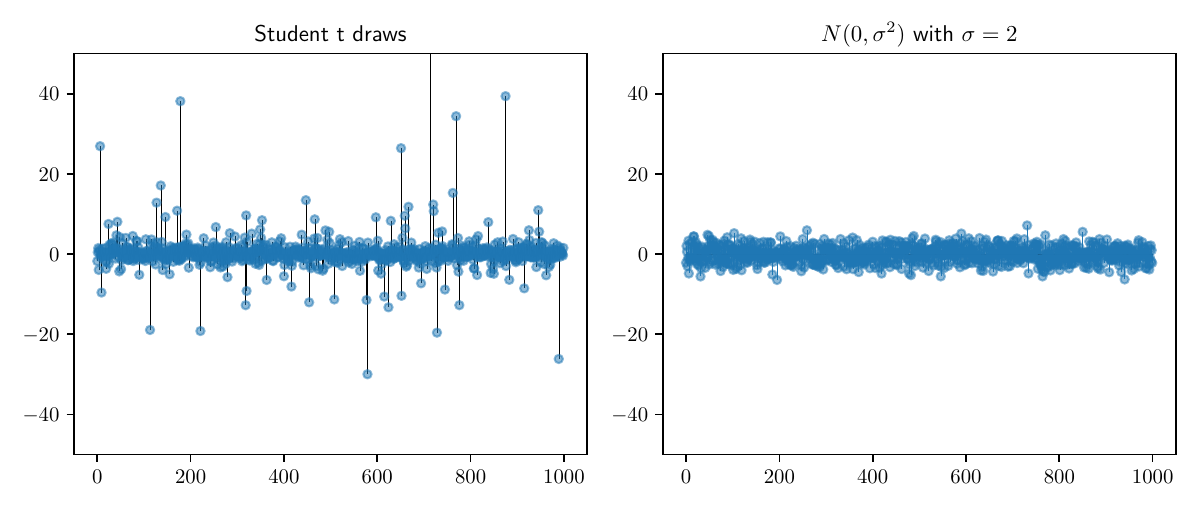}}
    \caption{\label{f:heavy_tailed_draws} Independent draws from Student's t- and normal distributions}
\end{figure}

Formally, a random variable $X$ on $\RR$ is called
\navy{light-tailed}\index{Light-tailed} if its \navy{moment generating
function}\index{Moment generating function} 
\begin{equation}
    \label{eq:defht}
    m(t) := \EE \me^{tX} 
    \qquad (t \geq 0)
\end{equation}
is finite for at least one $t > 0$.  Otherwise $X$ is called
\navy{heavy-tailed}\index{Heavy-tailed}.\footnote{Terminology on heavy tails
    varies across the literature but our choice is increasingly standard.
    See, for example, \cite{foss2011introduction} or
\cite{nair2021fundamentals}.}   

\begin{example}
    If $X \eqdist N(\mu, \sigma^2)$, the the moment generating function of $X$
    is known to be 
    \begin{equation*}
        m(t) = \exp \left( \mu t + \frac{t^2 \sigma^2}{2} \right)
        \qquad (t \geq 0).
    \end{equation*}
    Hence $X$ is light tailed.  
\end{example}

\begin{example}\label{eg:lognorm}
    A random variable $X$ on $(0, \infty)$ is said to have the \navy{lognormal
    density}\index{Lognormal distribution} and we write \navy{$X \eqdist
    LN(\mu, \sigma^2)$} if $\ln X \eqdist N(\mu, \sigma^2)$. The mean and variance of
    this distribution are, respectively,
    \begin{equation*}
        \EE \, X = \exp( \mu + \sigma^2/2)
        \quad \text{and} \quad
        \var \, X = (\exp(\sigma^2) - 1) \exp( 2 \mu + \sigma^2).
    \end{equation*}
    The moment generating function $m(t)$ is known to be infinite for all $t >
    0$, so any lognormally distributed random variable is heavy-tailed.
\end{example}

For any random variable $X$ and any $r \geq 0$, the (possibly infinite)
expectation $\EE |X|^r$ called the $r$-th \navy{moment}\index{Moment} of $X$.

\begin{lemma}
    \label{l:ltimp}
    Let $X$ be a random variable on $\RR_+$. If $X$ is light-tailed, then all
    of its moments are finite.
\end{lemma}

\begin{proof}
    Pick any $r > 0$.  We will show that $\EE X^r < \infty$.  Since $X$ is
    light-tailed, there exists a $t > 0$ such that $m(t) = \EE \exp(tX) <
    \infty$.  For a sufficiently large constant $\bar x$ we have $\exp(tx)
    \geq x^r$ whenever $x \geq \bar x$.  As a consequence, with $F$ as the
    distribution of $X$, we have
    \begin{equation*}
        \EE X^r
        = \int_0^{\bar x} x^r F(\diff x)
            + \int_{\bar x}^\infty x^r F(\diff x)
        \leq \bar x^r + m(t) < \infty.
        \qedhere
    \end{equation*}
\end{proof}

\begin{Exercise}\label{ex:lnfmo}
    Prove that the lognormal distribution has finite moments of every order.
\end{Exercise}

\begin{Answer}
    Fix $p > 0$ and let $X$ be $LN(\mu, \sigma^2)$.  We have
    \begin{equation*}
        \EE | X |^p = \EE X^p = \EE \exp(p \mu + p \sigma Z)
        \quad \text{ for } Z \eqdist N(0,1).
    \end{equation*}
    Since $p \mu + p \sigma Z \eqdist N(p\mu, p^2\sigma^2)$, we can apply the formula
    for the mean of a lognormal distribution to obtain $m_p = \exp(p \mu +
    (p\sigma)^2 / 2) < \infty$.
\end{Answer}

Together with Lemma~\ref{l:ltimp}, Exercise~\ref{ex:lnfmo} shows that
existence of an infinite moment is a sufficient but not necessary condition
for heavy tails.

\subsubsection{Pareto Tails}

Given $\alpha > 0$, a nonnegative random variable $X$ is said to have a
\navy{Pareto tail}\index{Pareto tail} with \navy{tail index}\index{Tail index}
$\alpha$ if there exists a $c > 0$ such that
\begin{equation}
    \label{eq:plrt}
    \lim_{t \to \infty} t^\alpha \, \PP\{X > t\} = c.
\end{equation}
In other words, the {\sc ccdf} $G$ of $X$ satisfies 
\begin{equation}
    \label{eq:plrt0}
    G(t) \approx c t^{-\alpha}  \text{ for large } t.
\end{equation}
If $X$ has a Pareto tail for some $\alpha > 0$, then $X$ is also said to obey
a \navy{power law}\index{Power law}.

\begin{example}\label{eg:ptp}
    A random variable $X$ on $\RR_+$ is said to have a \navy{Pareto
    distribution}\index{Pareto distribution} with parameters $\bar x, \alpha >
    0$ if its {\sc ccdf} obeys
    \begin{equation}
        \label{eq:pareto}
        G(t) =
        \begin{cases}
            1 & \text{ if } t < \bar x
            \\
            \left( \bar x/t \right)^{\alpha} & \text{ if } t \geq \bar x
        \end{cases}
    \end{equation}
    It should be clear that such an $X$ has a Pareto tail with tail
    index $\alpha$. 
\end{example}

Regarding Example~\ref{eg:ptp}, note that the converse is not true: 
Pareto-tailed random variables are not necessarily Pareto distributed, since
the Pareto tail property only restricts the far right hand tail.

\begin{Exercise}\label{ex:ptail}
    Show that, if $X$ has a Pareto tail with tail index $\alpha$, then $\EE[X^r] =
    \infty$ for all $r \geq \alpha$.  [Hint: Use~\eqref{eq:uexpid}.]
\end{Exercise}

\begin{Answer}
    Let $X$ have a Pareto tail with tail index $\alpha$ and let $G$ be its
    {\sc ccdf}.  Fix $r \geq \alpha$.   Under the Pareto tail assumption, we can take 
    positive constants $b$ and $\bar x$ such that $G(t) \geq b t^{- \alpha}$
    whenever $t \geq \bar x$.  Using \eqref{eq:uexpid} we have
    \begin{equation*}
        \EE X^r = r \int_0^\infty t^{r-1} G(t) \diff t
        \geq 
        r \int_0^{\bar x} t^{r-1} G( t ) \diff t
        + r \int_{\bar x}^\infty  t^{r-1} b t^{-\alpha} \diff t.
    \end{equation*}
    But $\int_{\bar x}^\infty t^{r-\alpha-1} \diff t = \infty$ whenever $r -
    \alpha - 1 \geq -1$.  Since $r \geq \alpha$, we have $\EE X^r = \infty$.
\end{Answer}

From Exercise~\ref{ex:ptail} and Lemma~\ref{l:ltimp}, we see that every
Pareto-tailed random variable is heavy-tailed. The converse is not true, since
the Pareto tail property~\eqref{eq:plrt} is very specific.  Despite this, it
turns out that many heavy-tailed distributions encountered in the study of
networks are, in fact, Pareto-tailed.

\begin{Exercise}
    Prove: If $X \eqdist \Exp(\lambda)$ for some $\lambda > 0$, then 
    $X$ does not obey a power law.
\end{Exercise}

\begin{Answer}
    Fix $\lambda > 0$ and suppose $X \eqdist \Exp(\lambda)$. A simple integration
    exercise shows that $\PP\{X > t\} = \me^{-\lambda t}$.  Now fix $\alpha > 0$.
    Since $\lim_{t \to \infty} t^\alpha  \me^{-\lambda t} = 0$, the random variable $X$
    does not obey a power law.
\end{Answer}

\subsubsection{Empirical Power Law Plots}

When the Pareto tail property holds, the {\sc ccdf} satisfies $\ln G(t)
\approx \ln c - \alpha \ln t$ for large $t$.  In other words, $G$ is
eventually log linear.  Figure~\ref{f:ccdf_comparison_1} illustrates this
using a Pareto distribution.  For comparison, the {\sc ccdf} of an exponential
distribution is also shown.

\begin{figure}
    \centering
    \scalebox{0.65}{\includegraphics{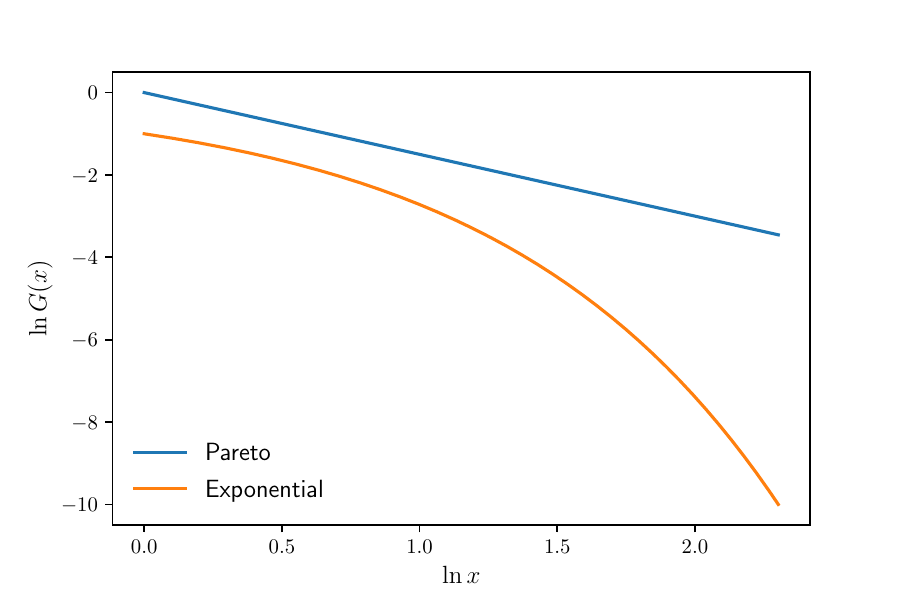}}
    \caption{\label{f:ccdf_comparison_1} {\sc ccdf} plots for the Pareto and exponential distributions}
\end{figure}

If we replace the {\sc ccdf} $G$ with its empirical counterpart---which
returns, for each $x$, the fraction of the sample with values greater than
$x$---we should also obtain an approximation to a straight line under the
Pareto tail assumption.

For example, consider the cross-sectional distribution 
of firm sizes.  While the precise nature of this distribution
depends on the measure of firm size, the sample of firms and other factors,
the typical picture is one of extreme heavy tails.  As an illustration,
Figure~\ref{f:empirical_powerlaw_firms_forbes} shows an empirical {\sc ccdf}
log-log plot for market values of the largest 500 firms in the Forbes Global
2000 list, as of March 2021.  The slope estimate and data distribution are
 consistent with a Pareto tail and infinite population variance.

\begin{figure}
    \centering
    \scalebox{0.74}{\includegraphics{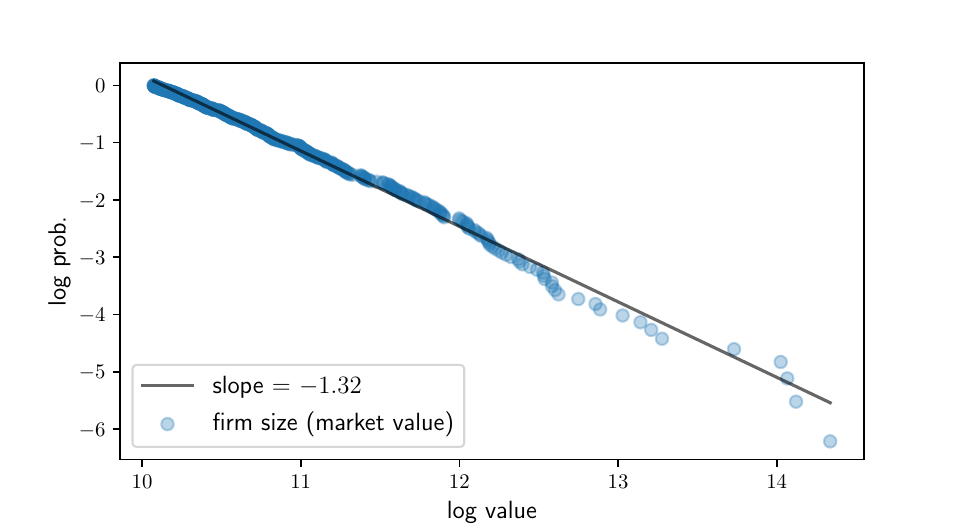}}
    \caption{\label{f:empirical_powerlaw_firms_forbes} Empirical {\sc ccdf} plots for largest firms (Forbes)}
\end{figure}

\subsubsection{Discrete Power Laws}\label{sss:dpls}

Let $X$ be a random variable with the Pareto distribution, as described in
Example~\ref{eg:ptp}.  The density of this random variable on the set $[\bar
x, \infty)$ is $p(t) = c t^{-\gamma}$ with $c := \alpha \bar x^\alpha$ and $\gamma
:= \alpha + 1$.  The next exercise extends this idea.

\begin{Exercise}\label{ex:dtpo}
    Let $X$ be a random variable with density $p$ on $\RR_+$.  Suppose that,
    for some constants $c > 0$, $\gamma > 1$ and $\bar x \in \RR_+$, we have
    \begin{equation}\label{eq:thiss}
        p(t) = c t^{-\gamma}
        \quad \text{whenever} \quad
        t \geq \bar x.
    \end{equation}
    Prove that $X$ is Pareto-tailed with tail index $\alpha := \gamma - 1$.
\end{Exercise}

\begin{Answer}
    Let $p$ and the constants $\gamma, c > 0$ and $\bar x \in \RR_+$ be as
    described in the exercise. Pick any $t \geq \bar x$.  By the usual rules of
    integration,
    \begin{equation*}
        \PP\{X > t\}
        = c \int_t^\infty u^{-\gamma} \diff u
        = - \frac{c}{1-\gamma} t^{1-\gamma}.
    \end{equation*}
    With $\alpha := \gamma - 1$, we then have $t^\alpha \PP\{X > t\} =
    c/\alpha$, and $X$ is Pareto-tailed with tail index $\alpha$.
\end{Answer}

The discrete analog of \eqref{eq:thiss} is a distribution on the positive
integers with
\begin{equation}\label{eq:zeta}
    f(k) = c k^{-\gamma}
\end{equation}
for large $k$. In the special case where this equality holds for all $k \in
\NN$, and $c$ is chosen so that $\sum_{k \in \NN} f(k) = 1$, we obtain the
\navy{zeta distribution}\index{Zeta distribution}.\footnote{Obviously the
    correct value of $c$ depends on $\gamma$, so we can write $c = H(\gamma)$ for
    some suitable function $H$.  The correct function for this normalization
is called the Riemann zeta function.}

In general, when we see a probability mass function with the specification
\eqref{eq:zeta} for large $k$, we can identify this with a Pareto tail, with
tail index $\alpha = \gamma - 1$.  Figure~\ref{f:zeta_1} illustrates with
$\gamma=2$.

\begin{figure}
    \centering
    \scalebox{0.7}{\includegraphics{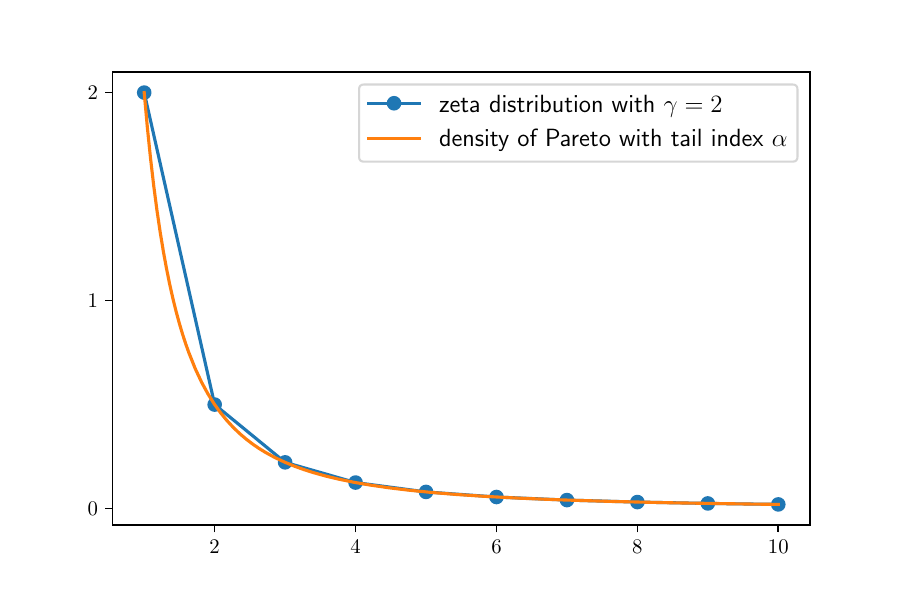}}
    \caption{\label{f:zeta_1} Zeta and Pareto distributions}
\end{figure}

\section{Graph Theory}\label{c:networks}\label{s:gandnet}

Graph theory is a major branch of discrete mathematics. It plays an essential
role in this text because it forms the foundations of network analysis.  This
section provides a concise and fast-moving introduction to graph theory
suitable for our purposes.\footnote{Graph theory is often regarded as originating from
work by the brilliant Swiss mathematician Leonhard Euler (1707--1783),
including his famous paper on the ``Seven Bridges of K\"onigsberg.''}

Graph theory has another closely related use: many economic models are
stochastic and dynamic, which means that they specify states of the world and
rates of transition between them.   One of the most natural ways to
conceptualize these notions is to view states as vertices in a graph and
transition rates as relationships between them.

We begin with definitions and fundamental concepts. We focus
on directed graphs, where there is a natural asymmetry in relationships (bank
$A$ lends money to bank $B$, firm $A$ supplies goods to firm $B$, etc.).  This
costs no loss of generality, since undirected graphs (where relationships are
symmetric two-way connections) can be recovered by insisting on symmetry
(i.e., existence of a connection from $A$ to $B$ implies existence of a
connection from $B$ to $A$).

\subsection{Unweighted Directed Graphs}\label{ss:uwdg}

We begin with unweighted directed graphs and examine standard properties, such
as connectedness and aperiodicity.

\subsubsection{Definition and Examples}

A \navy{directed graph}\index{Directed graph} or \navy{digraph}\index{Digraph}
is a pair $\gG = (V, E)$, where
\begin{itemize}
    \item $V$ is a finite nonempty set called the
        \navy{vertices}\index{Vertices} or \navy{nodes}\index{Nodes} of $\gG$, and 
    \item $E$ is a collection of ordered pairs $(u, v)$ in $V \times V$ called
        \navy{edges}\index{Edges}.
\end{itemize}
Intuitively and visually, $(u,v)$ is understood as an arrow from vertex $u$ to
vertex $v$.  

Two graphs are given in
Figures~\ref{f:rich_poor_no_label}--\ref{f:poverty_trap}.  Each graph has
three vertices.  In these cases, the arrows (edges) could be thought of as
representing positive possibility of transition over a given unit of time.

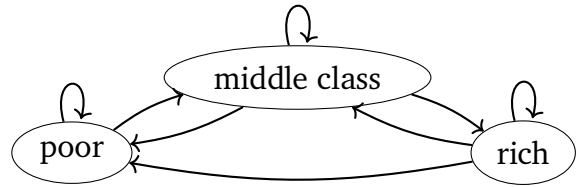
\begin{figure}
   \begin{center}
       \scalebox{1.0}{\input{tikz/rich_poor_no_label.tex}}
   \end{center}
   \caption{\label{f:rich_poor_no_label} A digraph of classes}
\end{figure}

\begin{figure}
   \centering
   \scalebox{1.0}{\input{tikz/poverty_trap.tex}}
   \caption{\label{f:poverty_trap} An alternative edge list}
\end{figure}
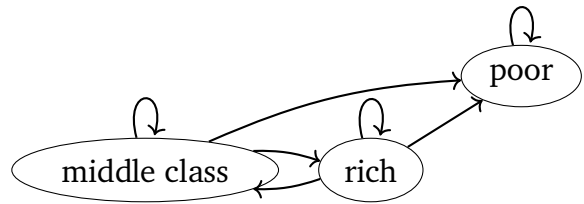

For a given edge $(u, v)$, the vertex $u$ is called the \navy{tail} of the
edge, while $v$ is called the \navy{head}.  Also,
$u$ is called a \navy{direct predecessor} of $v$ and $v$ is called a \navy{direct
successor} of $u$.  For $v \in V$, we use the following notation:
\begin{itemize}
    \item \navy{$\iI(v) :=$}  the set of all direct predecessors of $v$
    \item \navy{$\oO(v) :=$}  the set of all direct successors of $v$
\end{itemize}
Also,
\begin{itemize}
    \item the \navy{in-degree $i_d(v)$} $:= |\iI(v)| =$ the number of direct
        predecessors of $v$, and
    \item the \navy{out-degree $o_d(v)$} $:= |\oO(v)| =$ the number of
        direct successors of $v$.
\end{itemize}

If $i_d(v)=0$ and $o_d(v) > 0$, then $v$ is called a \navy{source}.  If either
$\oO(v)=\emptyset$ or $\oO(v)=\{v\}$, then 
$v$ is called a \navy{sink}.  For example, in Figure~\ref{f:poverty_trap},
``poor'' is a sink with an in-degree of 3.

\subsubsection{Digraphs in Networkx}\label{sss:nx}

Both Python and Julia provide valuable interfaces to numerical computing with
graphs.  Of these libraries, the Python package Networkx is probably the most
mature and fully developed.  It provides a convenient data structure for
representing digraphs and implements many common routines for analyzing them.
To import it into Python we run

\begin{minted}{python}
import networkx as nx
\end{minted}

In all of the code snippets shown below, we assume readers have executed this
import statement, as well as the additional two imports:

\begin{minted}{python}
import numpy as np
import matplotlib.pyplot as plt
\end{minted}

As an example, let us create the digraph in Figure~\ref{f:poverty_trap}, which
we denote henceforth by $\gG_p$.  To do so, we first create an empty
\texttt{DiGraph} object:

\begin{minted}{python}
G_p = nx.DiGraph()    
\end{minted}

Next we populate it with nodes and edges.  To do this we write down a list of
all edges, with \texttt{poor} represented by \texttt{p} and so on:

\begin{minted}{python}
edge_list = [
    ('p', 'p'),
    ('m', 'p'), ('m', 'm'), ('m', 'r'),
    ('r', 'p'), ('r', 'm'), ('r', 'r')
]
\end{minted}

Finally, we add the edges to our \texttt{DiGraph} object:

\begin{minted}{python}
for e in edge_list:
    u, v = e
    G_p.add_edge(u, v)
\end{minted}

Adding the edges automatically adds the nodes, so \texttt{G\_p} is now a
correct representation of $\gG_p$.  For our small digraph we can verify this
by plotting the graph via Networkx with the following code: 

\begin{minted}{python}
fig, ax = plt.subplots()
nx.draw_spring(G_p, ax=ax, node_size=500, with_labels=True, 
                 font_weight='bold', arrows=True, alpha=0.8,
                 connectionstyle='arc3,rad=0.25', arrowsize=20)
plt.show()
\end{minted}

This code produces Figure~\ref{f:networkx_basics_1}, which matches the
original digraph in Figure~\ref{f:poverty_trap}.

\begin{figure}
   \centering
   \scalebox{0.6}{\includegraphics[trim = 10mm 10mm 0mm 10mm, clip]{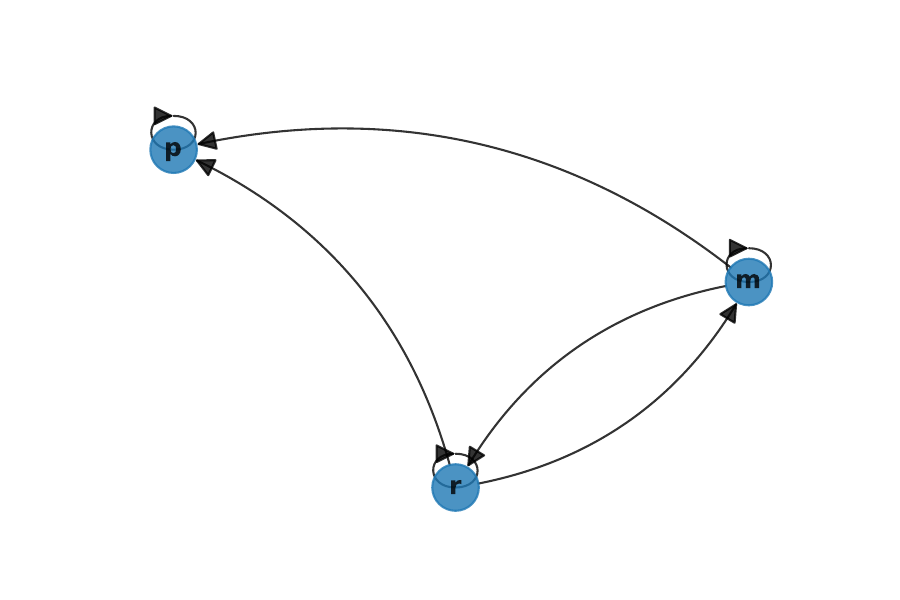}}
   \caption{\label{f:networkx_basics_1} Networkx digraph plot}
\end{figure}

\texttt{DiGraph} objects have methods that calculate in-degree and out-degree
of vertices.   For example,
\begin{minted}{python}
G_p.in_degree('p')
\end{minted}
prints 3.

\subsubsection{Communication}\label{sss:coms}

Next we study communication and connectedness, which have important
implications for production, financial, transportation and other networks, as
well as for dynamic properties of Markov chains.  

A \navy{directed walk}\index{Directed walk} from vertex $u$ to vertex $v$  of
a digraph $\gG$ is a finite sequence of vertices, starting with $u$ and ending
with $v$, such that any consecutive pair in the sequence is an edge of $\gG$.
A \navy{directed path}\index{Directed path} from $u$ to $v$ is a directed walk
from $u$ to $v$ such that all vertices in the path are distinct.  
For example, in Figure~\ref{f:strong_connected_components}, $(3, 2, 3, 2, 1)$
is a directed walk from $3$ to $1$, while $(3, 2, 1)$ is a directed path from
$3$ to $1$.

As is standard, the \navy{length} of a directed walk (or path) counts the
number of edges rather than vertices.   For example, the directed path $(3, 2,
1)$ from $3$ to $1$ in Figure~\ref{f:strong_connected_components} is said to have length 2.

Vertex $v$ is called \navy{accessible}\index{Accessible} (or \navy{reachable}) from vertex $u$, and
we write \navy{$u \to v$}, if either $u=v$ or there exists a directed path
from $u$ to $v$.  A set $U \subset V$ is called
\navy{absorbing}\index{Absorbing} for the directed graph $(V, E)$ if no
element of $V \setminus U$ is accessible from $U$.

\begin{example}
    Let $\gG = (V, E)$ be a digraph representing a production network, where
    elements of $V$ are sectors and $(i, j) \in E$ means that $i$ supplies to
    $j$.  Then sector $m$ is an upstream supplier of sector $\ell$ whenever $m
    \to \ell$.
\end{example}

\begin{figure}
   \begin{center}
       \input{tikz/strong_connected_components.tex}
   \end{center}
   \caption{\label{f:strong_connected_components} Strongly connected
   components of a digraph (rectangles)}
\end{figure}
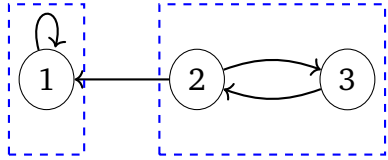

\begin{example}
    The vertex $\{ \text{poor} \}$ in the Markov digraph displayed in
    Figure~\ref{f:poverty_trap} is absorbing, since  $\{ \text{middle, rich}
    \}$ is not accessible from  $\{ \text{poor} \}$.
\end{example}

Two vertices $u$ and $v$ are said to \navy{communicate}\index{Communication
(graphs)} if $u \to v$ and $v \to u$.

\begin{Exercise}\label{ex:sceqrel}
    Let $(V, E)$ be a directed graph and write $u \sim v$ if $u$ and $v$
    communicate.  Show that $\sim$ is an equivalence relation (see  \S\ref{sss:eqclass}).
\end{Exercise}

Since communication is an equivalence relation, it induces a partition of $V$
into a finite collection of equivalence classes.  Within each of these
classes, all elements communicate.  These classes are called \navy{strongly
connected components}.  The graph itself is called \navy{strongly
connected}\index{Strongly connected} if there is only one such component; that
is, $v$ is accessible from $u$ for any pair $(u, v) \in V \times V$.  This
corresponds to the idea that any node can be reached from any other.  

\begin{example}
    Figure~\ref{f:strong_connected_components} shows a digraph with 
    strongly connected components $\{1\}$ and $\{2, 3\}$.
    The digraph is not strongly connected.
\end{example}

\begin{example}
    In Figure~\ref{f:rich_poor_no_label}, the digraph is strongly connected.  In
    contrast, in Figure~\ref{f:poverty_trap}, rich is not accessible from
    poor, so the graph is not strongly connected.  
\end{example}

Networkx can be used to test for communication and strong connectedness, as
well as to compute strongly connected components.  For example, applied to the
digraph in Figure~\ref{f:strong_connected_components}, the code
\begin{minted}{python}
G = nx.DiGraph()
G.add_edge(1, 1)
G.add_edge(2, 1)
G.add_edge(2, 3)
G.add_edge(3, 2)

list(nx.strongly_connected_components(G)) 
\end{minted}
prints \texttt{[\{1\}, \{2, 3\}]}.

\subsubsection{Aperiodicity}\label{sss:aperiodicgraph}

A \navy{cycle}\index{Cycle} $(u, v, w, \ldots, u)$ of a directed graph $\gG =
(V, E)$ is a directed walk in $\gG$ such that (i) the first and last vertices
are equal and (ii) no other vertex is repeated.  The graph is called
\navy{periodic}\index{Periodic digraph} if there exists a $k > 1$ such that
$k$ divides the length of every cycle.  A graph that fails to be periodic is
called \navy{aperiodic}.

\begin{example}
    In Figure~\ref{f:periodic}, the cycles are $(a, b, a)$, $(b, a, b)$, $(b, c, b)$,
    $(c, b, c)$, $(c, d, c)$ and $(d, c, d)$.  Hence the length of every cycle is $2$ and the graph is
    periodic.
\end{example}

\begin{figure}
   \begin{center}
       \input{tikz/periodic_mc.tex}
   \end{center}
   \caption{\label{f:periodic} A periodic digraph}
\end{figure}
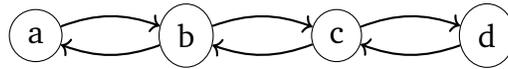

An obvious sufficient condition for aperiodicity is existence of even one
self-loop.  The digraphs in
Figures~\ref{f:rich_poor_no_label}--\ref{f:strong_connected_components} are
aperiodic for this reason.

The next result provides an easy way to understand aperiodicity for connected
graphs.  Proofs can be found in \cite{haggstrom2002finite} and
\cite{norris1998markov}.

\begin{lemma}\label{l:scaper}
    Let $\gG = (V, E)$ be a digraph.
    If $\gG$ is strongly connected, then $\gG$ is aperiodic if and only
    if, for all $v \in V$, there exists a $q \in \NN$ such that, for all $k
    \geq q$, there exists a directed walk of length $k$ from $v$ to $v$.
\end{lemma}

It is common to call a vertex $v$ satisfying the condition in
Lemma~\ref{l:scaper} \navy{aperiodic}\index{Aperiodic}.  With this
terminology, Lemma~\ref{l:scaper} states that a strongly connected digraph is
aperiodic if and only if every vertex is aperiodic.

Networkx can be used to check for aperiodicity of vertices or graphs.  For
example, if \texttt{G} is a \texttt{DiGraph} object, then
\texttt{nx.is\_aperiodic(G)} returns \texttt{True} or \texttt{False} depending
on aperiodicity of \texttt{G}.

\subsubsection{Adjacency Matrices}\label{sss:admat}

There is a simple map between edges of a graph with fixed vertices and a
binary matrix called an adjacency matrix. The benefit of viewing connections
through adjacency matrices is that they bring the power of linear algebra to
the analysis of digraphs.   We illustrate this briefly here and extensively in
\S\ref{ss:wdig}.

Let $\gG = (V, E)$ be a digraph and let $|V| = n$.
If we enumerate the elements of $V$, writing them as $(v_1, \ldots, v_n)$,
then the $n \times n$ \navy{adjacency matrix} corresponding to $(V, E)$ is
defined by\footnote{Note that, in some applied fields, the adjacency matrix is
    transposed: $a_{ij}=1$ if there is an edge from $j$ to $i$, rather than
from $i$ to $j$.  We will avoid this odd and confusing definition (which
contradicts both standard graph theory and standard notational conventions in
the study of Markov chains).}
\begin{equation}\label{eq:digrapham}
    A = (a_{ij})_{1 \leq i, j \leq n}
    \quad \text{with} \quad
    a_{ij} = \1\{(v_i, v_j) \in E\}.
\end{equation}
For example, with $\{$poor, middle, rich$\}$ mapped to $(1, 2, 3)$, the adjacency
matrix corresponding to the digraph in Figure~\ref{f:poverty_trap} is
\begin{equation}\label{eq:adjpt}
    A = 
    \begin{pmatrix}
        1 & 0 & 0 \\
        1 & 1 & 1 \\
        1 & 1 & 1 
    \end{pmatrix}.
\end{equation}

An adjacency matrix provides us with enough information to recover the edges
of a graph.  More generally, given a set of vertices $V = (v_1, \ldots, v_n)$,
an $n \times n$ matrix $A = (a_{ij})_{1 \leq i, j \leq n}$ with binary entries
generates a digraph $\gG$ with vertices $V$ and edges $E$ equal to all $(v_i, v_j)
\in V \times V$ with $a_{ij} = 1$.  The adjacency matrix of this graph $(V,
E)$ is $A$.

\begin{Exercise}\label{ex:uddg}
    A digraph $(V, E)$ is called \navy{undirected} if $(u, v) \in E$ implies
    $(v, u) \in E$.
    What property does this imply on the adjacency matrix?
\end{Exercise}

\begin{Answer}
    If $(V, E)$ is undirected, then the adjacency matrix is symmetric.
\end{Answer}

\begin{remark}
    The idea that a digraph can be undirected, presented in
    Exercise~\ref{ex:uddg}, seems contradictory.  After
    all, a digraph is a \emph{directed} graph.  Another way to introduce
    undirected graphs is to define them as a vertex-edge pair $(V, E)$, where
    each edge $\{u, v\} \in E$ is an unordered pair, rather than an ordered
    pair $(u, v)$.  However, the definition in Exercise~\ref{ex:uddg} is
    essentially equivalent and more convenient for our purposes, since we mainly
    study directed graphs.
\end{remark}

Like Networkx, the QuantEcon Python library \texttt{quantecon} supplies a
graph object that implements certain graph-theoretic algorithms.  
The set of available algorithms is more limited but each one is faster,
accelerated by just-in-time compilation. In the case of
QuantEcon's \texttt{DiGraph} object, an instance is created via the adjacency matrix.  For
example, to construct a digraph corresponding to Figure~\ref{f:poverty_trap}
we use the corresponding adjacency matrix \eqref{eq:adjpt}, as follows:

\begin{minted}{python}
import quantecon as qe
import numpy as np

A = ((1, 0, 0),
     (1, 1, 1,),
     (1, 1, 1))
A = np.array(A)       # Convert to NumPy array
G = qe.DiGraph(A)
\end{minted}

Let's print the set of strongly connected components, as a list of NumPy
arrays:

\begin{minted}{python}
G.strongly_connected_components 
\end{minted}

The output is \texttt{[array([0]), array([1, 2])]}.

\subsection{Weighted Digraphs}\label{ss:wdig}

Early quantitative work on networks tended to focus on unweighted digraphs,
where the existence or absence of an edge is treated as sufficient information
(e.g., following or not following on social media, existence or absence of a
road connecting two towns). However, for some networks, this binary measure is
less significant than the size or strength of the connection.  

As one illustration, consider
Figure~\ref{f:financial_network_analysis_visualization}, which shows flows of
funds (i.e., loans) between private banks, grouped by country of origin.  An
arrow from Japan to the US, say, indicates aggregate claims held by Japanese
banks on all US-registered banks, as collected by the Bank of International
Settlements (BIS). The size of each node in the figure is increasing in the
total foreign claims of all other nodes on this node. The widths of the arrows
are proportional to the foreign claims they represent.\footnote{Data for the
    figure was obtained from the BIS consolidated
    banking statistics, for Q4 of 2019. Our calculations used the immediate
    counterparty basis for financial claims of domestic and foreign banks,
    which calculates the sum of cross-border claims and local claims of
    foreign affiliates in both foreign and local currency. The foreign claim
    of a node to itself is set to zero.}
The country codes are given in Table~\ref{table:cfn}.

\begin{figure}
   \centering
   \scalebox{0.9}{\includegraphics[trim = 20mm 20mm 0mm 20mm, clip]{
   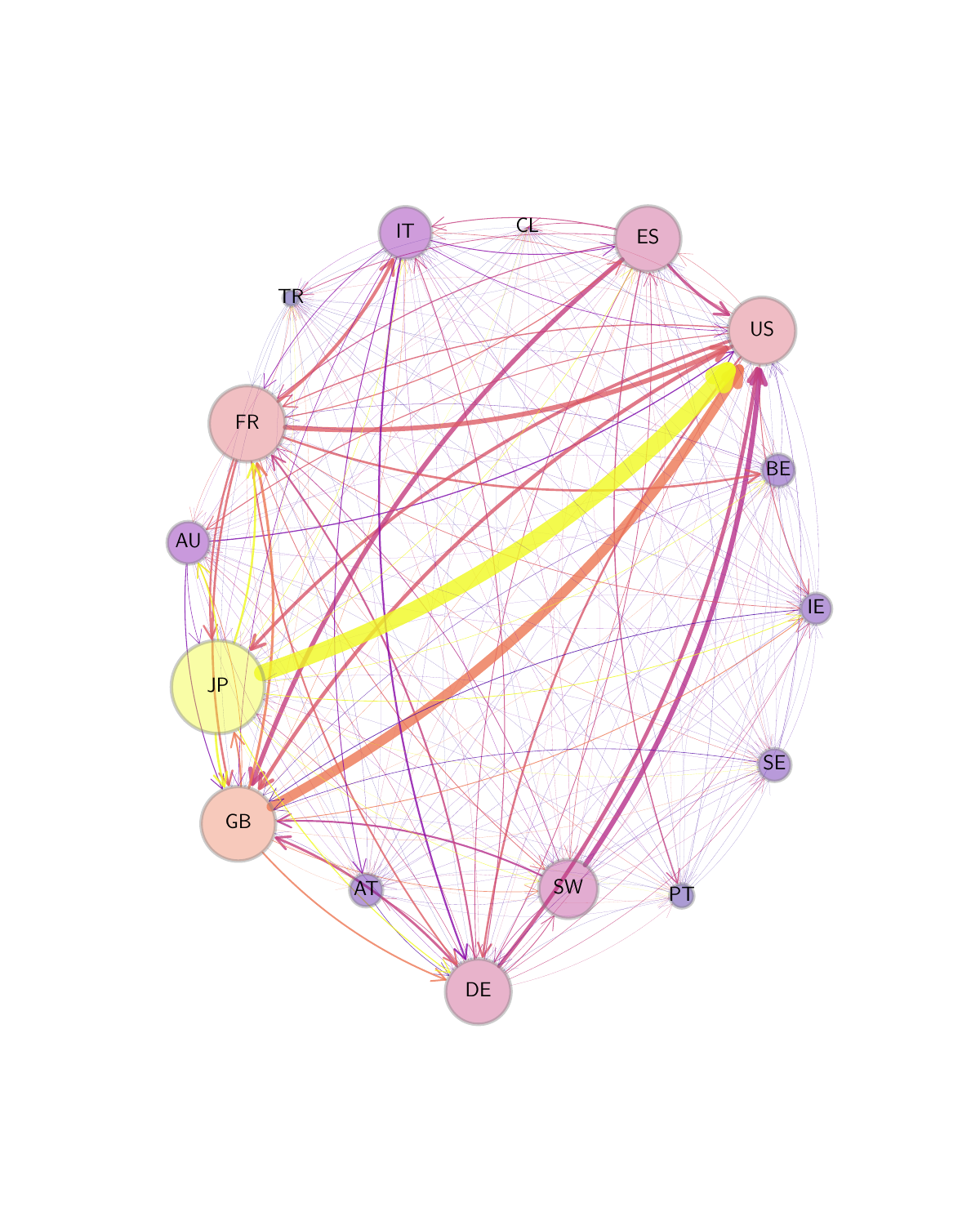}}
   \caption{\label{f:financial_network_analysis_visualization} International
   private credit flows by country}
\end{figure}

\begin{table}
    \small
    \centering
    \addtolength{\tabcolsep}{-2pt}
    \centering
    \begin{tabular}{ll|ll|ll|ll}
        \hline 
        AU  & Australia  & DE  & Germany & CL  & Chile & ES & Spain \\
        PT  &  Portugal & FR & France & TR  & Turkey & GB & United Kingdom \\
        US  & United States & IE & Ireland & AT  & Austria & IT & Italy \\
        BE  & Belgium & JP & Japan & SW & Switzerland & SE & Sweden \\
        \hline
    \end{tabular}
    \caption{\label{table:cfn} Codes for the 16-country financial network}
\end{table}

In this network, an edge $(u, v)$ exists for almost every choice of $u$ and
$v$ (i.e., almost every country in the network).\footnote{In fact arrows
    representing foreign claims less than US\$10 million are cut from
    Figure~\ref{f:financial_network_analysis_visualization}, so the network is
even denser than it appears.}  Hence existence of an edge is not
particularly informative.  To understand the network, we need to record not
just the existence or absence of a credit flow, but also the size of the flow.
The correct data structure for recording this information is a ``weighted
directed graph,'' or ``weighted digraph.'' In this section we define this
object and investigate its properties.

\subsubsection{Definitions}\label{sss:gd}

A \navy{weighted digraph}\index{Weighted digraph} $\gG$ is a triple $(V, E, w)$
such that $(V, E)$ is a digraph and $w$ is a function from $E$ to $(0,
\infty)$, called the \navy{weight function}\index{weight function}.

\begin{remark}
    Weights are traditionally regarded as nonnegative. In this text we insist
    that weights are also positive, in the sense that $w(u, v) > 0$ for all
    $(u, v) \in E$.  The reason is that the intuitive notion of zero weight is
    understood, here and below, as absence of a connection.  In other words,
    if $(u, v)$ has ``zero weight,'' then $(u, v)$ is not in $E$, so $w$ is
    not defined on $(u, v)$.
\end{remark}

\begin{example}
    As suggested by the discussion above, 
    the graph shown in Figure~\ref{f:financial_network_analysis_visualization}
    can be viewed as a weighted digraph.  Vertices are countries of origin
    and an edge exists between country $u$ and country $v$ when private banks
    in $u$ lend nonzero quantities to banks in $v$.  The weight assigned to
    edge $(u, v)$ gives total loans from $u$ to $v$ as measured according to
    the discussion of Figure~\ref{f:financial_network_analysis_visualization}.
\end{example}

\begin{example}
    Figure~\ref{f:rich_poor} shows a weighted digraph, with arrows 
    representing edges of the induced digraph (compare with the unweighted
    digraph in Figure~\ref{f:rich_poor_no_label}). The numbers next to
    the edges are the weights.  In this case, you can think of the numbers
    on the arrows as transition probabilities for a household over, say, one
    year.  For example, a rich household has a 10\% chance of becoming poor.  
\end{example}

\begin{figure}
   \begin{center}
       \scalebox{1.0}{\input{tikz/rich_poor.tex}}
   \end{center}
   \caption{\label{f:rich_poor} A weighted digraph}
\end{figure}
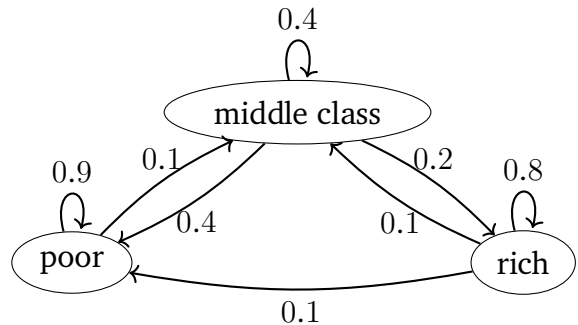

The definitions of \navy{accessibility}, \navy{communication},
\navy{periodicity} and \navy{connectedness} extend to any weighted digraph
$\gG
= (V, E, w)$ by applying them to $(V, E)$.  For example,
$(V, E, w)$ is called strongly connected if $(V, E)$ is strongly connected. The
weighted digraph in Figure~\ref{f:rich_poor} is strongly connected.

\subsubsection{Adjacency Matrices of Weighted Digraphs}

In \S\ref{sss:admat} we discussed adjacency matrices of unweighted digraphs.
The \navy{adjacency matrix}\index{Adjacency matrix} of a \emph{weighted}
digraph $(V, E, w)$ with vertices $(v_1, \ldots, v_n)$ is the matrix
\begin{equation*}
    A = (a_{ij})_{1 \leq i,j \leq n}
    \quad \text{with} \quad
    a_{ij} =
    \begin{cases}
        w(v_i, v_j) & \text{ if } (v_i, v_j) \in E
        \\
        0           & \text{ otherwise}.
    \end{cases}
\end{equation*}
Clearly, once the vertices in $V$ are enumerated, the weight function and
adjacency matrix provide essentially the same information.  We often work
with the latter, since it facilitates computations.

\begin{example}
    With $\{$poor, middle, rich$\}$ mapped to $(1, 2, 3)$, 
    the adjacency matrix corresponding to the weighted digraph in
    Figure~\ref{f:rich_poor} is
    \begin{equation}\label{eq:fnegwa0}
        A = 
        \begin{pmatrix}
            0.9 & 0.1 & 0 \\
            0.4 & 0.4 & 0.2 \\
            0.1 & 0.1 & 0.8
        \end{pmatrix}.
    \end{equation}
\end{example}

In QuantEcon's \texttt{DiGraph} implementation, weights are recorded via the
keyword \texttt{weighted}:

\begin{minted}{python}
A = ((0.9, 0.1, 0.0),
     (0.4, 0.4, 0.2),
     (0.1, 0.1, 0.8))
A = np.array(A)
G = qe.DiGraph(A, weighted=True)    # Store weights
\end{minted}

One of the key points to remember about adjacency matrices is that taking the
transpose ``reverses all the arrows'' in the associated digraph.  

\begin{example}
    The digraph in Figure~\ref{f:network_liabfin} can be interpreted as a
    stylized version of a financial network, with vertices as banks and edges
    showing flow of funds, similar to
    Figure~\ref{f:financial_network_analysis_visualization} on
    page~\pageref{f:financial_network_analysis_visualization}.  For example,
    we see that bank 2 extends a loan of size 200 to bank 3.
    The corresponding adjacency matrix is
    \begin{equation}\label{eq:fnegwa}
        A = 
        \begin{pmatrix}
            0 & 100 & 0 & 0 & 0 \\
            50 & 0 & 200 & 0 & 0 \\
            0 & 0 & 0 & 100 & 0 \\
            0 & 500 & 0 & 0 & 50 \\
            150 & 0 & 250 & 300 & 0 
        \end{pmatrix}.
    \end{equation}
    The transposition is
    \begin{equation}\label{eq:fnegwat}
        A^\top = 
        \begin{pmatrix}
            0   & 50  & 0   & 0   & 150 \\
            100 & 0   & 0   & 500 & 0 \\
            0   & 200 & 0   & 0   & 250 \\
            0   & 0   & 100 & 0   & 300 \\
            0   & 0   & 0   & 50  & 0 
        \end{pmatrix}.
    \vspace{0.3em}
    \end{equation}
    The corresponding network is visualized in
    Figure~\ref{f:network_liabfin_trans}.  This figure shows the network of
    liabilities after the loans have been granted.
    Both of these networks (original and transpose) are useful for analysis of
    financial markets (see, e.g., Chapter~\ref{c:fpms}).
\end{example}

\begin{figure}
   \begin{center}
    \input{tikz/network_liabfin.tex}
    \caption{\label{f:network_liabfin} A network of credit flows across institutions}
   \end{center}
\end{figure}
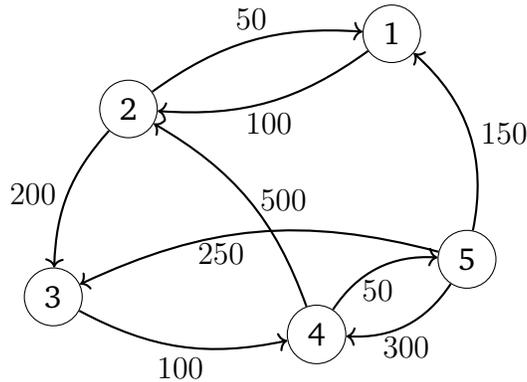

\begin{figure}
   \begin{center}
    \input{tikz/network_liabfin_trans.tex}
    \caption{\label{f:network_liabfin_trans} The transpose: a network of liabilities}
   \end{center}
\end{figure}
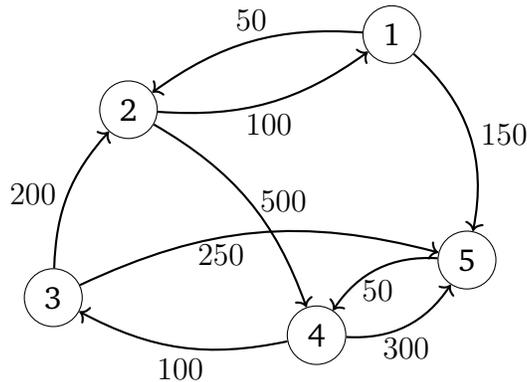

It is not difficult to see that each nonnegative $n \times n$ matrix $A = (a_{ij})$ can be
viewed as the adjacency matrix of a weighted digraph with vertices equal
to $\natset{n}$.  The weighted digraph $\gG = (V, E, w)$ in question is 
formed by setting 
\begin{equation*}
    V = \natset{n},
    \quad
    E = \setntn{(i,j) \in V \times V}{ a_{ij} > 0}
    \quad \text{and} \quad
    w(i, j) = a_{ij} \text{ for all } (i,j) \in E.
\end{equation*}
We call $\gG$ the \navy{weighted digraph induced by $A$}.

The next exercise helps to reinforce the point that transposes reverse the
edges.

\begin{Exercise}\label{ex:rar}
    Let $A= (a_{ij})$ be a nonnegative $n \times n$ matrix and let $\gG =
    (\natset{n}, E, w)$ and $\gG' = (\natset{n}, E', w')$ be the weighted
    digraphs induced by $A$ and $A^\top$, respectively.  Show that 
    \begin{enumerate}
        \item $(j, k) \in E'$ if and only if $(k, j) \in E$.
        \item $j \to k$ in $\gG'$ if and only if $k \to j$ in $\gG$.
    \end{enumerate}

\end{Exercise}

\begin{Answer}
    Let $A^{\top} = (a'_{ij})$, so that $a'_{ij} = a_{ji}$ for each $i, j$.  
    By definition, we have 
    \begin{equation*}
        (j, k) \in E' 
        \; \iff \;
        a'_{jk} > 0
        \; \iff \;
        a_{kj} > 0
        \; \iff \;
        (k, j) \in E,
    \end{equation*}
    which proves (i).  Regarding (ii), to say that 
    $k$ is accessible from $j$ in $\gG'$ means that we can find vertices
    $i_1, \ldots, i_m$ that form a directed path from $j$ to $k$ under $\gG'$, in the sense that 
    such that $i_1 = j$, $i_m=k$, and each successive
    pair $(i_{\ell}, i_{\ell+1})$ is in $E'$.  But then, by (i), 
    $i_m, \ldots, i_1$ provides a directed path from $k$ to $j$ under $\gG$,
    since and each successive pair $(i_{\ell+1}, i_{\ell})$ is in
    $E$.
\end{Answer}

\subsubsection{Application: Quadratic Network Games}\label{sss:quadgame}

\cite{acemoglu2016networks} and \cite{zenou2016key} consider quadratic
games with $n$ agents where agent $k$ seeks to maximize
\begin{equation}\label{eq:uing}
    u_k(x) 
    := -\frac{1}{2} x_k^2 + \alpha x^\top A x + x_k \epsilon_k.
\end{equation}
Here $x = (x_i)_{i=1}^n$, $A$ is a symmetric matrix with $a_{ii}=0$ for all
$i$, $\alpha \in (0,1)$ is a parameter and $\epsilon = (\epsilon_i)_{i=1}^n$
is a random vector.  (This is the set up for the quadratic game in \S21.2.1 of
\cite{acemoglu2016networks}.)  The $k$-th agent takes the decisions $x_j$ as
given for all $j \not= k$ when maximizing \eqref{eq:uing}.  

In this context, $A$ is understood as the adjacency matrix of a graph with
vertices $V = \natset{n}$, where each vertex is one agent.  We can reconstruct
the weighted digraph $(V, E, w)$ by setting $w(i, j) = a_{ij}$ and letting $E$
be all $(i,j)$ pairs in $\natset{n} \times \natset{n}$ with $a_{ij} > 0$.
The weights identify some form of relationship between the agents, such as
influence or friendship.

\begin{Exercise}
    A \navy{Nash equilibrium}\index{Nash equilibrium} for the quadratic network
    game is a vector $x^* \in \RR^n$ such that, for all $i \in \natset{n}$,
    the choice $x_i^*$ of agent $i$ maximizes~\eqref{eq:uing} taking $x_j^*$
    as given for all $j \not= i$.  Show that, whenever $r(A) < 1/\alpha$,
    a unique Nash equilibrium $x^*$ exists in $\RR^n$ and, moreover,
    $x^* := (I - \alpha A)^{-1} \epsilon$.
\end{Exercise}

\begin{Answer}
    Recalling that $\partial / (\partial x_k) x^\top A x = (x^\top A)_k$, the first
    order condition corresponding to \eqref{eq:uing}, taking the actions of other
    players as given, is
    \begin{equation*}
        x_k = \alpha (x^\top A)_k + \epsilon_k
        \qquad (k \in \natset{n}).
    \end{equation*}
    Concatenating into a row vector and then taking the transpose yields $x =
    \alpha A x + \epsilon$, where we used the fact that $A$ is symmetric.
    Since $r(\alpha A) = \alpha r(A)$, the condition $r(A) < 1/\alpha$ implies
    that $r(\alpha A) < 1$, so, by the Neumann series lemma, the unique solution 
    is $x^* = (I - \alpha A)^{-1} \epsilon$.
\end{Answer}

The network game described in this section has many interesting applications,
including social networks, crime networks and peer networks.  References are
provided in \S\ref{s:cnni}.

\subsubsection{Properties}

In this section, we examine some of the fundamental properties of and
relationships among digraphs, weight functions and adjacency matrices.
Throughout this section, the vertex set $V$ of any graph we examine will be
set to $\natset{n}$.  This costs no loss of generality, since, in this
text, the vertex set of a digraph is always finite and nonempty.

Also, while we refer to weighted digraphs for their additional generality, the
results below connecting adjacency matrices and digraphs are valid for
unweighted digraphs.  Indeed, an unweighted digraph $\gG = (V, E)$ can be mapped
to a weighted digraph by introducing a weight function that maps each element
of $E$ to unity.  The resulting adjacency matrix agrees with our original
definition for unweighted digraphs in \eqref{eq:digrapham}.

As an additional convention, if $A$ is an adjacency matrix, and $A^k$ is the
$k$-th power of $A$, then we write $a^k_{ij}$ for a typical element of $A^k$. 
With this notation, we observe that, since
$A^{(s+t)} = A^s A^t$, the rules of matrix multiplication imply
\begin{equation}\label{eq:accip}
    a^{s+t}_{ij}
    = \sum_{\ell=1}^n a^s_{i \ell} \, a^t_{\ell j}
    \qquad (i, j \in \natset{n}, \;\; s,t \in \NN).
\end{equation}
($A^0$ is the identity.) The next proposition explains the significance of the
powers.

\begin{proposition}\label{p:accesspos}
    Let $\gG$ be a weighted digraph with adjacency matrix $A$. For distinct
    vertices $i, j \in \natset{n}$ and $k \in \NN$, we have
    \begin{equation*}
        a^k_{i j} > 0
        \; \iff \;
        \text{ there exists a directed walk of length $k$ from $i$ to $j$}.
    \end{equation*}
\end{proposition}

\begin{proof}
    ($\Leftarrow$ ).  The statement is true by definition when $k=1$.  Suppose
    in addition that $\Leftarrow$ holds at $k-1$, and suppose
    there exists a directed walk $(i, \ell, m, \ldots, n, j)$ of length $k$
    from $i$ to $j$.  By the induction hypothesis we have $a^{k-1}_{i n} >
    0$.  Moreover, $(n, j)$ is part of a directed walk, so $a_{n j} > 0$.  Applying
    \eqref{eq:accip} now gives $a^k_{i j} > 0$.

    ($\Rightarrow$).  Left as an exercise (just use the same logic).
\end{proof}

\begin{example}
    In \S\ref{s:amcs} we show that if
    elements of $A$ represent
    one-step transition probabilities across states, then elements of $A^t$,
    the $t$-th power of $A$, provide $t$-step transition probabilities.  
    In Markov process theory, \eqref{eq:accip} is called the
    \emph{Chapman--Kolmogorov equation}.
\end{example}

In this context, the next result is fundamental.

\begin{theorem}\label{t:sconir}
    Let $\gG$ be a weighted digraph.  The following statements are equivalent:
    \begin{enumerate}
        \item $\gG$ is strongly connected. 
        \item The adjacency matrix generated by $\gG$ is irreducible.
    \end{enumerate}
\end{theorem}

\begin{proof}
    Let $\gG$ be a weighted digraph with adjacency matrix $A$.   By
    Proposition~\ref{p:accesspos}, strong connectedness of $\gG$ is equivalent
    to the statement that, for each $i, j \in V$, we can find a $k \geq 0$
    such that $a^k_{ij} > 0$. (If $i=j$  then set $k=0$.) This, in turn, is
    equivalent to $\sum_{m=0}^\infty A^m \gg 0$, which is irreducibility of
    $A$.
\end{proof}

\begin{example}
    Strong connectivity fails in the digraph in Figure~\ref{f:io_reducible}, since
    vertex 4 is a source.  By
    Theorem~\ref{t:sconir}, the adjacency matrix must be reducible.
\end{example}

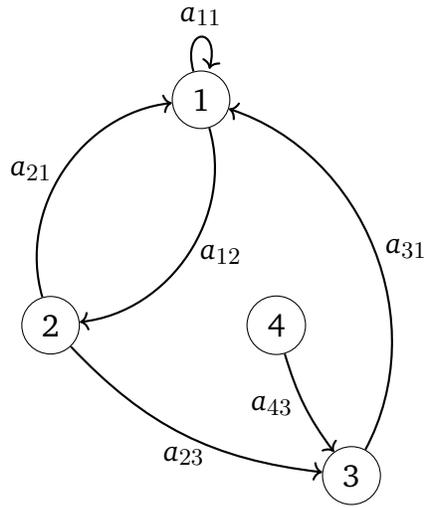
\begin{figure}
   \begin{center}
    \input{tikz/io_reducible.tex}
    \caption{\label{f:io_reducible} Failure of strong connectivity}
   \end{center}
\end{figure}

We will find that the property of being primitive is valuable for analysis.
(The Perron--Frobenius Theorem hints at this.)  What do we need to add to strong
connectedness to obtain primitiveness?

\begin{theorem}\label{t:scaperpr}
    For a weighted digraph $\gG=(V, E, w)$, the following statements are equivalent:
    \begin{enumerate}
        \item $\gG$ is strongly connected and aperiodic. 
        \item The adjacency matrix generated by $\gG$ is primitive.
    \end{enumerate}
\end{theorem}

\begin{proof}[Proof of Theorem~\ref{t:scaperpr}]
    Throughout the proof we set $V=\natset{n}$.
    First we show that, if $\gG$ is aperiodic and strongly connected, then, for all $i,
    j \in V$, there exists a $q \in \NN$ such that $a^k_{ij} > 0$
    whenever $k \geq q$.  To this end, pick any $i,j$ in $V$.   Since $\gG$ is
    strongly connected,
    there exists an $s \in \NN$ such that $a^s_{ij} > 0$.  Since $\gG$ is
    aperiodic, we can find an $m \in \NN$ such that $\ell \geq m$ implies
    $a^\ell_{jj} > 0$.  Picking $\ell \geq m$ and applying~\eqref{eq:accip},
    we have
    \begin{equation*}
        a^{s+\ell}_{ij}
        = \sum_{r \in V} a^s_{i r} a^\ell_{r j}
        \geq  a^s_{ij} a^\ell_{jj}
        > 0.
    \end{equation*}
    Thus, with $t = s + m$, we have $a^k_{ij} > 0$ whenever $k \geq t$.

    ((i) $\Rightarrow$ (ii)).  By the preceding argument, given any $i, j \in V$,
    there exists an $s(i,j) \in \NN$ such that $a^m_{ij} > 0$ whenever $m
    \geq s(i,j)$.  Setting $k := \max s(i, j)$ over all $(i,j)$ yields
    $A^k \gg 0$.

    ((ii) $\Rightarrow$ (i)). Suppose that $A$ is primitive.  Then, for some
    $k \in \NN$, we have $A^k \gg 0$.  Strong connectedness of the digraph
    follows directly from Proposition~\ref{p:accesspos}.  It remains to check
    aperiodicity.

    Aperiodicity will hold if we can establish that $a^{k+t}_{ii} > 0$ for all $t \geq 0$.
    To show this, it suffices to show that $A^{k + t}  \gg 0$ for all $t \geq 0$.
    Moreover, to prove the latter, we need only show that $A^{k + 1}  \gg 0$,
    since the claim then follows from induction.

    To see that $A^{k+1} \gg 0$, observe that, for any given
    $i, j$, the relation \eqref{eq:accip} implies
    \begin{equation*}
        a^{k+1}_{ij}
        = \sum_{\ell \in V} a_{i\ell} a^k_{\ell j}
        \geq \bar a  \sum_{\ell \in V} a_{i \ell}.
    \end{equation*}
    where $\bar a := \min_{\ell \in V} a^k_{\ell j} > 0$.
    The proof will be done if $\sum_{\ell \in V} a_{i \ell} > 0$.  But this
    must be true, since otherwise vertex $i$ is a sink, which contradicts
    strong connectedness.
\end{proof}

\begin{example}
    In Exercise~\ref{ex:pwprop} we worked hard to show that $P_w$ is
    irreducible if and only if $0 < \alpha, \beta \leq 1$, using the approach
    of calculating and then examining the powers of $P_w$ (as shown in
    \eqref{eq:pwpk}).  However, the result is trivial when we examine
    the corresponding digraph in Figure~\ref{f:worker_switching} and use the
    fact that irreducibility is equivalent to strong connectivity.
    Similarly, the result in Exercise~\ref{ex:pwprop} that $P_w$ is
    primitive if and only if $0 < \alpha, \beta \leq 1$ and $\min\{\alpha,
    \beta \} < 1$ becomes much easier to establish if we examine the digraph and use
    Theorem~\ref{t:scaperpr}.
\end{example}

\subsection{Network Centrality}\label{ss:netcen}

When studying networks of all varieties, a recurring topic is the relative
``centrality'' or ``importance'' of different nodes. One classic application
is the ranking of web pages by search engines.  Here are some examples related
to economics:

\begin{itemize}
    \item In which industry will one dollar of additional demand have the most
        impact on aggregate production, once we take into account all the
        backward linkages?  In which sector will a rise in productivity have
        the largest effect on national output?
    \item A negative shock endangers the solvency of the entire banking sector.
        Which institutions should the government rescue, if any?  
    \item In the network games considered in \S\ref{sss:quadgame},
        the Nash equilibrium is
        $x^* = (I - \alpha A)^{-1} \epsilon$. Players' actions are
        dependent on the topology of the network, as encoded in $A$.  A common
        finding is that the level of activity or effort exerted by an agent
        (e.g., severity of criminal activity by a participant in a criminal
        network) can be predicted from their ``centrality'' within the
        network.  
\end{itemize}

In this section we review essential concepts related to network
centrality.\footnote{Centrality measures are sometimes called ``influence
measures,'' particularly in connection with social networks.}

\subsubsection{Centrality Measures}

Let $G$ be the set of weighted digraphs. A \navy{centrality measure}
associates to each $\gG = (V, E, w)$ in $G$ a vector $m(\gG) \in
\RR^{|V|}$, where the $i$-th element of $m(\gG)$ is interpreted as the
centrality (or rank) of vertex $v_i$.  In most cases $m(\gG)$ is nonnegative.
In what follows, to simplify notation, we take $V = \natset{n}$.

(Unfortunately, the definitions and terminology associated with even the most
common centrality measures vary widely across the applied literature.  Our
convention is to follow the mathematicians, rather than the physicists.  For
example, our terminology is consistent with \cite{benzi2015limiting}.)

\subsubsection{Authorities vs Hubs}

Search engine designers recognize that web pages can be important in two
different ways.  Some pages have high \navy{hub centrality}\index{Hub
centrality}, meaning that they \emph{link to} valuable sources of information
(e.g., news aggregation sites) .  Other pages have high \navy{authority
centrality}, meaning that they contain valuable information, as indicated by
the number and significance of \emph{incoming} links (e.g., websites of
respected news organizations).  Figure~\ref{f:hub_and_authority} helps to
visualize the difference.

\begin{figure*}
   \begin{center}
    \input{tikz/hub_and_authority.tex}
    \caption{\label{f:hub_and_authority} Hub vs authority}
   \end{center}
\end{figure*}
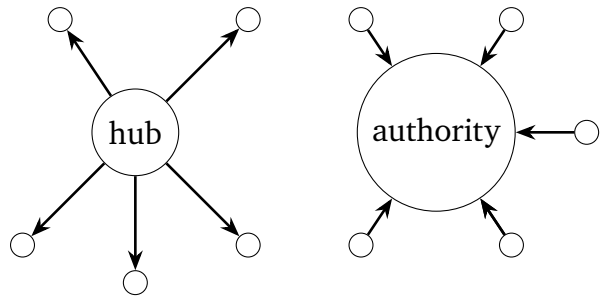

Similar ideas can and have been applied to economic networks (often using
different terminology).  For example, in production networks we study below,
high hub centrality is related to upstreamness: such sectors tend to supply
intermediate goods to many important industries. Conversely, a high authority
ranking will coincide with downstreamness.

In what follows we discuss both hub-based and authority-based centrality
measures, providing definitions and illustrating the relationship between
them.

\subsubsection{Degree Centrality}\label{sss:degcen}

Two of of the most elementary measures of ``importance''
of a vertex in a given digraph $\gG = (V, E)$ are its in-degree and
out-degree. Both of these provide a centrality measure.  
\navy{In-degree centrality}\index{In-degree centrality} $i(\gG)$ is defined as  
the vector $(i_d(v))_{v \in V}$.
\navy{Out-degree centrality}\index{out-degree centrality} $o(\gG)$ is defined
as $(o_d(v))_{v \in V}$.  If $\gG$ is expressed as a Networkx \texttt{DiGraph} called
\texttt{G} (see, e.g., \S\ref{sss:nx}), then $i(\gG)$ can be calculated via
\begin{minted}{python}
iG = [G.in_degree(v) for v in G.nodes()]    
\end{minted}

This method is relatively slow when $\gG$ is a large digraph. Since vectorized
operations are generally faster, let's look at an alternative method using
operations on arrays.  

To illustrate the method, recall the network of financial institutions in
Figure~\ref{f:network_liabfin}.  We can compute the in-degree and out-degree
centrality measures by first converting the adjacency matrix, which is shown in
\eqref{eq:fnegwa}, to a binary matrix that corresponds to the adjacency matrix
of the same network viewed as an unweighted graph:
\begin{equation}\label{eq:fnegwau}
    U = 
    \begin{pmatrix}
        0 & 1 & 0 & 0 & 0 \\
        1 & 0 & 1 & 0 & 0 \\
        0 & 0 & 0 & 1 & 0 \\
        0 & 1 & 0 & 0 & 1 \\
        1 & 0 & 1 & 1 & 0 
    \end{pmatrix}
\end{equation}
Now $U(i, j) = 1$ if and only if $i$ points to $j$.  The out-degree and
in-degree centrality measures can be computed as 
\begin{equation}\label{eq:ogig}
    o(\gG) = U \1
    \quad \text{and} \quad
    i(\gG) = U^\top \1,
\end{equation}
respectively.  That is, summing the rows of $U$ gives the out-degree
centrality measure, while summing the columns gives the in-degree measure.

The out-degree centrality measure is a hub-based ranking, while the vector of
in-degrees is an authority-based ranking.  For the financial network in
Figure~\ref{f:network_liabfin}, a high out-degree for a given institution
means that it lends to many other institutions.  A high in-degree indicates
that many institutions lend to it.

Notice that, to switch from a hub-based ranking to an authority-based ranking,
we need only transpose the (binary) adjacency matrix $U$.  We will see that
the same is true for other centrality measures.   This is intuitive, since
transposing the adjacency matrices reverses the direction of the edges
(Exercise~\ref{ex:rar}).

For a weighted digraph $\gG = (V, E, w)$ with adjacency matrix $A$, the
\navy{weighted out-degree centrality} and \navy{weighted in-degree centrality}
measures are defined as
\begin{equation}\label{eq:wogig}
    o(\gG) = A \1
    \quad \text{and} \quad
    i(\gG) = A^\top \1,
\end{equation}
respectively, by analogy with \eqref{eq:ogig}.  We present some intuition for
these measures in applications below.

Unfortunately, while in- and out-degree measures of centrality are simple to
calculate, they are not always informative.  As an example, consider again the
international credit network shown in
Figure~\ref{f:financial_network_analysis_visualization}.  There, an edge
exists between almost every node, so the in- or out-degree based centrality
ranking fails to effectively separate the countries. This can be seen in the
out-degree ranking of countries corresponding to that network in the top left panel of
Figure~\ref{f:financial_network_analysis_centrality}, and in the in-degree
ranking in the top right.

\begin{figure}
   \centering
   \scalebox{0.64}{\includegraphics[trim = 0mm 10mm 0mm 10mm, clip]{ 
   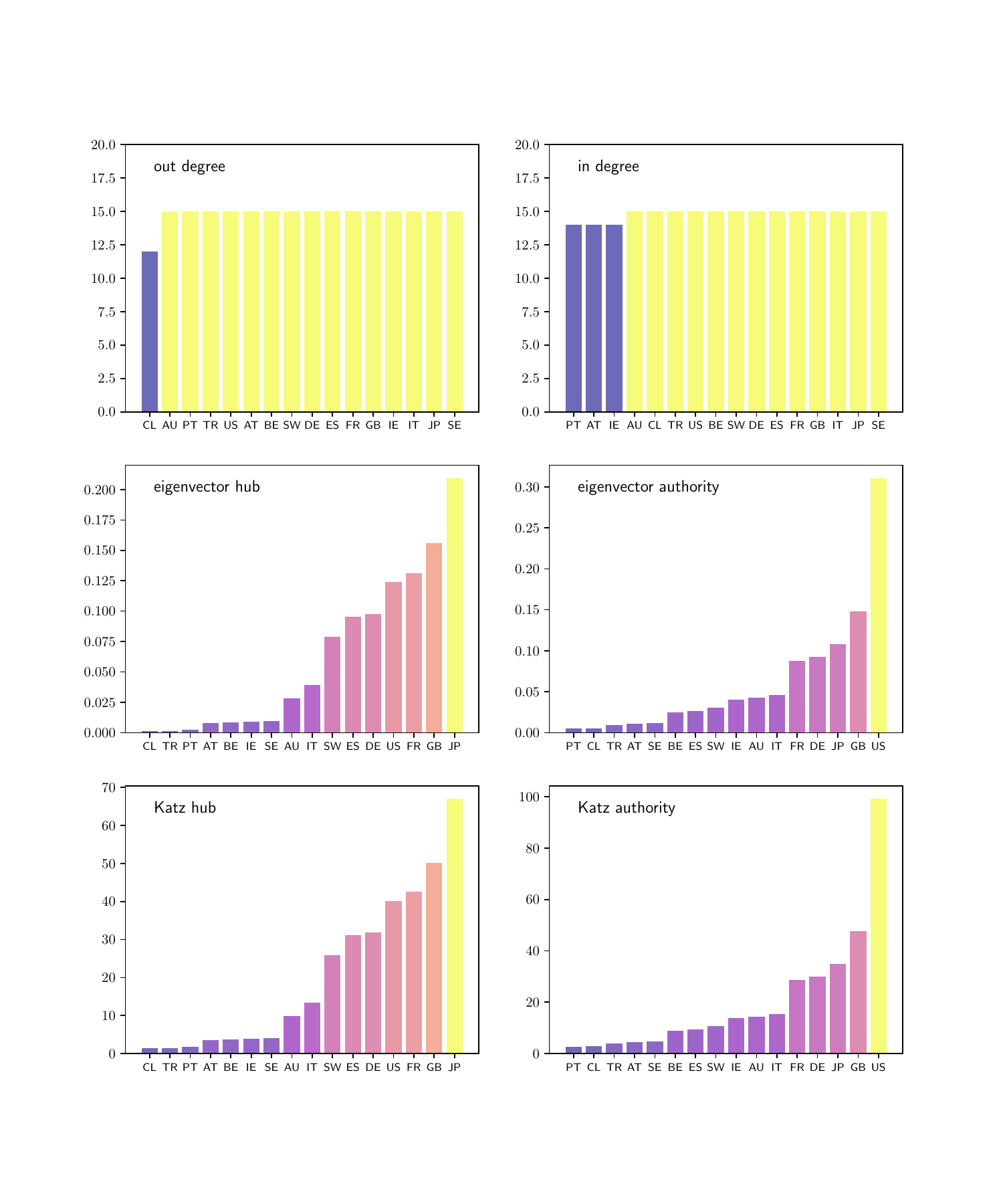}}
   \caption{\label{f:financial_network_analysis_centrality} Centrality measures for the credit network}
\end{figure}

There are other limitations of degree-based centrality rankings. For example,
suppose web page A has many inbound links, while page B has fewer.  Even
though page A dominates in terms of in-degree, it might be less important than
web page B to, say, a potential advertiser, when the links into B are from
more heavily trafficked pages.  Thinking about this point suggests
that importance can be recursive: the importance of a given node depends on
the importance of other nodes that link to it.  The next set of centrality
measures we turn to has this recursive property.

\subsubsection{Eigenvector Centrality}\label{sss:eigcen}

Let $\gG = (V, E, w)$ be a weighted digraph with adjacency matrix $A$.
Recalling that $r(A)$ is the spectral radius of $A$, the \navy{hub-based eigenvector
centrality}\index{Eigenvector centrality} of $\gG$ is
defined as the $e \in \RR^n_+$ that solves
\begin{equation}\label{eq:eicen0}
    e = \frac{1}{r(A)} A e.
\end{equation}
Element-by-element, this is
\begin{equation}\label{eq:eicen}
    e_i = \frac{1}{r(A)} \sum_{j \in \natset{n}} a_{ij} e_j
    \qquad  \text{for all } i \in \natset{n}.
\end{equation}
Note the recursive nature of the definition:
 the centrality obtained by vertex $i$ is proportional to a
sum of the centrality of all vertices, weighted by the ``rates of flow'' from
$i$ into these vertices.   A vertex $i$ is highly ranked if (a) there are many
edges leaving $i$, (b) these edges have large weights, and (c) the edges
point to other highly ranked vertices.

When we study demand shocks in \S\ref{ss:dshocks}, we will provide a more
concrete interpretation of eigenvector centrality.  We will see that, in
production networks, sectors with high hub-based eigenvector centrality are
important \emph{suppliers}.  In particular, they are activated by a wide array
of demand shocks once orders flow backwards through the network.

\begin{Exercise}
    Show that~\eqref{eq:eicen} has a unique solution, up to a positive scalar
    multiple, whenever $A$ is strongly connected.
\end{Exercise}

\begin{Answer}
    When $A$ is strongly connected, the Perron--Frobenius theorem tells us that
    $r(A)>0$ and $A$ has a unique (up to a scalar multiple) dominant right
    eigenvector satisfying $r(A) e = A e$.  Rearranging
    gives~\eqref{eq:eicen}.\footnote{While the dominant eigenvector is only
        defined up to a positive scaling constant, this is no reason for
        concern, since positive scaling has no impact on the ranking.  In most
        cases, users of this centrality ranking choose the dominant
    eigenvector $e$ satisfying $\| e \| = 1$.}
\end{Answer}

As the name suggests, hub-based eigenvector centrality is a measure of hub
centrality: vertices are awarded high rankings when they \emph{point to}
important vertices.  The next two exercises help to reinforce this point.

\begin{Exercise}\label{ex:sinkev}
    Show that nodes with zero out-degree always have zero hub-based
    eigenvector centrality. 
\end{Exercise}

To compute eigenvector centrality when the adjacency matrix $A$ is primitive,
we can employ the Perron--Frobenius Theorem, which tells us that 
$r(A)^{-m} A^m \to e \, \epsilon^\top$ as $m \to \infty$,
where $\epsilon$ and $e$ are the dominant left and right eigenvectors
of $A$.  This implies 
\begin{equation}\label{eq:ramm}
    r(A)^{-m} A^m \1 \to c e 
    \quad \text{where } c := \epsilon^\top \1.
\end{equation}
Thus, evaluating $r(A)^{-m} A^m \1$ at large $m$ returns a scalar multiple of
$e$.  The package Networkx provides a function for computing eigenvector
centrality via \eqref{eq:ramm}.

One issue problem with this method is the assumption of primitivity,
since the convergence in \eqref{eq:ramm} can fail without it.
The following function uses an alternative technique, based on Arnoldi
iteration, which generally works even when primitivity fails.
(The \texttt{authority} option is explained below.) 

\begin{minted}{python}
import numpy as np
from scipy.sparse import linalg

def eigenvector_centrality(A, m=40, authority=False):
    """
    Computes and normalizes the dominant eigenvector of A.  
    """
    A_temp = A.T if authority else A
    r, vec_r = linalg.eigs(A_temp, k=1, which='LR')
    e = vec_r.flatten().real
    return e / np.sum(e)
\end{minted}

\begin{Exercise}\label{ex:ha0}
    Show that the digraph in Figure~\ref{f:hub_vs_authorith} is not primitive.
    Using the code above or another suitable routine, compute the hub-based
    eigenvector centrality rankings. You should obtain values close to $e =
    (0.3694,0.2612,0.3694,0)$.  Note that the sink vertex (vertex 4) obtains
    the lowest
    rank.
\end{Exercise}

\begin{figure}
   \begin{center}
    \input{tikz/hub_vs_authorith.tex}
    \caption{\label{f:hub_vs_authorith} A network with a source and a sink}
   \end{center}
\end{figure}
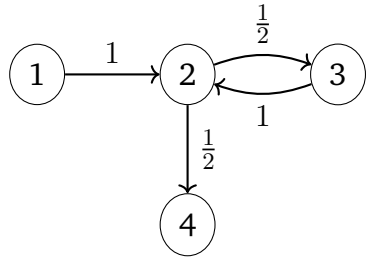

The middle left panel of Figure~\ref{f:financial_network_analysis_centrality}
shows the hub-based eigenvector centrality ranking for the international
credit network shown in
Figure~\ref{f:financial_network_analysis_visualization}.  Countries that are rated
highly according to this rank tend to be important players in terms of supply
of credit.  Japan takes the highest rank according to this measure, although
countries with large financial sectors such as Great Britain and France are
not far behind.  (The color scheme in
Figure~\ref{f:financial_network_analysis_visualization} is also matched to 
hub-based eigenvector centrality.)

The \navy{authority-based eigenvector centrality}\index{Eigenvector
centrality} of $\gG$ is defined as the $e \in \RR^n_+$ 
solving
\begin{equation}\label{eq:eicena0}
    e = \frac{1}{r(A)} A^\top e.
\end{equation}
The difference between~\eqref{eq:eicena0} and~\eqref{eq:eicen} is just transposition of $A$.
(Transposes do not affect the spectral radius of a matrix.)
Element-by-element, this is
\begin{equation}\label{eq:eicena}
    e_j = \frac{1}{r(A)} \sum_{i \in \natset{n}} a_{ij} e_i
    \qquad  \text{for all } j \in \natset{n}.
\end{equation}
We see $e_j$ will be high if many nodes with high authority
rankings link to $j$.

The middle right panel of Figure~\ref{f:financial_network_analysis_centrality}
shows the authority-based eigenvector centrality ranking for the international
credit network shown in
Figure~\ref{f:financial_network_analysis_visualization}.  Highly ranked
countries are those that attract large inflows of credit, or credit inflows
from other major players.  The US clearly dominates the rankings as a target
of interbank credit.  

\begin{Exercise}
    Assume that $A$ is strongly connected.  Show that authority-based eigenvector
    centrality is uniquely defined up to a positive scaling constant and
    equal to the dominant \emph{left} eigenvector of $A$.
\end{Exercise}

\subsubsection{Katz Centrality}\label{sss:hbkc}

Eigenvector centrality can be problematic.  Although the definition in
\eqref{eq:eicen} makes sense when $A$ is strongly connected (so that, by the
Perron--Frobenius theorem, $r(A) > 0$), strong connectedness fails in many real
world networks. We will see examples of this in \S\ref{ss:mutmod}, for
production networks defined by input-output matrices.

In addition, while strong connectedness yields strict positivity of the
dominant eigenvector, many vertices can be assigned a zero ranking when it
fails (see, e.g., Exercise~\ref{ex:sinkev}).   This zero ranking often runs
counter to our intuition when we examine specific networks.

Considerations such as these encourage use of an alternative notion of
centrality for networks called Katz centrality, originally due to
\cite{katz1953new}, which is positive under weaker conditions and uniquely
defined up to a tuning parameter.  Fixing $\beta$ in $(0, 1/r(A))$, the
\navy{hub-based Katz centrality}\index{Katz centrality} of weighted digraph
$\gG$ with adjacency matrix $A$, at parameter $\beta$, is defined as the
vector $\kappa := \kappa(\beta, A) \in \RR^n_+$ that solves
\begin{equation}\label{eq:katz}
    \kappa_i =  \beta \sum_{j \in \natset{n}} a_{ij} \kappa_j + 1
    \qquad  \text{for all } i \in \natset{n}.
\end{equation}
The intuition is very similar to that provided for eigenvector centrality: 
high centrality is conferred on $i$ when it is linked to by
vertices that themselves have high centrality.  The difference between
\eqref{eq:katz} and \eqref{eq:eicen} is just in the additive constant $1$.

\begin{Exercise}
    Show that, under the stated condition $0 < \beta < 1/r(A)$, hub-based Katz
    centrality is always finite and
    uniquely defined by
    \begin{equation}\label{eq:katzhub}
        \kappa 
        = (I - \beta A)^{-1} \1
        = \sum_{\ell \geq 0} (\beta A)^\ell \1,
    \end{equation}
    where $\1$ is a column vector of ones.
\end{Exercise}

\begin{Answer}
    When $\beta < 1/r(A)$ we have $r(\beta A) < 1$.  Hence,
     we can express~\eqref{eq:katz} as
    $\kappa = \1 + \beta A \kappa$ and employ the Theorem~\ref{t:nsl} to
    obtain the stated result.
\end{Answer}

\begin{Exercise}
    We know from the Perron--Frobenius theorem that the eigenvector centrality
    measure will be everywhere positive when the digraph is strongly
    connected.  A condition weaker than strong connectivity is that every
    vertex has positive out-degree.  Show that the Katz measure of centrality
    is strictly positive on each vertex under this condition.
\end{Exercise}

The attenuation parameter $\beta$ is used to ensure that $\kappa$ is finite
and uniquely defined under the condition $0 < \beta < 1/r(A)$.  It can be
proved that, when the graph is strongly connected, hub-based (resp.,
authority-based) Katz centrality converges to the hub-based (resp.,
authority-based) eigenvector centrality as $\beta \uparrow
1/r(A)$.\footnote{See, for example, \cite{benzi2015limiting}.} This is why, in
the bottom two panels of Figure~\ref{f:financial_network_analysis_centrality},
the hub-based (resp., authority-based) Katz centrality ranking is seen to be
close to its eigenvector-based counterpart.

When $r(A)<1$, we use $\beta=1$ as the default for Katz centrality computations.

\begin{Exercise}\label{ex:ha}
    Compute the hub-based Katz centrality rankings for the simple digraph
    in Figure~\ref{f:hub_vs_authorith} when $\beta=1$.   You should obtain $\kappa = (5, 4,
    5, 1)$.  Hence, the source vertex (vertex 1) obtains equal highest rank
    and the sink vertex (vertex 4) obtains the lowest rank.
\end{Exercise}

Analogously, the \navy{authority-based Katz centrality}\index{Katz centrality}
 of $\gG$ is defined as
the $\kappa \in \RR^n_+$ that solves
\begin{equation}\label{eq:katza}
    \kappa_j =  \beta \sum_{i \in \natset{n}} a_{ij} \kappa_i + 1
    \qquad  \text{for all } j \in \natset{n}.
\end{equation}

\begin{Exercise}
    Show that, under the restriction $0 < \beta < 1/r(A)$, 
    the unique solution to \eqref{eq:katza} is given by 
    \begin{equation}\label{eq:katzav}
        \kappa = (I - \beta A^\top)^{-1} \1
        \quad \iff \quad
        \kappa^\top = \1^\top (I - \beta A)^{-1}.
    \end{equation}
    (Verify the stated equivalence.)
\end{Exercise}

\begin{Exercise}\label{ex:ha2}
    Compute the authority-based Katz centrality rankings for the digraph
    in Figure~\ref{f:hub_vs_authorith} when $\beta=1$.   You should obtain
    $\kappa = (1, 6, 4, 4)$.  Notice that the source vertex now has the lowest
    rank. This is due to the fact that hubs are devalued relative to authorities.
\end{Exercise}

\subsection{Scale-Free Networks}\index{Scale-free networks}

What kinds of properties do large, complex networks typically possess?  One of
the most striking facts about complex networks is that many exhibit the
\navy{scale-free} property, which means, loosely speaking, that the number of
connections possessed by each vertex in the network follows a power law.  The
scale-free property is remarkable because it holds for a wide variety of
networks, from social networks to citation, sales, financial and production
networks, each of which is generated by different underlying
mechanisms. Nonetheless, they share this specific statistical structure.

We begin this section by defining the degree distribution and then discuss its properties,
including possible power law behavior.

\subsubsection{Empirical Degree Distributions}

Let $\gG = (V, E)$ be a digraph.
Assuming without loss of generality that $V = \natset{n}$ for some $n \in
\NN$, the \navy{in-degree distribution}\index{Degree distribution} of $G$ is
the sequence $(\phi_{in}(k))_{k=0}^n$ defined by
\begin{equation}\label{eq:degdi}
    \phi_{in}(k) = \frac{\sum_{v \in V} \1\{ i_d(v)=k\} }{n}
    \qquad (k = 0, \ldots, n),
\end{equation}
where $i_d(v)$ is the in-degree of vertex $v$. In other words, the in-degree
distribution evaluated at $k$ is the fraction of nodes in the network that
have in-degree $k$.   In Python, when $\gG$ is
expressed as a Networkx \texttt{DiGraph} called \texttt{G} and
\mintinline{python}{import numpy as np} has been executed, the in-degree
distribution can be calculated via
\begin{minted}{python}
def in_degree_dist(G):
    n = G.number_of_nodes()
    iG = np.array([G.in_degree(v) for v in G.nodes()])
    phi = [np.mean(iG == k) for k in range(n+1)]
    return phi
\end{minted}

The \navy{out-degree distribution} is defined analogously, replacing $i_d$ with
$o_d$ in \eqref{eq:degdi}, and denoted by $(\phi_{out}(k))_{k=0}^n$.

Recall that a digraph $\gG = (V, E)$ is called undirected if $(u, v) \in E$
implies $(v, u) \in E$. If $\gG$ is undirected, then $i_d(v) = o_d(v)$ for all
$v \in V$.
In this case we usually write $\phi$ instead of $\phi_{in}$ or $\phi_{out}$ and refer
simply to the \navy{degree-distribution} of the digraph.

A \navy{scale-free network}\index{Scale-free network} is a network whose
degree distribution obeys a power law, in the sense that there exist positive
constants $c$ and $\gamma$ with 
\begin{equation}\label{eq:apla}
    \phi(k) \approx c k^{-\gamma}
    \quad \text{for large } k.
\end{equation}
Here $\phi(k)$ can refer to in-degree or the out-degree (or both), depending on our
interest.  In view of the discussion in \S\ref{sss:dpls}, this can be
identified with the idea that the degree distribution is Pareto-tailed with tail
index $\alpha = \gamma -1$.

Although we omit formal tests, the degree distribution for the commercial
aircraft international trade network shown in
Figure~\ref{f:commercial_aircraft_2019_1} on
page~\pageref{f:commercial_aircraft_2019_1} is approximately scale-free.
Figure~\ref{f:commercial_aircraft_2019_2} illustrates this by plotting the
degree distribution  alongside $f(x) = c x^{-\gamma}$ with $c=0.2$ and $\gamma
= 1.1$.  (In this calculation of the degree distribution, performed by the Networkx function
\texttt{degree\_histogram}, directions are ignored and the network is treated as
an undirected graph.)

\begin{figure}
   \centering
   \scalebox{0.65}{\includegraphics[trim = 0mm 0mm 0mm 0mm, clip]{ 
       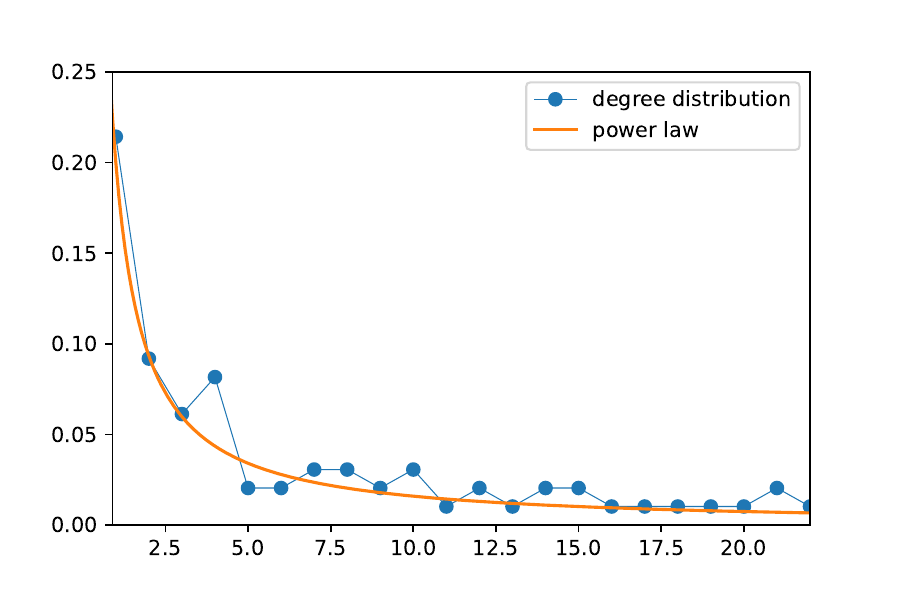}}
   \caption{\label{f:commercial_aircraft_2019_2} Degree distribution for international aircraft trade}
\end{figure}

Attention was drawn to the scale-free nature of many networks by
\cite{barabasi1999emergence}. They found, for example, that the in-degree and out-degree
distributions for internet pages connected by hyperlinks both follow power
laws.   In subsequent years, many networks have been found to have the
scale-free property, up to a first approximation, including networks of
followers on Twitter \citep{pearce2017triangle, punel2018using}, other social
networks \citep{rybski2009scaling} and academic collaboration networks (e.g.,
papers plus citations). 

Within economics and finance, \cite{carvalho2014micro} shows that the weighted
out-degree distribution for US input-output data (discussed further in
Chapter~\ref{c:prod}) obeys a power law, as does the Katz centrality measure.
\cite{carvalho2021supply} document power law tails for the in-degree
(suppliers) and out-degree (customers) distributions in a Japanese network of
interacting firms.  Scale-free degree distributions have also been observed in
a number of financial and inter-bank credit networks \citep{kim2007characteristics, ou2007power,
de2011analysis}.

In many cases, the scale-free property of a given network has significant
implications for economic outcomes and welfare.  For example, a power law in
input-output networks often typically indicates dominance by a small number of
very large sectors or firms.  This in turn affects both the dynamism of
industry and the likelihood of aggregate instability caused by firm-level
shocks.  We explore some of these issues in Chapter~\ref{c:prod}.

\subsubsection{Random Graphs}

One way to explore the implications of different dynamics for the degree
distribution of graphs is to specify a law for generating graphs randomly and
then examine the degree distribution that results.  This methodology 
leads to insights on the kinds of mechanisms that can generate scale-free networks.

We begin with one of the most popular and elementary ways of randomly
generating an undirected graph, originally examined by \cite{erdos1960evolution}. The 
process to generate a graph $\gG = (V, E)$ is 
\begin{enumerate}
    \item fix an integer $n \in \NN$ and a $p \in (0, 1)$,
    \item view $V := \natset{n}$ as a collection of vertices, 
    \item let $E = \{\emptyset\}$, and
    \item for each $(i, j) \in V \times V$ with $i \neq j$, add the undirected edge $\{i,j\}$
        to the set of edges $E$ with probability $p$.
\end{enumerate}

In the last step additions are independent---each time, we flip an unbiased
{\sc iid} coin with head probability $p$ and add the edge if the coin comes up
heads.  

The Python code below provides a function that can be called to randomly generate
an undirected graph using this procedure.
It applies the \texttt{combinations}
function from the \texttt{itertools} library, which, for the call
\texttt{combinations(A, k)}, returns a list of all subsets of $A$ of size $k$.
For example,
\begin{minted}{python}
import itertools
letters = 'a', 'b', 'c'
list(itertools.combinations(letters, 2))    
\end{minted}
returns \texttt{[('a', 'b'), ('a', 'c'), ('b', 'c')]}.

We use \texttt{combinations} to produce the set of all possible edges and then
add them to the graph with probability $p$:

\begin{minted}{python}
def erdos_renyi_graph(n=100, p=0.5, seed=1234):
    "Returns an Erdos-Renyi random graph."
    np.random.seed(seed)
    edges = itertools.combinations(range(n), 2)
    G = nx.Graph()
    
    for e in edges:
        if np.random.rand() < p:
            G.add_edge(*e)
    return G
\end{minted}

(The code presented here is a simplified version of functionality provided by
the library Networkx.  It is written for clarity rather than efficiency.  More
efficient versions can be found both in Networkx and in Julia's Graphs.jl library.)

The left hand side of Figure~\ref{f:rand_graph_experiments_1} shows one
instance of a graph that was generated by the \texttt{erdos\_renyi\_graph} function, with
$n=100$ and $p=0.05$. Lighter colors on a node indicate higher degree
(more connections).  The right hand side shows the degree distribution, which
exhibits a bell-shaped curve typical for Erdos--Renyi random graphs.
In fact one can show (see, e.g., \cite{bollobas1999linear} or
\cite{durrett2007random}) that the degree distribution is binomial, with
\begin{equation*}
    \phi(k) = \binom{n-1}{k} p^k (1-p)^{n-1-k}
    \qquad (k = 0, \ldots, n-1).
\end{equation*}

\begin{figure}
   \centering
   \scalebox{0.6}{\includegraphics[trim = 0mm 0mm 0mm 0mm, clip]{
   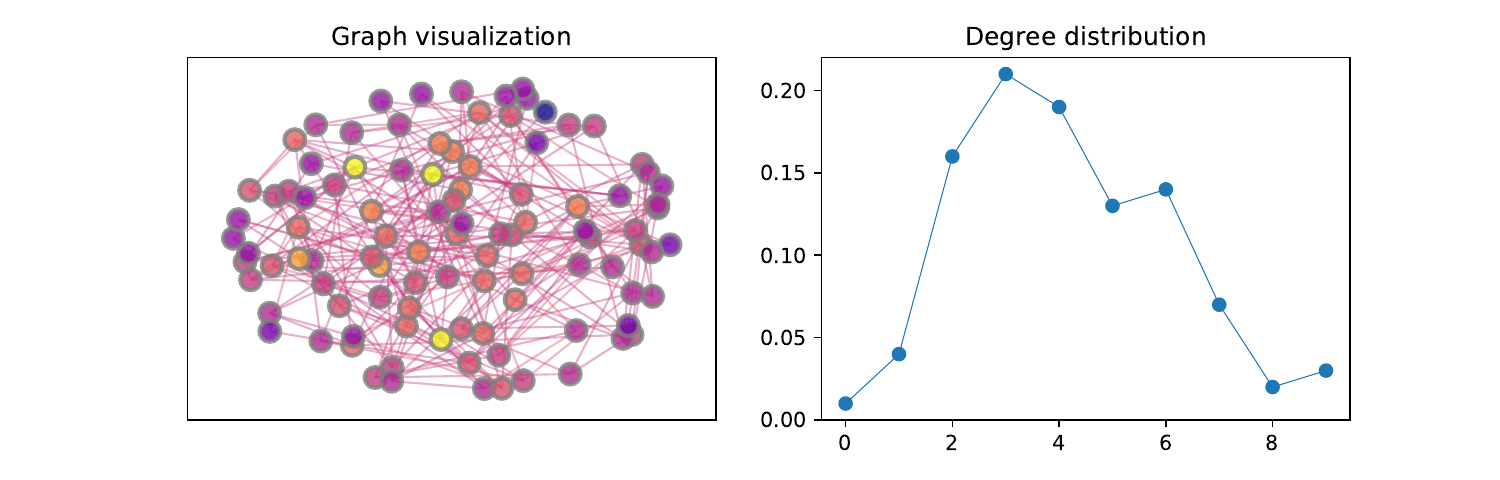}}
   \caption{\label{f:rand_graph_experiments_1} An instance of an Erdos--Renyi random graph}
\end{figure}

\subsubsection{Preferential Attachment}

Clearly Erdos--Renyi random graphs fail to replicate the heavy right hand tail
of the degree distribution observed in many networks.  In response to this,
\cite{barabasi1999emergence} proposed a mechanism for randomly generating
graphs that feature the scale-free property.

The stochastic mechanism they proposed is called \navy{preferential
attachment}\index{Preferential attachment}.  In essence, each time a new
vertex is added to an undirected graph, it is attached by edges to $m$ of the
existing vertices, where the probability of vertex $v$ being selected is
proportional to the degree of $v$.  \cite{barabasi1999emergence} showed that
the resulting degree distribution exhibits a Pareto tail in the limit, as the
number of vertices converges to $+\infty$.   A careful proof can be found in
Chapter~4 of \cite{durrett2007random}.

Although we omit details of the proof, we can see the power law emerge in
simulations.  For example, Figure~\ref{f:rand_graph_experiments_2} shows a
random graph with 100 nodes generated by Networkx's
\texttt{barabasi\_albert\_graph} function.  The number of attachments $m$ is
set to 5. The simulated degree distribution on the right hand side of
Figure~\ref{f:rand_graph_experiments_2} already exhibits a long right tail.

\begin{figure}
   \centering
   \scalebox{0.6}{\includegraphics[trim = 0mm 0mm 0mm 0mm, clip]{
   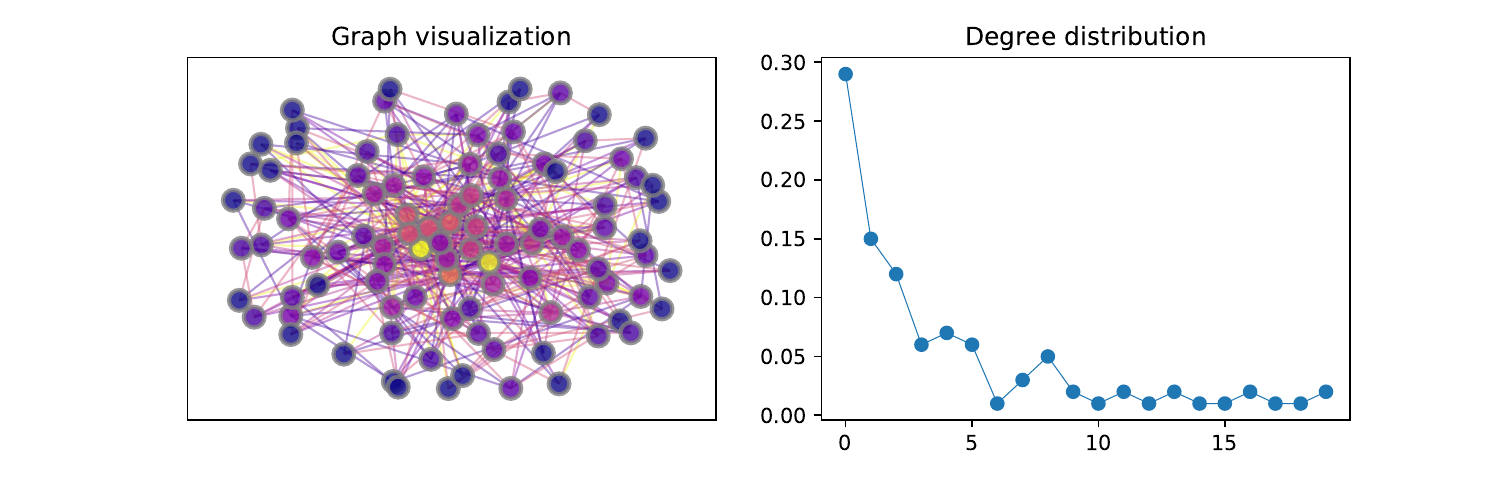}}
   \caption{\label{f:rand_graph_experiments_2} An instance of a preferential attachment random graph}
\end{figure}

The preferential attachment model is popular not just because it replicates
the scale-free property of many real-world networks, but also because its
mechanism is simple and plausible.  For example, in citation networks, we can 
imagine that a well-cited paper is more likely to attract additional citations
than a poorly-cited paper.  Similar intuition can be applied to an individual on a
social network, where the number of links is measured in terms of the number of followers.

\section{Chapter Notes}\label{s:cnni}

The Perron--Frobenius theorem is due to Oskar Perron (1880--1975) and
Ferdinand Georg Frobenius (1849--1917).  The main results were proved by
1912.  As early as 1915, D\'enes K\"onig (1884--1944) saw the connection
       between the Perron--Frobenius theorem and graph theory, and provided an
       alternative proof using bipartite graphs. Some of the history is
       discussed in \cite{schrijver2005history}.

We have already mentioned the textbooks on economic and social networks by
\cite{jackson2010social}, \cite{easley2010networks} and
\cite{borgatti2018analyzing}, as well as the handbook by
\cite{bramoulle2016oxford}.   \cite{jackson2014networks} gives a survey of the
literature.  Within the realm of network science, the high
level texts by \cite{newman2018networks}, \cite{menczer2020first} and
\cite{coscia2021atlas} are excellent.

One good text on graphs and graph-theoretic algorithms is
\cite{kepner2011graph}.  \cite{ballester2006s} provide an interpretation of
Katz centrality (which they call Bonacich centrality) in terms of Nash
equilibria of quadratic games. \cite{du2015competitive} show how PageRank can
be obtained as a competitive equilibrium of an economic problem.
\cite{calvo2009peer} develop a model in which the outcomes for agents embedded
in a network are proportional to the Katz centrality.
\cite{elliott2019network} show that, in a setting where agents can create
nonrival, heterogeneous public goods, an important set of efficient solutions
are characterized by contributions being proportional to agents' eigenvector
centralities in the network.  

\cite{kumamoto2018power} provide a detailed survey of power laws in
economics and social sciences, including a discussion of the preferential
attachment model of \cite{barabasi1999emergence}.  \cite{newman2005power} is
also highly readable.  The textbook of \cite{durrett2007random} is rigorous,
carefully written and contains interesting motivational background, as well as
an extensive citation list for studies of scale-free networks.

It should be clear from the symbol $\approx$ in \eqref{eq:apla} that the
definition of scale-free networks is not entirely rigorous.  Moreover, when
connecting the definition to observed networks, we cannot obtain complete
clarity by taking a limit, as we did when we defined power laws in
\S\ref{ss:pow}, since the number of vertices is always finite. This
imprecision in the definition has led to heated debate (see, e.g.,
\cite{holme2019rare}).  Given the preponderance of positive empirical studies,
we take the view that, up to a reasonable degree of approximation, the
scale-free property is remarkably widespread.

In \S\ref{sss:quadgame} we briefly mentioned network games, social networks
and key players.  These topics deserve more attention than we have been able
to provide.  An excellent overview is given in \cite{zenou2016key}.
\cite{amarasinghe2020key} apply these ideas to problems in economic
development.  Valuable related papers include 
\cite{allouch2015private}, \cite{belhaj2016efficient},
\cite{demange2017optimal}, \cite{belhaj2019group},
\cite{galeotti2020targeting}.

Another topic we reluctantly omitted in order to keep the textbook short is
endogenous network formation in economic environments.  Influential papers in
this field include \cite{bala2000noncooperative}, \cite{watts2001dynamic},
\cite{graham2017econometric}, \cite{galeotti2010law}, \cite{hojman2008core},
and \cite{jackson1996strategic}.

Finally, \cite{candogan2012optimal} study the profit maximization problem
for a monopolist who sells items to participants in a social network.  The
main idea is that, in certain settings, the monopolist will find it profitable
to offer discounts to key players in the network. \cite{atalay2011network}
argue that in-degrees observed in US buyer-supplier networks have lighter
tails than a power law, and supply a model that better fits their data.

%% file: tikz/worker_switching.tex
\begin{tikzpicture}
  \node[circle, draw] (1) at (0, 0) {unemployed};
  \node[circle, draw] (2) at (3, 0) {employed};
  \draw[->, thick, black]
  (1) edge [bend left=20, above] node {$\alpha$} (2)
  (2) edge [bend left=20, below] node {$\beta$} (1)
  (2) edge [loop below] node {$1-\beta$} (2)
  (1) edge [loop above] node {$1-\alpha$} (1);
\end{tikzpicture}

%% file: tikz/rich_poor_no_label.tex
\begin{tikzpicture}
  \node[ellipse, draw] (0) at (0, 0) {poor};
  \node[ellipse, draw] (1) at (3, 1) {middle class};
  \node[ellipse, draw] (2) at (6, 0) {rich};
  
  \draw[->, thick, black]
  (0) edge [bend left=10, above] node {} (1)
  (1) edge [bend left=10, above] node {} (0)
  (2) edge [bend left=10, above] node {} (0)
  (1) edge [bend left=10, above] node {} (2)
  (2) edge [bend left=10, above] node {} (1)
  (0) edge [loop above] node {}(0)
  (1) edge [loop above] node {} (1) 
  (2) edge [loop above] node {} (2);
\end{tikzpicture}

%% file: tikz/poverty_trap.tex
\begin{tikzpicture}
  \node[ellipse, draw] (0) at (0, 0) {middle class};
  \node[ellipse, draw] (1) at (3, 0) {rich};
  \node[ellipse, draw] (2) at (5, 1.25) {poor};
  \draw[->, thick, black]
  (0) edge [bend left=10, above] node {} (1)
  (0) edge [bend left=10, above] node {} (2)
  (1) edge [bend left=10, below] node {} (0)
  (1) edge [bend left=0, below] node {} (2)
  (0) edge [loop above] node {} (0)
  (1) edge [loop above] node {} (1)
  (2) edge [loop above] node {} (2) ;
\end{tikzpicture}

%% file: tikz/strong_connected_components.tex
\begin{tikzpicture}
  \node[ellipse, draw] (1) at (0, 0) {1};
  \node[ellipse, draw] (2) at (2, 0) {2};
  \node[ellipse, draw] (3) at (4, 0) {3};
  \draw[->, thick, black]
  (2) edge [bend left=0, above] node {} (1)
  (2) edge [bend left=20, above] node {} (3)
  (3) edge [bend left=20, above] node {} (2)
  (1) edge [loop above] node {} (1);
  
  \draw[blue,thick,dashed] (-0.5,-1) -- (-0.5,1) -- (0.5,1) -- (0.5,-1) -- (-0.5,-1);
  \draw[blue,thick,dashed] (1.5,-1) -- (1.5,1) -- (4.5,1) -- (4.5,-1) -- (1.5,-1);
\end{tikzpicture}

%% file: tikz/periodic_mc.tex
\begin{tikzpicture}
  \node[ellipse, draw] (1) at (0, 0) {a};
  \node[ellipse, draw] (2) at (2, 0) {b};
  \node[ellipse, draw] (3) at (4, 0) {c};
  \node[ellipse, draw] (4) at (6, 0) {d};
  \draw[->, thick, black]
  (1) edge [bend left=20, above] node {} (2)
  (2) edge [bend left=20, below] node {} (1)
  (2) edge [bend left=20, above] node {} (3)
  (3) edge [bend left=20, below] node {} (2)
  (3) edge [bend left=20, above] node {} (4)
  (4) edge [bend left=20, below] node {} (3);
\end{tikzpicture}

%% file: tikz/rich_poor.tex
\begin{tikzpicture}
  \node[ellipse, draw] (0) at (0, 0) {poor};
  \node[ellipse, draw] (1) at (3, 2) {middle class};
  \node[ellipse, draw] (2) at (6, 0) {rich};
  
  \draw[->, thick, black]
  (0) edge [bend left=10, above] node {$0.1$} (1)
  (1) edge [bend left=10, below] node {$0.4$} (0)
  (2) edge [bend left=10, below] node {$0.1$} (0)
  (1) edge [bend left=10, above] node {$0.2$} (2)
  (2) edge [bend left=10, below] node {$0.1$} (1)
  (0) edge [loop above] node {$0.9$}(0)
  (1) edge [loop above] node {$0.4$} (1) 
  (2) edge [loop above] node {$0.8$} (2);
\end{tikzpicture}

%% file: tikz/network_liabfin.tex
\begin{tikzpicture}
  \node[circle, draw] (1) at (2.5, 3) {1};
  \node[circle, draw] (2) at (-1, 2) {2};
  \node[circle, draw] (3) at (-2, -0.5) {3};
  \node[circle, draw] (4) at (1.5, -1) {4};
  \node[circle, draw] (5) at (3.5, 0) {5};
  \draw[->, thick, black]
  (1) edge [bend left=20, below] node {$100$} (2)
  (2) edge [bend left=20, above] node {$50$} (1)
  (2) edge [bend right=20, left] node {$200$} (3)
  (3) edge [bend right=20, below] node {$100$} (4)
  (4) edge [bend right=20, right] node {$500$} (2)
  (5) edge [bend right=20, below left] node {$250$} (3)
  (5) edge [bend left=30, below] node {$300$} (4)
  (4) edge [bend left=30, below] node {$50$} (5)
  (5) edge [bend right=30, right] node {$150$} (1);
\end{tikzpicture}

%% file: tikz/network_liabfin_trans.tex
\begin{tikzpicture}
  \node[circle, draw] (1) at (2.5, 3) {1};
  \node[circle, draw] (2) at (-1, 2) {2};
  \node[circle, draw] (3) at (-2, -0.5) {3};
  \node[circle, draw] (4) at (1.5, -1) {4};
  \node[circle, draw] (5) at (3.5, 0) {5};
  \draw[<-, thick, black]
  (1) edge [bend left=20, below] node {$100$} (2)
  (2) edge [bend left=20, above] node {$50$} (1)
  (2) edge [bend right=20, left] node {$200$} (3)
  (3) edge [bend right=20, below] node {$100$} (4)
  (4) edge [bend right=20, right] node {$500$} (2)
  (5) edge [bend right=20, below left] node {$250$} (3)
  (5) edge [bend left=30, below] node {$300$} (4)
  (4) edge [bend left=30, below] node {$50$} (5)
  (5) edge [bend right=30, right] node {$150$} (1);
\end{tikzpicture}

%% file: tikz/io_reducible.tex
\begin{tikzpicture}
  \node[circle, draw] (1) at (-1, 3) {1};
  \node[circle, draw] (2) at (-3, 0) {2};
  \node[circle, draw] (3) at (1, -2) {3};
  \node[circle, draw] (4) at (0, 0) {4};
  \draw[->, thick, black]
  (1) edge [bend left=50, right] node {$a_{12}$} (2)
  (2) edge [bend left=50, left] node {$a_{21}$} (1)
  (2) edge [bend right=20, below] node {$a_{23}$} (3)
  (3) edge [bend right=50, right] node {$a_{31}$} (1)
  (4) edge [bend right=10, left] node {$a_{43}$} (3)
  (1) edge [loop above] node {$a_{11}$} (1);
\end{tikzpicture}

%% file: tikz/hub_and_authority.tex
\begin{tikzpicture}
  \node[circle, draw, radius=6] (1) at (0, 0) {$  $hub$  $};
  \node[circle, draw, radius=2] (2) at (4, 0) {authority};
  \node[ellipse, draw, scale=0.8] (01) at (-1, 1.5) {};
  \node[ellipse, draw, scale=0.8] (02) at (-1.5, -1.5) {};
  \node[ellipse, draw, scale=0.8] (03) at (0, -2) {};
  \node[ellipse, draw, scale=0.8] (04) at (1.5, -1.5) {};
  \node[ellipse, draw, scale=0.8] (05) at (3, -1.5) {};
  \node[ellipse, draw, scale=0.8] (06) at (5, -1.5) {};
  \node[ellipse, draw, scale=0.8] (07) at (6, 0) {};
  \node[ellipse, draw, scale=0.8] (08) at (5, 1.5) {};
  \node[ellipse, draw, scale=0.8] (09) at (3, 1.5) {};
  \node[ellipse, draw, scale=0.8] (10) at (1.5, 1.5) {};

  \draw[->, thick, black]
  (1) edge [bend left=0, left, -{Stealth[scale=1]}, line width=1pt] node {}(01)
  (1) edge [bend left=0, below, -{Stealth[scale=1]}, line width=1pt] node {} (02)
  (1) edge [bend left=0, below, -{Stealth[scale=1]}, line width=1pt] node {} (03)
  (1) edge [bend left=0, below, -{Stealth[scale=1]}, line width=1pt] node {} (04)
  (05) edge [bend left=0, below, -{Stealth[scale=1]}, line width=1pt] node {} (2)
  (06) edge [bend left=0, below, -{Stealth[scale=1]}, line width=1pt] node {} (2)
  (06) edge [bend left=0, below, -{Stealth[scale=1]}, line width=1pt] node {} (2)
  (07) edge [bend left=0, below, -{Stealth[scale=1]}, line width=1pt] node {} (2)
  (08) edge [bend left=0, below, -{Stealth[scale=1]}, line width=1pt] node {} (2)
  (09) edge [bend left=0, below, -{Stealth[scale=1]}, line width=1pt] node {} (2)
  (1) edge [bend left=0, below, -{Stealth[scale=1]}, line width=1pt] node {} (10);
\end{tikzpicture}

%% file: tikz/hub_vs_authorith.tex
\begin{tikzpicture}
  \node[ellipse, draw] (1) at (0, 0) {1};
  \node[ellipse, draw] (2) at (2, 0) {2};
  \node[ellipse, draw] (3) at (4, 0) {3};
  \node[ellipse, draw] (4) at (2, -2) {4};
  \draw[->, thick, black]
  (1) edge [bend left=0, above] node {$1$} (2)
  (2) edge [bend left=20, above] node {$\frac{1}{2}$} (3)
  (3) edge [bend left=20, below] node {$1$} (2)
  (2) edge [bend left=0, right] node {$\frac{1}{2}$} (4);
\end{tikzpicture}

%% file: ch_production.tex
\chapter{Production}\label{c:prod}

In this chapter we study production in multisector environments.  The basic
framework is input-output analysis, which was initiated by Wassily Leontief
(1905--1999) and popularized in~\cite{leontief1941structure}. Input-output
analysis is routinely used to organize national accounts and study
inter-industry relationships. In 1973, Leontief received the Nobel Prize in
Economic Sciences for his work on input-output systems.

Input-output analysis is currently being incorporated into modern theories of
trade, growth, shock propagation and aggregate fluctuations in multisector
models (\S\ref{s:cnprod} provides a detailed list of references). One of the
reasons for renewed interest is that the introduction of concepts from network
analysis and graph theory has yielded new insights.  This chapter provides
an introduction to the main ideas.

\section{Multisector Models}\label{ss:mutmod}

In this section we introduce the basic input output model, explain the network
interpretation of the model, and connect traditional questions, such as the
relative impact of demand shocks across sectors, to network topology and
network centrality.

\subsection{Production Networks}

We begin with the foundational concepts of input-output tables and how they
relate to production networks.  To simplify the exposition, we ignore imports
and exports in what follows.  (References for the general case are discussed
in \S\ref{s:cnprod}.)

\subsubsection{Input-Output Analysis}\label{sss:prodnet}

Agencies tasked with gathering national and regional production accounts (such
as the US Bureau of Economic Analysis) compile input-output data based on
the structure set out by \cite{leontief1941structure}.  Firms are divided
across $n$ sectors, each of which produces a single homogeneous good.  These
sectors are organized into an input-output table, a highly simplified example
of which is 

\begin{center}
    \begin{tabular}{l|ccc}
                  & sector 1 & sector 2 & sector 3 \\
        \hline 
         sector 1 & $a_{11}$ & $a_{12}$ & $a_{13}$ \\
         sector 2 & $a_{21}$ & $a_{22}$ & $a_{23}$ \\
         sector 3 & $a_{31}$ & $a_{32}$ & $a_{33}$ \\
    \end{tabular}
\end{center}

Entries $a_{ij}$  are called the \navy{input-output coefficients}; 
\begin{equation*}
    a_{ij} = 
    \frac{\text{value of sector $j$'s inputs purchased from sector $i$}}
        {\text{total sales of sector $j$}}.
\end{equation*}

Thus, $a_{ij}$ is large if sector $i$ is an important supplier of intermediate
goods to sector $j$. The sum of the $j$-th column of
the table gives the value of all inputs to sector $j$.  The $i$-th row shows
how intensively each sector uses good $i$ as an intermediate good.

The \navy{production coefficient matrix} $A = (a_{ij})$ induces a weighted
digraph $\gG = (V, E, w)$, where $V = \natset{n}$ is the list of
sectors  and
\begin{equation*}
    E := \setntn{(i, j) \in V \times V}{a_{ij} > 0}
\end{equation*}
is the edge set.  The values $a_{ij}$ show backward linkages across sectors. 

Given $i \in V$, the set $\oO(i)$ of direct successors of $i$ is all sectors
to which $i$ supplies a positive quantity of output.  The set $\iI(i)$ is all
sectors that supply a positive quantity to $i$.

Figure~\ref{f:input_output_analysis_15} illustrates the weighted digraph associated
with the 15 sector version of the input-output tables provided by the Bureau
of Economic Analysis for the year
2019.  The data source is the US Bureau of Economic Analysis's
    2019 Input-Output Accounts Data.\footnote{We obtain input expenditures and total
         sales for each sector from the Make-Use Tables. The figure was
     created using Python's Networkx library.} An arrow from $i$ to $j$
     indicates a positive weight $a_{ij}$.  Weights are indicated by
     the widths of the arrows, which are proportional to the corresponding input-output
     coefficients.  The sector codes are provided in Table~\ref{tab:codes15}.
     The size of vertices is proportional to their share of total sales across
     all sectors.

A quick look at Figure~\ref{f:input_output_analysis_15} shows that
manufacturing (ma) is an important supplier for many sectors, including
construction (co) and agriculture (ag).  Similarly, the financial sector (fi) and 
professional services (pr) supply services to a broad range of sectors.
On the other hand, education (ed) is relatively downstream and only
a minor supplier of intermediate goods to other sectors.

The color scheme for the nodes is by hub-based eigenvector centrality, with
hotter colors indicating higher centrality.  Later, in~\S\ref{ss:dshocks} we will
give an interpretation of hub-based eigenvector centrality for this setting
that connects to relative impact of demand shocks.

\begin{figure}
   \begin{center}
    \scalebox{1.0}{\includegraphics[trim = 25mm 25mm 25mm 25mm, clip]{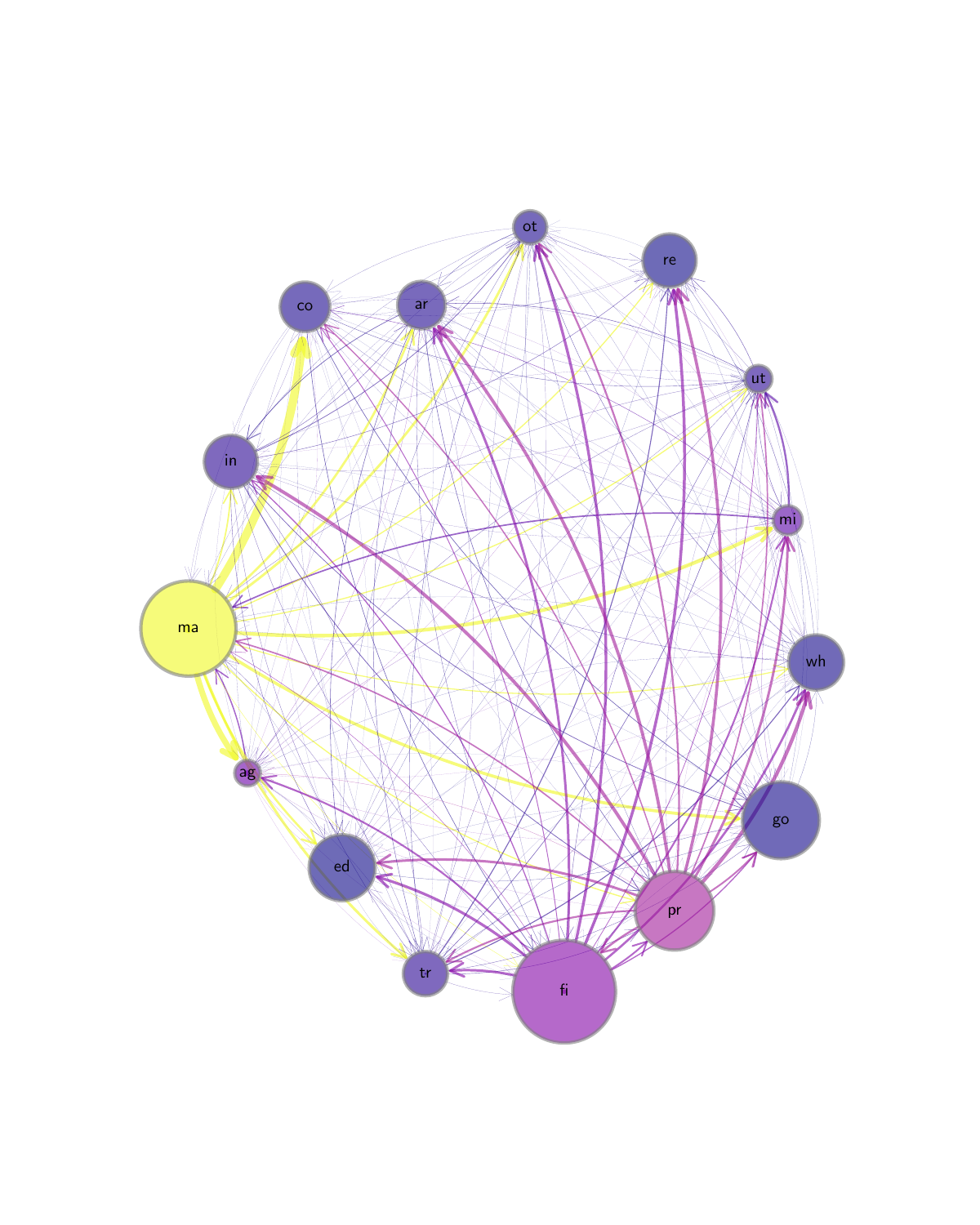}} \caption{\label{f:input_output_analysis_15} Backward linkages for 15 US sectors in 2019}
   \end{center}
\end{figure}

\begin{table}
    \centering
    \caption{\label{tab:codes15}Sector codes for the 15 good case}
    \begin{tabular}{ll}
        \hline 
         Label & Sector  \\ 
         \hline
         ag & Agriculture, forestry, fishing, and hunting \\
         mi &  Mining \\
         ut & Utilities \\
         co & Construction \\
         ma & Manufacturing \\
         wh & Wholesale trade \\
         re & Retail trade \\
         tr & Transportation and warehousing \\
         in & Information \\
         fi & Finance, insurance, real estate, rental, and leasing \\
         pr & Professional and business services \\
         ed & Educational services, health care, and social assistance \\
         ar & Arts, entertainment, accommodation, and food services \\
         ot & Other services, except government \\
         go & Government \\
        \hline 
    \end{tabular}
\end{table}

\subsubsection{Connectedness}

We will gain insights into input-output networks by applying some of the
graph-theoretic notions studied in Chapter~\ref{c:networks}.   One elementary
property we can investigate is connectedness.  We can imagine that demand and
productivity shocks diffuse more widely through a given production network
when the network is relatively connected.  Conversely, the impact of a demand
shock occurring within an absorbing set will be isolated to sectors in that
set.

The 15 sector network in Figure~\ref{f:input_output_analysis_15} is strongly
connected. Checking this visually is hard, so instead we use a graph-theoretic
algorithm that finds strongly connected components from
\href{https://quantecon.org/}{QuantEcon}'s \texttt{DiGraph} class.  
(This class is convenient for the current problem because instances are
created directly from the adjacency matrix.) Examining
the attributes of this class when the weights are given by the 15 sector
input-output model confirms its strong connectedness.  The same class can be used
to verify that the network is also aperiodic. Hence, the input-output
matrix $A$ is primitive.  This fact will be used in computations below.

\subsubsection{Disaggregation}

Figure~\ref{f:input_output_analysis_71} repeats the graphical representation
for the more disaggregated 71 sector case.  Sector codes are provided in
Table~\ref{tab:codes71}.  Input-output coefficients below $0.01$ were rounded
to zero to increase visual clarity.  As in the 15 sector case, the size of vertices and 
edges is proportional to share of sales and input-output coefficients
respectively.  Hotter colors indicate higher hub-based eigenvector centrality (which we
link to propagation of demand shocks in~\S\ref{ss:dshocks}).

\begin{figure}
   \begin{center}
    \scalebox{0.9}{\includegraphics[trim = 42mm 48mm 20mm 52mm, clip]{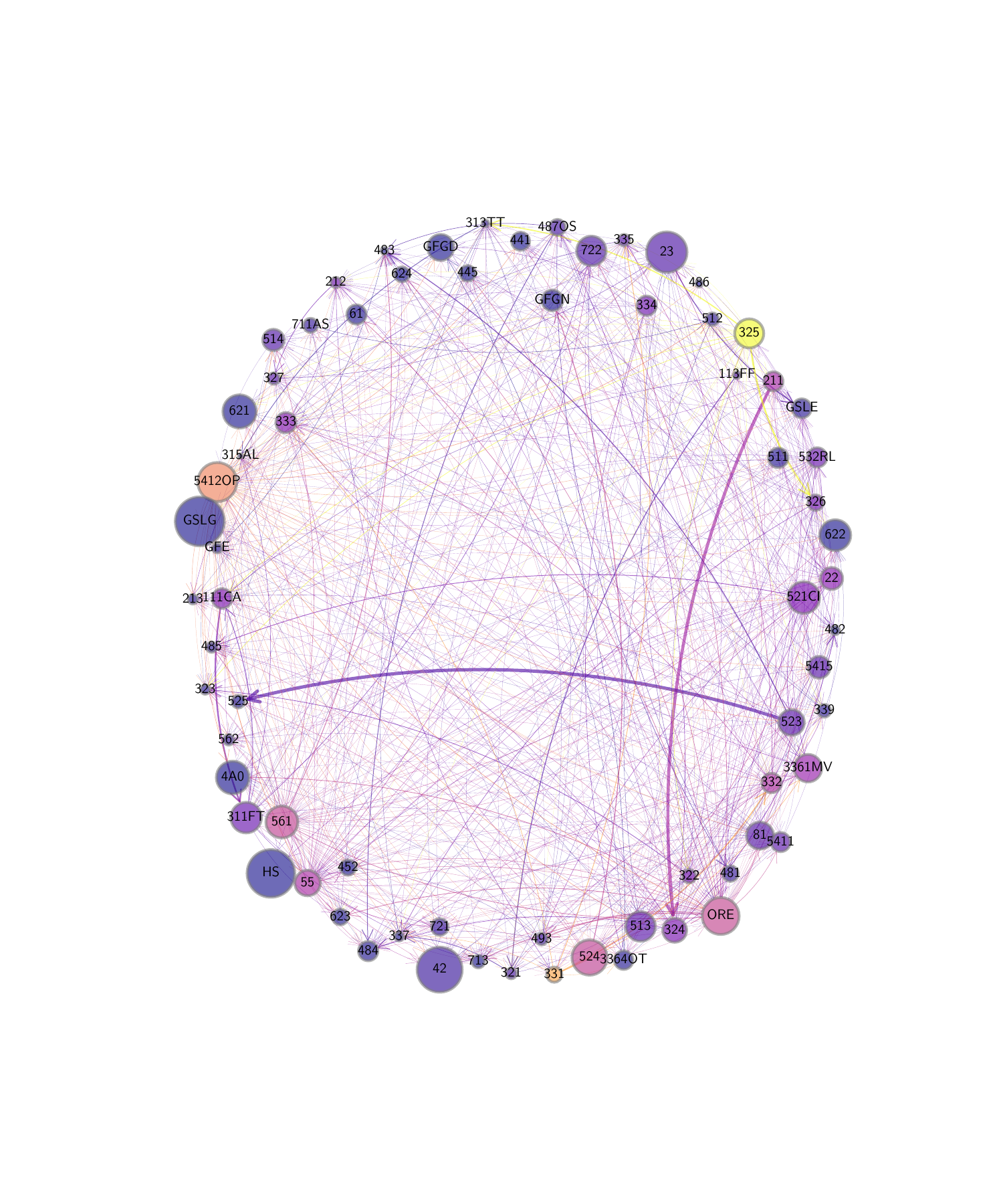}}
    \caption{\label{f:input_output_analysis_71} Network for 71 US sectors in 2019}
   \end{center}
\end{figure}

\begin{table}
    \centering
    \fontsize{9.5pt}{10.25pt}\selectfont
    \addtolength{\tabcolsep}{-2pt}
    \caption{\label{tab:codes71}Sector codes for the 71 good case}
    \centering
    \begin{tabular}{llll}
        \hline \hline 
        IO Code & Sector & IO Code & Sector  \\ \hline
        &  \\ 
        \vspace{0.3em}
        111CA  & Farms  & 486  & Pipeline
        transportation\\
        \vspace{0.3em}
        113FF  & Forestry, fishing & 487OS & Other transportation \\
        \vspace{0.3em}
        211  &  Oil and gas extraction & 493 & Warehousing and storage \\
        \vspace{0.3em}
        212  & Mining, except oil, gas & 511 & Publishing industries \\
        \vspace{0.3em}
        213  & Mining support activities & 512 & Motion picture and sound \\
        \vspace{0.3em}
        22  & Utilities & 513 & Broadcasting, telecommunications \\
        \vspace{0.3em}
        23  & Construction & 514 & Data processing, internet publishing \\
        \vspace{0.3em}
        321 & Wood products & 521CI & Reserve banks, credit intermediation \\
        \vspace{0.3em}
        327 & Nonmetallic mineral products & 523 & Securities and investments \\
        \vspace{0.3em}
        331 & Primary metals & 524  & Insurance carriers \\
        \vspace{0.3em}
        334 & Computer \& electronic products  & 525  & Funds, trusts, financial vehicles \\
        \vspace{0.3em}
        333 & Machinery & HS  &  Housing \\
        \vspace{0.3em}
        332 & Fabricated metal products & ORE  & Other real estate \\
        \vspace{0.3em}
        335 & Electrical equipment & 532RL  & Rental and leasing services \\
        \vspace{0.3em}
        337  & Furniture & 55 & Firm management \\
        \vspace{0.3em}
        3364OT  & Other transportation equipment & 5415 & Computer systems design \\
        \vspace{0.3em}
        3361MV & Motor vehicles, parts & 5412OP & Miscellaneous technical services \\
        \vspace{0.3em}
        339  &  Miscellaneous manufacturing & 5411  & Legal services \\
        \vspace{0.3em}
        311FT & Food, beverage, tobacco  & 561 & Administrative \\
        \vspace{0.3em}
        313TT  & Textile mills and products & 562 & Waste management \\
        \vspace{0.3em}
        315AL  & Apparel and leather & 61 & Educational services \\
        \vspace{0.3em}
        322  & Paper products & 621 & Ambulatory health care services \\
        \vspace{0.3em}
        323 & Printing & 622 & Hospitals \\
        \vspace{0.3em}
        324 & Petroleum and coal & 623 & Nursing and residential care facilities \\
        \vspace{0.3em}
        325 & Chemical products & 624  & Social assistance \\
        \vspace{0.3em}
        326 & Plastics, rubber & 711AS  & Arts, spectator sports, museums \\
        \vspace{0.3em}
        42 & Wholesale trade & 713  &  Amusements, gambling, recreation \\
        \vspace{0.3em}
        441 & Motor vehicle and parts dealers & 721  & Accommodation \\
        \vspace{0.3em}
        445 & Food and beverage stores & 722 & Food services and drinking places \\
        \vspace{0.3em}
        452 & General merchandise stores & 81  & Other services, except government \\
        \vspace{0.3em}
        4A0  & Other retail & GFGD  & Federal government (defense) \\
        \vspace{0.3em}
        481  & Air transportation & GSLE & State and local government enterprises \\
        \vspace{0.3em}
        482  &  Rail transportation & GFE & Federal government enterprises \\
        \vspace{0.3em}
        483  & Water transportation & GSLG & State and local government \\
        \vspace{0.3em}
        484  & Truck transportation & GFGN & Federal government (nondefense) \\
        \vspace{0.3em}
        485  & Passenger transportation & & \\ 
        &  \\ 
        \hline \hline
    \end{tabular}
\end{table}

Unlike the 15 sector case, the 71 sector 2019 input-output matrix is not
strongly connected.  This is because it contains sinks (sectors with
zero out-degree).  For example, according to the data, ``food and beverage
stores'' do not supply any intermediate inputs, although they do of course
supply products to final consumers.

\subsection{Equilibrium}\label{sss:ioeq}

Equilibrium in Leontief models involves tracing the impact of final demand as
it flows backward through different sectors in the economy. To illustrate the
challenges this generates, consider the simplified network shown in
Figure~\ref{f:io_irreducible}.  Suppose sector~3 receives a positive demand
shock.  Meeting this demand will require greater output from its immediate
suppliers, which are sectors~2 and
4.  However, an increase in production in sector~2 requires more output from
    sector~1, which then requires more output from sector~3, where the
    initial shock occurred. This, in turn, requires more output from sectors~2
    and 4, and so on.  Thus, the chain of backward linkages leads to an
    infinite loop. Resolving this tail chasing problem requires some analysis.  

\begin{figure}
   \begin{center}
    \input{tikz/io_irreducible.tex}
    \caption{\label{f:io_irreducible} A simple production network}
   \end{center}
\end{figure}
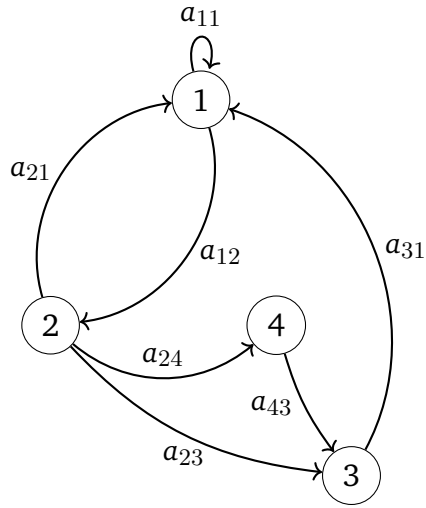

\subsubsection{Identities}

To start our search for equilibria, we set
\begin{itemize}
    \item $d_i :=$ final consumer demand for good $i$.
    \item $x_i :=$ total sales of sector $i$.
    \item $z_{ij} :=$ inter-industry sales from sector $i$ to sector $j$.
\end{itemize}
All numbers are understood to be in units of  national currency---dollars,
say.  For each sector $i$ we have the accounting identity
\begin{equation}\label{eq:sales}
    x_i = \sum_{j=1}^n z_{ij} + d_i,
\end{equation}
which states that total sales are divided between sales to other industries
and sales to final consumers.

Notice that
\begin{equation}\label{eq:leona}
    \frac{z_{ij}}{x_j}
    = \text{dollar value of inputs from $i$ per dollar output from $j$}
    = a_{ij} ,
\end{equation}
where the values $a_{ij}$ are the the input-output coefficients
discussed in \S\ref{sss:prodnet}. 
Using the coefficients, \eqref{eq:sales} can be rewritten as
\begin{equation}\label{eq:iorcon}
    x_i = \sum_{j=1}^n a_{ij} x_j + d_i,
    \qquad i = 1, \ldots n.
\end{equation}

The first term on the right hand side is the amount of good $i$ required as
inputs when the output vector is $x := (x_i)_{i=1}^n$.
We can combine the $n$ equations in
\eqref{eq:iorcon} into the linear system
\begin{equation}\label{eq:nnlinsys}
    x = A x + d.
\end{equation}

\subsubsection{Existence and Uniqueness}\label{sss:prodeu}

So far we have used no more than accounting identities and definitions.
However, we would also like to use \eqref{eq:nnlinsys} to determine output
vector $x$ given demand vector $d$, taking $A$ as fixed.  As a first step, we
seek conditions under which nonnegative solutions to~\eqref{eq:nnlinsys} exist
and are unique.

The \navy{value added}\index{Value added} of sector $j$ is defined as sales
minus spending on intermediate goods, or
\begin{equation*}
    v_j := x_j -\sum_{i=1}^n z_{ij} .
\end{equation*}

\begin{assumption}\label{a:pva}
    The input-output adjacency matrix $A$ obeys
    \begin{equation}
        \eta_j := \sum_{i=1}^n a_{ij} < 1
        \quad \text{for all } j \in \natset{n}.
    \end{equation}
\end{assumption}

\begin{Exercise}\label{ex:aimi}
    Prove that Assumption~\ref{a:pva} holds whenever value added is strictly
    positive in each sector.
\end{Exercise}

\begin{Answer}
    Fix $j \in \natset{n}$. Since $a_{ij} = z_{ij}/x_j$, we have $\eta_j = \frac{\sum_{i=1}^n
    z_{ij}}{x_j}$. Hence, if $v_j > 0$, then $\eta_j < 1$.
\end{Answer}

Exercise~\ref{ex:aimi} shows that Assumption~\ref{a:pva} is very mild.  For
example, in a competitive equilibrium, where firms make zero profits, positive
value added means that payments to factors of production other than
intermediate goods (labor, land, etc.) are strictly positive.  

\begin{Exercise}\label{ex:eara}
    Let $\eta(A) := \max_{j \in \natset{n}} \eta_j$.  Prove that $r(A) \leq
    \eta(A)<1$ whenever Assumption~\ref{a:pva} holds. 
\end{Exercise}

\begin{Answer}
    It follows easily from Assumption~\ref{a:pva} that $\eta(A) < 1$.  Moreover, since
    $A \geq 0$, the results in \S\ref{sss:someimp} imply that $r(A)$ is
    dominated by the maximum of the column sums of $A$.  But this is precisely
    $\eta(A)$.  Hence $r(A) \leq \eta(A)<1$. 
\end{Answer}

\begin{proposition}\label{p:usolnn}
    If Assumption~\ref{a:pva} holds, then, for each $d \geq 0$, the production
    system~\eqref{eq:nnlinsys} has the unique nonnegative output solution
    \begin{equation}\label{eq:leoninv}
        x^* = L d 
        \quad \text{where } 
        L := (I - A)^{-1}.
    \end{equation}
\end{proposition}

\begin{proof}
    By Exercise~\ref{ex:eara} and Assumption~\ref{a:pva} we have $r(A) < 1$.
    Hence the Neumann series lemma (NSL) 
    implies $x^*$ in~\eqref{eq:leoninv} is the unique solution in $\RR^n$. 
    Regarding nonnegativity, since $A$ is nonnegative, so is $A^i$ for all
    $i$.  Hence $x^* \geq 0$, by the power series 
    representation $L = \sum_{i=0}^\infty A^i$ provided by the NSL.
\end{proof}

The matrix $L = (\ell_{ij})$ in~\eqref{eq:leoninv} is called the \navy{Leontief
inverse}\index{Leontief inverse} associated with the coefficient matrix $A$.
We discuss its interpretation in~\S\ref{ss:dshocks}.

\begin{Exercise}
    A demand vector is called nontrivial if $d \not= 0$.  Let $d$ be
    nontrivial and suppose that $r(A) < 1$.  Show that, in equilibrium,
    every sector is active (i.e., $x^* \gg 0$) when $A$ is irreducible.
\end{Exercise}

\begin{Answer}
    Let $(A, d)$ be as stated.  When $A$ is irreducible, $L :=
    \sum_{i=1}^\infty A^i \gg 0$ and $x^* = L d$.  Since $d$ is nontrivial,
    $x^* \gg 0$ follows from $L \gg 0$ and the definition of matrix
    multiplication.
\end{Answer}

\begin{Exercise}
    A \emph{closed} input-output system is one where $d=0$.
    A nontrivial solution of a closed system $x = Ax$ is a
    nonzero $x \in \RR^n_+$ such that $Ax^* = x^*$.  Let $A$ be irreducible.
    Show that no nontrivial solution exists when $r(A) < 1$.  Show that a
    nontrivial solution exists and is unique up to constant multiples when
    $r(A) = 1$.
\end{Exercise}

\begin{Answer}
    If $r(A) < 1$, then $I-A$ is nonsingular.  At the same
    time, for the nontrivial solution $x$, we have $(I-A)x = 0$.  Contradiction.
    If, on the other hand, $r(A)=1$, then, since $r(A)$ is an eigenvalue (by the
    Perron--Frobenius theorem), we have $A x = x$ for some $x \gg 0$.
    The uniqueness claim follows from the Perron--Frobenius theorem.
\end{Answer}

\begin{Exercise}
    Consider a closed input-output system defined by input matrix $A$.
    Let $A$ be primitive.  Show that every nontrivial solution is everywhere
    positive.  Show that no nontrivial solution exists when $r(A) > 1$.
\end{Exercise}

\subsubsection{Assumptions}

It is common to interpret the expression $x^* = (I - A)^{-1} d$
from~\eqref{eq:leoninv} as meaning that supply is driven by demand.  While
this is not a universal truth, it does have plausibility in some
settings, such as when analyzing demand shifts in the short run.  Changes in
demand lead to changes in inventories, which typically cause firms to modify
production quantities.  We investigate these ideas in depth in~\S\ref{ss:dshocks}.

Another assumption concerns the production function in each sector.  You might
recall from elementary microeconomics that the \navy{Leontief production
function} takes the form
\begin{equation}\label{eq:leoprod}
    x 
    = f(z_1, \ldots, z_n)
    = \min\{ \gamma_1 z_1, \ldots, \gamma_n z_n \}.
\end{equation}
Here $x$ is output in a given sector, $\{\gamma_i\}$ is a set of parameters
and $\{z_i\}$ is a set of inputs.  To understand why \eqref{eq:leoprod} is
called a Leontief production function, note that by~\eqref{eq:leona} we have
\begin{equation}\label{eq:leonai}
    x_j = \frac{z_{ij}}{a_{ij}} 
    \quad \text{for all } i \in \natset{n} \text{ such that } a_{ij}> 0.
\end{equation}
If we interpret $z/0 = \infty$ for all $z \geq 0$, then~\eqref{eq:leonai}
implies $x_j = \min_{i \in \natset{n}} z_{ij}/a_{ij}$.  This is a version of
\eqref{eq:leoprod} specialized to sector $j$.  Hence~\eqref{eq:leoprod}
arises naturally from Leontief input-output analysis.

A final comment on assumptions is that, while the Leontief model is too simple for some purposes,  it  serves as a useful
building block for more sophisticated
models.  We discuss one such model in~\S\ref{ss:eqmulm}.

\subsection{Demand Shocks}\label{ss:dshocks}

In this section we study  impacts of changes in demand via a power series
representation  $\sum_{i \geq 0} A^i$ of the Leontief inverse $L$.  We assume
throughout that $r(A) < 1$, so that the series and $L$ are finite and equal.

\subsubsection{Response to Demand Shocks}\label{sss:rds}

Consider the impact of a demand shock of
size $\Delta d$, so that demand shifts from $d_0$ to $d_1 = d_0 + \Delta d$.  The
equilibrium output vector shifts from $x_0 = L d_0$ to 
$x_1 = L d_1$.  Subtracting the first of these equations from
the second and expressing the result in terms of differences gives
$\Delta x = L \Delta d$. Using the geometric sum version of
the Leontief inverse yields
\begin{equation}
    \Delta  x
    = \Delta d
    + A (\Delta d)
    + A^2 (\Delta d)
    + \cdots
\end{equation}
The sums in this term show how the shock propagates backward through the production
network:  
\begin{enumerate}
    \item $\Delta d$ is the initial response in each sector,
    \item $A (\Delta d)$ is the response generated by the first round of
        backward linkages,
    \item $A^2 (\Delta d)$ is the response generated by the second round, and so on.
\end{enumerate}
The total response is the sum of  responses at all  rounds. 

We can summarize the above by stating that a typical element $\ell_{ij}$ of $L
= \sum_{m \geq 0} A^m$ shows the total impact on sector $i$ of a unit change
in demand for good $j$, after taking into account all direct and indirect
effects.  $L$ itself is reminiscent of a Keynesian multiplier: changes in demand
are multiplied by this matrix to generate final output. 


Figure~\ref{f:input_output_analysis_15_leo} helps visualize the Leontief
inverse computed from the 15 sector network.  Hotter colors indicate larger
values for $\ell_{ij}$, with $i$ on the vertical axis and $j$ on the
horizontal axis.  We see, for example, that an increase in demand in almost
any sector generates a rise in manufacturing output.

\begin{figure}
   \begin{center}
    \scalebox{0.65}{\includegraphics[trim = 6mm 0mm 6mm 6mm, clip]{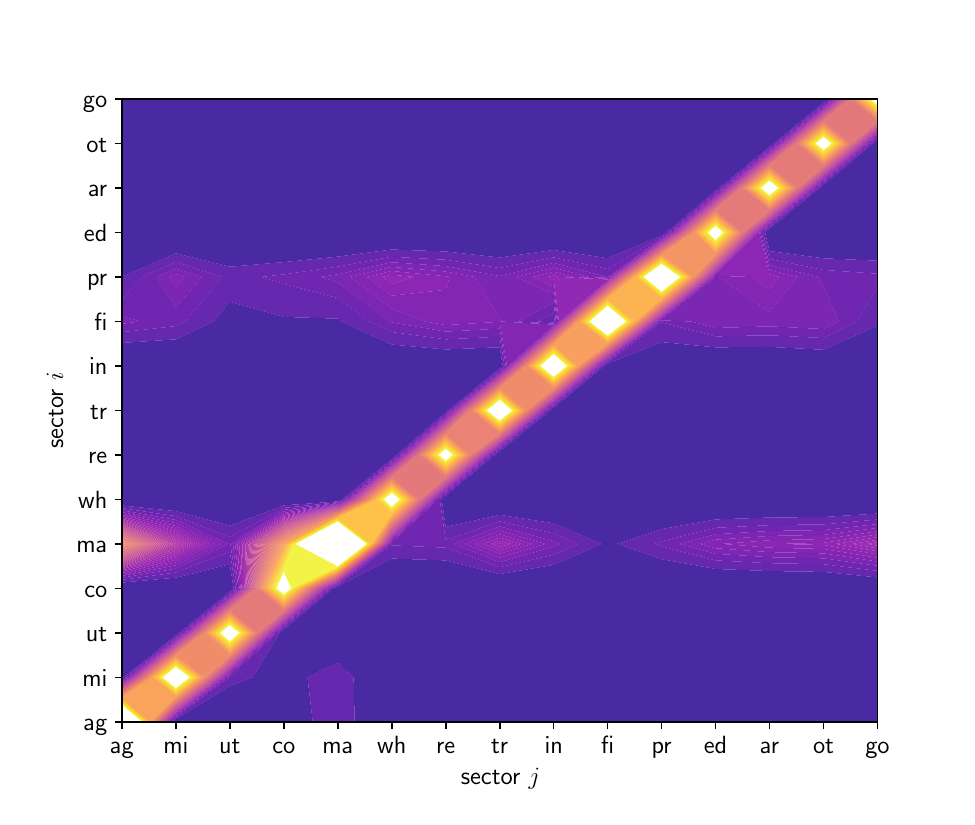}}
    \caption{\label{f:input_output_analysis_15_leo} The Leontief inverse $L$
    (hotter colors indicate larger values)}
   \end{center}
\end{figure}

\subsubsection{Shock Propagation}\label{sss:shockprop}

Figure~\ref{f:input_output_analysis_15_shocks} shows the impact of a given
vector of demand shocks $\Delta d$ on the 15 sector input-output model.  In
this simulation, each element of $\Delta d$ was drawn independently from a
uniform distribution.  The vector $\Delta d$ is shown visually in the panel
titled ``round 0,'' with hotter colors indicating larger values.  The shock
draw was relatively large in retail (re), agriculture (ag) and wholesale (wh).

\begin{figure}
   \begin{center}
    \scalebox{0.7}{\includegraphics{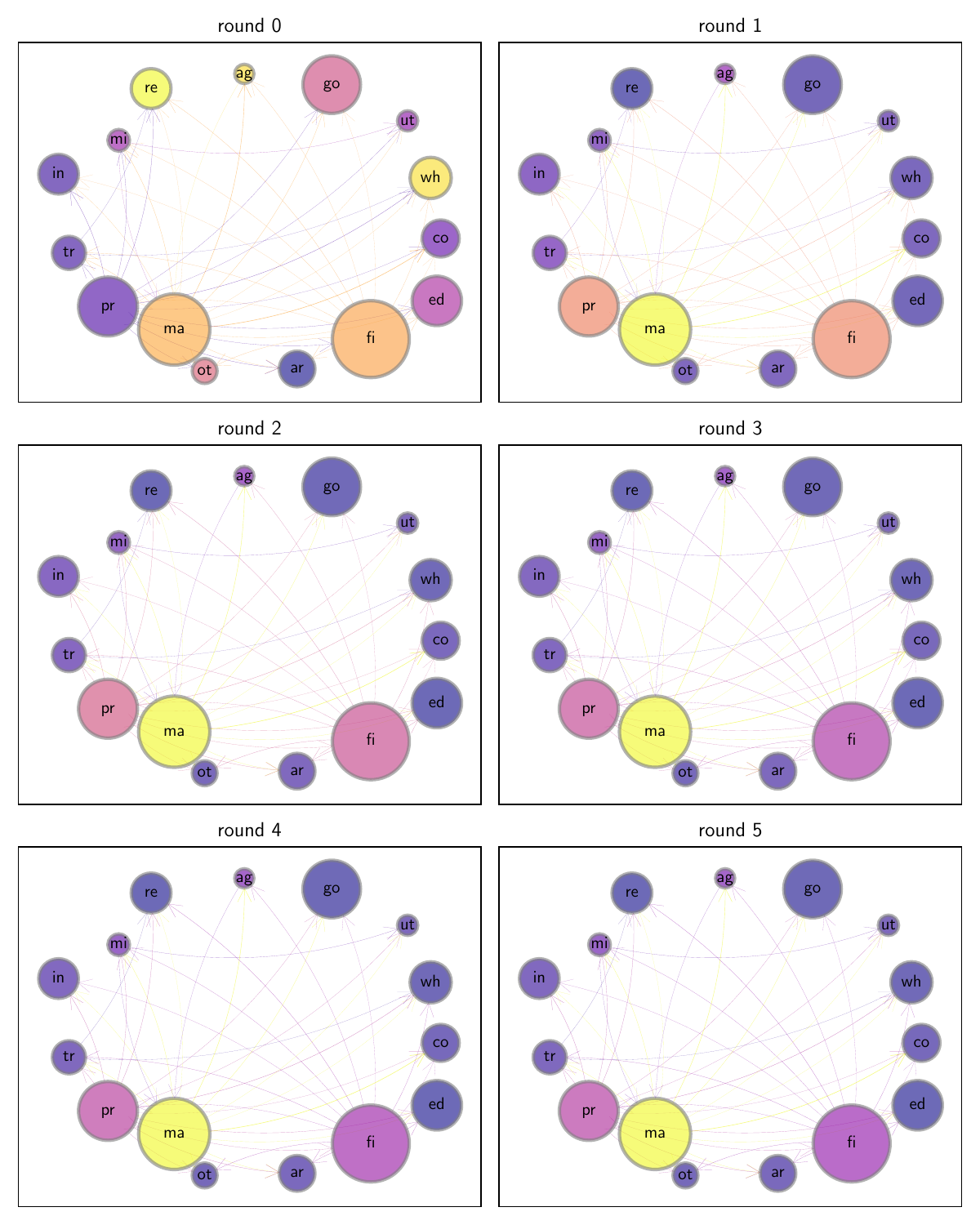}}
    \caption{\label{f:input_output_analysis_15_shocks} Propagation of demand shocks via backward linkages}
   \end{center}
\end{figure}

The remaining rounds then show the values $A (\Delta d)$, $A^2 (\Delta d)$, etc.,
with hotter colors indicating higher values.  In each round, to make the
within-round comparison between sectors clearer, values of the vector $A^i
(\Delta d)$ are rescaled into the $[0,1]$ interval before the color map is
applied.  

Note that, by round 4, the values of $A^i (\Delta d)$ have settled into a
fixed pattern.  (This is only up to a scaling constant, since values are
rescaled into $[0,1]$ as just discussed.)  Manufacturing is the most active
sector, while finance and professional services are also quite active.
In fact, if we repeat the simulation with a new draw for $\Delta d$, the 
pattern of active sectors quickly converges to exactly the same configuration.

We can explain this phenomenon using the Perron--Frobenius theorem.  Since $A$
is primitive (in the 15 sector case), we know that $r(A)^{-m} A^m$ converges
to $e \epsilon^\top$ as
$m \to \infty$, where $e$ and $\epsilon$ are the dominant left and right
eigenvectors respectively, normalized so that $\inner{\epsilon, e}=1$.  It follows that,
for large $m$, we have
\begin{equation}\label{eq:dshockk}
    A^m (\Delta d) \approx r(A)^m \inner{\epsilon, \Delta d} e.
\end{equation}
In other words, up to a scaling constant, the shock response $A^m (\Delta d)$
converges to the dominant right eigenvector, which is also the hub-based
eigenvector centrality measure.  

In Figure~\ref{f:input_output_analysis_15_shocks}, the scaling constant is not
visible because the values are rescaled to a fixed interval before the color
map is applied.  However, \eqref{eq:dshockk} shows us its value, as well as
the fact that the scaling constant converges to zero like $r(A)^m$.  Hence,
the dominant eigenpair $(r(A), e)$ gives us both the configuration of the
response to an arbitrary demand shock and the rate at which the response dies
out as we travel back through the linkages.

At this point, we recall that the sectors in
Figure~\ref{f:input_output_analysis_15} were colored according to hub-based
eigenvector centrality.  If you compare this figure to
Figure~\ref{f:input_output_analysis_15_shocks}, you will be above to confirm
that, at least for later rounds, the color schemes line up, as predicted by
the theory. Finance (fi) and manufacturing (ma) rank highly, as does the
professional services sector (pr), which includes consulting, accounting and
law.

\subsubsection{Eigenvector Centrality}

Let's look at hub-based eigenvector centrality more closely.  In a production
network, the hub property translates into being an important supplier.  Our
study of demand shocks in~\S\ref{ss:dshocks} highlighted the significance of
the eigenvector measure of hub-based centrality: if sector $i$ has high rank
under this measure, then it becomes active after a large variety of different
shocks.  Figure~\ref{f:input_output_analysis_15_ec} shows hub-based
eigenvector centrality as a bar graph for the 15 sector case. By this measure,
manufacturing is by far the most dominant sector in the US economy.

\begin{figure}
   \begin{center}
    \scalebox{0.64}{\includegraphics{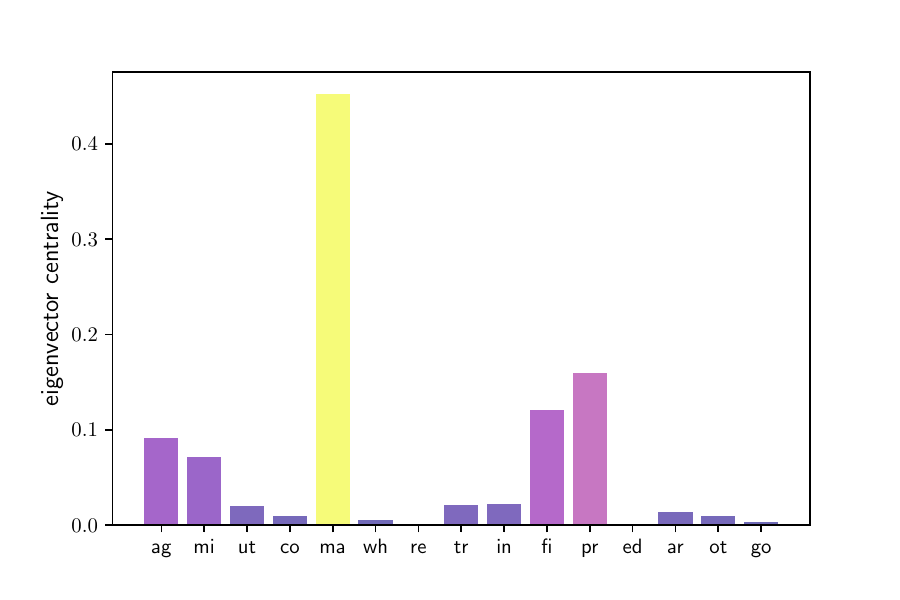}}
    \caption{\label{f:input_output_analysis_15_ec} Eigenvector centrality across US industrial sectors}
   \end{center}
\end{figure}

Reviewing the color scheme in Figure~\ref{f:input_output_analysis_71} based on
our current understanding of eigenvector centrality, we see that chemical
products (325) and primary metals (331) are both highly ranked, and hence a
wide range of demand shocks generate high activity in these sectors.

To provide some extra context, we show the analogous figure using Australian
2018 input-output data, collected by the Australian Bureau of Statistics.
     Node size is proportional to sales share and arrow width is proportional
     to the input-output coefficient.  The color map shows hub-based
     eigenvector centrality.  

By this measure, the highest ranked sector is 6901, which is ``professional,
scientific and technical services.'' This includes scientific research,
engineering, computer systems design, law, accountancy, advertising, market
research, and management consultancy.  The next highest sectors are
construction and electricity generation.  This is in contrast to the US, where
manufacturing sectors are at the top of the ranking.

\begin{figure}
   \begin{center}
    \scalebox{0.68}{\includegraphics[trim = 40mm 50mm 30mm 50mm, clip]{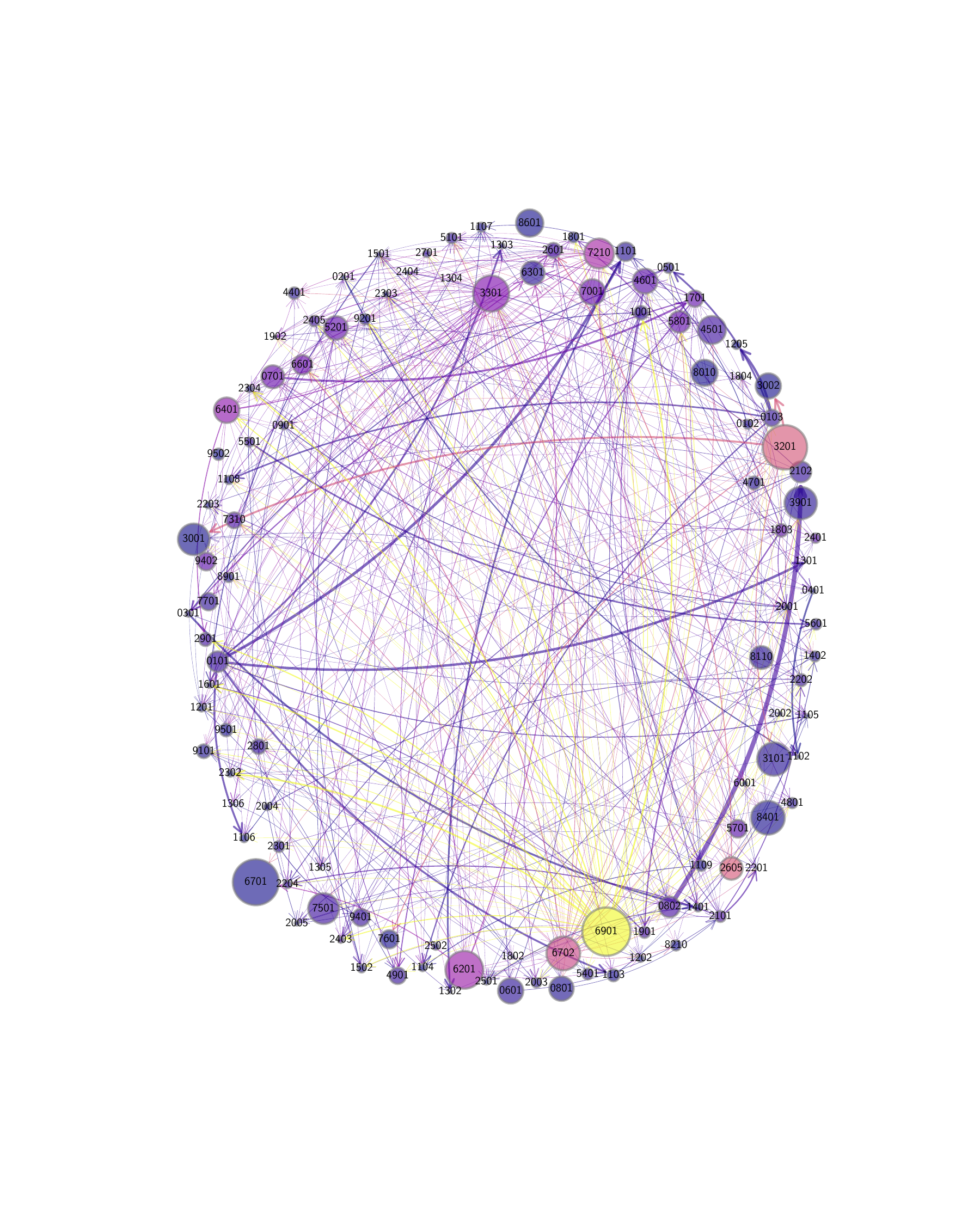}}
    \caption{\label{f:input_output_analysis_aus_114} Network for 114 Australian industry sectors in 2018}
   \end{center}
\end{figure}

\subsubsection{Output Multipliers}\label{sss:outmul}

One way to  rank sectors that has a long tradition in
input-output analysis is via output multipliers.  The \navy{output
multiplier}\index{Output multiplier} of
sector $j$, denoted below by $\mu_j$, is usually defined as the ``total
sector-wide impact of an extra dollar of demand in sector $j$,'' where total
means taking into account backward linkages.  This measure has historically
been of interest to policy makers considering  impacts of fiscal stimulus.

Recalling from~\S\ref{sss:rds} that $\ell_{ij}$ 
shows the total impact on sector $i$ of a unit change
in demand for good $j$, we come to the definition
\begin{equation*}
    \mu_j = \sum_{i=1}^n  \ell_{ij}
    \qquad (j \in \natset{n}).
\end{equation*}
In vector notation this is $\mu^\top = \1^\top L$ or,
\begin{equation}\label{eq:mudefio}
    \mu^\top = \1^\top (I - A)^{-1}.
\end{equation}
Comparing this with \eqref{eq:katzav}, we see that the vector of output
multipliers is equal to the authority-based Katz centrality measure (with
the parameter $\beta$ defaulting to unity).

The connection between the two measures makes sense: high authority-based
centrality means that a sector has many inward links, and that those links are
from other important sectors.  Loosely speaking, such a sector is an important
buyer of intermediate inputs. A sector highly ranked by this measure that
receives a demand shock will cause a large impact on the whole production
network.

Figure~\ref{f:input_output_analysis_15_omult} shows the size of output
multipiers across 15 US industrial sectors, calculated from the same
input-output data as previous 15 sector figures using~\eqref{eq:mudefio}.  The
highest ranks are assigned to manufacturing, agriculture and construction.

\begin{figure}
   \begin{center}
    \scalebox{0.64}{\includegraphics[trim = 0mm 0mm 0mm 0mm, clip]{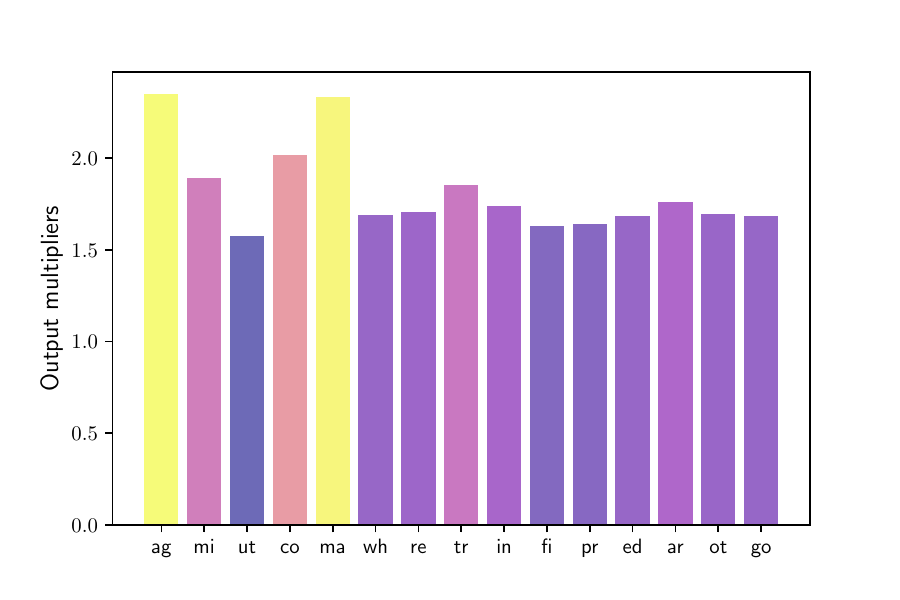}} \caption{\label{f:input_output_analysis_15_omult} Output multipliers across 15 US industrial sectors}
   \end{center}
\end{figure}

\subsection{Forward Linkages}\label{ss:forlink}

Several economic questions connect to relative ``upstreamness'' of a sector or
production good.  For example, \cite{olabisi2020input} finds that upstreamness
is related to sectoral volatility, while \cite{antras2012measuring} examine
the relationship between upstreamness and tendency to export.  Tariff changes
tend to have different aggregate effects when applied to upstream rather than
downstream industries \citep{martin2020downstream}.  Finally, since WWII, many
developing countries have systematically supported and encouraged upstream
industries \citep{liu2019industrial}.

In order to study upstreamness, we first introduce the Ghosh model for forward
linkages, which uses a rearrangement of terms from the original Leontief model.

\subsubsection{The Ghosh Model}

Recall that $a_{ij}=z_{ij}/x_j = $  the dollar value of inputs from $i$ per
dollar of sales from $j$.  Consider now the related quantities
\begin{equation}\label{eq:forlink}
    f_{ij} := 
    \frac{z_{ij}}{x_i}
    = \text{value of inputs from $i$ to $j$ per dollar output from $i$}.
\end{equation}
Let $F := (f_{ij})_{i, j \in \natset{n}}$.  The matrix $F$ is called the
\navy{direct-output} matrix or the \navy{Ghosh matrix}.  Element $f_{ij}$ can
be interpreted as the size of the ``forward linkage'' from $i$ to $j$. 
Analogous to $A$, the matrix $F$ can be viewed as a weight function over
output sectors and visualized as in
Figure~\ref{f:input_output_analysis_15_fwd}.  This digraph uses the same data
source as Figure~\ref{f:input_output_analysis_15}.

\begin{figure}
   \begin{center}
    \scalebox{0.92}{\includegraphics[trim = 20mm 20mm 4mm 20mm, clip]{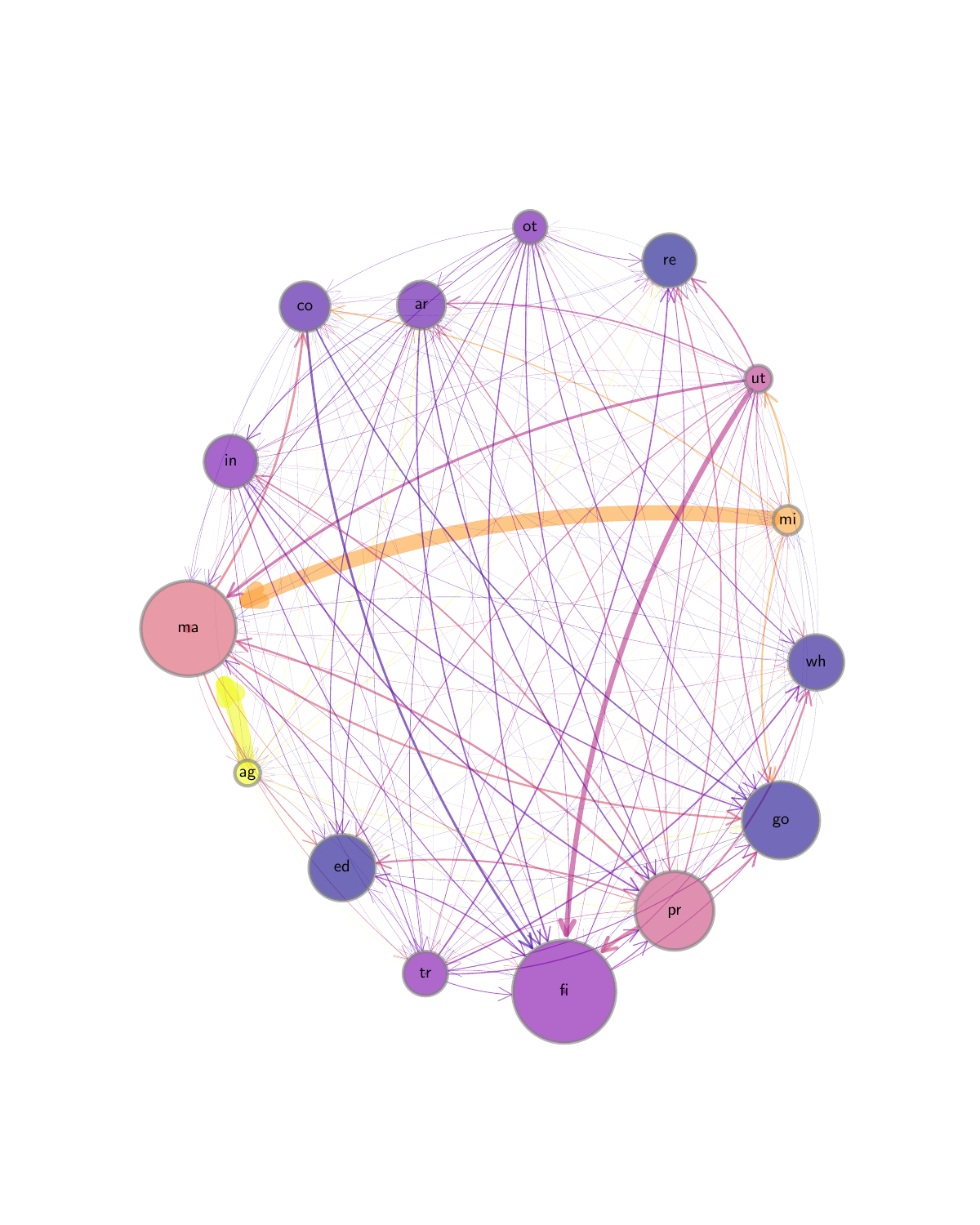}}
    \caption{\label{f:input_output_analysis_15_fwd} Forward linkages and upstreamness over US industrial sectors}
   \end{center}
\end{figure}

\begin{Exercise}
    Prove that $A$ and $F$ are similar matrices (see \S\ref{sss:sim}) when $x \gg 0$.
\end{Exercise}

Let $v_j$ be value added in sector $j$ (i.e., payments to factors of
production other than intermediate goods). We have
\begin{equation}\label{eq:vaid}
    x_j = \sum_{i=1}^n z_{ij} + v_j
    \qquad (j \in \natset{n}).
\end{equation}
This states that (under perfect competition), the revenue of sector $j$ is
divided between spending on intermediate goods, which is the first term
$\sum_{i=1}^n z_{ij}$, and payments to other factors of production
(value added).

Using the forward linkages, we can rewrite \eqref{eq:vaid}
as $x_j = \sum_i f_{ij} x_i + v_j$ for all $j$ or, in matrix form
\begin{equation}
    x^\top = x^\top F + v^\top.
\end{equation}
Taking transposes and solving under the assumption $r(F) < 1$ gives
\begin{equation}\label{eq:vaid2}
    x^* = (I - F^\top)^{-1} v.
\end{equation}
We can think of the solution $x^*$ in \eqref{eq:vaid2} as the amount of output
necessary to acquire a given amount of value added.
Since payments of value added are made to underlying factors of production,
the Ghosh model is also called a ``supply-side input output model''.

\begin{Exercise}\label{ex:rarf}
    In \S\ref{sss:prodeu} we argued that $r(A) < 1$ will almost always hold.  This
    carries over to $r(F)$, since $r(A) = r(F)$ whenever $x \gg 0$.  Provide a
    proof of the last statement.
\end{Exercise}

\begin{Answer}
    Let $\lambda$ be an eigenvalue of $A$ and let $e$ be the corresponding
    eigenvector. Then, for all $i \in \natset{n}$, we have
    \begin{equation*}
        \sum_j a_{ij} e_j = \lambda e_i
        \quad \iff \quad
        \sum_j \frac{z_{ij}}{x_j} e_j = \lambda e_i
        \quad \iff \quad
        \sum_j f_{ij} \frac{e_j}{x_j} = \lambda \frac{e_i}{x_i},
    \end{equation*}
    where we have used the fact that $x \gg 0$.  It follows that $\lambda$ is an
    eigenvalue of $F$.  The same logic runs in reverse, so $A$ and $F$
    share eigenvalues.  Hence $r(A)=r(F)$.
\end{Answer}

We omit a discussion of the relative merits of supply- and demand-driven
input-output models.  Our main interest in forward linkages is due to their
connection to the topic of ranking sectors by relative upstreamness.

\subsubsection{Upstreamness}\label{sss:upness}

Which industries are relatively upstream?  One proposed measure of upstreamness can
be found in \cite{antras2012measuring}.  With $f_{ij}$ as defined in
\eqref{eq:forlink}, the upstreamness $u_i$ of sector $i$ is defined
recursively by
\begin{equation}\label{eq:defup}
    u_i = 1 + \sum_{j=1}^n f_{ij} u_j.
\end{equation}
The recursive definition of the vector $u$ in \eqref{eq:defup} stems from the
idea that those sectors selling a large share of their output to upstream
industries should be upstream themselves.  

We can write \eqref{eq:defup} in vector form as $u = \1 + F u$ and solve for
$u$ as
\begin{equation}\label{eq:defup2}
    u = (I - F)^{-1} \1.
\end{equation}
A unique nonnegativity solution exists provided that $r(F) < 1$.
We expect this to hold in general, due to the findings in
Exercise~\ref{ex:rarf}.  

Maintaining the convention $\beta=1$, we see that the upstreamness
measure~\eqref{eq:defup2} proposed by \cite{antras2012measuring} is in fact
the hub-based Katz centrality measure \eqref{eq:katzhub} for the production
network with weights allocated by the forward linkage matrix $F$.

Figure~\ref{f:input_output_analysis_15_up} shows the result of computing $u$
via~\eqref{eq:defup2}, plotted as a bar
graph, for the 15 sector input-output network. Consistent with expectations,
the primary commodity producers (agriculture and mining) are the most
upstream, while retail, education and health services are typical downstream
sectors.

\begin{figure}
   \begin{center}
    \scalebox{0.64}{\includegraphics{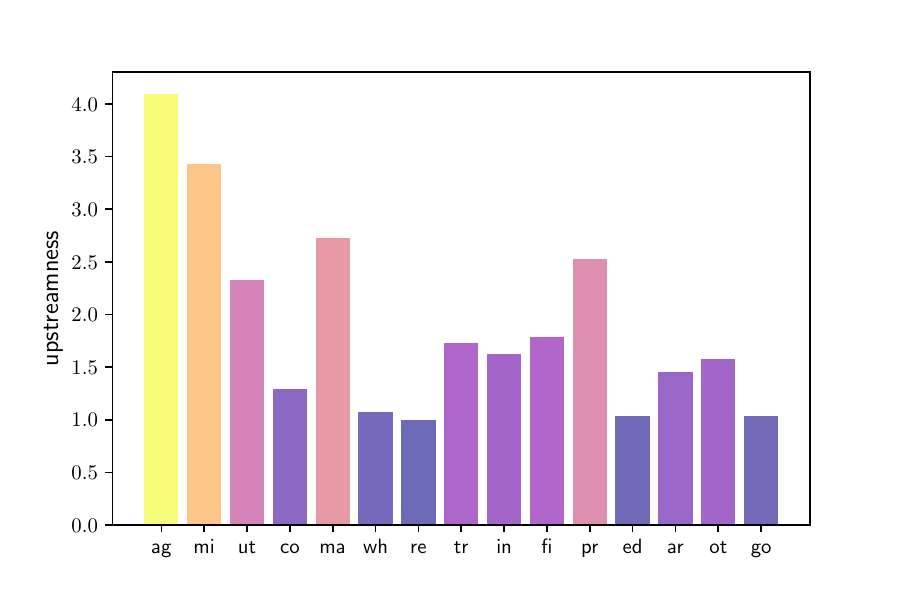}}
    \caption{\label{f:input_output_analysis_15_up} Relative upstreamness of US industrial sectors}
   \end{center}
\end{figure}

The nodes in Figure~\ref{f:input_output_analysis_15_fwd} are also colorized by
upstreamness.

\section{General Equilibrium}

One limitation of the Leontief input-output analysis from \S\ref{ss:mutmod}
is that demand is fixed and exogenous.  In this section we embed Leontief's
model in an equilibrium setting where output and prices are determined by a
combination of supply and demand.  One objective is to understand how an
input-output structure interacts with firm-level shocks to shape aggregate
volatility.   

\subsection{Supply and Demand}\label{ss:eqmulm}

Our first step is to introduce and solve a multisector general
equilibrium model based on \cite{acemoglu2012network} and
\cite{carvalho2019production}.

\subsubsection{Production and Prices}

As in the Leontief economy, there are $n$ sectors, also called industries,
each of which produces one good.  Real output in sector $j$ is given by 
\begin{equation}\label{eq:mscdp}
    y_j = s_j \ell_j^\alpha \prod_{i=1}^n q_{ij}^{a_{ij}}.
\end{equation}
Here 
\begin{itemize}
    \item $s_j$ is a sector-specific shock (independent across sectors),
    \item $\ell_j$ is labor input to sector $j$,
    \item $q_{ij}$ is the amount of good $i$ used in the production of good
        $j$, and
    \item $\alpha$ and $a_{ij}$ take values in $(0, 1)$ and
        satisfy $\alpha + \sum_i a_{ij}=1$ for all $j \in \natset{n}$.
\end{itemize}
The last condition implies constant returns to scale (CRS) in each
sector.\footnote{In order to be
consistent with traditional input-output notation (see~\S\ref{ss:mutmod}), we
transpose $i$ and $j$ relative to sources such as \cite{acemoglu2012network}
and \cite{carvalho2019production}.  This is just a matter of convention.}

\begin{Exercise}\label{ex:caral}
    Let $A = (a_{ij})$ be the $n \times n$ matrix of \navy{technical coefficients} from the
    Cobb--Douglas production function in~\eqref{eq:mscdp}.  Using the
    stated assumptions and the results in \S\ref{sss:someimp}, show that $r(A)
    < 1$.\footnote{Later, in \S\ref{s:mst}, we use additional spectral theory
    to prove the exact result $r(A)=1-\alpha$.}
\end{Exercise}

\begin{Exercise}\label{ex:1malpha}
    Prove that $\sum_i \sum_j a_{ij}^{(m)} = n (1-\alpha)^m$ for all $m \in
    \NN$, where $a_{ij}^{(m)}$ is the $(i,j)$-th element of $A^m$.
\end{Exercise}

\begin{Answer}
    We need to show that $\1^\top A^m \1 = n(1-\alpha)^m$ for any $m$.  We prove
    this by induction, noting that $\1^\top A = (1-\alpha) \1^\top$ by the CRS
    assumption.  It follows immediately that $\1^\top A \1 = n (1-\alpha)$.
    Now suppose also that $\1^\top A^m \1 = n(1-\alpha)^m$ holds.  Then
    \begin{equation*}
        \1^\top A^{m+1} \1
        = \1^\top A A^m \1
        = (1-\alpha) \1^\top A^m \1
        = n (1-\alpha)^{m+1},
    \end{equation*}
    where the last step is by the induction hypothesis.
\end{Answer}

Firms are price takers.  With $p_j$ being the price of good $j$, a firm in
sector $j$ maximizes profits
\begin{equation}\label{eq:profitsms}
    \pi_j := p_j y_j - w \ell_j - \sum_i p_i q_{ij}
\end{equation}
with respect to the $n+1$ controls $\ell_j$ and $q_{1j}, \ldots, q_{nj}$.

\begin{Exercise}
    Show that, when prices and wages are taken as given, the unique global
    maximizers of \eqref{eq:profitsms} are 
    \begin{equation}\label{eq:msmax}
        \ell_j = \alpha \frac{p_j y_j}{w}
        \quad \text{and} \quad
        q_{ij} = a_{ij} \frac{p_j y_j}{p_i}
        \qquad (i, j \in \natset{n}).
    \end{equation}
\end{Exercise}

\begin{Answer}
    Inserting~\eqref{eq:mscdp} into \eqref{eq:profitsms} and differentiating with
    respect to $\ell_j$ and $q_{ij}$ leads to the first order conditions given
    in \eqref{eq:msmax}.  It can be shown that these local maximizers are
    global maximizers, although we omit the details.   
\end{Answer}

\begin{remark}\label{r:tceio}
    From~\eqref{eq:msmax} we have $a_{ij} = (p_i q_{ij})/(p_j
    y_j)$,  which states that the $i,j$-th technical coefficient is 
    the dollar value of inputs from $i$ per dollar of sales from $j$.
    This coincides with the definition of $a_{ij}$ from the discussion of
    input-output tables in \S\ref{sss:prodnet}. Hence, in the current setting,
    the (unobservable) technical coefficient matrix equals the (observable)
    input-output coefficient matrix defined in \S\ref{ss:mutmod}.
\end{remark}

Substituting the maximizers~\eqref{eq:msmax} into the production function gives
\begin{equation}\label{eq:mspwf}
    y_j = c s_j \left( \frac{p_j y_j}{w} \right)^\alpha
    \prod_{i=1}^n \left( \frac{p_j y_j}{p_i} \right)^{a_{ij}},
\end{equation}
where $c$ is a positive constant depending only on parameters.  

\begin{Exercise}\label{ex:umspw}
    Using~\eqref{eq:mspwf}, show that 
    \begin{equation*}
        \rho_j = \sum_i a_{ij} \rho_i - \epsilon_j
        \quad \text{where }
        \rho_j := \ln \frac{p_j}{w}
        \quad \text{and} \quad
        \epsilon_j := \ln (c s_j).
    \end{equation*}
\end{Exercise}
Let $\rho$ and $\epsilon$ be the column vectors $(\rho_i)_{i=1}^n$ and
$(\epsilon_i)_{i=1}^n$ of normalized prices and log shocks from
Exericse~\ref{ex:umspw}. Collecting the equations stated there leads to
$\rho^\top = \rho^\top A - \epsilon^\top$, or
\begin{equation}\label{eq:msirho}
    \rho = A^\top \rho - \epsilon.
\end{equation}

\begin{Exercise}
    Prove that
    \begin{equation}
        \rho_j = - \sum_i \epsilon_i \ell_{ij} 
        \quad \text{where }
        L := (\ell_{ij}) := (I - A)^{-1}.
    \end{equation}
    Why is $L$ well defined?
\end{Exercise}

\begin{Answer}
    By Exercise~\ref{ex:caral} we have $r(A)<1$.  Hence 
    $L = (I - A)^{-1}$ is well defined
    and, moreover, we can solve~\eqref{eq:msirho} using the Neumann series lemma,
    yielding $\rho = - (I - A^\top)^{-1} \epsilon$.  Since the
    inverse of the transpose is the transpose of the inverse, we can write this as
    $\rho = - L^\top \epsilon$.  Unpacking gives the
    equation stated in the exercise.
\end{Answer}

As in Chapter~\ref{c:prod}, the matrix $L$ is the Leontief
inverse\index{Leontief inverse} generated by $A$.

\subsubsection{Consumpion}

Wages are paid to a representative household who chooses consumption  to
maximize utility $\sum_i \ln c_i$.  In equilibrium, profits are zero, so the
only income accruing to the household consists of wage income. The
household supplies one unit of labor inelastically.  Hence, the budget
constraint is $\sum_i p_i c_i = w$.  

\begin{Exercise}
    Show that the unique utility maximizer is the vector $(c_1, \ldots, c_n)$
    that satisfies $p_i c_i = w/n$ for all $i \in \natset{n}$.  (Equal amounts are
    spent on each good.)
\end{Exercise}

\subsubsection{Aggregate Output}

In this economy, aggregate value added (defined in \S\ref{ss:forlink}) is
equal to the wage bill.  This quantity is identified with real aggregate
output and referred to as GDP.  The \navy{Domar weight}\index{Domar weight} of
each sector is defined as its sales as a fraction of GDP:
\begin{equation*}
    h_i := \frac{p_i y_i}{w} .
\end{equation*}
From the closed economy market clearing condition $y_i = c_i + \sum_j q_{ij}$ and
the optimality conditions we obtain
\begin{equation}\label{eq:mseqcon}
    y_i = \frac{w}{n p_i} 
    + \sum_j a_{ij} \frac{p_j y_j}{p_i}.
\end{equation}

\begin{Exercise}\label{ex:dmrleo}
    Letting $L = (\ell_{ij})$ be the Leontief inverse and
    using~\eqref{eq:mseqcon}, show that  Domar weights satisfy
    \begin{equation*}
        h_i = \frac{1}{n} \sum_j \ell_{ij} 
        \qquad \text{for all } i \in \natset{n}.
    \end{equation*}
\end{Exercise}

\begin{Answer}
    Equation~\eqref{eq:mseqcon} can be expressed as $h_i = n^{-1} + \sum_j a_{ij}
    h_j$.  Letting $h = (h_i)$ be a column vector in $\RR^n$ and letting $\1$
    be a column vector of ones, these $n$ equations become $h = n^{-1} \1 + A h$.
    Since $r(A) <1$, the unique solution is $h = n^{-1} (I-A)^{-1} \1 =
    n^{-1} L \1$.  Unpacking the vector equation gives the stated result.
\end{Answer}

\begin{Exercise}\label{ex:sboh}
    Prove that $\sum_{i=1}^n h_i = 1/\alpha$.
\end{Exercise}

\begin{Answer}
    From the result in Exercise~\ref{ex:1malpha}, we have
    \begin{equation*}
        \1^\top h 
        = \frac{1}{n} \1^\top \sum_{m \geq 0} A^m \1
        = \frac{1}{n} \sum_{m \geq 0} \1^\top A^m \1
        = \frac{1}{n} n \sum_{m \geq 0} (1-\alpha)^m .
    \end{equation*}
    The last expression evaluates to $1/\alpha$, as was to be shown.
\end{Answer}

From the results of Exercise~\ref{ex:umspw} we obtain
$\ln w = \ln p_j + \sum_i \epsilon_i \ell_{ij}$.  Setting $g := \ln w$
and summing yields
\begin{equation*}
    n g = \sum_j  \ln p_j + \sum_i \epsilon_i \sum_j \ell_{ij} .
\end{equation*}
Normalizing prices so that $\sum_i \ln p_i = 0$, this simplifies to
\begin{equation}\label{eq:msagout}
    g = \sum_i \epsilon_i h_i .
\end{equation}
Thus, log GDP is the inner product of sectoral shocks and the Domar weights.

\subsection{The Granular Hypothesis}\label{ss:granor} 

We have just constructed a multisector model of production and output.  We
plan to use this model to study shock propagation and aggregate fluctuations.
Before doing so, however, we provide a relatively simple and network-free
discussion of shock propagation.  The first step is to connect the propagation
of shocks to the firm size distribution.  Later, in \S\ref{ss:avr}, we will
see how these ideas relate to the general equilibrium model and the
topology of the production network.

\subsubsection{Aggregate vs Idiosyncratic Shocks}

Some fluctuations in aggregate variables such as GDP growth and the
unemployment rate can be tied directly to large exogenous changes in the
aggregate environment.  One obvious example is the jump in the US unemployment
rate from 3.5\% to 14.8\% between February and April 2020, which was initiated
by the onset of the COVID pandemic and resulting economic shutdown.

Other significant fluctuations lack clear macro-level causes.  For example,
researchers offer mixed explanations for the 1990 US recession, including
``technology shocks,'' ``consumption shocks'' and loss of ``confidence''
\citep{cochrane1994shocks}.  However, these explanations are either difficult
to verify on the basis of observable outcomes or require exogenous shifts in
variables that should probably be treated as endogenous.

One way to account for at least some of the variability observed in output
growth across most countries is on the basis of firm-level and sector-specific
productivity and supply shocks.  Examples of sector-specific shocks include 
\begin{enumerate}
    \item the spread of African Swine Fever to China in 2018,
    \item the Great East Japan Earthquake of 2011 and resulting tsunami,
        which triggered meltdowns at three reactors in the Fukushima Daiichi
        Nuclear Power Plant, and
    \item the destruction of Asahi Kasei Microdevices' large scale IC factory
        in Miyazaki Prefecture in October 2020.
\end{enumerate}

In the discussion below, we investigate the extent to which 
firm-level shocks can drive fluctuations in aggregate productivity.

\subsubsection{The Case of Many Small Firms}\label{sss:gdpgrowth}

It has been argued that idiosyncratic, firm-level shocks can  account only for
a very small fraction of aggregate volatility (see, e.g.,
\cite{dupor1999aggregation}).  The logical heart of this argument is the  dampening
effect of averaging over independent random variables.  To illustrate the main
idea, we follow a simple model of production without linkages across firms
by \cite{gabaix2011granular}.  

Suppose there are $n$ firms, with the size of the $i$-th firm, measured by
sales, denoted by $S_i$.  Since all sales fulfill final demand, GDP is given
by $Y := \sum_{i=1}^n S_i$.  We use primes for next period values and $\Delta$
for first differences (e.g., $\Delta S_i = S_i' - S_i$).  We assume that firm
growth $\Delta S_i / S_i$ is equal to $\sigma_F \epsilon_i$, where
$\{\epsilon_i\}$ is a collection of {\sc iid} random variables corresponding
to firm-level idiosyncratic shocks.  We also assume that $\var(\epsilon_i) =
1$, so that $\sigma_F$ represents firm-level growth volatility.  

GDP growth is then
\begin{equation*}
    G 
      := \frac{\Delta Y}{Y} 
      = \frac{\sum_{i=1}^n \Delta S_i}{Y} 
      = \sigma_F \sum_{i=1}^n \frac{S_i}{Y} \epsilon_i.
\end{equation*}

\begin{Exercise}
    Treating the current firm size distribution $\{S_i\}$ and hence GDP as given,
    show that, under the stated assumptions, the standard deviation of GDP
    growth $\sigma_G := (\var G)^{1/2}$ is
    \begin{equation}\label{eq:gsdg}
        \sigma_G 
            = \sigma_F H_n
            \quad \text{where } \;
            H_n := \left( 
                    \sum_{i=1}^n \left( \frac{S_i}{Y} \right)^2 
                   \right)^{1/2}.
    \end{equation}
\end{Exercise}

\begin{Answer}
    This expression follows easily from the definition of variance and the
    independence of firm-level shocks, which allows us to pass the variance
    through the sum.
\end{Answer}

If, say, all firms are equal size, so that $n S_i = Y$, this means that
$\sigma_G = \sigma_F / \sqrt{n}$, so volatility at the aggregate level is very
small when the number of firms is large.  For example, if the number of firms
$n$ is $10^6$, which roughly matches US data, then 
\begin{equation}\label{eq:gonf}
    \frac{\sigma_G}{\sigma_F }
    = H_n = \left( \frac{1}{10^6} \right)^{1/2} = 10^{-3} = 0.001.
\end{equation}
Hence firm-level volatility accounts for only 0.1\% of aggregate volatility.

To be more concrete, \cite{gabaix2011granular} calculates $\sigma_F = 12$,
which means that, by \eqref{eq:gonf}, $\sigma_G = 0.012$\%.  But the
volatility of GDP growth is actually far higher.  Indeed,
Figure~\ref{f:gdp_growth} reports that, for the US, $\sigma_G$ is
approximately $2\%$, which is two orders of magnitude greater.  The core
message is that, under the stated assumptions, firm-level shocks explain only
a tiny part of aggregate volatility. 

\begin{figure}
    \centering
    \scalebox{0.74}{\includegraphics{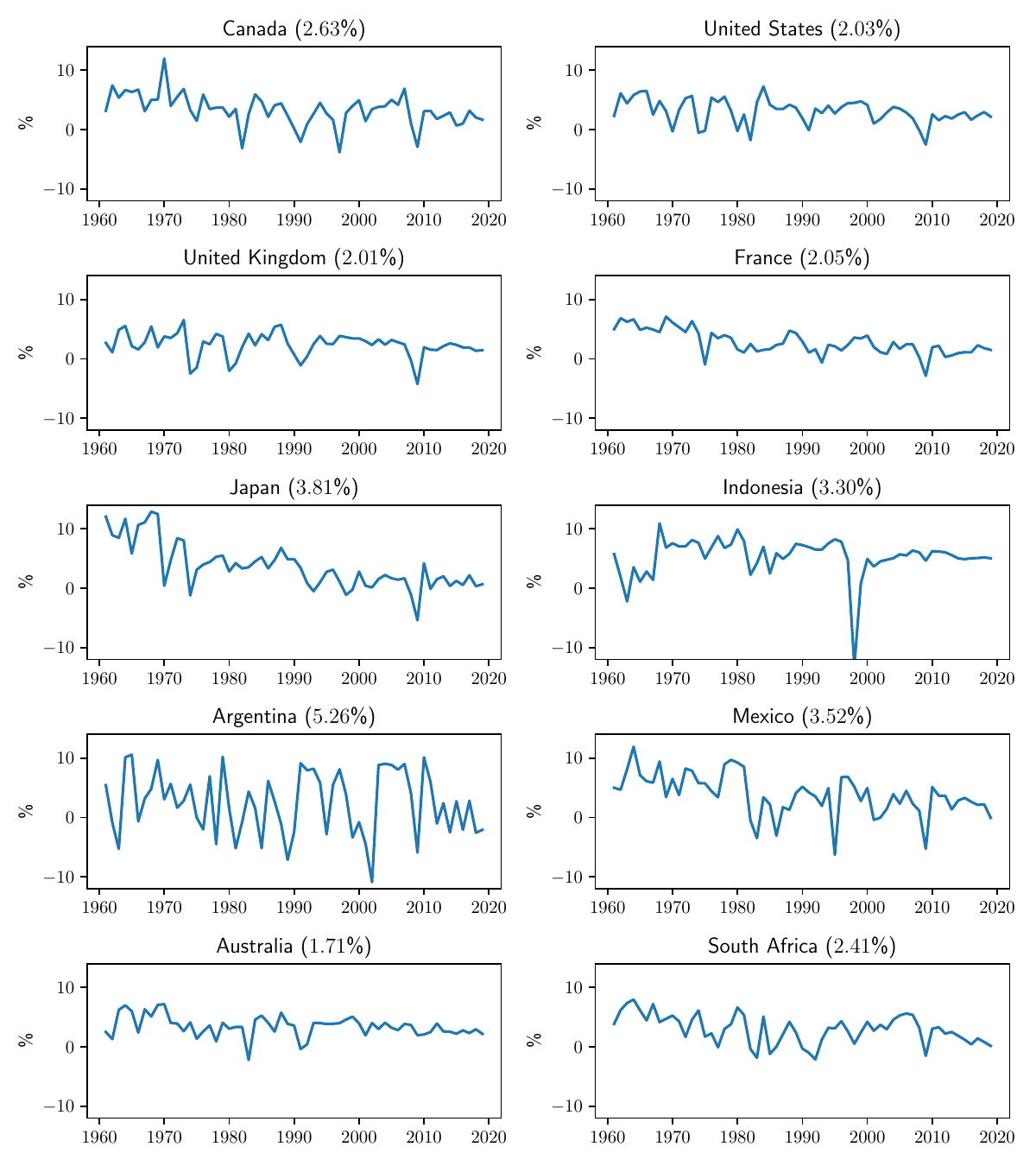}}
    \caption{\label{f:gdp_growth} GDP growth rates and std.\ deviations (in
    parentheses) for 10 countries}
\end{figure}

\subsubsection{The Effect of Heavy Tails}\label{sss:efht}

There are some obvious problems with the line of reasoning used in
\S\ref{sss:gdpgrowth}.  One is that firms are assumed to be of identical size.
In reality, most firms are small to medium, while a relative few are enormous.
For example, in the US, a small number of giants dominate technology,
electronics and retail. 

\cite{gabaix2011granular} emphasized that we can get closer to actual GDP
volatility by more thoughtful specification of the firm size distribution.
Altering the distribution $\{S_i\}_{i=1}^n$  changes the value $H_n$ in
\eqref{eq:gonf}, which is called the \navy{Herfindahl index}\index{Herfindahl
index}.  This index is often applied to a group of firms in a sector to
measure industry concentration.  For given aggregate output $Y$, the
Herfindahl index is minimized when $S_i = S_j$ for all $i, j$.  This is the
case we considered above.  The index is maximized at $H_n=1$ when a single
firm dominates all sales.  By~\eqref{eq:gonf}, a larger Herfindahl index will
increase $\sigma_G$ relative to $\sigma_F$, which allows firm-level shocks to
account for more of aggregate volatility.

Calculation of $H_n$ is challenging because the entire firm size distribution
$\{S_i\}_{i=1}^n$ is difficult to observe.  Nonetheless, we can estimate $H_n$
by (a) estimating a population probability distribution that fits the
empirical distribution $\{S_i\}_{i=1}^n$ and (b) using analysis or Monte Carlo
simulations to calculate typical values of $H_n$.

For step (a), \cite{gabaix2011granular} cites the study of
\cite{axtell2001zipf}, which finds the firm size distribution to be Pareto
with tail index $1.059$.  If we repeatedly draw $\{S_i\}_{i=1}^n$ from a
Pareto distribution with $\alpha = 1.059$ and $n=10^6$, record the value of
$H_n$ after each draw and then take the median value as our estimate, we
obtain $H_n \approx 0.88$.  In other words, under the Pareto assumption just
stated, firm-level volatility accounts for almost 90\% of aggregate
volatility. In essence, this means that, to explain aggregate volatility, we
need to look no further than firm-level shocks.

\subsubsection{Sensitivity Analysis}\label{sss:senan}

The finding in the previous paragraph is quite striking.  How seriously should
we take it?

One issue is that the figure $H_n \approx 0.88$ is not robust to small
changes in assumptions.  For example, the regression in
Figure~\ref{f:empirical_powerlaw_firms_forbes} suggests that we take $1.32$ as
our estimate for the tail index $\alpha$, rather than Axtell's value of
$1.059$.  If we rerun the same calculation with $\alpha=1.32$, the estimated
value of $H_n$ falls to $0.018$.  In other words, firm-level shocks account
for only 18\% of aggregate volatility.

Another issue is that the large value for $H_n$ obtained under Axtell's
parameterization is very sensitive to the parametric family chosen for the
firm size distribution.  The next exercise illustrates.

\begin{Exercise}
    Figure~\ref{f:empirical_powerlaw_firms_forbes} suggests only that the far
    right tail of the firm size distribution obeys a Pareto law.
    In fact, some authors argue that the lognormal distribution
    provides a better fit than the Pareto distribution (\cite{kondo2020heavy}
    provide a recent discussion).  So suppose now that $\{S_i\}$ is $n$ {\sc
        iid} draws from the $LN(\mu, \sigma^2)$ distribution (as given in
        Example~\ref{eg:lognorm}), where $\mu, \sigma$ are
        parameters.\footnote{In
        other words, each $S_i$ is an independent copy of the random variable
        $S := \exp(\mu + \sigma Z)$, where $Z$ is standard normal.}    Implement and run
    Algorithm~\ref{algo:genherf}.  Set $m=10^3$ and $n=10^6$.  Choose $\mu$
    and $\sigma$ so that the mean and median of the $LN(\mu, \sigma^2)$
    distribution agree with that of the
    standard Pareto distribution with tail index $\alpha$, which are
    $\alpha/(\alpha-1)$ and $2^{1/\alpha}$ respectively.  As in
    \cite{gabaix2011granular}, set $\alpha=1.059$.  What estimate do you
    obtain for $H_n$?  How much of aggregate volatility is explained?
\end{Exercise}

\begin{algorithm}
    \For{$j$ in $1, \ldots, m$ } 
    {
        generate $n$ independent draws $\{S_i^j \}$ from the $LN(\mu, \sigma^2)$
        distribution  \;
        compute the Herfindahl index $H_n^j$ corresponding to $\{S_i^j \}$ \; 
    }
    set $H_n$ equal to the median value of $\{H^j_n\}_{j=1}^m$ \;
    \Return{$H_n$ }
    \caption{\label{algo:genherf} Generate an estimate of $H_n$ under log-normality}
\end{algorithm}

\subsection{Network Structure and Shock Propagation}\label{ss:avr}

The sensitivity analysis in \S\ref{sss:senan} suggests we should be skeptical
of the claim that firm-level shocks explain most aggregate-level shocks that
we observe.  This means that either micro-level shocks account for only a
small fraction of aggregate volatility or, alternatively, that the model is
too simple, and micro-level shocks are amplified through some other mechanism.

An obvious way to explore further is to allow linkages between firms, in the
sense that the inputs for some firms are outputs for others. Such an
extension opens up the possibility that shocks propagate through the network.
This seems plausible even for the sector-specific shocks listed above, such as
the Great East Japan Earthquake. Although the initial impact was focused on
electricity generation, the flow-on effects for other sectors were rapid and
substantial \citep{carvalho2021supply}.

To investigate more deeply, we connect our discussion of the granular
hypothesis back to the multisector models with linkages studied above,
allowing us to study flow-on and multiplier effects across industries.

\subsubsection{Industry Concentration and Shocks}\label{sss:ics}

From~\eqref{eq:msagout} and the independence of sectoral shocks, the standard
deviation $\sigma_g$ of log GDP is given by
\begin{equation}\label{eq:mssg}
    \sigma_g = \sigma H_n
        \quad \text{where } \;
        H_n := \left( \sum_{i=1}^n h_i^2 \right)^{1/2}.
\end{equation}
where $\sigma$ is the standard deviation of each $\epsilon_i$.

Note that the expression for aggregate volatility takes the same form
as \eqref{eq:gsdg} from our discussion of the granular hypothesis in
\S\ref{ss:granor}, where $H_n$ was called the Herfindahl
index\index{Herfindahl index}.  Once again, this index is the critical
determinant of how much firm-level volatility passes through to aggregate
volatility.  In particular, as discussed in \S\ref{sss:efht}, independent
firm-level shocks cannot explain aggregate volatility unless $H_n$ is large,
which in turn requires that the components of the vector $h$ are relatively
concentrated in a single or small number of sectors.

To investigate an extreme case, we recall from Exercise~\ref{ex:sboh} that
$\sum_{i=1}^n h_i = 1/\alpha$.  The Herfindahl index is $H_n := \| h\|$, where
$\| \cdot \|$ is the Euclidean norm.  

\begin{Exercise}
    Show that the minimizer of $\| h \|$ given $\sum_{i=1}^n h_i = 1/\alpha$
    is the constant vector where $h_i = 1/(\alpha n)$ for all $i$.
\end{Exercise}

\begin{Answer}
    Let $h^* := \1 / (\alpha n)$.  The claim is that $h^*$ is the minimizer of
    $\| h \|$ on $\RR^n_+$ under the constraint $\sum_{i=1}^n h_i = 1/\alpha$.
    Squaring the objective function and substituting the constraint into the
    objective by taking $h_n = 1/\alpha - h_1 - \cdots - h_{n-1}$, we are led to
    the equivalent problem of finding the minimizer of 
    \begin{equation*}
        f(h_1, \ldots, h_{n-1})
        := h_1^2 + \cdots + h_{n-1}^2 + 
        \left( \frac{1}{\alpha} - h_1 - \cdots - h_{n-1} \right)^2.
    \end{equation*}
    Since $f$ is convex, any local minimizer is
    a global minimizer.  Moreover, the first order conditions give 
    \begin{equation*}
        h_i = \left( \frac{1}{\alpha} - h_1 - \cdots - h_{n-1} \right)
            = h_n.
    \end{equation*}
    for all $i$.  Hence, the solution vector is constant over $i$.  Letting $c$ be
    this constant and using the constraint gives $n c = 1/\alpha$.  The claim
    follows.
\end{Answer}

Under this configuration of sector shares, 
\begin{equation*}
    H_n 
    = \frac{1}{n \alpha} \| \1 \|
    = \frac{1}{n \alpha} \sqrt{n}
    = O \left( \frac{1}{\sqrt{n}} \right).
\end{equation*}
Hence, by \eqref{eq:mssg}, we have $\sigma_g = O ( n^{-1/2} )$.
This is the classic diversification result.  The standard deviation of log GDP
goes to zero like $n^{-1/2}$, as in the identical firm size case
in~\S\ref{sss:gdpgrowth}.

Now let's consider the other extreme:

\begin{Exercise}
    Show that the maximum of $H_n = \|h\|$ under the constraint
    $\sum_{i=1}^n h_i = 1/\alpha$ is $1/\alpha$, attained by setting $h_k =
    1/\alpha$ for some $k$ and $h_j = 0$ for other indices.
\end{Exercise}

\begin{Answer}
    Let $h^*$ be an $n$-vector with $h^*_k = 1/\alpha$ for some $k$ and $h^*_j = 0$
    for other indices.  Clearly $\| h^* \| = 1/\alpha$.  Hence it suffices to
    show that, for any $h \in \RR^n_+$ with $\sum_{i=1}^n h_i = 1/\alpha$, we
    have $\| h \| \leq 1/\alpha$.

    Fix $h \in \RR^n_+$ with $\sum_{i=1}^n h_i = 1/\alpha$.  Since we are
    splitting $1/\alpha$ into $n$ parts, we can express $h$ as $
    (w_1/\alpha, \ldots, w_n/\alpha) $, where $0 \leq w_i \leq 1$ and $\sum_i w_i = 1$.
    By Jensen's inequality (Exercise~\ref{ex:jendiscrete}), we have 
    \begin{equation*}
        \sum_i \left( \frac{w_i}{\alpha } \right)^2
        \leq \left( \sum_i \frac{w_i}{\alpha} \right)^2 
        = \frac{1}{\alpha^2}.
    \end{equation*}
    Taking the square root gives $\| h \| \leq 1/\alpha$, as was to be shown.
\end{Answer}

This is the extreme case of zero diversification. By \eqref{eq:msagout}, log
GDP is then
\begin{equation*}
    g 
    = \sum_i \epsilon_i h_i 
    = \frac{1}{\alpha} \epsilon_k .
\end{equation*}
The volatility of log GDP is constant in $n$, rather than declining in the number of
sectors.  In other words, idiosyncratic and aggregate shocks are identical.

\subsubsection{The Role of Network Topology} 

In the previous section we looked at two extreme cases, neither of which is
realistic.  Now we look at intermediate cases.  In doing so, we note an
interesting new feature: unlike the analysis in \S\ref{ss:granor}, where the
sector shares were chosen from some fixed distribution, the Herfindahl index
is now determined by the network structure of production.  

To see this, we can use the results of Exercise~\ref{ex:dmrleo} to obtain
\begin{equation*}
    h = \frac{1}{n} L \1 = \frac{1}{n} \sum_{m \geq 0} A^m \1.
\end{equation*}
Recalling our discussion in \S\ref{sss:hbkc}, we see that the vector of Domar
weights is just a rescaling of the vector of hub-based Katz centrality
rankings for the input-output matrix.    Thus, the propagation of
sector-specific productivity shocks up to the aggregate level depends on the
distribution of Katz centrality across sectors.  The more ``unequal'' is this
distribution, the larger is the pass through.

Figure~\ref{f:authority_symmetric_nodes} shows some examples of
different network configurations, each of which is associated with a different
Katz centrality vector.  Not surprisingly, the symmetric network has
a constant centrality vector, so that all sectors have equal centrality.  This
is the maximum diversification case, as discussed in \S\ref{sss:ics}. 
The other two cases have nonconstant centrality vectors and hence greater
aggregate volatility.\footnote{The captions in Figure~\ref{f:authority_symmetric_nodes}
refer to these two cases as ``star networks,'' in line with recent usage in
multisector production models.  In graph theory, a star network is an \emph{un}directed graph
that (a) is strongly connected and (b) has only one node with degree greater than
one.}

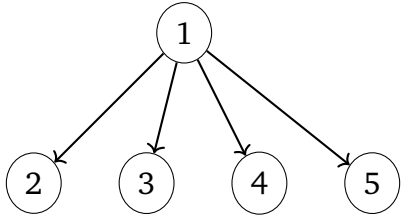
\begin{figure}
  \begin{center}
      \input{tikz/hub.tex}
  \end{center}
  \caption{\label{f:hub} Star network with single hub}
\end{figure}

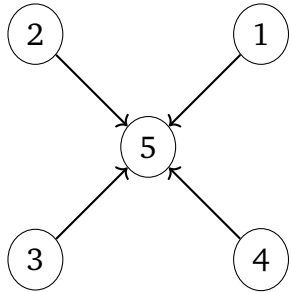
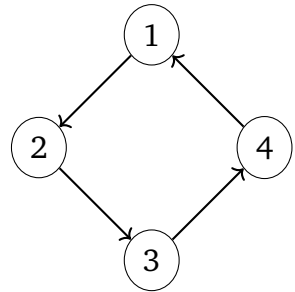
\begin{figure*}[t!]
    \centering
    \begin{subfigure}[t]{0.5\textwidth}
        \centering
          \input{tikz/authority.tex}
          \caption{Star network with single authority}
    \end{subfigure}%
    ~ 
    \begin{subfigure}[t]{0.5\textwidth}
        \centering
          \input{tikz/symmetric_nodes.tex}
          \caption{Symmetric network}
    \end{subfigure}
    \vspace{1em}
    \caption{\label{f:authority_symmetric_nodes} Symmetric and asymmetric networks}
\end{figure*}

\begin{Exercise}\label{ex:netophk}
    Let the nonzero input output coefficient $a_{ij}$ shown by arrows in
    Figure~\ref{f:authority_symmetric_nodes} all have equal value $0.2$.  Show
    computationally that the hub-based Katz centrality vectors for the hub and
    star networks are
    \begin{equation*}
        \kappa_h = (1.8, 1, 1, 1, 1)
        \quad \text{and} \quad
        \kappa_s = (1.2, 1.2, 1.2, 1.2, 1)
    \end{equation*}
    respectively (for nodes $1, \ldots, 5$).
\end{Exercise}

The results of Exercise~\ref{ex:netophk} show that the hub network has the
more unequal Katz centrality vector.  Not surprisingly, the source node in the
hub has a high hub-based centrality ranking.  Productivity shocks affecting
this sector have a large effect on aggregate GDP.

Figure~\ref{f:input_output_analysis_15_katz} shows hub-based Katz centrality
computed from the 15 sector input-output data for 2019.  We can see that
productivity shocks in manufacturing will have a significantly larger impact
on aggregate output than shocks in say, retail or education.

\begin{figure}
   \begin{center}
    \scalebox{0.64}{\includegraphics{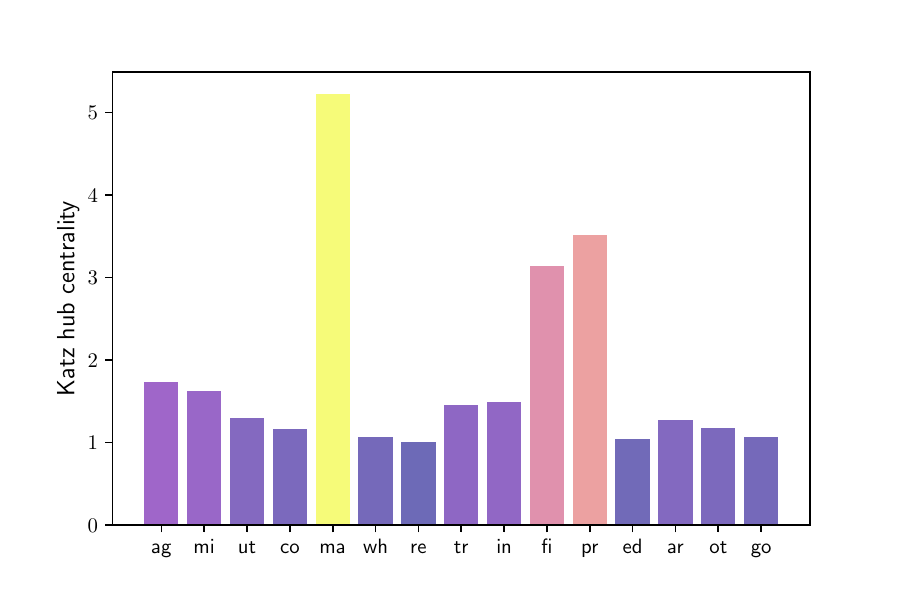}}
    \caption{\label{f:input_output_analysis_15_katz} Hub-based Katz centrality
    of across 15 US industrial sectors}
   \end{center}
\end{figure}

\subsubsection{Supply Shocks vs Demand Shocks}

In the previous section we saw that the degree to which a productivity shock
propagates through the economy depends on the hub-based centrality ranking of
the relevant sector.  This is intuitive.  Sectors that act like hubs
supply many sectors, so changes in productivity in these sectors will have
large flow-on effects.

This is in contrast with what we learned in~\S\ref{sss:outmul}, where a high
authority-based centrality measure lead to high shock propagation.  The
difference can be explained by the source of the shocks:
in~\S\ref{sss:outmul}, we were concerned with the impact of demand shocks.
Demand shocks to sector $i$ have large flow-on effects when many sectors
purchase inputs from sector $i$.  Hence authority-based centrality measures
are appropriate for studying this case.

\section{More Spectral Theory}\label{s:mst}

This is a relatively technical section, which analyzes aspects of vector
dynamics and spectral theory in more depth. It is aimed at readers who are
interested in further mathematical discussion of the theory treated above. The
ideas studied here are applied to production problems to generate additional
insights into existence and uniqueness of equilibria, as well as later topics
such as convergence of distributions for Markov models on networks.

\subsection{Vector Norms}\label{ss:vecnorms}

In this section we learn about abstract vector norms on $\RR^n$ and provide
several examples.  Later we will see how the different norms are related and how
they can be useful for some kinds of network analysis.

\subsubsection{Norms}\label{sss:norms}

A function $\| \cdot \| \colon \RR^n \to \RR$ is called a
\navy{norm}\index{Norm} on $\RR^n$ if, for any $\alpha \in \RR$ and $u, v \in \RR^n$, 
\begin{multicols}{2}
    \begin{enumerate}
        \item[(a)] $\| u \| \geq 0$
        \item[(b)] $\| u \| =0 \iff u=0$
        \item[(c)] $\| \alpha u \| = |\alpha| \| u\|$ and
        \item[(d)] $\| u + v \| \leq \| u \| + \| v \|$
        \item[] (nonnegativity) 
        \item[] (positive definiteness) 
        \item[] (positive homogeneity) 
        \item[] (triangle inequality) 
    \end{enumerate}
\end{multicols}

The Euclidean norm is a norm on $\RR^n$, as suggested by its name.     

\begin{example}\label{eg:ell1norm}
    The \navy{$\ell_1$ norm} of a vector $u \in \RR^n$ is defined by 
    \begin{equation}\label{eq:l1normfd}
        u = (u_1, \ldots, u_n) \mapsto \| u \|_1 := \sum_{i=1}^n |u_i|.
    \end{equation}
    In machine learning applications, $\| \cdot \|_1$ is sometimes called the
    ``Manhattan norm,'' and $d_1 (u, v) := \| u - v \|_1$ is called the
    ``Manhattan distance'' or ``taxicab distance'' between vectors $u$ and
    $v$. We will refer to it more simply as the \navy{$\ell_1$ distance} or
    \navy{$\ell_1$ deviation}.
\end{example}

\begin{Exercise}
    Verify that the $\ell_1$ norm on $\RR^n$ satisfies (a)--(d) above.
\end{Exercise}

The $\ell_1$ norm and the Euclidean norm are special cases of the so-called
\navy{$\ell_p$ norm}, which is defined for $p \geq 1$ by 
\begin{equation}\label{eq:lpnormfd}
    u = (u_1, \ldots, u_n) \mapsto 
    \| u \|_p := \left( \sum_{i=1}^n |u_i|^p \right)^{1/p}.
\end{equation}
It can be shown that $u \mapsto \| u \|_p$ is a norm for all $p \geq 1$, as
suggested by the name  (see, e.g., \cite{kreyszig1978introductory}). For this
norm, the subadditivity in (d) is called \navy{Minkowski's
inequality}\index{Minkowski's inequaltiy}.

Since the Euclidean case is obtained by setting $p=2$, the Euclidean norm is
also called the $\ell_2$ norm, and we write $\| \cdot \|_2$ rather than $\|
\cdot \|$ when extra clarity is required.

\begin{Exercise}\label{ex:ellinftyfd}
    Prove that $u \mapsto \| u \|_\infty := \max_{i=1}^n |u_i|$ is also a norm on $\RR^n$.
\end{Exercise}

(The symbol $\| u \|_\infty$ is used because, $\forall \, u \in \RR^n$, we have
$\| u \|_p \to \| u \|_\infty$ as $p \to \infty$.)  This norm is called the
\navy{supremum norm}\index{Supremum norm}.

\begin{Exercise}
    The so-called $\ell_0$ ``norm'' $\| u \|_0 := \sum_{i=1}^n \1 \{u_i \not=
    0\}$, routinely used in data science applications, is \emph{not} in fact a norm on
    $\RR^n$.  Prove this.
\end{Exercise}

\begin{Answer}
    For $\alpha > 0$ we always have $\| \alpha u \|_0 = \| u \|_0$, which
    violates positive homogeneity.
\end{Answer}

\subsubsection{Equivalence of Vector Norms}\label{sss:eqvecnorms}

When $u$ and $(u_m) := (u_m)_{m \in \NN}$ are all elements of $\RR^n$, we say
that $(u_m)$ \navy{converges}\index{Convergence of vectors} to $u$ and write
$u_m \to u$ if
\begin{equation*}
    \| u_m - u \| \to 0 
    \text{ as } m \to \infty \text{ for some norm } \| \cdot \|
    \text{ on } \RR^n.
\end{equation*}

It might seem that this definition is imprecise.  Don't we need to clarify
that the convergence is with respect to a particular norm?

In fact we do not.  This is because
any two norms $\| \cdot \|_a$ and $\| \cdot \|_b$ on
$\RR^n$ are \navy{equivalent}\index{Equivalence}, in the sense that
there exist finite constants $M, N$ such that
\begin{equation}\label{eq:eqvecnorms}
    M \|u\|_a \leq \| u\|_b \leq N \| u \|_a \quad \text{for all } u \in \RR^n.
\end{equation}

\begin{Exercise}
    Let us write $\| \cdot \|_a \sim \| \cdot \|_b$ if there exist finite $M, N$ such
    that \eqref{eq:eqvecnorms} holds.  Prove that $\sim$ is
    an equivalence relation on the set of norms on $\RR^n$.
\end{Exercise}

\begin{Exercise}\label{ex:cciseq}
    Let $\| \cdot \|_a$ and $\| \cdot \|_b$ be any two norms on $\RR^n$.
    Given a point $u$ in $\RR^n$ and a sequence $(u_m)$
    in $\RR^n$, use~\eqref{eq:eqvecnorms} to confirm that 
        $\| u_m - u \|_a \to 0$ implies $\| u_m - u \|_b \to 0$ as $m \to \infty$.
\end{Exercise}


Another way to understand $u_m \to u$ is via \navy{pointwise
convergence}\index{Pointwise convergence}: each element of the vector sequence
$u_m$ converges to the corresponding element of $u$.  Pointwise and norm
convergence are equivalent, as the next result makes clear.

\begin{lemma}\label{l:eqconvec}
    Fix $(u_m) \subset \RR^n$, $u \in \RR^n$  and norm
    $\| \cdot \|$ on $\RR^n$. Taking $m \to \infty$, the following
    statements are equivalent:
    \begin{enumerate}
        \item $\| u_m - u \| \to 0$,
        \item $\inner{a, u_m} \to \inner{a, u}$ for all $a \in \RR^n$, and
        \item $(u_m)$ converges pointwise to $u$.
    \end{enumerate}
\end{lemma}

\begin{Exercise}
    Prove Lemma~\ref{l:eqconvec}.
\end{Exercise}

\begin{Answer}
    Let $\| \cdot \|$ be a norm on $\RR^n$, let $(u_m)$ be a sequence in
    $\RR^n$ and let $u$ be a point in $\RR^n$.  To show
    that (i) implies (ii) we fix $a \in \RR^n$ and observe that
    \begin{equation*}
        \text{for all } m \in \NN, \;\;
        |\inner{a, u_m} - \inner{a, u}| 
            = |\inner{a, u_m -u}| 
            \leq \| u_m - u \|_1 \max_i | a_i |.
    \end{equation*}
    Hence convergence in $\ell_1$ norm implies (ii).  In view of
    Exercise~\ref{ex:cciseq}, convergence in $\| \cdot \|$ implies convergence
    in $\ell_1$, so (i) $\implies$ (ii) is confirmed.

    To show that (ii) implies (iii) at the $j$-th
    component, we just specialize $a$ to the $j$-th canonical basis vector.
    Finally, to show that (iii) implies (i), we first note that, by the
    equivalence of norms, it is enough to show that pointwise convergence
    implies $\ell_1$ convergence; that is,
    \begin{equation*}
        \text{(iii) }
        \implies
        \| u_m - u \|_1 
        = \sum_{j \in \natset{n}} |\inner{\delta_j, u_m - u}|
        \to 0,
    \end{equation*}
    where $\delta_j$ is the $j$-th canonical basis vector.
    To prove that this sum converges to zero it suffices to show that every element of
    the sum converges to zero (see \S\ref{sss:ocon}), which is true by (iii).
\end{Answer}

\begin{Exercise}
    Using Lemma~\ref{l:eqconvec}, provide a simple proof of the
    fact that convergence in $\RR^n$ is preserved under addition and
    scalar multiplication:  if $u_m \to x$ and $v_m \to y$ in $\RR^n$, while
    $\alpha_m \to \alpha$ in $\RR$, then $u_m + v_m \to x + y$ and $\alpha_m
    u_m \to \alpha x$.
\end{Exercise}

\begin{Answer}
    Consider the claim that $u_m \to x$ and $v_m \to y$ in $\RR^n$ implies
    $u_m + v_m
    \to x + y$.  We know this is true in the scalar case $n=1$.  Moreover,
    Lemma~\ref{l:eqconvec} tells us that convergence in $\RR^n$ holds 
    if and only if it holds componentwise---which is just the scalar case.
    The rest of the proof is similar.
\end{Answer}

\subsection{Matrix Norms}\label{ss:matnorip}

In some applications, the number of vertices $n$ of a given digraph is
measured in millions or billions.  This means that the adjacency matrix $A$ is
enormous. To control complexity, $A$ must often be replaced by a sparse
approximation $A_s$.  It is natural to require that $A$ and $A_s$ are close.
But how should we define this?

More generally, how can we impose a metric on the set of matrices to determine
similarity or distance between them? One option is to follow the example of
vectors on $\RR^n$ and introduce a norm on $\matset{n}{k}$.  With such a norm
$\| \cdot \|$ in hand, we can regard $A$ and $A_s$ as close when $\| A - A_s
\|$ is small.  

For this and other reasons, we now introduce the notion of a matrix norm.

\subsubsection{Definition}

Analogous to vectors on $\RR^n$, we will call a function $\| \cdot \|$ from
$\matset{n}{k}$ to $\RR_+$ a \navy{norm}\index{Matrix norm} if for any $A, B
\in \matset{n}{k}$,
\begin{multicols}{2}
    \begin{enumerate}
        \item[(a)] $\| A \| \geq 0$
        \item[(b)] $\| A \| =0 \iff u=0$
        \item[(c)] $\| \alpha A \| = |\alpha| \| A\|$ and
        \item[(d)] $\| A + B \| \leq \| A \| + \| B \|$
        \item[] (nonnegativity) 
        \item[] (positive definiteness) 
        \item[] (positive homogeneity) 
        \item[] (triangle inequality) 
    \end{enumerate}
\end{multicols}
The distance between two matrices $A, B$ is then specified as $\| A - B \|$.

Unlike vectors, matrices have a product operation defined over all conformable
matrix pairs.  We want matrix norms to interact with this product in a
predictable way.  For example, it is helpful for analysis when a matrix norm
$\| \cdot \|$ is \navy{submultiplicative}\index{Submultiplicative norm},
meaning that
\begin{equation}\label{eq:smnm}
    \| A B \| \leq \| A \| \cdot \| B \|
    \;\;
    \text{ for all conformable matrices } A, B.
\end{equation}
A useful implication of \eqref{eq:smnm} is that $\|A^i\| \leq \|A \|^i$ for
any $i \in \NN$ and $A \in \matset{n}{n}$, where $A^i$ is the $i$-th
power of $A$.

\subsubsection{The Frobenius Norm}\label{sss:frobnorm}

One way to construct a norm on matrix space $\matset{n}{k}$ is to first introduce
the \navy{Frobenius inner product}\index{Frobenius inner product}
of matrices $A = (a_{ij}), B=(b_{ij})$ as
\begin{equation}\label{eq:frobin}
    \inner{A, B}_F 
        := \sum_{i=1}^n \sum_{j=1}^n a_{ij} b_{ij}
        = \trace(A B^\top)
        = \trace(B A^\top).
\end{equation}
From this inner product, the \navy{Frobenius norm}\index{Frobenius norm} of $A
\in \matset{n}{k}$ is defined as
\begin{equation}\label{eq:frobnorm}
    \| A \|_F 
    :=  \inner{A, A}^{1/2}_F .
\end{equation}
In essence, the Frobenius norm converts an $n \times k$ matrix into an $nk$
vector and computes the Euclidean norm.

The Frobenius norm is submultiplicative.   The next exercise asks you to prove
this in one special case.

\begin{Exercise}\label{ex:smcsch}
    Suppose that $A$ is a row vector and $B$ is a column vector.  
    Show that \eqref{eq:smnm} holds in this case when $\| \cdot \|$ is the Frobenius norm.
\end{Exercise}

\begin{Answer}
    The Frobenius norm reduces to the Euclidean norm for column and row vectors, so 
    \eqref{eq:smnm} requires that $| \inner{A, B}_F| \leq
    \| A \|_F \|B\|_F$.  This bound certainly holds: it is the
    Cauchy--Schwarz inequality.
\end{Answer}

\subsubsection{The Operator Norm}\label{sss:onorm}

Another important matrix norm is the \navy{operator norm}\index{Operator
norm}, defined at $A \in \matset{n}{k}$ by
\begin{equation}
    \label{eq:defsn}
    \| A \| 
    := \sup \setntn{ \| A u \| }{ u \in \RR^k \text{ and } \| u \| = 1},
\end{equation}
where the norm $\| \cdot \|$ on the right hand size of \eqref{eq:defsn} is the
Euclidean norm on $\RR^n$.

\begin{example}\label{eg:normdiag}
    If $A = \diag(a_i)$, then, for any $u \in \RR^n$ we have
    $\|Au\|^2 = \sum_i (a_i u_i)^2$.  To maximize this value subject to $\sum_i u_i^2
    = 1$, we pick $j$ such that $a_j^2 \geq a_i^2$ for all $i$ and
    set $u_i = \1\{i=j\}$.  The maximized value of $\| Au\|$ is then 
    \begin{equation*}
        \| A\| = \sqrt{a_j^2} = \max_{i \in \natset{n}} \, |a_i|.
    \end{equation*}
\end{example}

\begin{Exercise}\label{ex:sures}
    Show that $\| A \|$ equals the supremum of $\| A u \|/ \| u \|$ over all
    $u \not= 0$.
\end{Exercise}

\begin{Answer}
    Let    
    \begin{equation*}
        a :=
        \sup_{u \not= 0} f(u) 
        \quad \text{where} \quad
        f(u) := \frac{\| A u \|}{ \| u \|} 
        \qquad \text{and let } 
        b := \sup_{\| u \| = 1} \| A u \|
    \end{equation*}
    Evidently $a \geq b$ because the supremum is over a larger domain.  To see
    the reverse fix $\epsilon > 0$ and let $u$ be a nonzero vector such that
    $f(u) > a - \epsilon$. Let $\alpha := 1 / \| u\|$ and let $u_b := \alpha u$.  Then 
    \begin{equation*}
        b 
        \geq 
        \| A u_b \|
        = \frac{ \| A u_b \| }{ \| u_b \|}
        = \frac{ \| \alpha A u \| }{ \| \alpha u \|}
        = \frac{\alpha}{\alpha} \frac{\| A u \|}{ \| u \|}
        = f(u) > a - \epsilon
    \end{equation*}
    Since $\epsilon$ was arbitrary we have $b \geq a$.
\end{Answer}

\begin{Exercise}
    It is immediate from the definition of the operator norm that 
    \begin{equation}
        \label{eq:smnv}
        \| A u \| \leq \| A \| \cdot \| u \|
        \qquad \forall \, u \in \RR^n.
    \end{equation}
    Using this fact, prove that the operator norm is submultiplicative.
\end{Exercise}

\begin{Answer}
    Let $A$ and $B$ be elements of $\matset{n}{k}$ and $\matset{n}{j}$ respectively.
    Fix $v \in \RR^n$.  
    Since $\| A u \| \leq \| A \| \cdot \| u \|$ for any vector $u$, we have
        $\| A B v \| 
        \leq \| A \| \cdot \| B v \|
        \leq \| A \| \cdot \| B \| \cdot \| v \|$, from which~\eqref{eq:smnm}
        easily follows.
\end{Answer}

\begin{Exercise}\label{ex:srsn}
    Let $\| \cdot \|$ be the operator norm on $\matset{n}{n}$. Show that, for
    each $A \in \matset{n}{n}$, we have
    \begin{enumerate}
        \item $\| A \|^2 = r(A^{\top} A)$,
        \item $r(A)^k \leq \| A^k \|$ for all $k \in \NN$, and
        \item $\| A^{\top} \| = \| A \|$.
    \end{enumerate}
\end{Exercise}

\subsubsection{Other Matrix Norms}

Two other useful matrix norms are the $\ell_1$ and $\ell_\infty$ norms given
by
\begin{equation*}
    \| A \|_1 := \sum_{i=1}^n \sum_{j=1}^k |a_{ij}|
    \quad \text{and} \quad
    \| A \|_\infty := \max_{i \in \natset{n}, \, j \in \natset{k}} |a_{ij}|.
\end{equation*}

\begin{Exercise}
    Prove that both are norms on $\matset{n}{k}$.
\end{Exercise}

\begin{Exercise}
    Prove that both norms are submultiplicative.
\end{Exercise}

\subsubsection{Equivalence of Matrix Norms}\label{sss:eqmatnorms}

In \S\ref{sss:eqvecnorms} we saw that all norms on $\RR^n$ are equivalent.
Exactly the same result holds true for matrix norms: any two norms
$\| \cdot \|_a$ and $\| \cdot \|_b$ on $\matset{n}{k}$ are
\navy{equivalent}\index{Equivalence}, in the sense that there exist finite
constants $M, N$ such that
\begin{equation}\label{eq:eqmatnorms}
    M \|A\|_a \leq \| A\|_b \leq N \| A \|_a \quad \text{for all } A \in \matset{n}{k}.
\end{equation}
A proof can be found (for abstract finite-dimensional vector space) in \cite{bollobas1999linear}.

Analogous to the vector case, given $A$ and $(A_m) := (A_m)_{m \in \NN}$ in
$\matset{n}{k}$, we say that $(A_m)$ \navy{converges}\index{Convergence of
matrices} to $A$ and write $A_m \to A$ if $\| A_m - A \| \to 0$ as $m \to
\infty$, where $\| \cdot \|$ is a matrix norm.  Once again, we do not need to
clarify the norm due to the equivalence property.
Also, norm convergence is equivalent to pointwise convergence:

\begin{Exercise}\label{ex:mnormpw}
    Prove that, given $A$ and $(A_m)$ in
    $\matset{n}{k}$, we have $A_m \to A$ as $m \to \infty$ if and only if 
    every element $a_{ij}^m$ of $A_m$ converges to the corresponding element
    $a_{ij}$ of $A$ 
\end{Exercise}

\begin{Exercise}\label{ex:rrmas}
    Prove the following result for matrices $A, B, C$ and matrix sequences
    $(A_m), (B_m)$, taking $m \to \infty$ and assuming sizes are conformable:
    \begin{enumerate}
        \item If $B_m \to A$ and $A_m - B_m \to 0$, then $A_m \to A$.
        \item If $A_m \to A$, then $B A_m C \to B A C$.
    \end{enumerate}
\end{Exercise}

\begin{Answer}
    Let $(A_m)$ and $(B_m)$ have the stated properties.  Regarding (i), we use the triangle
    inequality to obtain
    \begin{equation*}
        \| A_m - A \|
        = \| A_m - B_m + B_m - A \|
        \leq \| A_m - B_m \| +  \| B_m - A \|.
    \end{equation*}
    Both terms on the right converge to zero, which completes the proof.

    Regarding (ii), we use the submultiplicative property to obtain 
    $ \| B A_m C - B A C \| \leq \| B \| \| A_m - A \| \| C \| \to 0$.
\end{Answer}

\begin{example}
    Convergence of the Perron projection
    in \eqref{eq:patocon} of the Perron--Frobenius theorem was defined using
    pointwise convergence.  By Exercise~\ref{ex:mnormpw}, norm convergence
    also holds.  One of the advantages of working with norms is that we can
    give rates of convergence for norm deviation.  This idea is discussed
    further in \S\ref{sss:pfrate}.
\end{example}

\begin{Exercise}\label{ex:mnormqd}
    Given $A$ and $(A_m)$ in
    $\matset{n}{n}$, prove that $A_m \to A$ as $m \to \infty$ if and only if 
    $s^\top A s$ for every $s \in \RR^n$.
\end{Exercise}

\subsection{Iteration in Matrix Space}\label{sss:mnspec}

Results such as Proposition~\ref{p:accesspos} on page~\pageref{p:accesspos}
already showed us the significance of powers of adjacency matrices. The
Perron-Frobenius theorem revealed connections between spectral radii, dominant
eigenvectors and powers of positive matrices.   In this section we investigate
powers of matrices in more depth.

\subsubsection{Gelfand's Formula}

One very general connection between matrix powers and spectral radii
is \navy{Gelfand's formula}\index{Gelfand's formula} for the spectral radius:

\begin{theorem}\label{t:gf}
    For any matrix norm $\| \cdot \|$ and $A \in \matset{n}{n}$, we have
    \begin{equation}\label{eq:gelform}
        r(A) = \lim_{k \to \infty} \| A^k \|^{1/k}.
    \end{equation}
\end{theorem}

Proofs can be found in \cite{bollobas1999linear} or \cite{kreyszig1978introductory}.

\begin{Exercise}
    The references above show that the limit \eqref{eq:gelform} always exists.
    Using this fact, prove that every choice of norm over $\matset{n}{n}$
    yields the same value for this limit.
\end{Exercise}

\begin{Answer}
    Let $\| \cdot \|_a$ and $\| \cdot \|_b$ be two norms on $\matset{n}{n}$. 
    By the results in~\S\ref{sss:eqmatnorms}, these norms are equivalent, so
    there exist constants $M, N$ such that 
    $\|A^k \|_a \leq M \| A^k \|_b \leq N \| A^k  \|_a$ for all $k \in \NN$.
    \begin{equation*}
        \fore
        \|A^k \|_a^{1/k} \leq M^{1/k} \| A^k \|_b^{1/k} \leq N^{1/k} \| A^k \|_a^{1/k} 
    \end{equation*}
    for all $k \in \NN$.  Taking $k \to \infty$, we see that the definition of the
    spectral radius is independent of the choice of norm.
\end{Answer}

The next exercise shows that $r(A)<1$ implies $\|A^k\| \to 0$ at a geometric
rate.

\begin{Exercise}\label{ex:igfo}
    Using \eqref{eq:gelform}, show that $r(A) < 1$ implies the existence of a
    constant $\delta < 1$ and an $M < \infty$ such that $\| A^k \| \leq
    \delta^k M$ for all $k \in \NN$.
\end{Exercise}

\begin{Answer}
    Since $r(A) < 1$, we can find a constant $K$ and an $\epsilon > 0$ such
    that $k \geq K$ implies $\| A^k \| < (1 - \epsilon)^k$.  Setting $M :=
    \max_{k \leq K} \| A^k \|$ and $\delta := 1-\epsilon$ produces the desired
    constants.
\end{Answer}

\begin{Exercise}
    Consider the dynamic system $x_t = A x_{t-1} + d$ with $x_0$ given, where
    each $x_t$ and $d$ are vectors in $\RR^n$ and $A$ is $n \times n$.  (If
    you like, you can think of this process as orders flowing backwards
    through a production network.)  Show that, when $r(A)<1$, the sequence
    $(x_t)_{t \geq 0}$ converges to $x^* := (I - A)^{-1} d$, independent of
    the choice of $x_0$.
\end{Exercise}

\begin{Answer}
    Suppose $r(A) < 1$.
    Iterating backwards on $x_t = A x_{t-1} + d$ yields $x_t = d + Ad + \cdots
    + A^{t-1} d + A^t x_0$.  By the Neumann series lemma, we have $x^* =
    \sum_{t \geq 0} A^t d$, so 
    \begin{equation*}
        x^* - x_t 
        = \sum_{j > t} A^j d - A^t x_0.
    \end{equation*}
    Hence, with $\|\cdot \|$ as both the Euclidean vector norm and the matrix
    operator norm, we have
    \begin{equation*}
        \| x^* - x_t \|
        \leq \left\| \sum_{j > t} A^j d - A^t x_0 \right\|
        \leq \sum_{j > t} \|  A^j \| \| d \| - \| A^t \| \| x_0 \|.
    \end{equation*}
    Using $r(A) < 1$ again, it now follows from Exercise~\ref{ex:igfo} that $\|
    x^* - x_t \| \to 0$ as $t \to \infty$.
\end{Answer}

\begin{Exercise}
    In \S\ref{ss:eqmulm} we studied a production coefficient matrix of the form
    $A = (a_{ij})$ in $\matset{n}{n}$ where each $a_{ij}$ takes values in
    $(0,\infty)$ and, for each $j$, $\alpha + \sum_i a_{ij}=1$ for some $\alpha > 0$.
    We can calculate $r(A)$ using the following strategy.
    In Exercise~\ref{ex:1malpha} we saw that 
        $\sum_i \sum_j a_{ij}^{(m)} = n (1-\alpha)^m$ for
            all $m \in \NN$, where $a_{ij}^{(m)}$ is the $(i,j)$-th element of $A^m$.
    Using the fact that $\| B \|_1 := \sum_i \sum_j |b_{ij}|$ is a matrix
            norm, apply Gelfand's formula to obtain $r(A) = 1-\alpha$.
\end{Exercise}

\begin{Answer}
    We have $\| A^m \|_1 = n (1-\alpha)^m$ and so
    $\| A^m \|_1^{1/m} = n^{1/m} (1-\alpha)$.  Taking $m \to \infty$ gives
    $r(A)=1-\alpha$.
\end{Answer}

\subsubsection{A Local Spectral Radius Theorem}\label{sss:lsr}

The next theorem is a ``local'' version of Gelfand's formula that relies on
positivity.  It replaces matrix norms with vector norms, which
are easier to compute.

\begin{theorem}\label{t:local_sr}
    Fix $A \in \matset{n}{n}$ and let $\| \cdot \|$ be any norm on $\RR^n$.
    If $A \geq 0$ and $x \gg 0$, then
    \begin{equation}\label{eq:local_sr}
        \| A^m x\|^{1/m} \to r(A)
        \qquad (m \to \infty).
    \end{equation}
\end{theorem}

Theorem~\ref{t:local_sr} tells us that,  eventually, for any positive $x$, the
norm of the vector $A^m x$ grows at rate $r(A)$.  A proof can be found in
\cite{krasnoselskiipositive}.\footnote{A direct proof of a generalized version
    of Theorem~\ref{t:local_sr} is provided in Theorem~B1 of
\cite{borovivcka2020necessary}.}

\begin{example}
    In \S\ref{ss:dshocks} we studied how the impact of a given demand shock
    $\Delta d$ flows backward through a production network, with $A^m (\Delta
    d)$ giving the impact on sectors at $m$ steps (backward linkages).  When
    $\Delta d \gg 0$ and $r(A) < 1$,  Theorem~\ref{t:local_sr}
    tells us that $\| A^m \Delta d\| = O(r(A)^m)$.  If we set the norm to $\|
    \cdot \|_\infty$, this tells us that 
    the maximal impact of demand shocks through backward linkages fades at rate $r(A)$.
\end{example}

\begin{Exercise}
    Prove Theorem~\ref{t:local_sr} in the case where $A$ is primitive.
\end{Exercise}

\begin{Answer}
    Let $\| \cdot \|$ be a norm on $\RR^n$.  From the Perron--Frobenius
    theorem (Theorem~\ref{t:pf}), when $A$ is primitive, $\| r(A)^{-m} A^m x
    \| \to c$ as $m \to \infty$, where $c > 0$ whenever $x \gg 0$. Hence
    $\lim_{m \to \infty} \| A^m x \|^{1/m} = \lim_{m \to \infty} r(A) c^{1/m}
    = r(A)$.
\end{Answer}

\subsubsection{Convergence to the Perron Projection}\label{sss:pfrate}

The local spectral radius theorem assumes $A \geq 0$.  Now we further
strengthen this assumption by requiring that $A$ is primitive.  In this case,
$r(A)^{-m} A^m$ converges to the Perron projection as $m \to \infty$ (see
\eqref{eq:patocon}).  We want  \emph{rates} of convergence in high-dimensional settings.

We fix $A \in \matset{n}{n}$ and label the eigenvalues so that
$|\lambda_{i+1}| \leq |\lambda_i|$ for all $i$.  Note that $|\lambda_1| =
\lambda_1 = r(A)$.  Let $E := e \, \epsilon^\top$ be the Perron projection.

\begin{proposition}
    If $A$ is diagonalizable and primitive, then $\alpha := |\lambda_2/\lambda_1| < 1$ and
    \begin{equation}\label{eq:rconpp}
        \| r(A)^{-m} A^m - E  \| = O \left( \alpha^m \right).
    \end{equation}
\end{proposition}

Thus, an upper bound on the rate of convergence to the Perron projection is
determined by the  modulus of  the ratio of the first two eigenvalues.

\begin{proof}
    We saw in~\eqref{eq:aspecrep} that $A^m =
    \sum_{i=1}^n \lambda_i^m e_i \epsilon_i^\top$ for all $m \in
    \NN$.  From this spectral representation we obtain
    \begin{equation*}
        r(A)^{-m} A^m - e_1 \, \epsilon_1^\top 
            = r(A)^{-m} \left( A^m - r(A)^m e_1 \, \epsilon_1^\top \right)
            = \sum_{i=2}^n \theta_i^m  e_i
            \epsilon_i^\top
    \end{equation*}
    when $\theta_i := \lambda_i/r(A)$.  Let $\| \cdot \|$ be the operator norm on
    $\matset{n}{n}$.  The triangle inequality gives
    \begin{equation*}
        \| r(A)^{-m} A^m - e_1 \, \epsilon_1^\top  \|
        \leq \sum_{i=2}^n | \theta_i| ^m  \| e_i \epsilon_i^\top \|
        \leq | \theta_2| ^m  \sum_{i=2}^n  \| e_i \epsilon_i^\top \|.
    \end{equation*}
    Since $A$ is primitive,  the Perron--Frobenius theorem tells us that
    $|\lambda_2| < r(A)$. Hence $\alpha := |\theta_2| < 1$.  The proof is now
    complete.
\end{proof}

\subsection{Exact Stability Conditions}

The Neumann series lemma tells us that the linear system $x = Ax + d$ has a
unique solution whenever $r(A) < 1$.  We also saw that, in the input-output
model, where $A$ is the adjacency matrix, the condition $r(A)<1$ holds
whenever every sector has positive value added (Assumption~\ref{a:pva} and
Exercise~\ref{ex:eara}).  Hence we have sufficient conditions for stability.

This analysis, while important, leaves open the question of how tight the
conditions are and what happens when they fail.  For example, we might ask
\begin{enumerate}
    \item To obtain $r(A) < 1$, is it necessary that each sector has positive
        value added?  Or can we obtain the same result under weaker
        conditions?
    \item What happens when $r(A) < 1$ fails?  Do we always lose existence of
        a solution, or uniqueness, or both?
\end{enumerate}
In \S\ref{sss:srsm} and \S\ref{ss:nslconv} below we address these two questions.

\subsubsection{Spectral Radii of Substochastic Matrices}\label{sss:srsm}

To reiterate, the results in \S\ref{sss:pwconst} relied on the assumption that
every sector has positive value added, which in turn gave us the property
$r(A) < 1$ for the adjacency matrix of the input-output production network.
Positive value added in every sector is not necessary, however. Here we
investigate a weaker condition that is exactly necessary and sufficient for
$r(A)<1$.  This weaker condition is very helpful for understanding other kinds
of networks too, including financial networks, as discussed in
\S\ref{s:finnet}.

To begin, recall that a matrix $P \in \matset{n}{n}$ is called stochastic if
$P \geq 0$ and $P \1 = \1$.   Similarly, $P \in \matset{n}{n}$ is called
\navy{substochastic}\index{Substochastic matrix} if $P \geq 0$ and $P \1 \leq
\1$.    Thus, a substochastic matrix is a nonnegative matrix with less
than unit row sums.  Clearly the transpose $Q^\top$ of a nonnegative matrix $Q$
is substochastic if and only if $Q$ has less than unit column sums.

A natural example of a substochastic matrix is the transpose $A^\top$ of an
adjacency matrix $A$ of an input-output network. Indeed, such an $A$ is
nonnegative and, for $j$-th column sum we have
\begin{equation*}
   \sum_i  a_{ij} 
    =  \frac{\sum_i z_{ij}}{x_j}
    = \frac{\text{spending on inputs by sector $j$}}{\text{total sales of sector $j$}}.
\end{equation*}
Hence $\sum_i  a_{ij} \leq 1$, which says that spending on intermediate goods
and services by a given industry does not exceed total sales revenue, is a
necessary condition for nonnegative profits in sector $j$.
When this holds for all $j$, the adjacency matrix has less than unit column
sums.

From Lemma~\ref{l:rscsbounds} on page~\pageref{l:rscsbounds}, we see that any
substochastic matrix $P$ has 
\begin{equation}\label{eq:rprs}
    r(P) \leq \max_i \rsum_i(P) \leq 1.
\end{equation}
We wish to know when we can tighten this to $r(P) < 1$.

From \eqref{eq:rprs}, one obvious sufficient condition is that $\rsum_i(P) <
1$ for all $i$.  This is, in essence, how we used Assumption~\ref{a:pva}
(every sector has positive value added) in the input-output model.  Under that
condition we have $\sum_i  a_{ij} < 1$ for all $j$, which says all column sums
are strictly less than one.  Hence
\begin{equation*}
    \max_i \csum_i(A) < 1
    \; \iff \;
    \max_i \rsum_i(A^\top) < 1
    \; \implies \;
    r(A^\top) < 1
    \; \iff \;
    r(A) < 1,
\end{equation*}
where the middle implication is by \eqref{eq:rprs} and the last equivalence is
by $r(A)=r(A^\top)$.

Now we provide a weaker---in fact necessary and sufficient---condition for
$r(P)<1$, based on network structure.  To do so, we define an $n \times n$
substochastic matrix $P = (p_{ij})$ to be \navy{weakly chained substochastic}
if, for all $m \in \natset{n}$, there exists an $i \in \natset{n}$ such that
$m \to i$ and $\sum_j p_{ij} < 1$.  Here accessibility of $i$
from $m$ is in terms of the induced weighted digraph.\footnote{See
    \S\ref{sss:gd} for the definition of the induced weighted digraph.  By
    Proposition~\ref{p:accesspos}, accessibility of $i$ from $m$ is equivalent
to existence of a $k$ with $p^k_{mi} > 0$, where $p^k_{mi}$ is the $m,i$-th
element of $P^k$.}

\begin{Exercise}\label{ex:awcs}
    Let $A = (a_{ij}) \in \matset{n}{n}$ be nonnegative.  Show that $A^\top$
    is weakly chained substochastic if and only if $A$ has less than unit column
    sums and, for each $m \in \natset{n}$, there exists an $i \in \natset{n}$
    such that $i \to m$ under (the digraph induced by) $A$ and $\sum_k a_{ki} < 1$.
\end{Exercise}

\begin{Answer}
    We prove only that the stated conditions on $A$ imply that $A^\top$ is
    weakly chained substochastic, since the proof of the reverse implication
    is very similar.  We set $a'_{ij}:=a_{ji}$, so that $A^\top = (a'_{ij})$.  

    Let $A$ have the stated properties.  Since $A$ has less than unit column
    sums, $A^\top$ has less than unit rows, so $A^\top$ is substochastic.  

    Now fix $m \in \natset{n}$ and take
    $i \in \natset{n}$ such that $i \to m$ under $A$ and $\sum_k a_{ki} < 1$.
    By Exercise~\ref{ex:rar} on page~\pageref{ex:rar}, 
    $i \to m$ under $A$ is equivalent to $m \to i$ under $A^\top$.
    Moreover $\sum_k a_{ki} < 1$ is equivalent to $\sum_k a'_{ik} < 1$.
    Hence $A^\top$ is weakly chained substochastic.
\end{Answer}

\begin{proposition}\label{p:wcs}
    For a substochastic matrix $P$, we have
    \begin{equation*}
        r(P) < 1 
        \quad \iff \quad
        \text{$P$ is weakly chained substochastic}.
    \end{equation*}
\end{proposition}

A proof can be found in Corollary~2.6 of \cite{azimzadeh2019fast}.

Now we return to the input output model.  Let's agree to say that sector $i$
is an \navy{upstream supplier} to sector $j$ if $i \to j$  in the input-output
network.  By Proposition~\ref{p:accesspos} on page~\pageref{p:accesspos}, an
equivalent statement is that there exists a $k \in \NN$ such that $a^k_{ij} >
0$.

\begin{Exercise}
    Let $A$ be the adjacency matrix of an input-output network and assume that
    value added is nonnegative in each sector.  Using Proposition~\ref{p:wcs},
    show that $r(A) < 1$ if and only if each sector in the network
    has an upstream supplier with positive value added.
\end{Exercise}

\begin{Answer}
    Let $A$ be the adjacency matrix of an input-output network such that value
    added is nonnegative in each sector.  In what follows, we write $a'_{ij}$ for
    the $i,j$-th element of $A^\top$, so that $a'_{ij} = a_{ji}$.  

    Let's say that $A$ has property U if,
    for each sector in the network, there exists an upstream supplier with
    positive value added.  Property U is equivalent to the statement that,
    for all $m \in \natset{n}$, there is an $i \in \natset{n}$ with $i \to m$
    under $A$ and $\sum_k a_{ki} < 1$.  By Exercise~\ref{ex:awcs}, this is
    equivalent to the statement that $A^\top$ is weakly chained substochastic. 
    Since $A^\top$ is substochastic, this is, in turn equivalent to $r(A^\top)
    < 1$. But $r(A)=r(A^\top)$, so property U is equivalent
    to $r(A)<1$.
\end{Answer}

\subsubsection{A Converse to the Neumann Series Lemma}\label{ss:nslconv}

Since the Neumann series lemma is a foundational result with many
economic applications, we want to know what happens when the conditions of
the lemma fail.  Here is a partial converse:

\begin{theorem}\label{t:nslconv}
    Fix $A \in \matset{n}{n}$ with $A \geq 0$.  If $A$ is irreducible, then the
    following statements are equivalent:
    \begin{enumerate}
        \item $r(A) < 1$.
        \item $x = A x + b$ has a unique everywhere positive solution for all $b \geq 0$
             with $b \not= 0$.
        \item $x = A x + b$ has a nonnegative solution for at least one $b
            \geq 0$ with $b \not= 0$.
        \item There exists an $x \gg 0$ such that $Ax \ll x$.
    \end{enumerate}
\end{theorem}

\begin{remark}
    If $A$ is irreducible and one of (and hence all of) items (i)--(iv) are true,
    then, by the Neumann series lemma, the unique solution is $x^* := \sum_{m \geq
    0} A^m b$.  
\end{remark}

\begin{remark}
    Statement (iii) is obviously weaker than statement (ii).  It is
    important, however, in the case where $r(A) < 1$ fails.  In this setting,
    from the negation of (iii),
    we can conclude that there is \emph{not even one} nontrivial $b$ in $\RR^n_+$
    such that a nonnegative solution to $x=Ax+b$ exists.
\end{remark}

\begin{proof}[Proof of Theorem~\ref{t:nslconv}] We show (i) $\iff$
        (iv) and then (i) $\iff$ (ii) $\iff$ (iii).

    ((i) $\Rightarrow$ (iv)). 
    For $x$ in (iv) we can use the
    Perron--Frobenius theorem to obtain a real eigenvector $e$ satisfying $Ae
    = r(A) e \ll e$ and $e \gg 0$. 

    ((iv) $\Rightarrow$ (i)). Fix $x \gg 0$ such that $x \gg Ax$.  Through
    positive scaling, we can assume that $\|x\|=1$.  Choose 
    $\lambda < 1$ such that $\lambda x \geq Ax$.  Iterating on this
    inequality gives, for all $k \in \NN$, 
    \begin{equation*}
        \lambda^k x \geq A^k x 
        \quad \implies \quad
        \lambda^k = \lambda^k \|x\|  \geq \|A^k x\|
         \quad \iff \quad
        \| A^k x \|^{1/k} \leq \lambda   
    \end{equation*}
    Hence, by the local spectral radius result in Theorem~\ref{t:local_sr},
    $r(A) \leq \lambda < 1$.

    ((i) $\Rightarrow$ (ii)). Existence of a unique solution $x^* = \sum_{i \geq 0} A^i b$
    follows from the NSL.  Positivity follows from irreducibility of $A$,
    since $b \not= 0$ and $\sum_i A^i \gg 0$.  

    ((ii) $\Rightarrow$ (iii)). Obvious.  

    ((iii) $\Rightarrow$ (i)). Suppose there is a $b \geq 0$ with $b \not= 0$ and
    an $x \geq 0$ such that $x = Ax + b$.  By the Perron--Frobenius theorem,
    we can select a left eigenvector $e$ such that $e \gg 0$ and $e^\top A=
    r(A) e^\top$.  For this $e$ we have
    \begin{equation*}
        e^\top x 
        = e^\top Ax + e^\top b 
        = r(A) e^\top x + e^\top b .
    \end{equation*}
    Since $e \gg 0$ and $b \not= 0$,  we must have $e^\top b > 0$.  In
    addition, $x \not= 0$ because $b \not= 0$ and $x = Ax+b$, so $e^\top x >
    0$.  Therefore $r(A)$ satisfies $(1-r(A)) \alpha = \beta$ for positive
    constants $\alpha, \beta$.  Hence $r(A) < 1$.
\end{proof}

\begin{Exercise}
    For the production system $x = Ax + d$, what do we require on $A$ for the
    condition $r(A) < 1$ to be necessary for existence of a nonnegative output
    solution $x^*$, for each nontrivial $d$?
\end{Exercise}

\begin{Exercise}
    Irreducibility cannot be dropped from Theorem~\ref{t:nslconv}.  Provide an
    example demonstrating that, without irreducibility, we can have $r(A) \geq
    1$ for some $A \geq 0$ and yet find a nonzero $b \geq 0$ and an $x \geq 0$
    such that $x = A x + b$.
\end{Exercise}

\begin{Answer}
    Let $A = \diag(a, 1)$ where $0 < a < 1$.  Clearly $r(A) = 1$.  Let $b^\top = (1,
    0)$ and let $x^\top = (1/(1-a), 0)$.  Simple algebra shows that $x = Ax + b$.  
\end{Answer}

\section{Chapter Notes}\label{s:cnprod}

High quality foundational textbooks on input-output analysis and multisector
production networks include \cite{nikaido1968convex}, \cite{miller2009input}
and \cite{antras2020global}.  References on production networks and aggregate
shocks include \cite{acemoglu2012network}, \cite{antras2012measuring},
\cite{di2014firms}, \cite{carvalho2014micro}, \cite{barrot2016input},
\cite{baqaee2018cascading}, \cite{carvalho2019production},
\cite{acemoglu2020endogenous}, \cite{miranda2021production} and
\cite{carvalho2021supply}.

For other network-centric analysis of multisector models, see, for example,
\cite{bernard2019origins}, who use buyer-supplier relationship data from
Belgium to investigate the origins of firm size heterogeneity when firms are
interconnected in a production network. \cite{dew2022tail} studies tail risk
and aggregate fluctuations in a nonlinear production network.  
\cite{herskovic2018networks} analyzes asset pricing implications of production
networks.  \cite{cai2018interfirm} consider the effect of interfirm
relationships on business performance.

%% file: tikz/io_irreducible.tex
\begin{tikzpicture}
  \node[circle, draw] (1) at (-1, 3) {1};
  \node[circle, draw] (2) at (-3, 0) {2};
  \node[circle, draw] (3) at (1, -2) {3};
  \node[circle, draw] (4) at (0, 0) {4};
  \draw[->, thick, black]
  (1) edge [bend left=50, right] node {$a_{12}$} (2)
  (2) edge [bend left=50, left] node {$a_{21}$} (1)
  (2) edge [bend right=20, below] node {$a_{23}$} (3)
  (3) edge [bend right=50, right] node {$a_{31}$} (1)
  (4) edge [bend right=10, left] node {$a_{43}$} (3)
  (2) edge [bend right=40, above] node {$a_{24}$} (4)
  (1) edge [loop above] node {$a_{11}$} (1);
\end{tikzpicture}

%% file: tikz/hub.tex
\begin{tikzpicture}
  \node[ellipse, draw] (2) at (0, 0) {2};
  \node[ellipse, draw] (3) at (1.5, 0) {3};
  \node[ellipse, draw] (4) at (3, 0) {4};
  \node[ellipse, draw] (5) at (4.5, 0) {5};
  \node[ellipse, draw] (1) at (2, 2) {1};
  \draw[->, thick, black]
  (1) edge [bend left=0, above] node {} (2)
  (1) edge [bend left=0, above] node {} (3)
  (1) edge [bend left=0, below] node {}(4)
  (1) edge [bend left=0, right] node {} (5);
\end{tikzpicture}

%% file: tikz/authority.tex
\begin{tikzpicture}
  \node[ellipse, draw] (1) at (1.5, 1.5) {1};
  \node[ellipse, draw] (2) at (-1.5, 1.5) {2};
  \node[ellipse, draw] (3) at (-1.5, -1.5) {3};
  \node[ellipse, draw] (4) at (1.5, -1.5) {4};
  \node[ellipse, draw] (5) at (0, 0) {5};
  \draw[->, thick, black]
  (1) edge [bend left=0, above] node {} (5)
  (2) edge [bend left=0, above] node {} (5)
  (3) edge [bend left=0, below] node {}(5)
  (4) edge [bend left=0, right] node {} (5);
\end{tikzpicture}

%% file: tikz/symmetric_nodes.tex
\begin{tikzpicture}
  \node[ellipse, draw] (1) at (0, 1.5) {1};
  \node[ellipse, draw] (2) at (-1.5, 0) {2};
  \node[ellipse, draw] (3) at (0, -1.5) {3};
  \node[ellipse, draw] (4) at (1.5, 0) {4};
  \draw[->, thick, black]
  (1) edge [bend left=0, above] node {} (2)
  (2) edge [bend left=0, above] node {} (3)
  (3) edge [bend left=0, below] node {}(4)
  (4) edge [bend left=0, below] node {}(1);
\end{tikzpicture}

%% file: ch_opt.tex
\chapter{Optimal Flows}\label{c:ofd}

Up until now we have analyzed problems where network structure is either fixed
or generated by some specified random process.  In this chapter, we
investigate networks where connections are determined endogenously via
equilibrium or optimality conditions.  In the process, we cover some of the
most powerful methods available for solving optimization problems in networks
and beyond, with applications ranging from traditional
graph and network problems, such as trade, matching, and communication,
through to machine learning, econometrics and finance.

\section{Shortest Paths}\label{ss:shp} 

As a preliminary step, we study the
\navy{shortest path}\index{Shortest path} problem---a topic that has
applications in production, network design, artificial intelligence,
transportation  and many other fields.  The solution method we adopt also
happens to be one of the clearest illustrations of Bellman's principle of
optimality, which is one of the cornerstones of optimization theory and modern economic analysis.

\subsection{Definition and Examples}

We start proceedings by introducing simple examples.  
(In the next section we will formalize the problem and consider solution
methods.)

\subsubsection{Set Up}\label{sss:spsu}

Consider a firm that wishes to ship a container from $A$ to $G$ at minimum
cost, where $A$ and $G$ are vertices of the weighted digraph $\gG$ shown in
Figure~\ref{f:graph}.  Arrows (edges) indicate paths that can be used for
freight, while weights indicate costs of traversing them.  In this context,
the weight function is also called the \navy{cost function}, and we denote it
by $c$.  For example, $c(A, B)$ is the cost of traveling from vertex $A$ to
vertex $B$.

Since this graph is small, we can find the minimum cost path visually. A quick
scan shows that the minimum attainable cost  is 8.  Two paths realize this
cost: $(A, C, F, G)$ and $(A, D, F, G)$, as shown in Figure~\ref{f:graph2}
and Figure~\ref{f:graph3}, respectively.  

\begin{figure}
    \centering
    \input{tikz/shortest_paths_1.tex}
    \caption{\label{f:graph} Graph for the shortest path problem}
\end{figure}
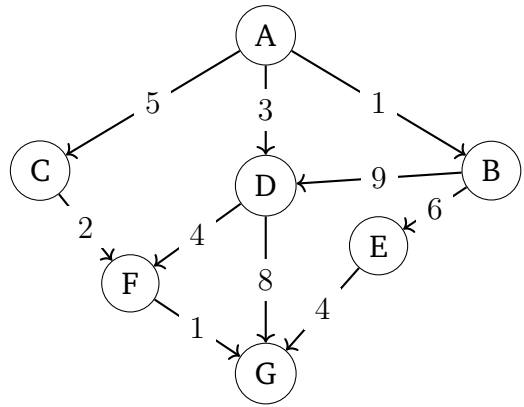

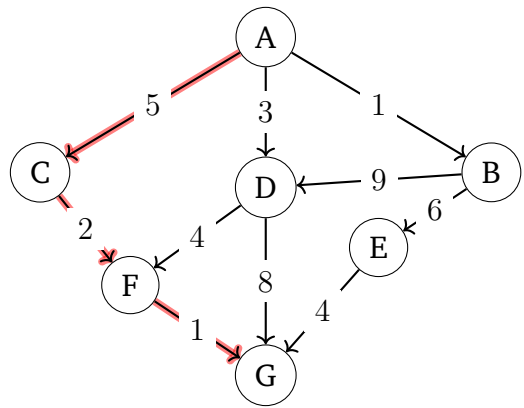
\begin{figure}
    \centering
    \input{tikz/shortest_paths_2.tex}
    \caption{\label{f:graph2} Solution 1}
\end{figure}

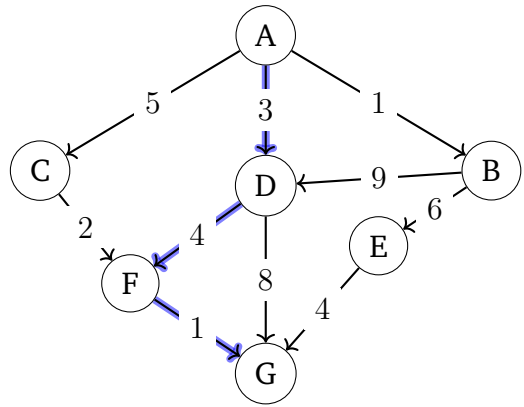
\begin{figure}
    \centering
    \input{tikz/shortest_paths_3.tex}
    \caption{\label{f:graph3} Solution 2}
\end{figure}

\subsubsection{A Recursive View}\label{sss:arecv}

Let's now consider a systematic solution that can be applied to larger graphs.
Let $q^*(x)$ denote the \navy{minimum cost-to-go}\index{Minimum cost-to-go}
from vertex $x$.  That is, $q^*(x)$ is the total cost of traveling from $x$ to
$G$ \emph{if} we take the best route.    Its values are shown at each vertex
in Figure~\ref{f:graph4}.  We can represent $q^*$ in vector form via
\begin{equation}
    \label{eq:vf_ssp}
    (q^*(A), q^*(B), q^*(C), q^*(D), q^*(E), q^*(F), q^*(G))
    = (8, 10, 3, 5, 4, 1, 0) \in \RR^7.
\end{equation}

\begin{figure}
    \centering
    \input{tikz/shortest_paths_4.tex}
    \caption{\label{f:graph4} The cost-to-go function, with $q^*(x)$ indicated
    by red digits at each $x$}
\end{figure}
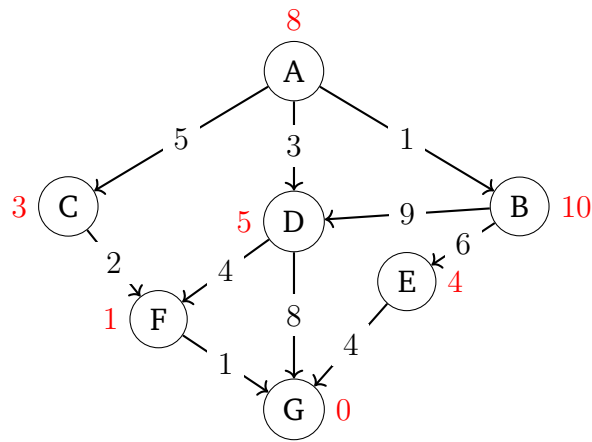

As is clear from studying Figure~\ref{f:graph4}, once 
$q^*$ is known, the least cost path can be computed as follows: Start at $A$
and, from then on, at arbitrary vertex $x$, move to any $y$ that solves
\begin{equation}
    \label{eq:spbc}
    \min_{y \in \oO(x)} \{ c(x, y) + q^*(y) \}.
\end{equation}
Here $\oO(x) = \setntn{y \in V}{(x, y) \in E}$ is the set of direct successors
of $x$, as defined in \S\ref{ss:uwdg}, while
$c(x, y)$ is the cost of traveling from $x$ to $y$.  In other words,
to minimize the cost-to-go, we choose the next path to minimize current traveling cost
plus cost-to-go from the resulting location.  

Thus, if we know $q^*(x)$ at each $x$, then finding
the best path reduces to the simple two stage optimization problem in
\eqref{eq:spbc}.  

But now another problem arises: how to find $q^*$ in more complicated cases,
where the graph is large?  One approach to this
problem is to exploit the fact that
\begin{equation}
    \label{eq:besp}
    q^*(x) = \min_{y \in \oO(x)} \{ c(x, y) + q^*(y) \}
\end{equation}
must hold for every vertex $x$ in the graph apart from $G$ (where $q^*(G)=0$).  

Take the time to convince yourself that, for our example, the function $q^*$
satisfies~\eqref{eq:besp}.  In particular, check that \eqref{eq:besp} holds at
each nonterminal $x$ in Figure~\ref{f:graph4}.

We can understand \eqref{eq:besp}, which is usually called the \navy{Bellman
equation}\index{Bellman equation}, as a restriction on $q^*$ that helps us
identify this vector.  The main difficulty with our plan is that the Bellman
equation is nonlinear in the unknown function $q^*$.   Our strategy will be to
convert this nonlinear equation into a fixed point problem, so that fixed
point theory can be applied.  

We do this in the context of a more general version of the shortest path
problem.

\subsection{Bellman's Method}\label{ss:belm}

We consider a generic shortest path problem on a \navy{flow
network}\index{Flow network}, which consists of a weighted digraph $\gG = (V,
E, c)$ with a sink $d \in V$ called the \navy{destination} and a weight
function $c \colon E \to (0, \infty)$ that associates a positive cost to
each edge $(x,y) \in E$.  We consider how to find the shortest (i.e., minimum
cost) path from $x$ to $d$ for every $x \in V$.  For this problem to make
sense we impose

\begin{assumption}\label{a:expath}
    For each $x \in V$, there exists a directed path from $x$ to $d$.
\end{assumption}

To make Assumption~\ref{a:expath} hold at $x=d$ and facilitate neat proofs, we
add a self-loop at $d$ (i.e., we append $(d, d)$ to $E$) and extend $c$ to
this self-loop by setting $c(d, d) = 0$.  This just means that ``terminating
at $d$'' is the same as ``staying at $d$,''  since no more cost accrues after
arrival at the destination.  Figure~\ref{f:shortest_paths_5} illustrates in
the context of the previous example.

\begin{figure}
   \centering
   \input{tikz/shortest_paths_5.tex}
   \caption{\label{f:shortest_paths_5} Addition of a self loop to the destination}
\end{figure}
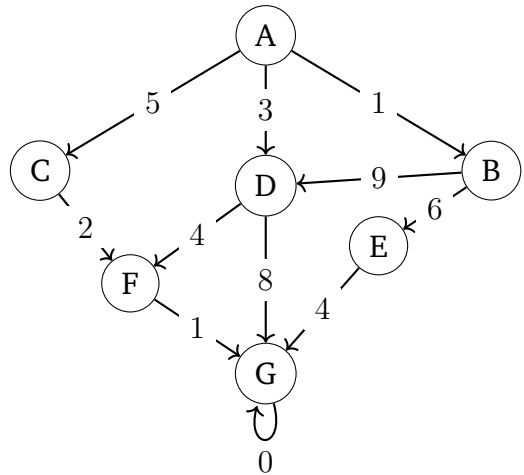

We will use the intuition from the previous section, based around the Bellman
equation, to construct a solution method.  
In what follows we take $|V|=n+1$.

\subsubsection{Policies}

Instead of optimal paths, it turns out to be more convenient to aim for
optimal \emph{policies}.  In general, a policy is a specification of how to
act in every state.  In the present setting, a \navy{policy}\index{Policy
function} is a map $\sigma \colon V \to V$ with $\sigma(x) = y$ understood as
meaning ``when at vertex $x$, go to $y$.'' A policy is called \navy{feasible}
if $\sigma(x) \in \oO(x)$ for all $x \in V$.

For any feasible policy $\sigma$ and $x \in V$, the \navy{trajectory} of $x$
under $\sigma$ is the path from $x$ to the destination indicated by the
feasible policy.  In other words, it is the sequence $(p_0, p_1, p_2, \ldots)$
defined by $p_0 = x$ and $p_i = \sigma(p_{i-1})$  for all $i \in \NN$.  It can
also be expressed as $(\sigma^i(x)) := (\sigma^i(x))_{i \geq 0}$, where
$\sigma^i$ is the $i$-th composition of $\sigma$ with itself.

Let $\Sigma$ be the set of all policies that are feasible and do not cycle, in
the sense that there exists no $x \in V \setminus \{d\}$ such that
$\sigma^i(x)=x$ for some $i \in \NN$.  The policies in $\Sigma$ have the
property that every trajectory they generate reaches $d$ in finite time (and
necessarily stays there, by our assumptions on $d$).  The last statement
implies that 
\begin{equation}\label{eq:tracd}
    \sigma^i(x) = d  \text{ for all } i \geq n \text{ and all } x \in V.
\end{equation}

\begin{Exercise}
    Convergence to $d$ in \eqref{eq:tracd} assumes that any
    trajectory reaching $d$ eventually will reach it in $n$ steps or less.
    Explain why this is the case.
\end{Exercise}

\begin{Answer}
    If $(\sigma^i(x))$ fails to reach $d$ in $n$ steps, then, since $|V|=n+1$, there exists a vertex
    $y \in V$ that appears twice in $(x, \sigma(x), \ldots, \sigma^n(x))$.
    For this $y$ we have $\sigma^i(y) = y$ for some $i \leq n$.  The cycle from
    $y$ back to itself does not contain $d$ and repeats forever.
    Hence $(\sigma^i(x))$ never reaches $d$.
\end{Answer}

Given $q \in \RR^V_+$, we call $\sigma \in \Sigma$ \navy{$q$-greedy} if  
\begin{equation*}
    \sigma(x) \in 
    \argmin_{y \in \oO(x)} \{c(x, y) + q(y)\} \quad \text{for all } x \in V.
\end{equation*}
In essence, a greedy policy treats $q$ as the minimum cost-to-go function
and picks out an optimal path under that assumption.  

Using our new terminology, we can rephrase the discussion in \S\ref{sss:arecv}
as follows: the shortest path problem can be solved by finding the true
minimum cost-to-go function $q^*$ and then following a $q^*$-greedy policy.
In the remainder of this section, we prove this claim more carefully.

\subsubsection{Cost of Policies}

We need to be able to assess the cost of any given policy.  To this end,
for each $x \in V$ and $\sigma \in \Sigma$, let $q_\sigma(x)$ denote the cost of
following $\sigma$ from $x$.  That is,
\begin{equation}\label{eq:qsdef}
    q_\sigma(x) 
    = \sum_{i=0}^\infty c(\sigma^i(x), \sigma^{i+1}(x))
    = \sum_{i=0}^{n-1} c(\sigma^i(x), \sigma^{i+1}(x)).
\end{equation}
This second equality holds because $\sigma^i(x) = d$ for all $i
\geq n$ and $c(d, d) =0$.  The function $q_\sigma \in \RR^V_+$ is called the
\navy{cost-to-go} under $\sigma$.

It will be helpful in what follows to design an operator such that $q_\sigma$
is a fixed point.  For this purpose we let $U$ be all $q \in \RR_+^V$ with
$q(d)=0$ and define $T_\sigma \colon U \to U$ by
\begin{equation*}
    (T_\sigma \, q)(x) = c(x, \sigma(x)) + q(\sigma(x))
    \qquad (x \in V).
\end{equation*}
Here and below, with $k \in \NN$, the expression $T_\sigma^k$ indicates the
$k$-th composition of $T_\sigma$ with itself (i.e., $T_\sigma$ is applied $k$ times).

\begin{Exercise}
    Prove that $T_\sigma$ is a self-map on $U$ for all $\sigma \in \Sigma$.
\end{Exercise}

\begin{Answer}
    Fix $q \in U$ and $\sigma \in \Sigma$.  
    Nonnegativity of $T_\sigma \, q$ is obvious.  Also, 
    $(T_\sigma q)(d) = c(d, d) + q(d) = 0$, where the first equality is by $\sigma(d)=d$
    and the second is by $q(d)=0$ and $c(d, d)=0$. Hence $T_\sigma \, q \in U$,
    as required.
\end{Answer}

\begin{proposition}\label{p:sptsp}
    For each $\sigma \in \Sigma$, the function $q_\sigma$ is the unique fixed
    point of $T_\sigma$ in $U$ and $T^k_\sigma \, q = q_\sigma$ for all $k \geq
    n$ and all $q \in U$.
\end{proposition}

\begin{proof}
    Fix $\sigma \in \Sigma$ and $q \in U$.  For each $x \in V$ we have
    \begin{equation*}
        (T^2_\sigma \, q)(x)
         = c(x, \sigma(x)) + (T_\sigma q)(\sigma(x)) 
         = c(x, \sigma(x)) + c(\sigma(x), \sigma^2(x)) + q(\sigma^2(x)).
    \end{equation*}
    More generally, for $k \geq n$, we have
    \begin{equation*}
        (T_\sigma^k \, q)(x) 
         = \sum_{i=0}^{k-1} c(\sigma^i(x), \sigma^{i+1}(x)) + q(\sigma^k(x))
         = \sum_{i=0}^{n-1} c(\sigma^i(x), \sigma^{i+1}(x)).
    \end{equation*}
    The second equality holds because $\sigma^k(x)=d$ for all $k \geq n$
    and $q(d)=0$.  Hence $(T_\sigma^k \, q)(x) = q_\sigma(x)$ by \eqref{eq:qsdef}.
    The fact that $q_\sigma$ is the unique fixed point of $T_\sigma$ now
    follows from Exercise~\ref{ex:ufnflow} in the appendix (page~\pageref{ex:ufnflow}).
\end{proof}

\subsubsection{Optimality}

The \navy{minimum cost-to-go} function $q^*$ is defined by
\begin{equation*}
    q^*(x) = \min_{\sigma \in \Sigma} q_\sigma(x)
    \qquad (x \in V).
\end{equation*}
A policy $\sigma^* \in \Sigma$ is called \navy{optimal}\index{Optimal policy} if 
it attains the minimum in this expression, so that $q^* = q_{\sigma^*}$ on
$V$. Note that the definition of $q^*$ matches our intuitive definition from
\S\ref{sss:arecv}, in the sense that $q^*(x)$ is, for each $x \in V$, the
minimum cost of traveling from $x$ to the destination $d$.

Our main aims now are to
\begin{enumerate}
    \item obtain a method for calculating $q^*$ and
    \item verify that a $q^*$-greedy policy is in fact optimal, as suggested
        in our discussion of greedy policies above.
\end{enumerate}

Regarding the first step, we claim that $q^*$ satisfies the Bellman
equation~\eqref{eq:besp}.  To prove this claim we introduce the \navy{Bellman
operator}\index{Bellman operator} $T$ via
\begin{equation}
    \label{eq:bespt}
    (Tq)(x) = \min_{y \in \oO(x)} \{ c(x, y) + q(y) \}
    \qquad (x \in V).
\end{equation}
By construction, $q^*$ satisfies the Bellman equation~\eqref{eq:besp} if and only
if $Tq^* = q^*$.

\begin{Exercise}\label{ex:ttsor0}
    Show that $T$ is a self-map on $U$ and, moreover, $T q \leq T_\sigma \, q$
    for all $q \in U$ and $\sigma \in \Sigma$.
\end{Exercise}

\begin{Exercise}\label{ex:ttsor}
    Show that $T$ and $T_\sigma$ are both order-preserving on $\RR^V_+$ with respect to the
    pointwise partial order $\leq$.  Prove that $T^k q \leq T_\sigma^k q$ for all $q
    \in U$, $\sigma \in \Sigma$ and $k \in \NN$.
\end{Exercise}

\begin{Exercise}\label{ex:qg}
    Fix $q \in U$ and let $\sigma$ be a $q$-greedy policy.  Show that $T q =
    T_\sigma q$.
\end{Exercise}

The next result is central.  It confirms that the minimum cost-to-go function
satisfies the Bellman equation and also provides us with a means to compute
it: pick any $q$ in $U$ and then iterate with $T$.

\begin{proposition}
    The function $q^*$ is the unique fixed point of $T$ in $U$ and, in
    addition, $T^k q \to q^*$ as $k \to \infty$ for all $q \in U$.
\end{proposition}

\begin{proof}
    In view of Exercise~\ref{ex:ufnflow} in the appendix
    (page~\pageref{ex:ufnflow}), it suffices to
    verify existence of an $m \in \NN$ that $T^k q = q^*$ for all $k \geq m$
    and all $q \in U$.  To this end, let $\gamma$ be the minimum of $c(x,y)$ over
    all $(x, y) \in E$ except $(d, d)$.  Since $c$ is positive on such edges
    and $E$ is finite, $\gamma > 0$.

    Fix $q \in U$.  We claim first that $T^k q \geq q^*$ for sufficiently
    large $k$.  To see this, fix $k \in \NN$ and iterate with $T$ to get
    \begin{equation*}
        (T^k q)(x) = c(x, p_1) + c(p_1, p_2) + \cdots + c(p_{k-1}, p_k) + q(p_k)
    \end{equation*}
    for some path $(x, p_1, \ldots, p_k)$.  If this path leads to $d$,
    then $(T^k q)(x) \geq q^*(x)$ when $k \geq n$, since $q^*(x)$ is the 
    minimum cost of reaching $d$ from $x$.\footnote{Recall that $q^*(x)$ is
        the minimum cost-to-go under a policy leading to $d$.
        What if the path $(x, p_1, \ldots, p_k)$ cannot be realized as the
        trajectory of any policy?  This will only be the case if the path
        contains a cycle.  If so we can find a shorter path leading to $d$ 
        by dropping cycles.  The cost of this path is greater that $q^*(x)$,
        so the inequality $(T^k q)(x) \geq q^*(x)$ still holds.}
        If not then $(T^k q)(x) \geq k \gamma$, so
    $(T^k q)(x) \geq q^*(x)$ for $k$ such that $k \gamma \geq \max_{x \in V}
    q^* (x)$.

    For the reverse inequality, fix $k \geq n$ and observe that, by
    Proposition~\ref{p:sptsp} and the inequality from
    Exercise~\ref{ex:ttsor},
    \begin{equation*}
        T^k q \leq T_{\sigma^*}^k q = q_{\sigma^*} = q^*.
    \end{equation*}
    We have now shown that $T^k q = q^*$ for sufficiently large $k$, as
    required.
\end{proof}

We now have a means to compute the minimum cost-to-go function (by iterating
with $T$, starting at any $q \in U$) and, in addition, a way to verify the
following key result.

\begin{theorem}
    A policy $\sigma \in \Sigma$ is optimal if and only if $\sigma$ is
    $q^*$-greedy.
\end{theorem}

\begin{proof}
    By the definition of greedy policies,
    \begin{equation*}
        \sigma \text{ is $q^*$-greedy}
        \quad \iff \quad 
            c(x, \sigma(x)) + q^*(\sigma(x)) 
            = \min_{y \in \oO(x)} \{ c(x, y) + q^*(y) \},
            \quad \forall \, x \in V.
    \end{equation*}
    Since $q^*$ satisfies the Bellman equation, we then have
    \begin{equation*}
        \sigma \text{ is $q^*$-greedy}
        \quad \iff \quad 
        c(x, \sigma(x)) + q^*(\sigma(x)) 
        = q^*(x),
            \quad \forall \, x \in V.
    \end{equation*}
    The right-hand side is equivalent to $T_\sigma \, q^* = q^*$.  At the same
    time, $T_\sigma$ has only one fixed point in $U$, which is $q_\sigma$.
    Hence $q_\sigma = q^*$. Hence, by this chain of logic and the definition of optimality,
    \begin{equation*}
        \sigma \text{ is $q^*$-greedy}
        \quad \iff \quad
        q^* = q_\sigma
        \quad \iff \quad
        \text{ $\sigma$ is optimal}.
        \qedhere
    \end{equation*}
\end{proof}

\subsubsection{An Implementation in Julia}

Let's use the ideas set out above to solve the original shortest path problem,
concerning shipping, which we introduced in \S\ref{sss:spsu}.  We will
implement in Julia.

Our first step is to set up the cost function, which we store as an array
called \texttt{c}.  We identify the vertices $A$, \ldots, $G$, with the integers $1,
\ldots, 7$.  We set \mintinline{julia}{c[i, j] = Inf} when no edge exists from \texttt{i} to \texttt{j}, so
that such a path is never chosen when evaluating the 
Bellman operator, as defined in \eqref{eq:besp}.  When an edge does exist, we
enter the cost shown in Figure~\ref{f:graph}.

\begin{minted}{julia}
c = fill(Inf, (7, 7))

c[1, 2], c[1, 3], c[1, 4] = 1, 5, 3
c[2, 4], c[2, 5] = 9, 6
c[3, 6] = 2
c[4, 6] = 4
c[5, 7] = 4
c[6, 7] = 1
c[7, 7] = 0
\end{minted}

Next we define the Bellman operator:

\begin{minted}{julia}
function T(q)
    Tq = similar(q)
    n = length(q)
    for x in 1:n
        Tq[x] = minimum(c[x, :] + q[:]) 
    end
    return Tq
end
\end{minted}

Now we arbitrarily set $q \equiv 0$, generate the sequence of iterations $Tq, T^2 q, T^3 q$ and plot them:

\begin{minted}{julia}
using PyPlot
fig, ax = plt.subplots()

n = 7
q = zeros(n)
ax.plot(1:n, q)

for i in 1:3
    new_q = T(q)
    ax.plot(1:n, new_q, "-o", alpha=0.7)
    q = new_q
end

\end{minted}

\begin{figure}
   \centering
   \scalebox{0.6}{\includegraphics[trim = 0mm 0mm 0mm 0mm,
   clip]{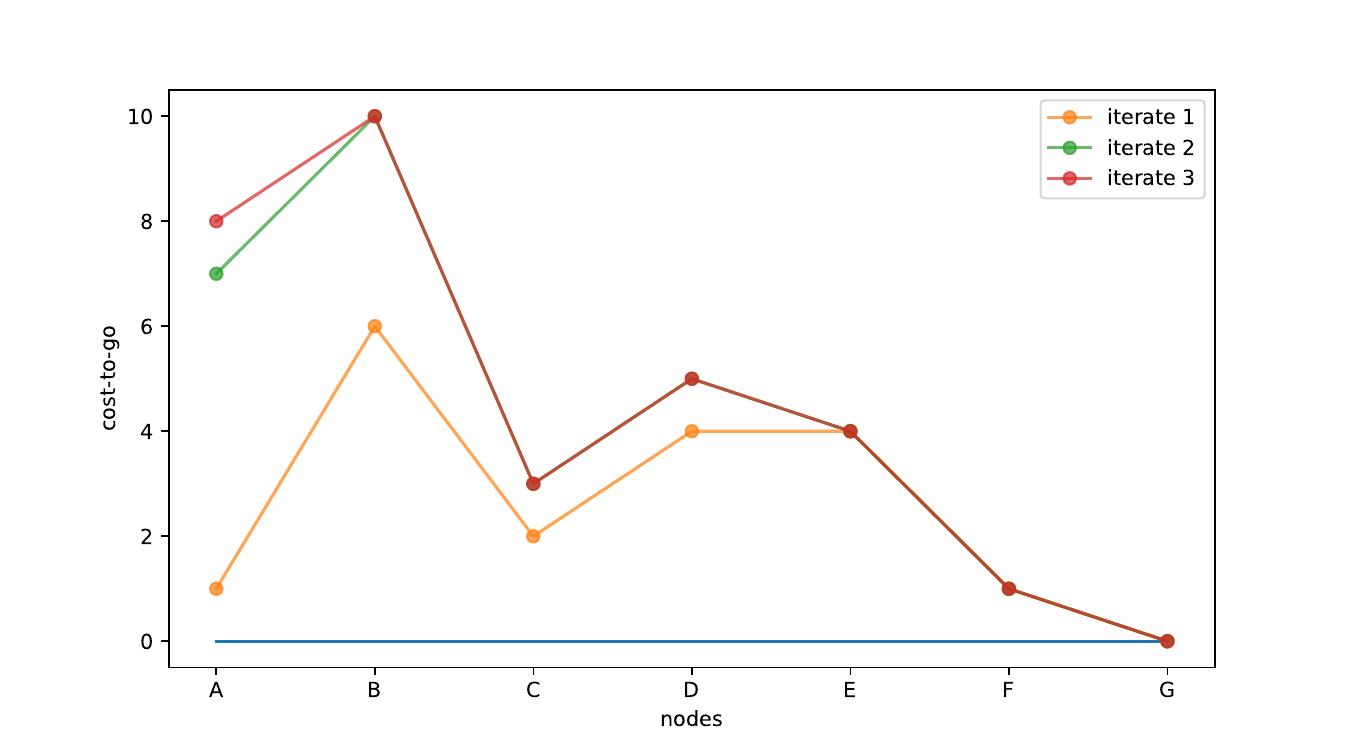}} \caption{\label{f:shortest_path_iter_1} Shortest path Bellman iteration}
\end{figure}

After adding some labels, the output looks like the image in
Figure~\ref{f:shortest_path_iter_1}.  Notice that, 
by $T^3 q$, we have already converged on $q^*$.  You can confirm this by 
checking that the values of $T^3 q$ line up with those we obtained manually in
Figure~\ref{f:graph4}.

\subsection{Betweenness Centrality}

In \S\ref{ss:netcen} we discussed a range of centrality measures for networks,
including degree, eigenvector and Katz centrality.  Aside from these, there is
another well-known centrality measure, called betweenness centrality, that builds on the
notion of shortest paths and is particularly popular in analysis of social and
peer networks.  

Formally, for a given graph $\gG = (V, E)$, directed or undirected, the
\navy{betweenness centrality} of vertex $v \in V$ is
\begin{equation*}
    b(v) := \sum_{x,y \in V \setminus \{v\} }\frac{ |S(x, v, y)|}{|S(x,y)|},
\end{equation*}
where $S(x,y)$ is the set of all shortest paths from $x$ to $y$ and $S(x, v,
y)$ is the set of all shortest paths from $x$ to $y$ that pass through $v$.
(As usual, $|A|$ is the cardinality of $A$.)  Thus, $b(v)$ is proportional to
the probability that, for a randomly selected pair of nodes $x, y$, a randomly
selected shortest path from $x$ to $y$ contains $v$.
Thus, the measure will be high for nodes that ``lie between'' a large number
of node pairs.

For example, Networkx stores a graph called \texttt{florentine\_families\_graph} that
details the marriage relationships between 15 prominent Florentine families
during the 15th Century.  The data can be accessed via
\begin{minted}{python}
import networkx as nx
G = nx.florentine_families_graph()    
\end{minted}
The network is shown in Figure~\ref{f:betweenness_centrality_1}, which was
created using
\begin{minted}{python}
nx.draw_networkx(G, [params])    
\end{minted}
where \texttt{[params]} stands for
parameters listing node size, node color and other features.
For this graph, node size and node color are scaled by betweenness centrality,
which is calculated via
\begin{minted}{python}
nx.betweenness_centrality(G)    
\end{minted}
Although this graph is very simple, the output helps to illustrate the
prominent role of the Medici family, consistent with their great wealth and
influence in Florence and beyond.

\begin{figure}
   \centering
   \scalebox{0.5}{\includegraphics[trim = 20mm 40mm 0mm 40mm, clip]{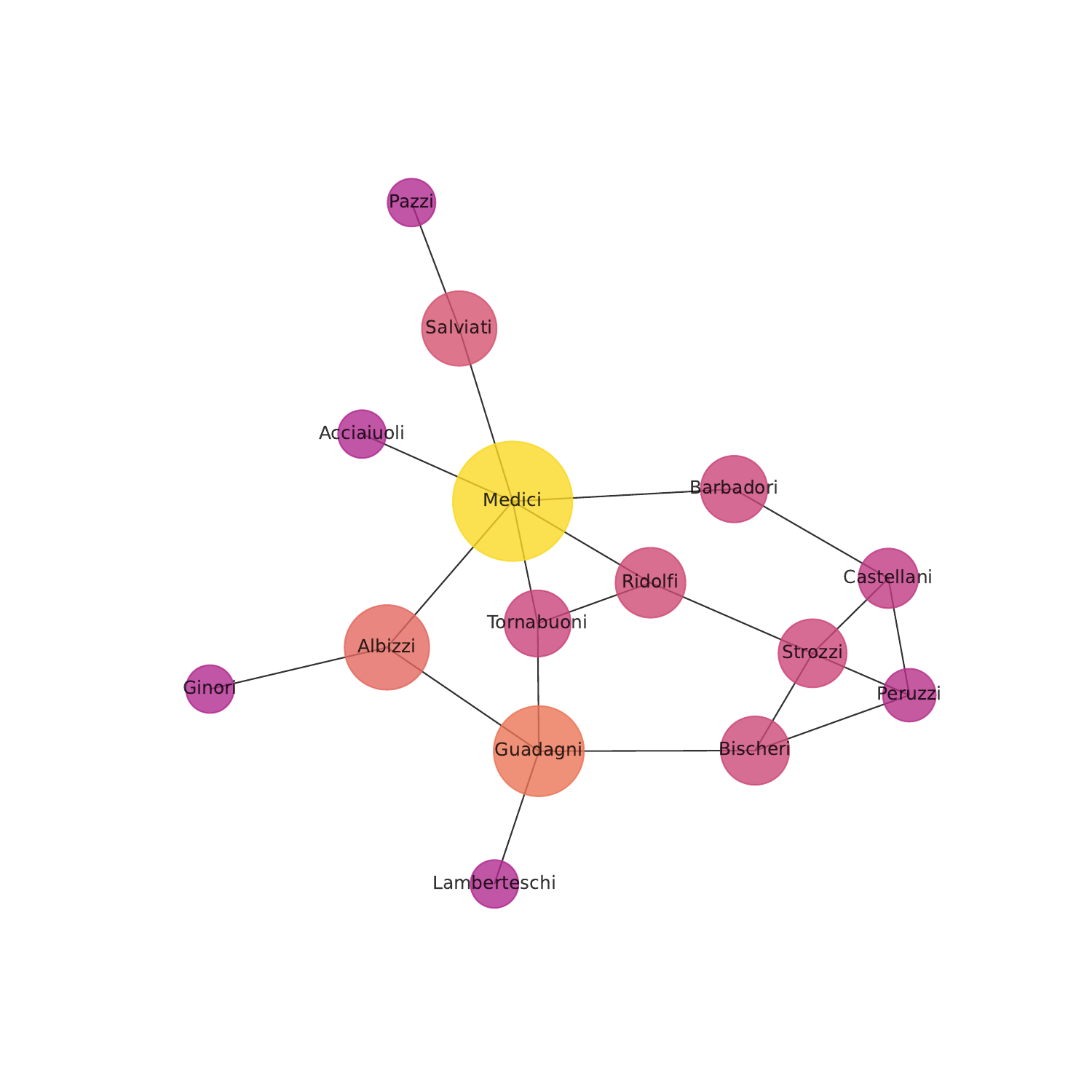}}
   \caption{\label{f:betweenness_centrality_1} Betweenness centrality (by
   color and node size) for the Florentine families}
\end{figure}

\section{Linear Programming and Duality}\label{ss:spe}

Our study of shortest paths in \S\ref{ss:shp} used a relatively specialized
optimization method.  In this section we cover more general results in
optimization and duality, which will then be applied to endogenous networks,
optimal transport and optimal flows.  Part of the strategy is to take
challenging optimization problems and regularize them, often by some form of
convexification.  

Before diving into theory, we use \S\ref{ss:wahdp} below to provide motivation
via one very specific application. This application involves what is typically
called \emph{matching} in economics and \emph{linear assignment} in
mathematics. We illustrate the key ideas in the context of matching workers
to jobs. Later, when we discuss how to solve the problem, we will see the
power of linear programming, convexification and duality.

\subsection{Linear Assignment}\label{ss:wahdp}

Behold the town of Springfield.  A local employer is shutting down and 40
workers stand idled.  Fortunately for these workers, Springfield
lies within a political battleground state and, as a result, the mayor
receives backing to attract a new employer.  She succeeds by
promising a certain firm that the 40 workers will be retrained for the new
skills they require.  For mathematical convenience, let us suppose that there
are exactly 40 new positions, each with distinct skill requirements.

The team set up by the mayor records the individual skills of the 40 workers,
along with the requirements of the new positions,
and estimates the cost $c(i, j)$ of retraining individual $i$ for position
$j$.  The team's challenge is to minimize the total cost of retraining.  In
particular, they wish to solve
\begin{equation}\label{eq:matchcost}
    \min_{\sigma \in \pP} \sum_{i=1}^{40} c(i, \sigma(i)),
\end{equation}
where $\pP$ is the set of all permutations (i.e., bijective self-maps) on
the integers $1, \ldots, 40$.  Figure~\ref{f:worker_job_matching} illustrates
one possible permutation.

\begin{figure*}
   \begin{center}
    \input{tikz/worker_job_matching.tex}
    \caption{\label{f:worker_job_matching} One possible assignment (i.e., permutation of $\natset{40}$)}
   \end{center}
\end{figure*}
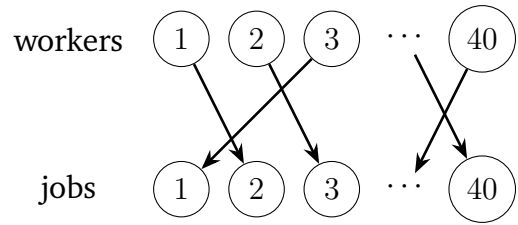

The problem is discrete, so first order conditions are unavailable.  Unsure
of how to proceed but possessing a powerful computer, the team 
instructs its workstation to step through all possible permutations and record
the one that leads to the lowest total retraining cost.  The instruction set
is given in Algorithm~\ref{algo:searchperm}. 

\begin{algorithm}
    set $m = +\infty$ \;  
    \For{$\sigma$ in $\pP$ } 
    {
        set $t(\sigma) = \sum_{i=1}^{40} c(i, \sigma(i))$ \;
        \If{$t(\sigma) < m$}
        {
            set $m = t(\sigma)$ \;
            set $\sigma^* = \sigma$ \;
        }
    }
    \Return{$\sigma^*$ }
    \caption{\label{algo:searchperm} Minimizing total cost by testing all permutations }
\end{algorithm}

After five days of constant execution, the workstation is still running and the
mayor grows impatient. The team starts to calculate how long execution will
take.  The main determinant is the size of the set $\pP$.  Elementary
combinatorics tells us that the number of permutations of a set of size $n$ is
$n!$.  Quick calculations show that $40!$ exceeds $8 \times 10^{47}$.  A
helpful team member points out that this is much less than the
number of possible AES-256 password keys (approximately $10^{77}$).  The mayor is not appeased 
and demands a runtime estimate.

Further calculations reveal the following: if, for each $\sigma$ in $\pP$,
the workstation can evaluate the cost
$t(\sigma) = \sum_{i=1}^{40} c(i, \sigma(i))$ in $10^{-10}$-th of a second (which is
extremely optimistic), then the total run time would be
\begin{equation*}
    10^{-10} \times 8 \times 10^{47} 
    = 8 \times 10^{37} \text{ seconds }
    \approx 2.5 \times 10^{30} \text{ years}.
\end{equation*}
Another helpful team member provides perspective by noting that the sun will
expand into a red giant and vaporize planet Earth in less than $10^{10}$
years. 

The great computational cost of solving this problem by direct calculations is
an example of what is often called the \emph{curse of dimensionality}. This
phrase, coined by Richard Bellman (1920-1984) during his fundamental research
into dynamic optimization, refers to the exponential increase in processor
cycles needed to solve computational problems to a given level of accuracy as
the number of dimensions increases.  The matching problem we have just
described is high-dimensional because the choice variable $\sigma$, a
permutation in $\pP$, is naturally associated with the vector $(\sigma(1),
\ldots, \sigma(40))$.  This, in turn, is a point in 40-dimensional vector
space.\footnote{Readers familiar with high performance computing techniques
    might suggest that the curse of dimensionality is no cause for concern for
    the matching problem,  since the search algorithm be parallelized.
    Unfortunately, even the best parallelization methods cannot save the
    workstation from being vaporized by the sun with the calculation
    incomplete.  The best-case scenario is that adding another execution
    thread doubles effective computations per second and hence halves
execution time.   However, even with $10^{10}$ such threads (an enormous
number), the execution time would still be $2.5 \times 10^{20}$ years.}

Fortunately, clever algorithms for this matching problem have been found and,
for problems such as this one, useful approximations to the optimal allocation
can be calculated relatively quickly.  For example,
\cite{dantzig1951application} showed how such problems can be cast as a
\emph{linear program}, which we discuss below.  Further progress has been made
in recent years by adding regularization terms to the optimization problem
that admit derivative-based methods. 

In the sections that follow, we place matching problems in a more general
setting and show how they can be solved efficiently.  Our first step is to
review the basics of linear programming.

\subsection{Linear Programming}\label{ss:lp} 

A linear program is an optimization problem with a linear objective function
and linear constraints.  If your prior belief is that all interesting problems
are nonlinear, then let us reassure you that linear programming is
applicable to a \emph{vast} array of interesting applications.  One of these
is, somewhat surprisingly, the assignment problem in \S\ref{ss:wahdp}, as 
famously demonstrated by the American mathematician George Bernard Dantzig
(1914--2005) in \cite{dantzig1951application}.  Other applications include
optimal flows on networks and optimal transport problems, which in turn have
diverse applications in economics, finance, engineering and machine learning.

We explain the key ideas below, beginning with an introduction to linear
programming.

\subsubsection{A Firm Problem}

One way to define linear programs is in terms of what 
open source and commercial solvers for linear programming problems handle.
Typically, for these solvers, the problem takes the form 
\begin{align}
    \label{eq:linprogsolver}
    & \min_{x \in \RR^n} \, c^\top x  \\
    \label{eq:linprogsolverc}
    \text{ subject to }  & A_1 x = b_1, \; A_2 x \leq b_2, \text{ and } 
    d_\ell \leq x \leq d_u.
\end{align}
Here each $A_i$ is a matrix with $n$ columns and $c, b_1, b_2, d_\ell$ and $d_u$ are
conformable column vectors.

To illustrate, let's consider a simple example, which concerns a firm that fabricates
products labeled $1, \ldots, n$.  To make each product requires machine hours and labor.
Product $i$ requires $m_i$ machine hours and $\ell_i$ labor hours per unit of
output, as shown in the table below
\begin{center}
    \begin{tabular}{l|c|c|c}
     product  & machine hours & labor hours & unit price \\
     \hline 
     1        & $m_1$         &    $\ell_1$ &   $p_1$ \\
     $\vdots$        & $\vdots$         &    $\vdots$ &   $\vdots$ \\
     $n$        & $m_n$         &    $\ell_n$ &   $p_n$  \\
     \hline
    \end{tabular}
\end{center}
A total of $M$ machine hours and $L$ labor hours are available.  Letting $q_i$
denote output of product $i$, the firm's problem
is 
\begin{equation*}
    \max_{q_1, \ldots, q_n} \pi(q_1, \ldots, q_n) 
    := \sum_{i=1}^n p_i q_i - w L - r M 
\end{equation*}
subject to
\begin{equation}\label{eq:lpfc}
      \sum_{i=1}^n m_i q_i \leq M ,
      \quad
      \sum_{i=1}^n \ell_i q_i \leq L  \text{ and }
      q_1, \ldots, q_n \geq 0.
\end{equation}
Here $\pi$ is profits and $w$ and $r$ are the wage rate and rental rate of
capital.  We are taking these values as fixed, along with $L$ and $M$, so
choosing the $q \in \RR^n_+$ that maximizes profits is the same as choosing
the $q$ that maximizes revenue $\sum_{i=1}^n p_i q_i$.   This will be our
objective in what follows.

(Why are total labor supply and machine hours held fixed here?  We might
think of this problem as one of designing a daily or weekly production plan,
which optimally allocates existing resources, given current prices.)

Figure~\ref{f:linear_programming_1} shows an illustration of a simple case
with two products.  Consistent with Exercise~\ref{ex:conpoly}, the feasible
set is a polyhedron, shaded in blue. The green lines are contour lines of the
revenue function $(q_1, q_2) \mapsto p_1 q_1 + p_2 q_2$, with $p_1=3$ and
$p_2=4$.  By inspecting this problem visually, and recognizing that the
contour lines are increasing as we move to the northeast, it is clear that
the maximum is obtained at the extreme point indicated in the figure.

\begin{figure}
   \centering
   \scalebox{0.65}{\includegraphics[trim = 0mm 0mm 0mm 0mm, clip]{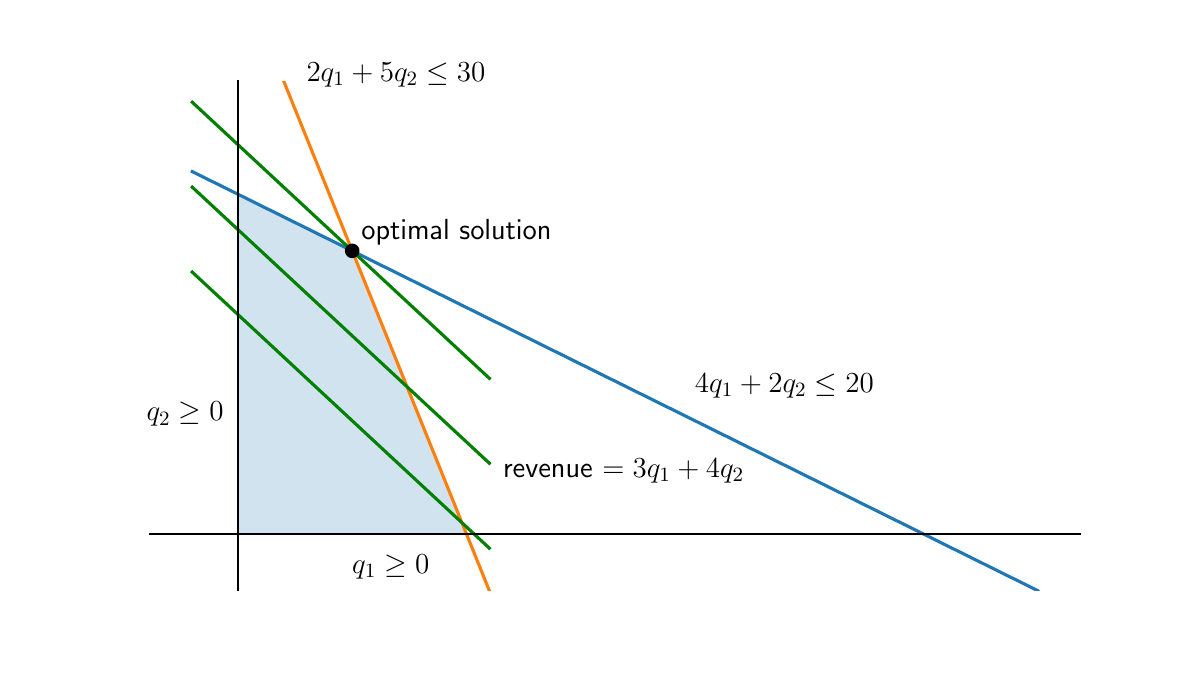}}
   \caption{\label{f:linear_programming_1} Revenue maximizing quantities}
\end{figure}

Maximizing revenue is equivalent to minimizing $\sum_i (- p_i) q_i$, so with
$c = (-p_1, -p_2)$ we have a linear programming problem to which we can apply
a solver. 

\subsubsection{A Python Implementation}

Let's look at one option for solving this problem with Python, via SciPy's
open source solver \texttt{linprog}. For the simple two-product firm problem,
we ignore
the unnecessary equality constraint in \eqref{eq:linprogsolverc} and set
\begin{equation*}
    A_2 =  
    \begin{pmatrix}
        m_1 & m_2 \\
        \ell_1 & \ell_2 
    \end{pmatrix}
    \quad \text{and} \quad
    b_2 =
    \begin{pmatrix}
        M \\
        L
    \end{pmatrix}.
\end{equation*}
The bound $0 \leq q$ is imposed by default so we do not need to specify $d_\ell$
and $d_u$.

We apply the numbers in Figure~\ref{f:linear_programming_1}, so the
maximization version of the problem is 
\begin{equation}\label{eq:fsimpp}
    \max_{q_1, q_2} 
    3q_1 + 4 q_2
    \; \st \;
    q_1 \geq 0, \; q_2 \geq 0, \;\;
    2q_1 + 5q_2 \leq 30
    \text{ and }
    4q_1 + 2q_2 \leq 20.
\end{equation}

Now we set up primitives and call \texttt{linprog} as follows:

\begin{minted}{python}
A = ((2, 5),
     (4, 2))
b = (30, 20)
c = (-3, -4)  # minus in order to minimize

from scipy.optimize import linprog
result = linprog(c, A_ub=A, b_ub=b)
print(result.x)    
\end{minted}

The output is \texttt{[2.5, 5.0]}, indicating that we should set $q_1=2.5$ and
$q_2=5.0$.  The result is obtained via the simplex algorithm, discussed in
\S\ref{sss:sial}.  

\begin{Exercise}
    As a way to cross-check the solver's output, 
    try to derive the solution $q=(2.5, 5.0)$ in a more intuitive way, 
    from examining Figure~\ref{f:linear_programming_1}.
\end{Exercise}

\begin{Answer}
    Since prices are positive (more output means more revenue), we expect that
    both inequality constraints (see \eqref{eq:lpfc}) will hold with equality.
    Reading from the figure, the equalities are $2q_1 + 5 q_2 = 30$ and $4q_1 + 2 q_2 =
    20$.  Solving simultaneously leads to $q=(2.5, 5.0)$.
\end{Answer}

\begin{Exercise}
    Consider the same problem with the same parameters, but 
    suppose now that, in addition to the previous constraints, output of $q_2$
    is bounded above by $4$.  Use \texttt{linprog} or another numerical linear
    solver to obtain the new solution.
\end{Exercise}

\subsubsection{A Julia Implementation}

When solving linear programs, one option is to use a domain specific modeling
language to set out the objective and constraints in the optimization problem.
In Python this can be accomplished using the open source libraries such as 
Pyomo and Google's OR-Tools. In Julia we can use JuMP.  

The following code illustrates the Julia case, using JuMP, applied to the firm
problem \eqref{eq:fsimpp}.

\begin{minted}{julia}
using JuMP
using GLPK

m = Model()
set_optimizer(m, GLPK.Optimizer)

@variable(m, q1 >= 0)
@variable(m, q2 >= 0)
@constraint(m, 2q1 + 5q2 <= 30)
@constraint(m, 4q1 + 2q2 <= 20)
@objective(m, Max, 3q1 + 4q2)

optimize!(m)    
\end{minted}

Notice how the JuMP modeling language allows us to write objectives and
constraints as expressions, such as \mintinline{julia}{2q1 + 5q2 <= 30}.  This brings the
code closer to the mathematics and makes it highly readable.

If we now run \mintinline{julia}{value.(q1)} and \mintinline{julia}{value.(q2)}, we get 2.5 and 5.0,
respectively, which are the same as our previous solution.

\subsubsection{Standard Linear Programs}\label{sss:slps}

The programming problem in \eqref{eq:linprogsolver}--\eqref{eq:linprogsolverc} is convenient for
applications but somewhat cumbersome for theory. A more canonical version of
the linear programming problem is
\begin{equation}\label{eq:linprog}
    P := \min_{x \in \RR^n} \, c^\top x  
    \; \text{ subject to } \;
    A x = b  \text{ and } x \geq 0.
\end{equation}
Here $x$ is interpreted as a column vector in $\RR^n$, $c$ is
also a column vector in $\RR^n$, $A$ is $m \times n$, and $b$ is $m \times
1$.  A linear program in the form of \eqref{eq:linprog} is said to be in
\navy{standard equality form}\index{Standard form}.
In preparation for our discussion of duality below, we also call
\eqref{eq:linprog} the \navy{primal problem}.

\begin{Exercise}\label{ex:conpoly}
    Prove that the \navy{feasible set} $F = \setntn{x \in \RR^n}{Ax=b \text{
    and } x \geq 0}$ for the linear program \eqref{eq:linprog} is a polyhedron.
\end{Exercise}

\begin{Answer}
    The equality constraint $Ax = b$ can be replaced by the two inequality
    constraints $Ax \leq b$ and $-Ax \leq -b$.  The constraint $x \geq 0$ can be
    replaced by $-x \leq 0$.  If we unpack these matrix inequalities, row by row,
    we obtain a collection of constraints, each of which has the form $h^\top x
    \leq g$ for suitable $h \in \RR^n$ and $g \in \RR$.  The claim now follows
    from the definition of a polyhedron on page~\pageref{eq:polyh}.
\end{Answer}

Standard equality form is more general than it first appears. In fact the
original formulation \eqref{eq:linprogsolver}--\eqref{eq:linprogsolverc} can be manipulated into standard
equality form via a sequence of transformations. Hence, when treating theory
below, we can specialize to standard equality form without losing generality.

Although we omit full details on the set of necessary transformations (which
can be found in \cite{bertsimas1997introduction} and many other sources),
let's gain some understanding by converting the firm optimization problem into
standard equality form. To simplify notation, we address this task when $n=3$,
although the general case is almost identical.

As above, we switch to minimization of a linear constraint by using the fact
that maximizing revenue is equivalent to minimizing $\sum_i (- p_i) q_i$. 
Next, we need to convert the two inequality constraints \eqref{eq:lpfc} into equality
constraints.  We do this by introducing \navy{slack variables}\index{Slack
variables} $s_m$ and $s_\ell$ and rewriting the constraints as 
\begin{equation*}
      \sum_{i=1}^3 m_i q_i + s_m = M ,
      \quad
      \sum_{i=1}^3 \ell_i q_i + s_\ell = L  
      \quad \text{and} \quad
      q_1, q_2, q_3, s_m, s_\ell \geq 0.
\end{equation*}
Indeed, we can see that requiring $\sum_{i=1}^3 m_i q_i + s_m = M$ and $s_m
\geq 0$ is the same as imposing $\sum_{i=1}^3 m_i q_i \leq M$, and the same
logic extends to the labor constraint.

Setting $x := (q_1, q_2, q_3, s_m, s_\ell)$, we can now express the problem as
\begin{equation*}
    \min_x c^\top x 
    \quad \text{where} \quad
    c^\top := (-p_1, -p_2, -p_3, 0, 0)
\end{equation*}
subject to
\begin{equation*}
    \begin{pmatrix}
        m_1 & m_2 & m_3 & 1 & 0 \\
        \ell_1 & \ell_2 & \ell_3 & 0 & 1 
    \end{pmatrix}
    \begin{pmatrix}
        q_1 \\
        q_2 \\
        q_3 \\
        s_m \\
        s_\ell
    \end{pmatrix}
    =
    \begin{pmatrix}
        M \\
        L
    \end{pmatrix}
      \quad \text{and} \quad
      x \geq 0.
\end{equation*}
This is a linear program in standard equality form.

\subsubsection{Duality for Linear Programs}

One of the most important facts concerning linear programming is that strong
duality always holds.  Let us state the key results.  
The \navy{dual problem}\index{Dual problem} corresponding to the standard equality
form linear program \eqref{eq:linprog} is
\begin{equation}\label{eq:linprogdual}
    D = \max_{\theta \in \RR^m} \, b^\top \theta
    \; \text{ subject to } \;
    A^\top \theta \leq c.
\end{equation}

Readers who have covered \S\ref{ss:ld} in the appendix will be able to see the
origins of this expression.  In particular, by formula \eqref{eq:lagdual} in
the appendix, the dual problem corresponding to the standard equality form
linear program \eqref{eq:linprog} can be expressed as 
\begin{equation}
    D =
    \max_{\theta \in \RR^m}
    \min_{x \in E} \; L(x, \theta)
    \; \text{ where } \;
    L(x, \theta) 
        := c^\top x + \theta^\top (b - A x)
\end{equation}
and $E = \RR^n_+$.  (We can also treat the inequality $x \geq 0$ via a
multiplier but this turns out to be unnecessary.)
Now observe that
\begin{equation*}
    \min_{x \in E} \; L(x, \theta)
    =
    b^\top \theta + \min_{x \geq 0} x^\top (c - A^\top \theta) 
    = 
    \begin{cases}
        b^\top \theta & \text{ if } A^\top \theta \leq c \\
        -\infty & \text{ otherwise }
    \end{cases}.
\end{equation*}
Since the dual problem is to maximize this expression over $\theta \in
\RR^n$, we see immediately that a $\theta$ violating $A^\top \theta \leq c$
will never be chosen.  Hence the dual to the primal problem~\eqref{eq:linprog} is
\eqref{eq:linprogdual}.

\subsubsection{Strong Duality}

As shown in \S\ref{sss:theo} of the appendix, the inequality $D \leq P$ always holds,
where $P$ is the primal value in \eqref{eq:linprog}.  This is called
\navy{weak duality}\index{Weak duality}.  If $P = D$,  then \navy{strong
duality}\index{Strong duality} is said to hold.  Unlike weak duality, strong
duality requires conditions on the primitives.

The next theorem states that, for linear programs, strong duality holds whenever
a solution exists.  A proof can be obtained either through Dantzig's simplex
method or via Farkas' Lemma.  See, for example, Ch.~4 of
\cite{bertsimas1997introduction} or Ch.~6 of \cite{matousek2007understanding}.

\begin{theorem}[Strong duality for linear programs]\label{t:strongdual}
    The primal problem~\eqref{eq:linprog} has a finite minimizer $x^*$ if and only
    if the dual problem~\eqref{eq:linprogdual} has a finite maximizer $\theta^*$. 
    If these solutions exist, then $c^\top x^* = b^\top \theta^*$.  
\end{theorem}

Strong duality of linear programs has many important roles.  Some of these are
algorithmic: duality can be used to devise efficient solution methods for
linear programming problems.  Another way that duality matters for economists
is that it provides deep links between optimality and competitive equilibria,
as we show below.

\subsubsection{Complementary Slackness}

We say that $x^* \geq 0$ and $\theta^* \in \RR^m$ satisfy
the \navy{complementary slackness}\index{Complementary slackness} conditions
for the linear program \eqref{eq:linprog} when
\begin{align}
    \theta_i^* 
        \left(
            b_i - \sum_{j=1}^n a_{ij} x^*_j 
        \right) 
    & = 0
    \quad \text{for all } i \in \natset{m}
    \label{eq:lpcs1}
    \\
    x^*_j 
        \left(
            c_j - \sum_{i=1}^m a_{ij} \theta^*_i
        \right)
    & = 0
    \quad \text{for all } j \in \natset{n}.
    \label{eq:lpcs2}
\end{align}
While it is possible to derive these expressions from the complementary slackness
in the KKT conditions in \S\ref{sss:kkt}, a better
approach is to connect them directly to the saddle point condition.  

To see how this works, suppose that $x^* \geq 0$ is feasible for the primal
problem and $\theta^* \in \RR^m$ is feasible for the dual problem.  If $(x^*,
\theta^*)$ is a saddle point of $L$, then the complementary
slackness conditions \eqref{eq:lpcs1}--\eqref{eq:lpcs2} must hold.  Indeed,
\eqref{eq:lpcs1} is trivial when $x^*$ is feasible, since $Ax^* = b$. At the
same time, \eqref{eq:lpcs2} must be true because violation implies that 
\begin{equation*}
    x^*_j 
        \left(
            c_j - \sum_{i=1}^m a_{ij} \theta^*_i
        \right) > 0
    \quad \text{for some } j \in \natset{n},
\end{equation*}
due to dual feasibility (i.e., $A^\top \theta^* \leq c$) and $x^* \geq 0$.
But then $(x^*, \theta^*)$ is not a saddle point of $L$, since changing the
$j$-th element of $x^*$ to zero strictly reduces the Lagrangian
\begin{equation*}
    L(x, \theta) 
        = c^\top x + \theta^\top (b - A x)
        = x^\top (c  - A^\top \theta) + \theta^\top b .
\end{equation*}

By sharpening this argument, it is possible to show that, for linear
programs, the complementary slackness conditions are exact necessary and
sufficient conditions for a saddle point of the Lagrangian.  This leads to the
next theorem.

\begin{theorem}\label{t:csolp}
    If $x^* \geq 0$ is feasible for the primal problem and $\theta^*$ is
    feasible for the dual problem, then the following statements are
    equivalent:
    \begin{enumerate}
        \item $x^*$ is optimal for the primal problem and $\theta^*$ is
            optimal for the dual problem.
        \item The pair $(x^*, \theta^*)$ obeys the complementary slackness
            conditions~\eqref{eq:lpcs1}--\eqref{eq:lpcs2}.
    \end{enumerate}
\end{theorem}

Chapter~4 of \cite{bertsimas1997introduction} provides more discussion and a
full proof of Theorem~\ref{t:csolp}.  Below, we illustrate some of the
elegant connections between complementary slackness and equilibria in
competitive economic environments.

\subsubsection{The Simplex Algorithm}\label{sss:sial}

Linear programming is challenging in high-dimensional settings, partly because
the linear objective function implies that solutions are not interior.  The
first efficient solution methods for linear programs appeared in the 1930s and 1940s,
starting with the simplex method of Kantorovich and Dantzig.  As the simplex
method began to prove its worth, linear programming grew into a technique of
enormous practical importance. Operations research, communication and
production systems began to make heavy use of linear programs.  

The simplex algorithm makes use of the following result, which is proved in
Theorem~2.7 of \cite{bertsimas1997introduction}.

\begin{theorem}\label{t:lpeoe}
    If the linear program~\eqref{eq:linprog} has an optimal solution, then it
    also has an optimal solution that is an extreme point of the feasible set.
\end{theorem}

(An extreme point\index{Extreme point} of a polyhedron was defined in \S\ref{sss:cc}.  The
feasible set was shown to be a polyhedron in Exercise~\ref{ex:conpoly}.)

The simplex algorithm makes use of Theorem~\ref{t:lpeoe} by walking along
edges of the polyhedron that forms the feasible set, from one extreme point to
another, seeking at each step a new extreme point (which coincides with a
vertex of the polyhedron) that strictly lowers the value of the objective
function.  Details can be found in \cite{bertsimas1997introduction} and
\cite{matousek2007understanding}.

\section{Optimal Transport}\label{s:ot}

Next we turn to optimal flows across networks.  One simple---but
computationally nontrivial---example of a network flow problem is the linear
assignment problem we analyzed in \S\ref{ss:wahdp}. There, the vertices are
either workers or jobs and the edges are assignments, chosen optimally to
minimize aggregate cost.  More general network flow problems extend these
ideas, allowing endogenous formation of networks in more sophisticated
settings. The general structure is that vertices are given, while edges and
weights are chosen to optimize some criterion. There are natural applications
of these ideas in trade, transportation and communication, as well as less
obvious applications within economics, finance, statistics and machine
learning.

In our study of network flows, we begin with the optimal transport problem,
which is the most important special case of the general network flow problem.
(In fact, as we show in \S\ref{sss:redmfot}, there exists a technique by which
the general network flow problem can always be reduced to an optimal transport
problem.)

\subsection{The Monge-Kantorovich Problem}\label{ss:ot} 

Optimal transport is a classic problem dating back to the work of Gaspard
Monge (1746--1818), who studied, among many other things, the transport of earth for 
construction of forts. This simple-sounding problem---how to transport a pile
of given shape into a new pile of given shape at minimum cost---can,
after normalization, be identified with the problem of least cost
transformation of one distribution into another distribution.  
Figure~\ref{f:ot_figs_1} gives a visualization of transforming one
distribution into another in one dimension (although the cost function is not
specified).

\begin{figure}
   \centering
   \scalebox{0.5}{\includegraphics[trim = 0mm 0mm 0mm 0mm, clip]{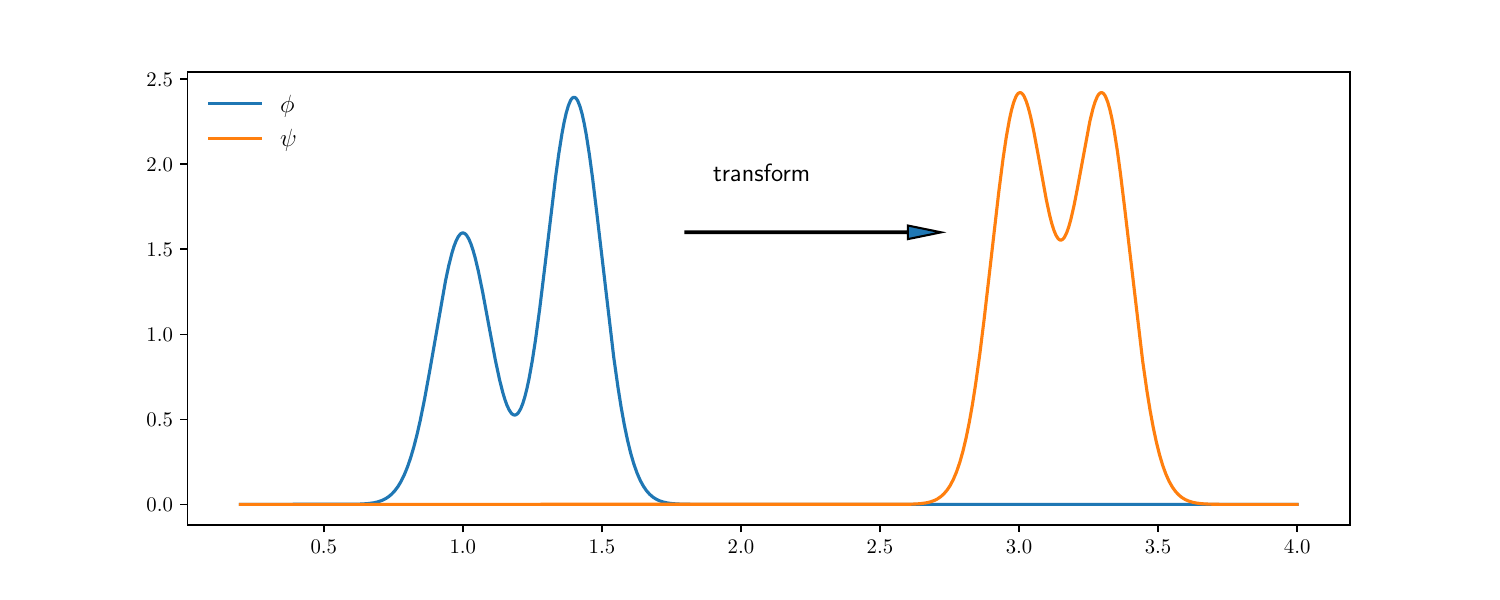}}
   \caption{\label{f:ot_figs_1} Transforming distribution $\phi$ into distribution $\psi$}
\end{figure}

It turns out that, by varying the notion of cost,
this transportation problem provides a highly flexible method for comparing
the distance between two distributions, with the essential
idea being that distributions are regarded as ``close'' if one can be
transformed into the other at low cost.  The resulting distance metric finds
wide-ranging and important applications in statistics, machine
learning and various branches of applied mathematics.\footnote{For example, in
    image processing, two images might be regarded as close if the cost of
    transforming one into the other by altering individual pixels is small.
    Even now, in image processing and some branches of machine learning, one
    of the metrics over the set of probability distributions arising from
    optimal transport methods is referred to as ``earth mover's distance'' in
honor of the work of Monge.}

In economics, optimal transport has important applications in transportation
and trade networks, as well as in matching problems, econometrics, finance and
so on (see, e.g., \cite{galichon2018optimal}).  As such, it is not surprising that economists have
contributed a great deal to the optimal transport problem, with deep and
fundamental work being accomplished by the great Russian mathematical
economist Leonid Kantorovich (1912--1986), as well as the Dutch economist
Tjalling Koopmans (1910--1985), who shared the Nobel Prize with Kantorovich in
1975 for their work on linear programming and optimal transport.\footnote{The
    optimal transport problem continues to attract the interest of many
    brilliant mathematicians and economists, with two recent Fields Medals
    being awarded for work on optimal transport.  The first was awarded to
    Cedric Villani in 2010, while the second was to Alessio Figalli in 2018.
    Note that the Field Medal is only awarded every four years (unlike the
Nobel Prize, which is annual).}

We start our discussion with a straightforward presentation of the
mathematics.  Then, in \S\ref{ss:otce}, we will connect the mathematics to
economic problems, and show the deep connections between transport, duality,
complementary slackness and competitive equilibria.

\subsubsection{Monge's Formulation}\label{sss:monge}

We start with the classical problem of Monge, which is simple to explain. We
take as given two finite sets $\Xsf$ and $\Ysf$, paired with distributions
$\phi \in \dD(\Xsf)$ and $\psi \in \dD(\Ysf)$.  Elements of $\Xsf$ and $\Ysf$
are called \navy{locations}.  To avoid tedious side cases, we assume
throughout that $\phi$ and $\psi$ are strictly positive on their domains. In
addition, we are supplied with a cost function $c \colon \Xsf \times \Ysf \to
\RR_+$.  Our goal is to ``transport'' $\phi$ into $\psi$ at minimum cost.
That is, we seek to solve
\begin{equation}\label{eq:mongeob}
    \min_T \sum_x c(x, T(x))
\end{equation}
over the set of all maps $T$ from $\Xsf$ onto $\Ysf$ satisfying
\begin{equation}\label{eq:mongecon}
    \sum_x \phi(x)  \1\{T(x) = y\} = \psi(y)
    \quad \text{for all } y \in \Ysf.
\end{equation}
The constraint says $T$ must be such that, for each target location $y$, the
sum of all probability mass sent to $y$ is equal to the specified quantity
$\psi(y)$.  The symbol $\sum_x$ is short for $\sum_{x \in \Xsf}$. In this
context, $T$ is often called a \navy{Monge map}\index{Monge map}.

\begin{Exercise}
    While we required $T$ to map $\Xsf$ \emph{onto} $\Ysf$, meaning that every
    $y \in \Ysf$ has some preimage, this condition is already implied by
    \eqref{eq:mongecon}.  Explain why.
\end{Exercise}

\begin{Answer}
    Fix $y \in \Ysf$.
    The distributions $\phi$ and $\psi$ are assumed to be everywhere positive, so
    if \eqref{eq:mongecon} holds for $T$, then $\phi(x)  \1\{T(x) = y\} > 0$
    for some location $x$.  Hence, there exists an $x \in \Xsf$ such that $T(x) = y$.
\end{Answer}

The problem is easily illustrated in the current discrete setting.
Figure~\ref{f:optimal_transport_unsplitting} gives a visualization when
locations are enumerated as $\Xsf = \{x_1, \ldots, x_7\}$ and $\Ysf = \{y_1,
y_2, y_3, y_4\}$, with both $\Xsf$ and $\Ysf$ being subsets of $\RR^2$.  For
simplicity, $\phi(x_i)$ is written as $\phi_i$ and similarly for $\psi(y_j)$.
Vertex size is proportional to probability mass assigned to the vertex.  The
edges represent one feasible Monge map.  

\begin{figure*}
   \begin{center}
    \input{tikz/optimal_transport_unsplitting.tex}
    \caption{\label{f:optimal_transport_unsplitting} A Monge map transporting $\phi$ to $\psi$}
   \end{center}
\end{figure*}
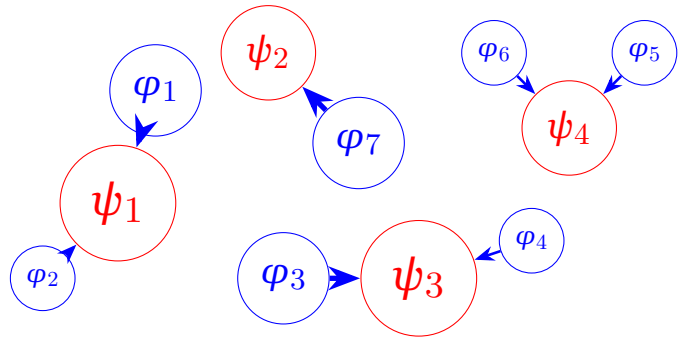

Discreteness and lack of convexity in the constraint \eqref{eq:mongecon} imply
that the Monge problem is, in general, hard to solve.  Truly fundamental
progress had to wait until Kantorovich showed how convexification can
greatly simplify the problem. The convexification process requires shifting
the problem to a higher dimensional space, but the cost of higher dimensions
is outweighed by the regularization provided by convexity and lack of
discreteness. We study the Kantorovich formulation in \S\ref{sss:mkp}.

\begin{Exercise}
    Another issue with the Monge formulation of optimal transport is that
    existence of a solution can easily fail.  Provide an example where no
    Monge map exists in the setting where $\Xsf$ and $\Ysf$ are finite.
\end{Exercise}

\begin{Answer}
    One scenario where no Monge map exists is when $|\Ysf| > |\Xsf|$.  For
    example, suppose $\Xsf = \{x_1\}$ and $\Ysf = \{y_1, y_2\}$, with
    $\phi(x_1) = 1$ and $\psi(y_i) \in (0,1)$ for $i=1,2$.  Either $x_1$ is mapped
    to $y_1$ or it is mapped to $y_2$.  In either case, the Monge map condition
    \eqref{eq:mongecon} fails for both $y_1$ and $y_2$.
\end{Answer}

\subsubsection{Assignment as Optimal Transport}\label{sss:asla}

The linear assignment problem studied in \S\ref{ss:wahdp}, with cost $c(i, j)$ of
training worker $i$ for job $j$, is a special case of optimal
transport.  All we have to do is set $\Xsf = \Ysf = \natset{n}$ and take
$\phi$ and $\psi$ to be discrete uniform distributions on $\natset{n}$.

\begin{Exercise}
    Show that, in this setting, $T$ is a Monge map if and only if $T$ is a
    bijection from $\natset{n}$ to itself.
\end{Exercise}

\begin{Answer}
    Let $T$ be a self-map on $\natset{n}$.
    If $T$ is a Monge map, then, by the definition in \eqref{eq:mongecon}, we must have
    $\sum_i (1/n)  \1\{T(i) = j\} = 1/n$ for all $j$.
    If $T$ is not a bijection, then
    there exist indices $i, k, j$ such that $i\not= k$ and $T(i)=T(k)=j$.  This clearly
    violates the previous equality.

    Conversely, if $T$ is a bijection on $\natset{n}$, then $T$ 
    satisfies $\sum_i (1/n)  \1\{T(i) = j\} = 1/n$ for all $j$.  Hence $T$ is
    a Monge map.
\end{Answer}

Since $T$ must be a bijection on $\natset{n}$, which is also a permutation of
$\natset{n}$, the objective of the optimal transport problem under the current
configuration is 
\begin{equation}\label{eq:matchcost2}
    \min_{T \in \pP} \sum_{i=1}^n c(i, T(i)),
\end{equation}
where $\pP$ is the set of all permutations of $\natset{n}$.
This is the same optimization problem as the linear assignment problem in
\eqref{eq:matchcost}.

\subsubsection{Kantorovich's Relaxation of the Monge Problem}\label{sss:mkp}

The basic idea in Kantorovich's relaxation of the Monge problem is to allow
mass located at arbitrary $x \in \Xsf$ to be mapped to multiple locations in
$\Ysf$, rather than just one.  This means that we are no longer seeking a
function $T$, since, by definition, a function can map a given point to only
one image.  Instead, we seek a ``transport plan'' that sends some
fraction $\pi(x, y)$ of the mass at $x$ to $y$.  The plan is constrained by
the requirement that, for all $x$ and $y$, 
\begin{enumerate}
    \item total probability mass sent to $y$ is $\psi(y)$ and
    \item total probability mass sent from $x$ is $\phi(x)$.
\end{enumerate}
These constraints on this transport plan mean that it takes the form of a
``coupling,'' which we define below.

Figure~\ref{f:optimal_transport_splitting_experiment} illustrates a feasible transport
plan in the discrete setting.  As with
Figure~\ref{f:optimal_transport_unsplitting}, $\phi(x_i)$ is written as
$\phi_i$ and similarly for $\psi(y_j)$, while vertex size is proportional to
probability mass.  Unlike the Monge setting of
Figure~\ref{f:optimal_transport_unsplitting}, the mass at each vertex $\phi_i$
can be shared across multiple $\psi_j$, as long as the constraints are
respected.

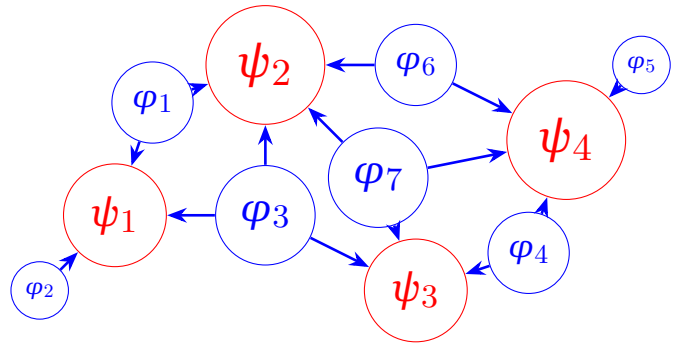
\begin{figure*}
   \begin{center}
    \input{tikz/optimal_transport_splitting_experiment.tex}
    \caption{\label{f:optimal_transport_splitting_experiment} Kantorovich relaxation of the Monge problem}
   \end{center}
\end{figure*}


Let's write the constraints more carefully. We recall that, in probability theory, a coupling
is a joint distribution with specific marginals.  More precisely, given $\phi$
in  $\dD(\Xsf)$ and $\psi$ in $\dD(\Ysf)$, a \navy{coupling}\index{coupling}
of $(\phi, \psi)$ is an element $\pi$ of $\dD(\Xsf \times \Ysf)$ with
marginals $\phi$ and $\psi$.  This restriction on marginals means that
\begin{align}
    \label{eq:romar0}
    \sum_y  \pi(x, y) 
    & = \phi(x)
    \quad \text{for all } x \in \Xsf 
    \quad \text{ and}
    \\
    \label{eq:romar1}
    \sum_x  \pi(x, y) 
    & = \psi(y)
    \quad \text{for all } y \in \Ysf
\end{align}
The constraints in \eqref{eq:romar0}--\eqref{eq:romar1} require that
\begin{enumerate}
    \item for any $x \in \Xsf$, the total amount of probability mass flowing
        out of $x$ is $\phi(x)$ and
    \item for any $y \in \Ysf$, the total amount of probability mass flowing
        into $y$ is $\psi(y)$.
\end{enumerate}

In the present setting, a coupling is also called a \navy{transport
plan}\index{Transport plan}.

\begin{Exercise}
    The constraints \eqref{eq:romar0}--\eqref{eq:romar1}, which define a
    coupling $\pi$ of $\phi$ and $\psi$, generalize the Monge constraint in
    \eqref{eq:mongecon}.  To see this, let $T$ be a map satisfying
    \eqref{eq:mongecon} and set 
    \begin{equation*}
        \pi(x, y) = \phi(x) \1\{T(x)=y\} 
    \end{equation*}
    on $\Xsf \times \Ysf$, so that $\pi$ is the joint distribution that puts
    all mass on the image of $T$.  Prove that
    \eqref{eq:romar0}--\eqref{eq:romar1} both hold.
\end{Exercise}

Let $\Pi(\phi, \psi)$ be the set of all couplings of
$\psi$ and $\phi$.    Taking $\phi, \psi$ and the cost function $c$ as given,
the general \navy{Monge--Kantorovich problem}, also called the \navy{optimal
transport problem}, is to solve
\begin{equation}\label{eq:ot0}
    P :=
    \min_\pi \inner{c, \pi}_F
    \quad \text{subject to} \quad
    \pi \in \Pi(\phi, \psi).
\end{equation}
where
\begin{equation*}
    \inner{c, \pi}_F := \sum_x \sum_y c(x,y) \pi(x, y)
\end{equation*}
is the \navy{Frobenius inner product}\index{Inner product!Frobenius} of $c$
and $\pi$, treated as $|\Xsf| \times |\Ysf|$ matrices.
The sum in $\inner{c, \pi}_F$ measures the total cost of transporting $\phi$
into $\psi$ under the plan $\pi$.  There is linearity embedded in this cost
formulation, since doubling the amount sent from $x$ to $y$ scales the
associated cost at rate $c(x, y)$.

We call any $\pi$ solving~\eqref{eq:ot0} an \navy{optimal plan}\index{Optimal
plan}.  Since we are maximizing over a finite set, at least one such plan
exists.

\begin{remark}
    The problem \eqref{eq:ot0} is sometimes expressed in terms of random
    variables, as follows.  In this setting, a coupling $\pi$ in
    $\Pi(\phi, \psi)$ is identified with a pair of random elements $(X, Y)$ such that 
    $X \eqdist \phi$ and $Y \eqdist \psi$.  
    We can then write
    \begin{equation*}
        P = \min_{(X,Y)} \EE \, c(X, Y)
        \quad \text{subject to} \quad
        (X, Y) \in \dD(\Xsf \times \Ysf)
        \text{ with } X \eqdist \phi \text{ and } Y \eqdist \psi.
    \end{equation*}

\end{remark}

\begin{Exercise}\label{ex:otconvexcon}
    One of the most important features of the Kantorovich relaxation is that,
    for given $\phi$ and $\psi$, the constraint set is convex.  
    To verify this, we let $n = |\Xsf|$ and $m=|\Ysf|$, associate each $x \in \Xsf$ with some $i
    \in \natset{n}$, associate each $y \in \Ysf$ with some $j \in \natset{m}$, and
    treat $c$ and $\pi$ as $n \times m$ matrices, with typical elements $c_{ij}$
    and $\pi_{ij}$.  The
    constraints are 
    \begin{equation}\label{eq:conpkv}
        \pi \1_m = \phi  
        \quad \text{and} \quad
        \pi^\top \1_n = \psi  ,
    \end{equation}
    where $\1_k$ is a $k \times 1$ vector of ones.  With this notation, prove
    that the set $\Pi(\phi, \psi)$ of $\pi \in \matset{n}{m}$ satisfying the
            constraints is convex, in the sense that 
            \begin{equation*}
                \pi, \hat \pi \in \Pi(\phi, \psi)
                \text{ and }
                \alpha \in [0, 1]
                \; \implies \;
                \alpha \pi + (1 - \alpha) \hat \pi \in \Pi(\phi, \psi).
            \end{equation*}
\end{Exercise}

\subsubsection{Optimal Transport as a Linear Program}\label{sss:otplp}

With some relatively simple manipulations, the general optimal transport
problem can be mapped into a standard equality form\index{Standard form}
linear program.  This provides two significant benefits.  First, we can apply
duality theory, which yields important insights.  Second, on the computational
side, we can use linear program solvers to calculate optimal plans.

To map the optimal transport problem into a linear program, we need to convert
matrices into vectors.  Will use the $\vecop$ operator, which
takes an arbitrary $A \in \matset{n}{m}$ and maps it to a vector in $\RR^{nm}$
by stacking its columns vertically.  For example,
\begin{equation*}
    \vecop
    \begin{pmatrix}
        a_{11} & a_{12} 
        \\
        a_{21} & a_{22} 
    \end{pmatrix}
    = 
    \begin{pmatrix}
        a_{11} \\
        a_{21} \\
        a_{12} \\
        a_{22}
    \end{pmatrix}
    .
\end{equation*}

In this section we adopt the notational conventions in
Exercise~\ref{ex:otconvexcon}. The objective function $\inner{c, \pi}_F$ for the
optimal transport problem can now be expressed as $\vecop(c)^\top \vecop(\pi)$. 

To rewrite the constraints in terms of $\vecop(\pi)$, we use the Kronecker
product, which is denoted by $\otimes$ and defined as follows.  Suppose $A$ is
an $m \times s$ matrix with entries $(a_{ij})$ and that $B$ is an $n \times t$
matrix.  The \navy{Kronecker product}\index{Kronecker product} $A \otimes B$
of $A$ and $B$ is the $mn \times st$ matrix defined, in block matrix form, by
\begin{equation*}
    A \otimes B = 
    \begin{pmatrix}
    a_{11}B & a_{12}B & \dots & a_{1s}B \\ 
    a_{21}B & a_{22}B & \dots & a_{2s}B \\ 
      &   & \vdots &   \\ 
    a_{m1}B & a_{m2}B & \dots & a_{ms}B \\ 
    \end{pmatrix}.
\end{equation*}
It can be shown that Kronecker products and the $\vecop$ operator are
connected by the following relationship:  for conformable matrices $A$, $B$
and $M$, we have 
\begin{equation}\label{eq:krovec}
    \vecop(A M B) = (B^\top \otimes A) \vecop(M).
\end{equation}
Using \eqref{eq:krovec} and the symbol $I_k$ for the $k \times k$ identity
matrix, we can rewrite the first constraint in \eqref{eq:conpkv} as
\begin{equation}\label{eq:phikv}
    \phi 
    = I_n \pi \1_m 
    = \vecop(I_n \pi \1_m )
    = (\1_m^\top \otimes I_n) \vecop( \pi ).
\end{equation}

\begin{Exercise}
    Show that the second constraint in \eqref{eq:conpkv} can be expressed as 
    \begin{equation}\label{eq:psikv}
        \psi 
        = (I_m \otimes \1_n^\top) \vecop( \pi ).
    \end{equation}
\end{Exercise}

\begin{Answer}
    Applying \eqref{eq:krovec} on page~\pageref{eq:krovec}, we have
    \begin{equation*}
        \psi 
        = \vecop(\psi^\top)
        = \vecop(\1_n^\top \pi I_m  )
        = (I_m \otimes \1_n^\top) \vecop( \pi ).
    \end{equation*}
\end{Answer}

Now, using block matrix notation and setting
\begin{equation*}
    A := 
    \begin{pmatrix}
        \1_m^\top \otimes I_n \\
        I_m \otimes \1_n^\top
    \end{pmatrix}
    \quad \text{and} \quad
    b :=
    \begin{pmatrix}
        \phi \\
        \psi
    \end{pmatrix},
\end{equation*}
the optimal transport problem can be expressed as the standard equality form
linear program 
\begin{equation}\label{eq:otplp}
    \min_x \vecop(c)^\top x
    \quad
    \text{ over } x \in \RR^{nm}_+
    \text{ such that }
    A x = b.
\end{equation}
Finally, for a given solution $x$, the transport plan is recovered by
inverting $x = vec(\pi)$.

\subsubsection{Implementation}\label{sss:iotpl}

Listing \ref{l:otshm} is a function that implements the above steps, given flat (one-dimensional)
arrays \texttt{phi} and \texttt{psi} representing the distributions over the
source and target locations, plus a two-dimensional array \texttt{c}
representing transport costs.  (The \texttt{method} argument
\textcolor{red}{\texttt{highs-ipm}} tells \texttt{linprog} to use a particular
interior point method, details of which can be found in the \texttt{linprog}
documentation. Simplex and other methods give similar results.) 

\begin{listing}
\begin{minted}{python}
import numpy as np
from scipy.optimize import linprog

def ot_solver(phi, psi, c, method='highs-ipm'):
    """
    Solve the OT problem associated with distributions phi, psi 
    and cost matrix c.

    Parameters
    ----------
    phi : 1-D array
        Distribution over the source locations.
    psi : 1-D array
        Distribution over the target locations.
    c : 2-D array
        Cost matrix.
    """
    n, m = len(phi), len(psi)

    # vectorize c
    c_vec = c.reshape((m * n, 1), order='F') 

    # Construct A and b
    A1 = np.kron(np.ones((1, m)), np.identity(n))
    A2 = np.kron(np.identity(m), np.ones((1, n)))
    A  = np.vstack((A1, A2))
    b  = np.hstack((phi, psi))

    # Call solver
    res = linprog(c_vec, A_eq=A, b_eq=b, method=method)

    # Invert the vec operation to get the solution as a matrix
    pi = res.x.reshape((n, m), order='F')
    return pi
\end{minted}
\caption{\label{l:otshm} Function to solve a transport problem via linear programming}
\end{listing}

Notice that in Listing~\ref{l:otshm}, the reshape order is specified to
\texttt{F}.  This tells NumPy to reshape with Fortran \navy{column-major} order,
which coincides with the definition of the $\vecop$ operator described in
\S\ref{sss:otplp}.  (The Python vectorize operation defaults to
\navy{row-major} order, which concatenates rows rather than stacking columns.
In contrast, Julia uses column-major by default.)

Let's call this function for the very simple problem
\begin{equation*}
    \phi =
    \begin{pmatrix}
        0.5 \\
        0.5
    \end{pmatrix},
    \quad
    \psi =
    \begin{pmatrix}
        1 \\
        0
    \end{pmatrix}
    \quad \text{and} \quad
    c =
    \begin{pmatrix}
        1 & 1 \\
        1 & 1
    \end{pmatrix}
\end{equation*}
With these primitives, all mass from $\phi_1$ and $\phi_2$ should be sent to
$\psi_1$. To implement this problem we set
\begin{minted}{python}
phi = np.array((0.5, 0.5))
psi = np.array((1, 0))
c = np.ones((2, 2))
\end{minted}
and then call \texttt{ot\_solver} via
\begin{minted}{python}
ot_solver(phi, psi, c)
\end{minted}
The output is as expected:
\begin{minted}{python}
array([[0.5, 0. ],
       [0.5, 0. ]])    
\end{minted}

\subsubsection{Python Optimal Transport}\label{sss:imot}

In the case of Python, the steps above have been automated by the Python
Optimal Transport package, due to \cite{flamary2021pot}.  For the simple
problem from \S\ref{sss:iotpl} we run
\begin{minted}{python}
import ot
ot.emd(phi, psi, c)   # Use simplex method via the emd solver
\end{minted}
The output is again equal to
\begin{minted}{python}
array([[0.5, 0. ],
       [0.5, 0. ]])    
\end{minted}

Figure~\ref{f:ot_large_scale_1} shows an example of an optimal transport
problem solved using the Python Optimal Transport package.  The interpretation
is similar to Figure~\ref{f:optimal_transport_splitting_experiment}, although
the number of vertices is larger. In addition, the edges show the optimal
transport configuration, in the sense that $\pi^*$, the optimal transport
plan, is treated as the adjacency matrix for the graph.  (The figure shows the
unweighted graph, with an arrow drawn from $\phi_i$ to $\psi_j$ whenever
$\pi^*_{ij} > 0$.) The optimal transport plan is obtained by converting the
transport problem into a linear program, as just described, and applying the
simplex method. Although there are 32 nodes of each type, the problem is
solved by the simplex routine in less than one millisecond.

\begin{figure}
   \centering
   \scalebox{0.92}{\includegraphics[trim = 0mm 30mm 0mm 40mm, clip]{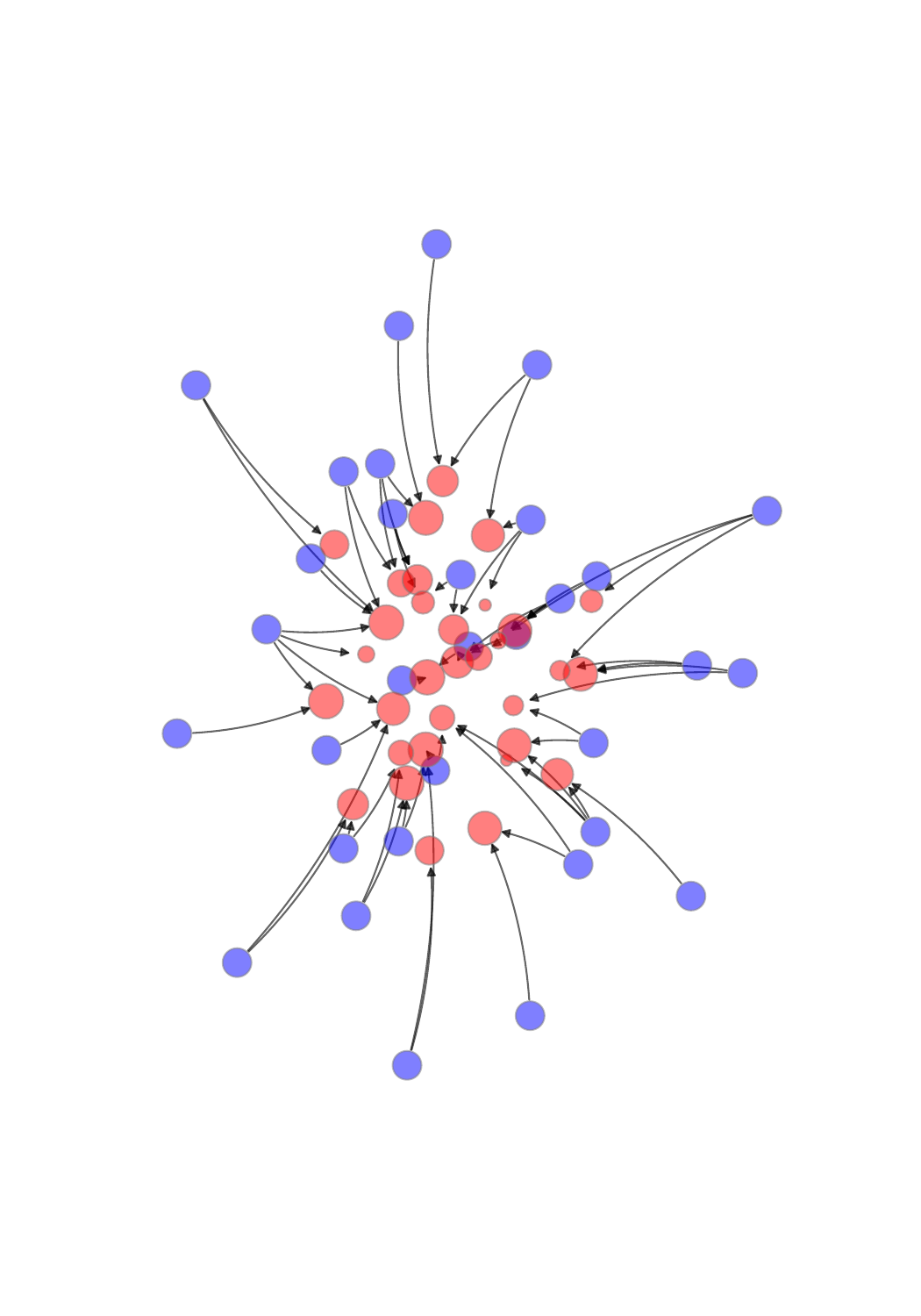}}
   \caption{\label{f:ot_large_scale_1} An optimal transport problem solved by linear programming}
\end{figure}

\subsubsection{Kantorovich Relaxation and Linear Assignment}

We showed in \S\ref{sss:asla} that the linear assignment problem studied in
\S\ref{ss:wahdp} is a special case of optimal transport, where $\phi$ and
$\psi$ become discrete uniform distributions on $\natset{n}$.  Moreover, as we
have just seen, Kantorovich's relaxation method allows us to apply linear
programming. This leads to fast solutions.

In the discussion in \S\ref{ss:wahdp} we used $n=40$.  Let's start here with
$n=4$, to illustrate the method, and then try with $n=40$.
The matrix $c(i,j)$ of costs, which is the only other primitive, will be
generated randomly as an independent array of uniform random variables: Then we
apply the Python Optimal Transport (POT) library, as in \S\ref{sss:imot}.

Here is our set up:
\begin{minted}{python}
import numpy as np
import ot

n = 4
phi = np.ones(n) 
psi = np.ones(n)
\end{minted}

We have broken the rule that $\phi$ and $\psi$ should sum to one.  This could
be fixed easily by using \texttt{np.ones(n)/n} instead of \texttt{np.ones(n)}, but the POT library
does not care (as long as \texttt{np.sum(phi)} equals \texttt{np.sum(psi)})
and, moreover, the idea of putting unit mass everywhere is natural, since each
element of $\phi$ represents one worker, and each element of $\psi$ represents
one job.

Now we build the cost matrix:
\begin{minted}{python}
c = np.random.uniform(size=(n, n))
\end{minted}
The output is 
\begin{minted}{python}
array([[0.03841, 0.32896, 0.55989, 0.41682],
       [0.91527, 0.24566, 0.26022, 0.64510],
       [0.96275, 0.44089, 0.79274, 0.93065],
       [0.40454, 0.87307, 0.43555, 0.54903]])
\end{minted}
(For example, the cost of retraining worker 1 for job 2 is 0.32896.) Finally, we call the solver:
\begin{minted}{python}
ot.emd(phi, psi, c)    
\end{minted}
The output is
\begin{minted}{python}
array([[1., 0., 0., 0.],
       [0., 0., 1., 0.],
       [0., 1., 0., 0.],
       [0., 0., 0., 1.]])    
\end{minted}

This is a permutation matrix, which provides another way to express a
permutation of $\natset{n}$.  The first row tells us that worker 1 is assigned
to job 1, the second tells us that worker 2 is assigned to job 3, and so on.

If we now set $n=40$ and rerun the code, the line \texttt{ot.emd(phi, psi,
c)}, which calls the simplex-based solver, runs in less than 1 millisecond on
a mid-range laptop.
This is a remarkable improvement on the $2.5 \times 10^{30}$ year estimate for
the brute force solver we obtained in \S\ref{ss:wahdp}.

\subsubsection{Tight Relaxation}

Notice that the solution we obtained for the linear assignment problem using
the simplex method does not split mass, as permitted by the Kantorovich
relaxation.  For example, we do not send half of a worker to one job and the
other half to another.  This is convenient but why does it hold?

While we omit the details, the basic idea is that a general Kantorovich
transport plan is a bistochastic matrix, and all such matrices can be
formed as convex combinations of permutation matrices. (This is called
Birkhoff's Theorem.) In other words, the permutation matrices are extreme
points of the set of bistochastic matrices.   Moreover,
Theorem~\ref{t:lpeoe} tells us that any optimizer of a linear program will be
an extreme point---in this case, a permutation matrix.

\subsection{Kantorovich Duality}\label{ss:kd}\index{Kantorovich duality}

One of the greatest achievements of Kantorovich was to show that the optimal
transport problem can be connected to a dual problem, and how that dual
problem can be used to characterize solutions. This work anticipated much of
the later development of duality theory for arbitrary linear programs.  

Throughout this section, in stating the main results, we use the notation
\begin{equation*}
    \inner{f, \phi} = \sum_x f(x) \phi(x) 
    \quad \text{for }
    f \in \RR^\Xsf
    \text{ and }
    \phi \in \dD(\Xsf).
\end{equation*}
This is just the usual inner product, when we think of $f$ and $\phi$ as
vectors in $\RR^{|\Xsf|}$.   Also, given a cost function $c$ on $\Xsf \times
\Ysf$, let $\fF_c$ be all pairs $(w, p)$ in $\RR^\Xsf \times \RR^\Ysf$ such
that 
\begin{equation}\label{eq:fcdef}
    p(y) \leq c(x, y) + w(x) \text{ on } \Xsf \times \Ysf.
\end{equation}

One part of Kantorovich's duality results runs as follows.

\begin{theorem}\label{t:kdual0}
    For all $\phi \in \dD(\Xsf)$ and $\psi \in \dD(\Ysf)$, we have $P = D$,
    where
    \begin{equation}\label{eq:kdual0}
        D
        := \max_{(w, p)}
        \left\{
            \inner{p, \psi} - \inner{w,  \phi}
        \right\}
        \quad \text{subject to} \quad
        (w, p) \in \fF_c.
    \end{equation}
\end{theorem}

Theorem~\ref{t:kdual0} can now be understood as a special case of the more
general result that strong duality holds for linear programs, which we stated
in Theorem~\ref{t:strongdual}.  In that spirit, let us verify
Theorem~\ref{t:kdual0} using Theorem~\ref{t:strongdual}, by working with the
linear programming formulation of the optimal transport problem provided in
\S\ref{sss:otplp}.  

To do this, we take that formulation, which is stated in \eqref{eq:otplp},
and apply the dual formula in \eqref{eq:linprogdual}, which yields
\begin{equation*}
    D = \max_{\theta \in \RR^{n+m}}
    \begin{pmatrix}
        \phi \\
        \psi
    \end{pmatrix}^\top  \theta
    \; \text{ subject to } \;
    \begin{pmatrix}
        \1_m^\top \otimes I_n \\
        I_m \otimes \1_n^\top
    \end{pmatrix}^\top 
    \theta \leq \vecop(c).
\end{equation*}
If we write the argument $\theta \in \RR^{n+m}$ as $(-w, p)$, so that we now
maximize over the two components $-w \in \RR^n$ and $p \in \RR^m$, as well as
transposing the constraint, we get
\begin{equation*}
    \max_{w, \, p}
    \left\{
        p^\top \psi - w^\top \phi
    \right\}
    \; \text{ subject to } \;
    p^\top (I_m \otimes \1_n^\top)
    - w^\top (\1_m^\top \otimes I_n)
    \leq \vecop(c)^\top
\end{equation*}
where $w \in \RR^n$ and $p \in \RR^m$.  By using the definition of the
Kronecker product and carefully writing out the individual terms, it can be
shown that the constraint in this expression is equivalent to requiring that
$p_j - w_i \leq c_{ij}$ for all $(i,j) \in \natset{n} \times \natset{m}$.
Recalling that $\Xsf$ has been mapped to $\natset{n}$ and 
$\Ysf$ has been mapped to $\natset{m}$, this is exactly the same restriction
as \eqref{eq:fcdef}.

At this point it is clear that \eqref{eq:kdual0} is nothing but the dual 
of the linear program formed from the optimal transport problem.
The claims in Theorem~\ref{t:kdual0} now follow directly from the strong
duality of linear programs (Theorem~\ref{t:strongdual}).

\begin{Exercise}\label{ex:wdotp}
    Show that 
    \begin{equation}\label{eq:wdotp}
        \inner{c, \pi} \geq \inner{p, \psi} - \inner{w, \phi}
        \text{ whenever }
        \pi \in \Pi(\phi, \psi)
        \text{ and } 
        (w, p) \in \fF_c.   
    \end{equation}
    Use this fact to provide a direct proof that weak duality holds for the
    optimal transport problem, in the sense that $D \leq P$.   (Here $P$ is
    defined in the primal problem \eqref{eq:ot0} and $D$ is defined in the
    dual problem \eqref{eq:kdual0}.)  
\end{Exercise}

\begin{Answer}
    Let $\pi$ be feasible for the primal problem and let $(w, p)$ be feasible for
    the dual.  By dual feasibility, we have $c(x, y) \geq p(y) - w(x)$ for all $x,
    y$, so
    \begin{equation*}
        \inner{\pi, c}
        \geq \sum_x \sum_y  \pi(x, y) [p(y) - w(x)]
        = \sum_x \sum_y  \pi(x, y) p(y) 
            - \sum_x \sum_y  \pi(x, y)w(x)
    \end{equation*}
    Rearranging and using primal feasibility now gives $\inner{\pi, c} \geq \inner{p, \psi} -
    \inner{w, \phi}$.   This proves the first claim.

    To see that $D \leq P$ follows from the last inequality, just fix $\pi \in
    \Pi(\phi, \psi)$ and maximize over all feasible dual pairs to obtain
    $\inner{\pi, c} \geq D$.   Now minimize over $\pi \in \Pi(\phi, \psi)$.
\end{Answer}

\subsection{Optimal Transport and Competitive Equilibria}\label{ss:otce}

The other major achievement of Kantorovich in the context of duality theory
for optimal transport was to connect optimality of transport plans with the
existence of functions $w, p$ from the dual problem such that a version of the
complementary slackness conditions holds.  Here we present this result, not in
the original direct formulation, but rather through the lens of a competitive
equilibrium problem.  In doing so, we illustrate some of the deep connections
between prices, decentralized equilibria and efficient allocations.

\subsubsection{The Advisor's Problem}

We imagine the following scenario.  Iron is mined at a finite collection of
sites, which we denote by $\Xsf$.  We identify an element $x \in \Xsf$ with a
point $(a, b) \in \RR^2$, which can be understood as the location of the mine
in question on a map.  At the wish of the queen, who seeks to defend the
empire from greedy rivals, this iron is converted to swords by blacksmiths.
There are a number of talented blacksmiths in this country, located at sites
given by $\Ysf$. As for $\Xsf$, each $y \in \Ysf$ indicates a point in $\RR^2$.
Henceforth, we refer to ``mine $x$'' rather than ``the mine at $x$'' and so
on.

Each month, mine $x$ produces $\phi(x)$ ounces of iron ore, while blacksmith
$y$ consumes $\psi(y)$ ounces.  We take these quantities as fixed.  We assume
that total supply equals total demand, so that $\sum_x \phi(x) = \sum_y
\psi(y)$.  For convenience, we normalize this sum to unity.  As a result,
$\phi$ and $\psi$ are elements of $\dD(\Xsf)$ and $\dD(\Ysf)$ respectively.

The cost of transporting from $x$ to $y$ is known and given by $c(x,y)$ per
ounce.  The king's chief advisor is tasked with allocating and transporting
iron from the mines to the blacksmiths, such that each blacksmith $y$ receives
their desired quantity $\psi(y)$, at minimum cost.  A small amount of thought
will convince you that the advisor's problem is a version of the optimal transport
problem \eqref{eq:ot0}.  We call this the primal problem in what follows.

Operating in the days before Kantorovich, Dantzig and the electronic
computer, the advisor employs a large team of bean counters, instructing them
to find the allocation with least cost by trying different combinations.
However, after a few days, she realizes the futility of the task.  (With
infinite divisibility, which corresponds to our mathematical model, the number
of allocations is infinite.  If we replace infinite divisibility with a finite
approximation, the scale can easily be as large as that of the matching
problem discussed in \S\ref{ss:wahdp}, with only a moderate number of mines
and blacksmiths.)

\subsubsection{The Guild's Problem}

At this point she has another idea.  There is a guild of traveling salesmen,
who buy goods in one town and sell them in another.  She seeks out the guild
master and asks him to bid for the project along the following lines.  The
guild will pay the queen's treasury $w(x)$ per ounce for iron ore at mine $x$.
It will then sell the iron at price $p(y)$ per ounce to the queen's
representative at blacksmith $y$.  The difference can be pocketed by the
guild, as long as all blacksmiths are provided with their desired quantities.
The guild master is asked to propose price functions $w$ and $p$.

The guild master sees at once that $p$ and $w$ must satisfy $p(y) - w(x) \leq
c(x, y)$ at each $x, y$, for otherwise the advisor, who is no ones fool, will
see immediately that money could be saved by organizing the transportation
herself.  Given this constraint, the guild master seeks to maximize aggregate
profits, which is $\sum_y p(y) \psi(y) - \sum_x w(x) \phi(x)$.  At this point
it will be clear to you that the problem of the guild master is exactly that
of Kantorovich's dual problem, as given in Theorem~\ref{t:kdual0}.

Since the advisor has given up on her team of bean counters, the guild master
employs them, and asks them to produce the optimal pair of prices.
The bean counters set to work, trying different combinations of prices that
satisfy the constraints.  However, without a systematic methodology to follow
or fast computers to turn to, their progress is slow.  The advisor begins to
fear that the coming war will be over before the guild master replies.

\subsubsection{Decentralized Equilibrium}

At this point, it occurs to the advisor that yet another approach exists:
privatize the mines, abolish the guild, and let the traveling salesmen, mine
owners and blacksmiths make individual choices in order to maximize their
profits.   Purchase and sales prices, as well as quantities transported from
each mine to each blacksmith, will be determined by the free market.  

Although the advisor predates Kantorovich, she reasons that competition will
prevent each salesman from profiteering, while the desire for profits will
encourage high levels of transportation and minimal waste.  It turns out that
this idea works amazingly well, in the sense that we now describe.

For the record, we define a \navy{competitive equilibrium}\index{Competitive
equilibrium} for this market as pair of price vectors $(w, p)$ in $\RR^\Xsf
\times \RR^\Ysf$ and a set of quantities $\pi \colon \Xsf \times \Ysf \to
\RR_+$ such that the following three conditions hold: For all $(x, y)$ in
$\Xsf \times \Ysf$,
\begin{align}
    & \sum_{v \in \Ysf}  \pi(x, v) = \phi(x) 
    \text{ and } 
    \sum_{u \in \Xsf} \pi(u, y) = \psi(y)
    \tag{RE}
    \\
    & p(y) \leq c(x, y) + w(x) 
    \tag{NA}
    \\ 
    & p(y) = c(x, y) + w(x) \text{ whenever } \pi(x, y) > 0.
    \tag{IC}
\end{align}

Condition (RE) is a resource constraint that builds in the assumption no ore is wasted or
disposed. Condition (NA) imposes no arbitrage.  If it is violated
along route $(x,y)$, then another salesman, of which we assume there are
many, will be able to gain business without suffering losses by offering a
slightly higher purchase prices at $x$ or a slightly lower sales prices at $y$.  Finally,
condition (IC) is an incentive constraint, which says that, whenever a route
is active (in the sense that a nonzero quantity is transported), prices are
such that the salesmen do not lose money.
  
We do not claim that a competitive equilibrium will hold immediately and at
every instant in time.  However, we reason, as the advisor does, that
competitive equilibrium has natural stability properties, as described in the
previous paragraph.  As such, we predict it as a likely outcome of
decentralized trade, provided that private property rights are enforced (e.g.,
bandits are eliminated from the routes) and noncompetitive behaviors are
prevented (e.g., collusion by mine owners is met by suitably painful
punishments).

Taking $c$, $\phi$ and $\psi$ as given, we can state the following key
theorem, which states that any competitive equilibrium simultaneously solves
both the advisor's quantity problem \emph{and} the guild master's price
problem. 

\begin{theorem}\label{t:optiffce}
    If prices $(w, p)$ and $\pi \in \matset{n}{m}$ form a competitive equilibrium,
    then 
    \begin{enumerate}
        \item $\pi$ is an optimal transport plan, solving
            the primal problem \eqref{eq:ot0}, and
        \item $(w, p)$ solves the Kantorovich dual problem~\eqref{eq:kdual0}.
    \end{enumerate}
\end{theorem}

To prove the theorem we will use the results from the next exercise.

\begin{Exercise}\label{ex:mmm}
    Let $A$ and $B$ be nonempty sets.  Let $f$ and $g$ be real-valued on $A$
    and $B$ such that $f(a) \geq g(b)$ for all $(a, b) \in A \times B$ and,
    in addition, $\min_{a \in A} f(a) = \max_{b \in B} g(b)$.  Prove the
    following statement: If there exist $\bar a \in A$ and $\bar b \in B$ such
    that $f(\bar a) = g(\bar b)$, then $\bar a$ is a minimizer of $f$ on $A$
    and $\bar b$ is a maximizer of $g$ on $B$.
\end{Exercise}

\begin{Answer}
    From $\min_{a \in A} f(a) = \max_{b \in B} g(b)$ we have $f(a) \geq
    g(b)$ for all $(a, b) \in A \times B$.  Taking $(\bar a, \bar b)$ 
    with $f(\bar a) = g(\bar b)$, we have $f(\bar a) = g(\bar b) \leq f(a)$
    for any given $a \in A$.  In particular, 
    $\bar a$ minimizes $f$ on $A$.  The argument
    for $\bar b$ is similar.
\end{Answer}

We will use Exercise~\ref{ex:mmm} in the following way.  Let $A = \Pi(\phi,
\psi)$ and $B = \fF_c$.  Let $f$ be the value of the primal and $g$ be the
value of the dual.  By \eqref{eq:wdotp}, the ordering $f(\pi) \geq g(w, p)$
holds over all feasible $\pi \in A$ and $(w,p) \in B$ pairs.    By strong
duality, we also have $\min_{\pi \in A} f(\pi) = \max_{(w,p) \in B} g(w, p)$.
Hence we need only show that, when $(w, p)$ and $\pi$ form a competitive
equilibrium, we have $f(\pi) = g(w,p)$.

\begin{proof}[Proof of Theorem~\ref{t:optiffce}]
    Suppose that $(w, p)$ and $\pi$ form a competitive equilibrium.
    From (RE) we know that $\pi$ is feasible for the primal problem.
    From (NA) we know that $(w, p)$ is feasible for the dual.  Since the
    equality in the (IC) condition holds when $\pi(x, y) > 0$, we can multiply
    both sides of this equality by $\pi(x, y)$ and sum over all $x,y$ to
    obtain
    \begin{equation}\label{eq:cgf}
        \sum_{x,y} c(x, y) \pi(x, y) 
        = \sum_y p(y) \psi(y) - \sum_x w(x) \psi(x).
    \end{equation}
    The result of Exercise~\ref{ex:mmm} now applies, so $\pi$ attains the minimum
    in the primal problem and $(w, p)$ attains the maximum in the dual.
\end{proof}

We also have the following converse:

\begin{theorem}
    If $\pi$ is an optimal transport plan, then there exists a pair $(w, p)
    \in \RR^\Xsf \times \RR^\Ysf$ such that the quantities determined by $\pi$
    and the prices in $(w, p)$ form a competitive equilibrium.
\end{theorem}

\begin{proof}
    Let $\pi$ be an optimal plan.  To be optimal, $\pi$ must be feasible, so
    $\pi \in \Pi(\phi, \psi)$, which implies that (RE) holds.

    By Kantorovich's duality theorem (Theorem~\ref{t:kdual0}, we can obtain $(w, p) \in \fF_c$ such
    that~\eqref{eq:cgf} holds.    Since $(w, p) \in \fF_c$ (NA) holds.  From
    (NA) we have $c(x, y) + w(x) - p(y) \geq 0$ for all $x,y$.  From this
    and~\eqref{eq:cgf} we see that (IC) must be valid.  We conclude that $\pi$ and $(w, p)$ form
    a competitive equilibrium.
\end{proof}

\subsection{The General Flow Problem}

We now describe a general network flow problem that can be used to analyze a
large range of applications, from international trade to communication
and assignment.  This general problem includes optimal transport as a special
case.

Once we have introduced the problem, we show two results.  First, the problem
can easily be formulated as a linear program and solved using standard linear
programming methods.  Second, even though optimal transport is a strict subset
of the general flow problem, every general flow problem can be solved using a
combination of optimal transport and shortest path methods.

\subsubsection{Problem Statement}\label{sss:probnf}

We are interested in flow of a good or service across a network with $n$
vertices.  This network can be understood as a weighted directed graph $(V, E,
c)$. To simplify notation, we label the nodes from $1$ to $n$ and let $V =
\natset{n}$.    Existence of an edge $e = (i, j) \in E$ with weight $c(i, j)$
indicates that the good can be shipped from $i$ to $j$ at cost $c(i, j)$. We
recall from \S\ref{ss:uwdg} that $\iI(i)$ is the set of direct predecessors of
vertex $i$ (all $u \in V$ such that $(u, i) \in E$) and $\oO(i)$ is the set of
direct successors (all $j \in V$ such that $(i, j) \in E$).

A classic example is the famous Silk Road of antiquity, part of which is
illustrated in Figure~\ref{f:sr}.   Silk was produced in eastern cities such
as Loyang and Changan, and then transported westward to satisfy final demand
in Rome, Constantinople and Alexandria.  Towns such as Yarkand acted as trade
hubs.  Rather than covering the whole route, traders typically traveled
backward and forward between one pair of hubs, where they knew the language
and customs.\footnote{Our use of the Silk Road as an example of a network flow
    problem is not original. \cite{galichon2018optimal} provides a highly
readable treatment in the context of optimal transport.}

\begin{figure}
   \centering
   \scalebox{0.5}{\includegraphics{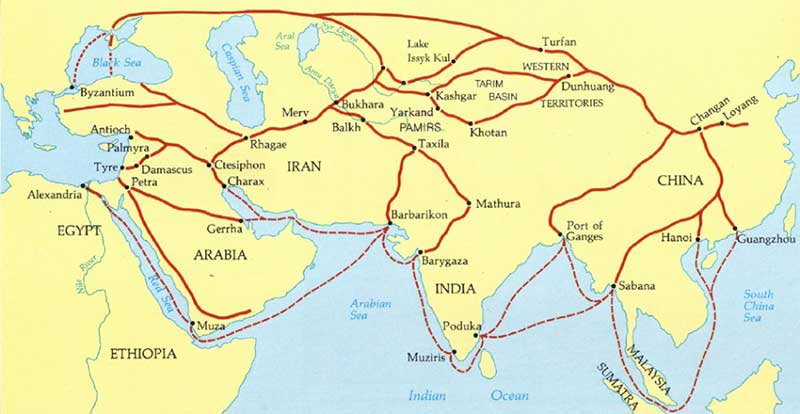}}
   \caption{\label{f:sr} The Silk Road}
\end{figure}

Returning to the model, we allow for both \navy{initial supply} of and
\navy{final demand} for the good at every node (although one or both could be
zero).  Let $s(i)$ and $d(i)$ be supply and demand at node $i$ respectively.
Aggregate supply and demand over the network are assumed to be equal, so that
\begin{equation}\label{eq:asead}
    \sum_{i \in V} s(i) = \sum_{i \in V} d(i) .
\end{equation}
This can be understood as an equilibrium condition: prices have adjusted to
equalize initial supply and final demand in aggregate.  We assume throughout
that the vectors $s$ and $d$ are nonnegative with at least one positive
element.

Let $q(i, j)$ be the amount of the good shipped from node $i$ to node $j$ for
all $i, j \in V$.  
The minimum cost network flow problem is to minimize \navy{total shipping
    cost}
\begin{equation}\label{eq:flot}
    \sum_{i \in V} \sum_{j \in V} c(i,j) q(i, j),
\end{equation}
subject to the restriction that $q \geq 0$ and  
\begin{equation}\label{eq:nflowsd}
    s(i) + \sum_{v \in \iI(i)} q(v, i)
    = d(i) + \sum_{j \in \oO(i)} q(i, j)
    \quad \text{for all } i \in V.
\end{equation}
The left hand side of \eqref{eq:nflowsd} is total supply to node $i$ (initial
supply plus inflow from other nodes), while the right hand side is 
total demand (final demand plus outflow to
other nodes).

\begin{Exercise}
    Although we presented them separately, the node-by-node restriction
    \eqref{eq:nflowsd} implies the aggregate restriction \eqref{eq:asead}.
    Explain why this is the case.
\end{Exercise}

\begin{Answer}
   Since any outflow from some node $i$ is matched by equal inflow into some
   node $j$, summing both sides of \eqref{eq:nflowsd} across all $i \in V$
   yields \eqref{eq:asead}.
\end{Answer}

\subsubsection{Optimality}

There are several ways to transform the network flow problem into a linear
program.  We follow the presentation in \cite{bertsimas1997introduction}.
We take $m = |E|$ to be the total number of edges and enumerate them (in any
convenient way) as $e_1, \ldots, e_m$.  Let's say that $e_k$ \navy{leaves}
node $i$ if $e_k = (i, j)$ for some $j \in \natset{n}$, and that $e_k$
\navy{enters} node $i$ if $e_k = (\ell, i)$ for some $\ell \in \natset{n}$.
Then we define the $n \times m$ \navy{node-edge incidence
matrix}\index{Incidence matrix} $A$ by 
\begin{equation*}
   A = (a_{i k})
   \quad \text{with} \quad
   a_{i k}
   :=
   \begin{cases}
       \; 1 & \text{if $e_k$ leaves $i$} \\
       -1   & \text{if $e_k$ enters $i$} \\
       \; 0 & \text{otherwise}.
   \end{cases}
\end{equation*}

\begin{example}\label{eg:vsmcf}
    Consider the very simple minimum cost flow problem in
    Figure~\ref{f:optimal_flow_1}.
    The shipment costs $c(i,j)$ are listed next to each existing edge.  Initial
    supply is 10 at node 1 and zero elsewhere.  Final demand is 10 at node 4 and
    zero elsewhere.  We enumerate the edges as
    \begin{equation}\label{eq:vsel}
        E 
        = \{e_1, \ldots, e_4\}
        = \{(1, 2), (1, 4), (2, 3), (3, 4)\}.
    \end{equation}
    The node-edge incidence matrix is
    \begin{equation*}
        A =
        \begin{pmatrix}
            1 & 1 & 0 & 0    \\
            -1 & 0 & 1 & 0    \\
            0 & 0 & -1 & 1    \\
            0 & -1 & 0 & -1
        \end{pmatrix}.
    \end{equation*}
\end{example}

\begin{figure}
    \centering
    \input{tikz/optimal_flow_1.tex}
    \caption{\label{f:optimal_flow_1} A simple network flow problem}
\end{figure}
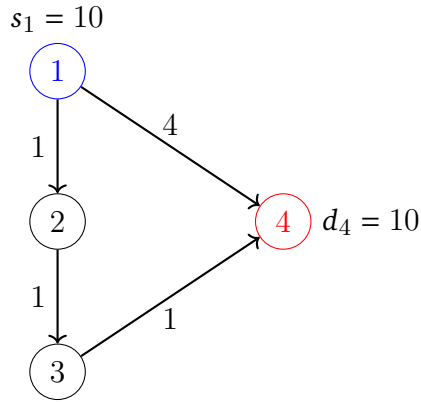

Now, returning to the general case, we rearrange $q$ and $c$ into $m \times 1$
vectors $(q_k)$ and $(c_k)$, where $q_k$ is the amount shipped along edge $k$
and $c_k$ is the cost.  For example, if $e_k = (i,j)$, then $q_k = q(i, j)$
and $c_k = c(i,j)$. In addition, we set $b$ to be the vector in $\RR^n$ with
$i$-th element $s(i) - d(i)$, which is net exogenous supply at node $i$.

\begin{Exercise}
    In this set up, show that \eqref{eq:nflowsd} is equivalent to $A q = b$ in
    the special case of Example~\ref{eg:vsmcf}.
\end{Exercise}

\begin{Exercise}\label{ex:immq}
    Let $(Aq)(i)$ be the $i$-th row of $Aq$.  Show that 
    \begin{equation}\label{eq:immq}
        (Aq)(i) = \sum_{j \in \oO(i)} q(i, j) - \sum_{v \in \iI(i)} q(v, i).
    \end{equation}
\end{Exercise}

\begin{Answer}
    Verifying the claim is just a matter of working with the definition of
    $A$.  Fixing $i \in \natset{n}$, we have
    \begin{equation*}
        (Aq)(i) 
         = \sum_{k=1}^m a_{ik} q_k 
         = \sum_{k=1}^m \1\{e_k \text{ leaves } i \} q_k 
                - \sum_{k=1}^m \1\{e_k \text{ points to } i \} q_k,
    \end{equation*}
    This is equal to $\sum_{j \in \oO(i)} q(i, j) - \sum_{v \in \iI(i)} q(v, i)$, as
    was to be shown.
\end{Answer}

Equation~\eqref{eq:immq} tells us that the $i$-th row of $Aq$ give us the net
outflow from node $i$ under the transport plan $q$. Now, with $\inner{c, q} :=
\sum_{k=1}^m c_k q_k$, the minimum cost network flow problem can now be
expressed as
\begin{equation}\label{eq:nflowlp}
    \min \inner{c, q} 
    \quad \st \quad
    q \geq 0
    \; \text{ and } \;
    A q = b.
\end{equation}
This is a linear program in standard equality form, to which we can apply any
linear programming solver.  For Example~\ref{eg:vsmcf}, we run the following:

\begin{minted}{python}
A = (( 1,  1,  0,  0),
     (-1,  0,  1,  0),
     ( 0,  0, -1,  1),
     ( 0, -1,  0, -1))

b = (10, 0, 0, -10)
c = (1, 4, 1, 1)

result = linprog(c, A_eq=A, b_eq=b, method='highs-ipm')
print(result.x)
\end{minted}

The output is \texttt{[10.  0. 10. 10.]}.  
 Recalling the order of the
paths in \eqref{eq:vsel}, this means that the optimal transport plan is
$q(1,4)=0$ and $q(1,2)=q(2,3)=q(3,4)=10$, as our intuition suggests.

\begin{Exercise}
    Some network flow problems have \navy{capacity constraints}, which can be
    modeled as a map $g \colon E \to [0, \infty]$, along with the restriction
    $q(e) \leq g(e)$ for all $e \in E$.  (If $g(e) =
    +\infty$, there is no capacity constraint over shipping on edge $e$.)
    Formulate this as a linear program and modify the code above, which solves
    Example~\ref{eg:vsmcf}, to include
    the capacity constraint $g(1,2)=5$.  Solve for the optimal plan.
\end{Exercise}

\begin{Answer}
    To the code that solves the original version of Example~\ref{eg:vsmcf}, we
    need to add
    \begin{minted}{python}
    bounds = ((0, 5),
              (0, None),
              (0, None),
              (0, None))    
    \end{minted}
    and then change the function call to 
    \begin{minted}{python}
    result = linprog(c, A_eq=A, b_eq=b, method='highs-ipm', bounds=bounds)
    print(result.x)    
    \end{minted}

    The output is \texttt{[5. 5. 5. 5.]}, which also agrees with our intuition.
\end{Answer}

\begin{Exercise}\label{ex:otsc}
    Explain how the generic optimal transport problem treated in \S\ref{ss:ot}
    is a special case of the minimum cost network flow problem.
\end{Exercise}

\subsubsection{Reduction to Optimal Transport}\label{sss:redmfot}

In Exercise~\ref{ex:otsc}, we saw how every optimal transport problem is a
special kind of minimum cost network flow problem.  There is a sense in which
the converse is also true.  In particular, we can use optimal transport
methods to solve any network flow problem, provided that we first modify the
network flow problem via an application of shortest paths.

To explain how this works, we take the abstract network flow problem described
in \S\ref{sss:probnf}, on the weighted digraph $(V, E, c)$, with $V =
\natset{n}$, initial supply vector $s \in \RR^n_+$ and final demand vector $d
\in \RR^n_+$.  For the purposes of this section, we agree to call a node $i$
with $s(i) - d(i) > 0$ a \navy{net supplier}. A node $i$ with $d(i) - s(i) >
0$ will be called a \navy{net consumer}.  Nodes with $s(i) = d(i)$ will be
called \navy{trading stations}.  

\begin{example}
    In the left hand side of Figure~\ref{f:optimal_flow_vs_transportation_1},
    nodes 1 and 2 are net suppliers, 3 is a trading station and 4 and 5 are
    net consumers.
\end{example}

\begin{example}
    In the Silk Road application, Rome would be a net consumer, where final
    demand is large and positive, while initial supply is zero.  A city such
    as Yarkand should probably be modeled as a trading station, with $s(i)=d(i)
    = 0$.
\end{example}

The idea behind the reduction is to treat the net supplier nodes as source
locations and the net consumer nodes as target locations in an optimal
transport problem.  The next step is to compute the shortest path (if there
are multiple, pick any one) from each net supplier $i$ to each net consumer
$j$.  Let $\rho(i, j)$ denote this path, represented as a sequence of edges in
$E$.  The cost of traversing $\rho(i, j)$ is 
\begin{equation*}
    \hat c(i, j)  
    := \sum_{k=1}^m c_k \1\{e_k \in \rho(i, j)\}.
\end{equation*}
Now the trading
stations are eliminated and we solve the optimal transport problem with
\begin{itemize}
    \item $\Xsf = $ the set of net suppliers,
    \item $\Ysf = $ the set of net consumers,
    \item $\phi(i) = s(i) - d(i)$ on $\Xsf$,
    \item $\psi(j) = d(j) - s(j)$ on $\Ysf$, and
    \item cost function $\hat c(i, j)$ as defined above.\footnote{If no path exists from $i$ to
            $j$ then we set $\hat c(i, j) = \infty$.  Such settings
            can be handled in linear programming solvers by adding capacity
            constraints.  See, for example, \cite{peyre2019computational},
            Section~10.3.}
\end{itemize}
After we find the optimal transport plan $\pi$, the network minimum cost flow
$q_k$ along arbitrary edge $e_k \in E$ is recovered by setting
\begin{equation*}
    q_k = \sum_{i \in \Xsf} \sum_{j \in \Ysf} \pi(i, j)
        \1\{e_k \in \text{ the shortest path from $i$ to $j$}\}.
\end{equation*}

\begin{Exercise}
    Recalling our assumptions on $s$ and $d$, prove that $\sum_{i \in \Xsf}
    \phi(i) = \sum_{j \in \Ysf} \psi(j)$ and that this sum is nonzero.
\end{Exercise}

\begin{remark}
    We have not imposed $\sum_i \phi(i) = \sum_j \psi(j) = 1$, as required for
    the standard formulation of the optimal transport problem.  But this
    normalization is only used for convenience in exposition and most solvers
    do not require it.
\end{remark}

Figure~\ref{f:optimal_flow_vs_transportation_1} illustrates the method.
Trading station 3 is eliminated after the shortest paths are computed.

\begin{figure}
   \centering
   \input{tikz/optimal_flow_vs_transportation_1.tex}
   \caption{\label{f:optimal_flow_vs_transportation_1} 
       Reducing minimum cost optimal flow to optimal transport}
\end{figure}
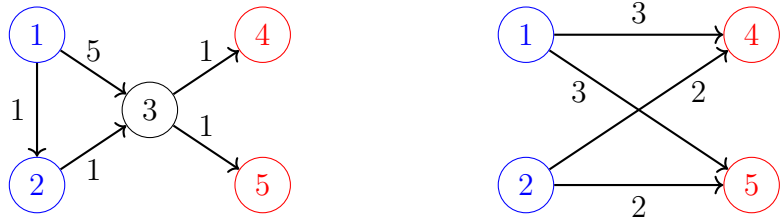

\section{Chapter Notes}

Our treatment of shortest paths can be understood as a simplified version of
both the Bellman--Ford algorithm and Dijkstra's algorithm, which are routinely
used to solve large shortest path planning problems. Our approach is intended
to emphasize recursive solution methods, which are valuable for analyzing a
vast range of economic problems, from intertemporal modeling (see, e.g.,
\cite{stokey1989recursive} or \cite{ljungqvist8012recursive}) to production
chains \citep{kikuchi2021coase}.

For a more in-depth treatment of linear programming, we recommend the
excellent textbooks by \cite{bertsimas1997introduction} and
\cite{matousek2007understanding}. For bedtime reading, \cite{cook2011pursuit}
provides an entertaining introduction to some of the main ideas and
applications to network problems, including a review of computation.

While \S\ref{ss:wahdp} provided a rather whimsical introduction to matching
and assignment problems, such problems have great real world importance.
Examples include assigning kidney donors to recipients, mothers to maternity
wards, doctors to hospitals, students to schools, delivery drivers to orders
and autonomous vehicles to riders.\footnote{The latter problem occurs in the
field of Autonomous Mobility on Demand (AMoD). See, for example,
\cite{ascenzi2021control} or \cite{simonetto2019real}.} A brief history of
assignment, matching problems and combinatorial optimization can be found in
\cite{schrijver2005history}. \cite{greinecker2021pairwise} study existence of
stable matchings in a setting with many agents.

\cite{villani2008optimal} and \cite{vershik2013long} provide extensive
historical background on the optimal transport problem. \cite{vershik2013long}
mentions some of the problems that Kantorovich faced, as a Soviet
mathematician working in the time of Stalin and Khrushchev, given that his
main duality theorem for optimal transport can be seen as a proof that
competitive market equilibria attain the maximal transport plan.
  
In \S\ref{ss:kd}, we mentioned that Kantorovich's work anticipated much of the
later development of duality theory for arbitrary linear programs.  In fact,
according to \cite{vershik2013long}, Kantorovich anticipated much of the general
theory of linear programming itself, including providing a version of the
simplex algorithm later rediscovered and extended by Dantzig.  

Optimal transport has a remarkably wide variety of applications, spread across
economics, econometrics, finance, statistics, artificial intelligence, machine
learning and other fields.  Within economics, \cite{galichon2018optimal}
provides an excellent overview. \cite{fajgelbaum2020optimal} consider optimal
transport in spacial equilibria. \cite{beiglbock2022bachelier} review some of
the major milestones of modern finance theory and show their connections via
optimal transport.  Connections to machine learning are surveyed in
\cite{kolouri2017optimal}.

The computational theory of optimal transport is now a major field.  A high
quality exposition can be found in \cite{peyre2019computational}.
\cite{blanchet2018computation} use computational optimal transport to solve
for Cournot--Nash equilibria in mean-field type games.

%% file: tikz/shortest_paths_1.tex
\begin{tikzpicture}
  \node[circle, draw] (A) at (0, 2) {A};
  \node[circle, draw] (B) at (3, 0.2) {B};
  \node[circle, draw] (C) at (-3, 0.2) {C};
  \node[circle, draw] (D) at (0, 0) {D};
  \node[circle, draw] (E) at (1.5, -0.8) {E};  
  \node[circle, draw] (F) at (-1.8, -1.3) {F};   
  \node[circle, draw] (G) at (0, -2.5) {G};
  \draw[->, thick, black]
  (A) edge [bend right=0] node [midway, fill=white] {$1$} (B)
  (A) edge [bend left=0] node [midway, fill=white] {$5$} (C)
  (A) edge [bend left=0] node [midway, fill=white] {$3$} (D)
  (B) edge [bend left=0] node [midway, fill=white] {$9$} (D)
  (B) edge [bend left=0] node [midway, fill=white] {$6$} (E)
  (C) edge [bend left=0] node [midway, fill=white] {$2$} (F)
  (D) edge [bend left=0] node [midway, fill=white] {$4$} (F)
  (D) edge [bend left=0] node [midway, fill=white] {$8$} (G)
  (E) edge [bend left=0] node [midway, fill=white] {$4$} (G)
  (F) edge [bend left=0] node [midway, fill=white] {$1$} (G);
\end{tikzpicture}

%% file: tikz/shortest_paths_2.tex
\begin{tikzpicture}
  \node[circle, draw] (A) at (0, 2) {A};
  \node[circle, draw] (B) at (3, 0.2) {B};
  \node[circle, draw] (C) at (-3, 0.2) {C};
  \node[circle, draw] (D) at (0, 0) {D};
  \node[circle, draw] (E) at (1.5, -0.8) {E};  
  \node[circle, draw] (F) at (-1.8, -1.3) {F};   
  \node[circle, draw] (G) at (0, -2.5) {G};
  \draw[-{>[scale=0.5]}, thick, red, line width=3pt, opacity=0.5]
  (A) edge [bend left=0] node [midway, fill=white] {} (C)
  (C) edge [bend left=0] node [midway, fill=white] {} (F)
  (F) edge [bend left=0] node [midway, fill=white] {} (G);
  \draw[->, thick, black]
  (A) edge [bend right=0] node [midway, fill=white] {$1$} (B)
  (A) edge [bend left=0] node [midway, fill=white] {$5$} (C)
  (A) edge [bend left=0] node [midway, fill=white] {$3$} (D)
  (B) edge [bend left=0] node [midway, fill=white] {$9$} (D)
  (B) edge [bend left=0] node [midway, fill=white] {$6$} (E)
  (C) edge [bend left=0] node [midway, fill=white] {$2$} (F)
  (D) edge [bend left=0] node [midway, fill=white] {$4$} (F)
  (D) edge [bend left=0] node [midway, fill=white] {$8$} (G)
  (E) edge [bend left=0] node [midway, fill=white] {$4$} (G)
  (F) edge [bend left=0] node [midway, fill=white] {$1$} (G);
\end{tikzpicture}

%% file: tikz/shortest_paths_3.tex
\begin{tikzpicture}
  \node[circle, draw] (A) at (0, 2) {A};
  \node[circle, draw] (B) at (3, 0.2) {B};
  \node[circle, draw] (C) at (-3, 0.2) {C};
  \node[circle, draw] (D) at (0, 0) {D};
  \node[circle, draw] (E) at (1.5, -0.8) {E};  
  \node[circle, draw] (F) at (-1.8, -1.3) {F};   
  \node[circle, draw] (G) at (0, -2.5) {G};
  \draw[-{>[scale=0.5]}, thick, blue, line width=3pt, opacity=0.5]
  (A) edge [bend left=0] node [midway, fill=white] {$3$} (D)
  (D) edge [bend left=0] node [midway, fill=white] {$4$} (F)
  (F) edge [bend left=0] node [midway, fill=white] {$1$} (G);
  \draw[->, thick, black]
  (A) edge [bend right=0] node [midway, fill=white] {$1$} (B)
  (A) edge [bend left=0] node [midway, fill=white] {$5$} (C)
  (A) edge [bend left=0] node [midway, fill=white] {$3$} (D)
  (B) edge [bend left=0] node [midway, fill=white] {$9$} (D)
  (B) edge [bend left=0] node [midway, fill=white] {$6$} (E)
  (C) edge [bend left=0] node [midway, fill=white] {$2$} (F)
  (D) edge [bend left=0] node [midway, fill=white] {$4$} (F)
  (D) edge [bend left=0] node [midway, fill=white] {$8$} (G)
  (E) edge [bend left=0] node [midway, fill=white] {$4$} (G)
  (F) edge [bend left=0] node [midway, fill=white] {$1$} (G);
\end{tikzpicture}

%% file: tikz/shortest_paths_4.tex
\begin{tikzpicture}
  \node[circle, draw, label={[red]above:$8$}] (A) at (0, 2) {A};
  \node[circle, draw, label={[red]right:$10$}] (B) at (3, 0.2) {B};
  \node[circle, draw, label={[red]left:$3$}] (C) at (-3, 0.2) {C};
  \node[circle, draw, label={[red]left:$5$}] (D) at (0, 0) {D};
  \node[circle, draw, label={[red]right:$4$}] (E) at (1.5, -0.8) {E};  
  \node[circle, draw, label={[red]left:$1$}] (F) at (-1.8, -1.3) {F};   
  \node[circle, draw, label={[red]right:$0$}] (G) at (0, -2.5) {G};
  \draw[->, thick, black]
  (A) edge [bend right=0] node [midway, fill=white] {$1$} (B)
  (A) edge [bend left=0] node [midway, fill=white] {$5$} (C)
  (A) edge [bend left=0] node [midway, fill=white] {$3$} (D)
  (B) edge [bend left=0] node [midway, fill=white] {$9$} (D)
  (B) edge [bend left=0] node [midway, fill=white] {$6$} (E)
  (C) edge [bend left=0] node [midway, fill=white] {$2$} (F)
  (D) edge [bend left=0] node [midway, fill=white] {$4$} (F)
  (D) edge [bend left=0] node [midway, fill=white] {$8$} (G)
  (E) edge [bend left=0] node [midway, fill=white] {$4$} (G)
  (F) edge [bend left=0] node [midway, fill=white] {$1$} (G);
\end{tikzpicture}

%% file: tikz/shortest_paths_5.tex
\begin{tikzpicture}
  \node[circle, draw] (A) at (0, 2) {A};
  \node[circle, draw] (B) at (3, 0.2) {B};
  \node[circle, draw] (C) at (-3, 0.2) {C};
  \node[circle, draw] (D) at (0, 0) {D};
  \node[circle, draw] (E) at (1.5, -0.8) {E};  
  \node[circle, draw] (F) at (-1.8, -1.3) {F};   
  \node[circle, draw] (G) at (0, -2.5) {G};
  \draw[->, thick, black]
  (A) edge [bend right=0] node [midway, fill=white] {$1$} (B)
  (A) edge [bend left=0] node [midway, fill=white] {$5$} (C)
  (A) edge [bend left=0] node [midway, fill=white] {$3$} (D)
  (B) edge [bend left=0] node [midway, fill=white] {$9$} (D)
  (B) edge [bend left=0] node [midway, fill=white] {$6$} (E)
  (C) edge [bend left=0] node [midway, fill=white] {$2$} (F)
  (D) edge [bend left=0] node [midway, fill=white] {$4$} (F)
  (D) edge [bend left=0] node [midway, fill=white] {$8$} (G)
  (E) edge [bend left=0] node [midway, fill=white] {$4$} (G)
  (F) edge [bend left=0] node [midway, fill=white] {$1$} (G)
  (G) edge [loop below] node {$0$} (G);
\end{tikzpicture}

%% file: tikz/worker_job_matching.tex
\begin{tikzpicture}
  \node[] (w0) at (-1, 1) {workers};
  \node[circle, draw] (w1) at (0.5, 1) {$1$};
  \node[circle, draw] (w2) at (1.5, 1) {$2$};
  \node[circle, draw] (w3) at (2.5, 1) {$3$};
  \node[] (w4) at (3.5, 1) {$\cdots$};
  \node[circle, draw] (w5) at (4.5, 1) {$40$};
  \node[] (j0) at (-1, -1) {jobs};
  \node[circle, draw] (j1) at (0.5, -1) {$1$};
  \node[circle, draw] (j2) at (1.5, -1) {$2$};
  \node[circle, draw] (j3) at (2.5, -1) {$3$};
  \node[] (j4) at (3.5, -1) {$\cdots$};
  \node[circle, draw] (j5) at (4.5, -1) {$40$};
  \draw[->, thick, black]
  (w1) edge [bend left=0, above, -{Stealth[scale=1]}, line width=1pt] node {} (j2)
  (w2) edge [bend left=0, above, -{Stealth[scale=1]}, line width=1pt] node {} (j3)
  (w3) edge [bend left=0, above, -{Stealth[scale=1]}, line width=1pt] node {} (j1)
  (w4) edge [bend left=0, above, -{Stealth[scale=1]}, line width=1pt] node {} (j5)
  (w5) edge [bend left=0, above, -{Stealth[scale=1]}, line width=1pt] node {} (j4);
\end{tikzpicture}

%% file: tikz/optimal_transport_unsplitting.tex
\begin{tikzpicture}
  \node[circle, draw, scale=1.5588, red] (1) at (1, 1) {$\psi_1$};
  \node[circle, draw, scale=1.2726, red] (2) at (3, 3) {$\psi_2$};
  \node[circle, draw, scale=1.5588, red] (3) at (5, 0) {$\psi_3$};
  \node[circle, draw, scale=1.2726, red] (4) at (7, 2) {$\psi_4$};
  \node[circle, draw, scale=1.2726, blue] (01) at (1.5, 2.5) {${\varphi}_1$};
  \node[circle, draw, scale=0.9, blue] (02) at (0, 0) {${\varphi}_2$};
  \node[circle, draw, scale=1.2726, blue] (03) at (3.2, 0) {${\varphi}_3$};
  \node[circle, draw, scale=0.9, blue] (04) at (6.5, 0.5) {${\varphi}_4$};
  \node[circle, draw, scale=0.9, blue] (05) at (8, 3) {${\varphi}_5$};
  \node[circle, draw, scale=0.9, blue] (06) at (6, 3) {${\varphi}_6$};
  \node[circle, draw, scale=1.2726, blue] (07) at (4.2, 1.8) {${\varphi}_7$};

  \draw[->, thick, blue]
  (01) edge [bend left=0, left, -{Stealth[scale=1]}, line width=2pt] node {}(1)
  (02) edge [bend left=0, below, -{Stealth[scale=1]}, line width=1pt] node {} (1)
  (03) edge [bend left=0, below, -{Stealth[scale=1]}, line width=2pt] node {} (3)
  (04) edge [bend left=0, below, -{Stealth[scale=1]}, line width=1pt] node {} (3)
  (05) edge [bend left=0, below, -{Stealth[scale=1]}, line width=1pt] node {} (4)
  (06) edge [bend left=0, below, -{Stealth[scale=1]}, line width=1pt] node {} (4)
  (07) edge [bend left=0, below, -{Stealth[scale=1]}, line width=2pt] node {} (2);
\end{tikzpicture}

%% file: tikz/optimal_transport_splitting_experiment.tex
\begin{tikzpicture}
  \node[circle, draw, scale=1.3856, red] (1) at (1, 1) {$\psi_1$};
  \node[circle, draw, scale=1.6, red] (2) at (3, 3) {$\psi_2$};
  \node[circle, draw, scale=1.3856, red] (3) at (5, 0) {$\psi_3$};
  \node[circle, draw, scale=1.6, red] (4) at (7, 2) {$\psi_4$};
  \node[circle, draw, scale=1.1312, blue] (01) at (1.5, 2.5) {${\varphi}_1$};
  \node[circle, draw, scale=0.8, blue] (02) at (0, 0) {${\varphi}_2$};
  \node[circle, draw, scale=1.3856, blue] (03) at (3, 1) {${\varphi}_3$};
  \node[circle, draw, scale=1.1312, blue] (04) at (6.5, 0.5) {${\varphi}_4$};
  \node[circle, draw, scale=0.8, blue] (05) at (8, 3) {${\varphi}_5$};
  \node[circle, draw, scale=1.1312, blue] (06) at (5, 3) {${\varphi}_6$};
  \node[circle, draw, scale=1.3856, blue] (07) at (4.5, 1.5) {${\varphi}_7$};

  \draw[->, thick, blue]
  (01) edge [bend left=0, left, -{Stealth[scale=1]}, line width=1pt] node {}(1)
  (01) edge [bend left=0, below, -{Stealth[scale=1]}, line width=1pt] node {} (2)
  (02) edge [bend left=0, below, -{Stealth[scale=1]}, line width=1pt] node {} (1)
  (03) edge [bend left=0, below, -{Stealth[scale=1]}, line width=1pt] node {} (1)
  (03) edge [bend left=0, below, -{Stealth[scale=1]}, line width=1pt] node {} (2)
  (03) edge [bend left=0, below, -{Stealth[scale=1]}, line width=1pt] node {} (3)
  (04) edge [bend left=0, below, -{Stealth[scale=1]}, line width=1pt] node {} (3)
  (04) edge [bend left=0, below, -{Stealth[scale=1]}, line width=1pt] node {} (4)
  (05) edge [bend left=0, below, -{Stealth[scale=1]}, line width=1pt] node {} (4)
  (06) edge [bend left=0, below, -{Stealth[scale=1]}, line width=1pt] node {} (2)
  (06) edge [bend left=0, below, -{Stealth[scale=1]}, line width=1pt] node {} (4)
  (07) edge [bend left=0, below, -{Stealth[scale=1]}, line width=1pt] node {} (2)
  (07) edge [bend left=0, below, -{Stealth[scale=1]}, line width=1pt] node {} (3)
  (07) edge [bend left=0, below, -{Stealth[scale=1]}, line width=1pt] node {} (4);
\end{tikzpicture}

%% file: tikz/optimal_flow_1.tex
\begin{tikzpicture}
  \node[circle, draw, label={[black]above:$s_1=10$}, blue] (1) at (0, 2) {$1$};
  \node[circle, draw] (2) at (0, 0) {$2$};
  \node[circle, draw] (3) at (0, -2) {$3$};
  \node[circle, draw, label={[black]right:$d_4=10$}, red] (4) at (3, 0) {$4$};
  \draw[->, thick, black]
  (1) edge [bend right=0, above] node  {$4$} (4)
  (1) edge [bend left=0, left] node {$1$} (2)
  (3) edge [bend left=0, below] node {$1$} (4)
  (2) edge [bend left=0, left] node {$1$} (3);
\end{tikzpicture}

%% file: tikz/optimal_flow_vs_transportation_1.tex
\begin{tikzpicture}
  \node[circle, draw, blue] (1) at (-1.5, 1) {$1$};
  \node[circle, draw, blue] (2) at (-1.5, -1) {$2$};
  \node[circle, draw] (3) at (0, 0) {$3$};
  \node[circle, draw, red] (4) at (1.5, 1) {$4$};
  \node[circle, draw, red] (5) at (1.5, -1) {$5$};
  \node[circle, draw, blue] (6) at (5, 1) {$1$};
  \node[circle, draw, blue] (7) at (5, -1) {$2$};
  \node[circle, draw, red] (8) at (8, 1) {$4$};
  \node[circle, draw, red] (9) at (8, -1) {$5$};
  \draw[->, thick, black]
  (1) edge [bend right=0, left] node  {$1$} (2)
  (1) edge [bend left=0, above] node {$5$} (3)
  (2) edge [bend left=0, below] node {$1$} (3)
  (3) edge [bend left=0, above] node {$1$} (4)
  (3) edge [bend left=0, above] node {$1$} (5)
  (6) edge [bend left=0, below] node [pos=0.15] {$3$} (9)
  (6) edge [bend left=0, above] node {$3$} (8)
  (7) edge [bend left=0, below] node [pos=0.85] {$2$} (8)
  (7) edge [bend left=0, below] node {$2$} (9);
\end{tikzpicture}

%% file: ch_mcs.tex
\chapter{Markov Chains and Networks}\label{c:mcs}

Markov chains evolving on finite sets are a foundational class of stochastic
processes, routinely employed in quantitative modeling within economics,
finance, operations research and social science.  A Markov chain is most
easily understood as a weighted digraph, with graph-theoretic properties such
as connectedness and periodicity being key determinants of dynamics.

\section{Markov Chains as Digraphs}\label{s:amcs}

We begin with fundamental definitions and then investigate dynamics.

\subsection{Markov Models}\label{sss:reps}

A \navy{finite Markov model}\index{Markov model} is a weighted directed graph
$\mM = (S, E, p)$, where $S$ is the (finite) set of vertices, $E$ is the set
of edges and $p$ is the weight function, with the additional restriction that 
\begin{equation}\label{eq:sto}
    \sum_{y \, \in \, \oO(x)} p(x, y) = 1
    \quad \text{for all } x \in S.
\end{equation}
Figure~\ref{f:rich_poor} on page~\pageref{f:rich_poor}  presented an example
of such a digraph.  

The set of vertices $S$ of a finite Markov model $\mM = (S, E, p)$ is also called
the \navy{state space}\index{State space} of the model, and vertices are
called \navy{states}.  The two standard interpretations are
\begin{enumerate}
    \item $S$ is a set of possible states
        for some random element (the state) and the weight $p(x, y)$ represents the
        probability that the state moves from $x$ to $y$ in one step.
    \item $S$ is a set of possible values for some measurement over a large
        population (e.g., hours worked per week measured across a large
        cross-section of households) and $p(x, y)$ is the fraction of agents
        that transition from state $x$ to state $y$ in one unit of time.
\end{enumerate}
These two perspectives are related in ways that we explore below.

\subsubsection{Transition Matrices}

If $\mM$ is a finite Markov model, then the restriction~\eqref{eq:sto} is
equivalent to the statement that the adjacency matrix associated with $\mM$ is
stochastic (see \S\ref{sss:stochmat} for the definition). Identifying $S =
\{\text{poor}, \text{middle}, \text{rich}\}$ with $\{1, 2, 3\}$, the adjacency
matrix for this weighted digraph is
\begin{equation}\label{eq:pia}
    P_a = 
    \begin{pmatrix}
        0.9 & 0.1 & 0.0 \\
        0.4 & 0.4 & 0.2 \\
        0.1 & 0.1 & 0.8 
    \end{pmatrix}
\end{equation}
Since $P_a \geq 0$ and rows sum to unity, $P_a$ is stochastic as required.
In the context of finite Markov models, the adjacency matrix of $\mM$ is also called the
\navy{transition matrix}\index{Transition matrix}.

Regarding notation, when $S$ has typical elements $x, y$, it turns out to be
convenient to write elements of the transition matrix $P$ as $P(x, y)$ rather
than $P_{ij}$ or similar.  We can think of $P$ as extending the weight
function $p$ from $E$ to the set of all $(x,y)$ pairs in $S \times S$,
assigning zero whenever $(x,y) \notin E$. As such, 
 for every possible choice of $(x,y)$, the value
$P(x,y)$ represents the
probability of transitioning from $x$ to $y$ in one step.  

The requirement that $P$ is stochastic can now be written as $P \geq 0$ and
\begin{equation}\label{eq:psum1}
    \sum_{y \in S}  P(x, y) = 1 \text{ for all } x \in S.
\end{equation}
The restriction in \eqref{eq:psum1} just says that the
state space is ``complete:'' after arriving at $x \in S$, the state must now
move to some $y \in S$.

Using notation from \S\ref{ss:spt}, to say that $P$ is stochastic is the same
as requiring that each row of $P$ is in $\dD(S)$.  

\begin{example}\label{eg:quah}
    A Markov model is estimated in the
    international growth dynamics study of \cite{quah1993empirical}.  The state is
    real GDP per capita in a given country relative to the world average.  Quah
    discretizes the possible values to 0--1/4, 1/4--1/2, 1/2--1, 1--2 and
    2--$\infty$, calling these states 1 to 5 respectively.  The transitions are
    over a one year period.  Estimated one step transition probabilities are
    represented as a weighted digraph in Figure~\ref{f:quah_graph}, where
    \begin{itemize}
        \item $S = \{1, \ldots, 5\}$ is the state space
        \item the set of edges $E$ is represented by arrows and
        \item transition probabilities are identified with weights attached to
            these edges.
    \end{itemize}
\end{example}

\begin{figure}
   \begin{center}
       \scalebox{1.0}{\input{tikz/quah_graph.tex}}
   \end{center}
   \caption{\label{f:quah_graph} Cross-country GDP dynamics as a digraph}
\end{figure}
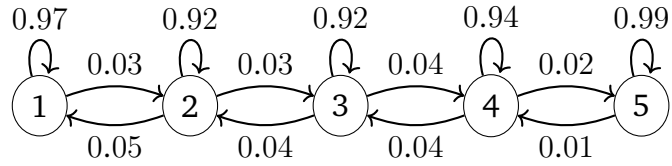

The transition matrix for the Markov model in Example~\ref{eg:quah} is
\begin{equation}
    \label{eq:pq}
    P_Q =
    \left(
      \begin{array}{ccccc}
        0.97 & 0.03 & 0.00 & 0.00 & 0.00 \\
        0.05 & 0.92 & 0.03 & 0.00 & 0.00 \\
        0.00 & 0.04 & 0.92 & 0.04 & 0.00 \\
        0.00 & 0.00 & 0.04 & 0.94 & 0.02 \\
        0.00 & 0.00 & 0.00 & 0.01 & 0.99
      \end{array}
    \right)
\end{equation}
Note the large values on the principal diagonal of $P_Q$.  These indicate strong
persistence:  the state stays constant from period to period with high
probability.

\cite{quah1993empirical} estimated $P_Q$ by maximum likelihood, pooling
transitions over the years 1960--1984.  (In this case maximum likelihood
estimation amounts to recording the relative frequency of transitions between states.)  
Figure~\ref{f:quah_graph_2} shows how the numbers change if we repeat the
exercise using World Bank GDP data from 1985--2019.  The numbers are quite
stable relative to the earlier estimate.  Below we will examine how long run
predictions are affected.

\begin{figure}
   \begin{center}
       \scalebox{1.0}{\input{tikz/quah_graph_2.tex}}
   \end{center}
   \caption{\label{f:quah_graph_2} Cross-country GDP dynamics as a digraph, updated data}
\end{figure}
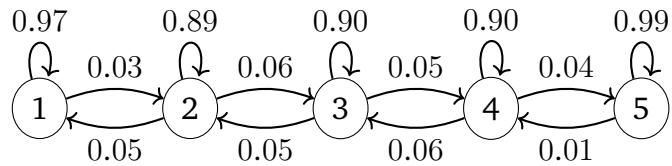

As another example,~\cite{benhabib2019wealth} estimate the following
transition matrix for intergenerational social mobility:

\begin{equation}
    \label{eq:piben}
    P_B :=
    \begin{pmatrix}
      0.222 & 0.222 & 0.215 & 0.187 & 0.081 & 0.038 & 0.029 & 0.006\\
      0.221 & 0.22 & 0.215 & 0.188 & 0.082 & 0.039 & 0.029 & 0.006\\
      0.207 & 0.209 & 0.21 & 0.194 & 0.09 & 0.046 & 0.036 & 0.008\\
      0.198 & 0.201 & 0.207 & 0.198 & 0.095 & 0.052 & 0.04 & 0.009\\
      0.175 & 0.178 & 0.197 & 0.207 & 0.11 & 0.067 & 0.054 & 0.012\\
      0.182 & 0.184 & 0.2 & 0.205 & 0.106 & 0.062 & 0.05 & 0.011\\
      0.123 & 0.125 & 0.166 & 0.216 & 0.141 & 0.114 & 0.094 & 0.021\\
      0.084 & 0.084 & 0.142 & 0.228 & 0.17 & 0.143 & 0.121 & 0.028
    \end{pmatrix}
\end{equation}

Here the states are percentiles of the wealth distribution. In particular,
with the states represented by $1, 2, \ldots, 8$, the corresponding
percentiles are
\begin{center}
    0--20\%,
    20--40\%,
    40--60\%,
    60--80\%,
    80--90\%,
    90--95\%,
    95--99\%,
    99--100\%
\end{center}
Transition probabilities are estimated from US 2007--2009 Survey of Consumer
Finances data. Relative to the highly persistent matrix $P_Q$, less weight on
the principle diagonal suggests more mixing---the influence of initial
conditions is relatively short-lived.

Additional insight about the dynamics can be obtained from a contour plot of
the matrix $P_B$, as in Figure~\ref{f:markov_matrix_visualization}.  Here
$P_B$ is plotted as a heat map after rotating it 90 degrees anticlockwise. The
rotation is so that the dynamics are comparable to the 45 degree diagrams
often used to understand discrete time dynamic systems.  A vertical line from
state $x$ corresponds to the next period conditional distribution $P(x,
\cdot)$.

In this case, we see that, for example, lower states are quite persistent,
whereas households in the highest state tend to fall back towards the middle.

\begin{figure}
   \begin{center}
       \scalebox{0.6}{\includegraphics[trim = 0mm 5mm 0mm 0mm, clip]{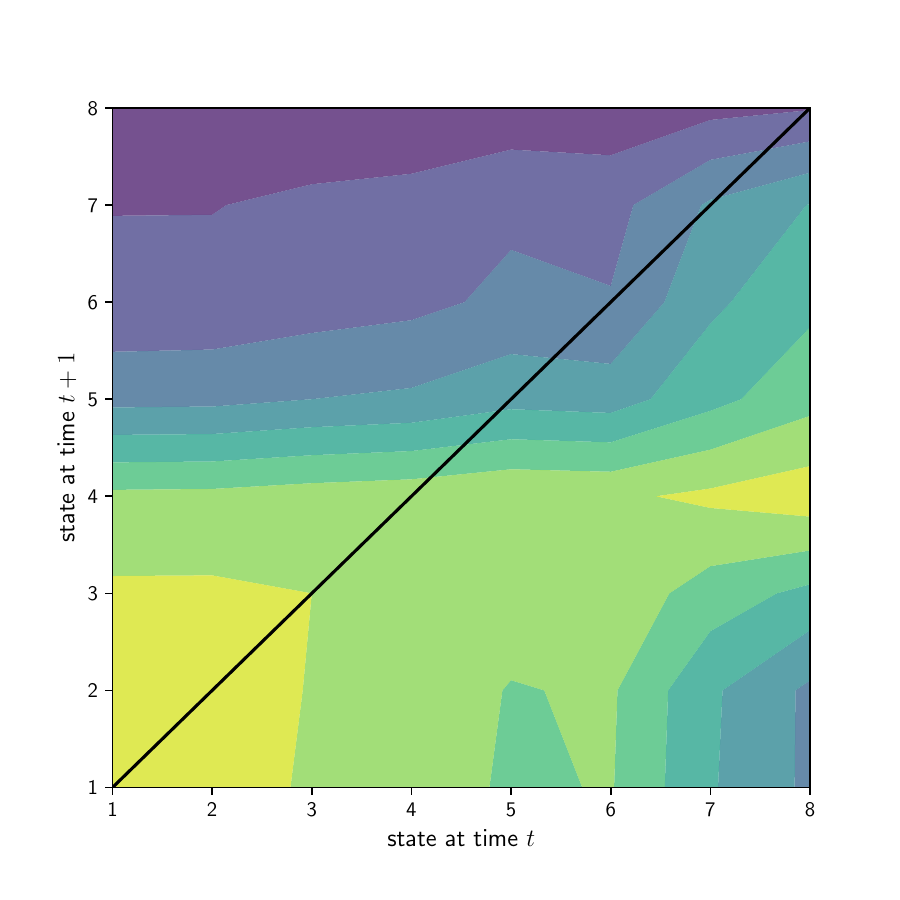}}
   \end{center}
   \caption{\label{f:markov_matrix_visualization} Contour plot of transition matrix $P_B$}
\end{figure}

\begin{Exercise}\label{ex:absorb}
    Let $\mM$ be a finite Markov model with state space $S$ and transition
    matrix $P$.  Show that $U \subset S$ is absorbing (see \S\ref{sss:coms})
    for the digraph $\mM$ if and only if 
    \begin{equation}
        \sum_{y \in U} P(x, y) = 1
        \quad \text{for all} \quad
        x \in U.
    \end{equation}
\end{Exercise}

\subsubsection{Markov Chains}\label{sss:mcs}

Consider a finite Markov model $\mM$ with state space $S$ and transition matrix $P$.
As before, $P(x, y)$ indicates the probability of transitioning from $x$ to
$y$ in one step.  Another way to say this is that, when in state $x$, we
update to a new state by choosing it from $S$ via the distribution $P(x,
\cdot)$.   The resulting stochastic process is called a Markov chain.

We can state this more formally as follows. Let $(X_t)_{t \in \ZZ_+}$ be a
sequence of random variables taking values in $S$.  We say that $(X_t)$ is a
\navy{Markov chain}\index{Markov chain} on $S$ if there exists a stochastic
matrix $P$ on $S$ such that
\begin{equation}
    \label{eq:mcdef}
    \PP \{ X_{t+1} = y \mid X_0, X_1, \ldots, X_t \} = P(X_t, y)
    \quad \text{for all} \quad
    t \geq 0, \; y \in S.
\end{equation}
To simplify terminology, we also call $(X_t)$
\navy{$P$-Markov}\index{$P$-Markov} when it satisfies \eqref{eq:mcdef} .  We
call either $X_0$ or its distribution $\psi_0$ the \navy{initial
condition}\index{Initial condition} of $(X_t)$ depending on context.

The definition of a Markov chain says two things:
\begin{enumerate}
    \item When updating to $X_{t+1}$ from $X_t$, earlier states are not
        required.
    \item The matrix $P$ encodes all of the information required to perform
        the update, given $X_t$.
\end{enumerate}

One way to think about Markov chains is algorithmically: Let $P$ be a
stochastic matrix and let $\psi_0$ be an element of $\dD(S)$.  Now generate $(X_t)$
via Algorithm~\ref{algo:mc}.  The resulting sequence is $P$-Markov with
initial condition $\psi_0$.

\begin{algorithm}
    \DontPrintSemicolon
  set $t=0$ and draw $X_t$ from $\psi_0$ \;
  \While{$t < \infty$}
  {
      draw $X_{t+1}$ from the distribution $P(X_t, \cdot)$   \;
      let $t = t + 1$ \;
  }
  \caption{\label{algo:mc} Generation of $P$-Markov $(X_t)$ with initial
  condition $\psi_0$}
\end{algorithm}

\subsubsection{Simulation}

For both simulation and theory, it is useful to be able to translate 
Algorithm~\ref{algo:mc} into a stochastic difference equation governing the
evolution of $(X_t)_{t \geq 0}$.  We now outline
the procedure, which uses inverse transform sampling (see \S\ref{sss:its}).
For simplicity, we assume that $S = \natset{n}$, with typical elements
$i,j$.   The basic idea is to apply the inverse transform method to each row
of $P$ and then sample by drawing a uniform random variable at each update.

To this end, we set
\begin{equation*}
    F(i, u) := \sum_{j=1}^{n} j \1\{q(i, j-1)  < u \leq q(i, j)\}
    \qquad
    (i \in S, \; u \in (0, 1)),
\end{equation*}
where, for each $i, j  \in S$, the value $q(i, j)$ is defined recursively by 
\begin{equation*}
    q(i, j) := q(i, j-1) + P(i, j)  \quad \text{with } q(i, 0) = 0.
\end{equation*}
Let $U(0, 1)$ represent the uniform distribution on $(0, 1)$ and take
\begin{equation}
    \label{eq:msrs}
    X_{t+1} = F(X_t, U_{t+1})
    \quad \text{where } \;
    (U_t ) \iidsim U(0, 1).
\end{equation}
If $X_0$ is an independently drawn random variable with distribution $\psi_0$
on $S$, then $(X_t)$ is $P$-Markov on $S$ with initial condition $\psi_0$, as
Exercise~\ref{ex:showppc} asks you to show.

\begin{Exercise}\label{ex:showppc}
    Conditional on $X_t = i$, show that, for given $j \in S$,
    \begin{enumerate}
        \item $X_{t+1} = j$ if and only if $U_{t+1}$ lies in the interval
            $(q(i, j-1), q(i, j)]$.
        \item This event has probability $P(i, j)$.
    \end{enumerate}
    Conclude that $X_{t+1}$ in~\eqref{eq:msrs} is a draw from $P(i, \cdot)$.
\end{Exercise}

\begin{Answer}
    Point (i) is immediate from the definition of $F$.  Regarding
    (ii), from $q(i, j) = q(i, j-1) + P(i, j)$ we have $P(i, j) = q(i,
    j) - q(i, j-1)$, which is the length of the interval $(q(i, j-1), q(i, j)]$.
    The probability that $U_{t+1}$ falls in this interval is its length, which
    we just agreed is $P(i,j)$.  The proof is now complete.
\end{Answer}

Each subfigure in Figure~\ref{f:benhabib_mobility_mixing} shows  realizations of two Markov
chains, both generated using the stochastic difference equation
\eqref{eq:msrs}.  The sequences are generated each with its  own independent
sequence of draws $(U_t)$.  The underlying transition matrices are $P_B$
from~\eqref{eq:piben} in the top panel and $P_Q$ from~\eqref{eq:pq} in the
bottom panel.  In both panels, one chain starts from the lowest state and
the other from the highest.  Notice that time series generated by $P_B$ mix faster
than those generated by $P_Q$: the difference in initial states is not a
strong predictor of outcomes after an initial ``burn in'' period.  We discuss
mixing and its connection to stability below.

\begin{figure}
   \begin{center}
       \scalebox{0.7}{\includegraphics{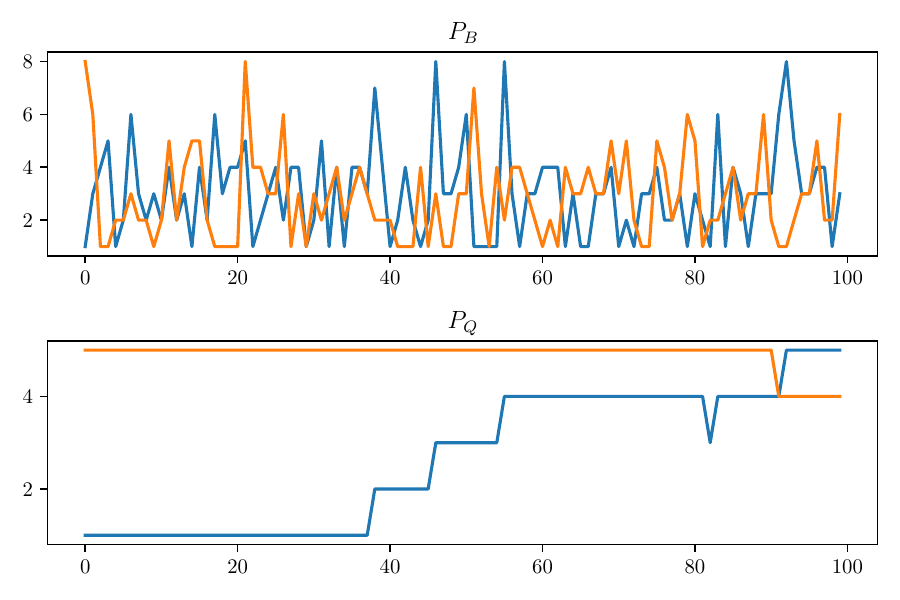}}
   \end{center}
   \caption{\label{f:benhabib_mobility_mixing} Wealth percentile over time}
\end{figure}

\subsubsection{Higher Order Transition Matrices}

Given a finite Markov model $\mM$ with state space $S$ and transition matrix
$P$, define $(P^k)_{k \in \NN}$ by $P^{k+1} = P P^k$ for all $k$, with the
understanding that $P^0 = I =$ the identity matrix. In other words, for each
$k$, the matrix $P^k$ is the $k$-th power of $P$.  If we spell out the matrix
product $P^{k+1} = P P^k$ element-by-element, we get
\begin{equation}
    \label{eq:updatepi}
    P^{k+1}(x,y) := \sum_z  P(x,z) P^k(z,y)
    \qquad (x, y \in S, \; k \in \NN).
\end{equation}

\begin{Exercise}
    Prove that $P^k$ is a stochastic matrix on $S$ for all $k \in \NN$. 
\end{Exercise}

\begin{Answer}
    $P^k$ is stochastic for all $k \in \NN$ by induction and the fact that the
    set of stochastic matrices is closed under multiplication
    (see~\ref{sss:stochmat}).
\end{Answer}

In this context, $P^k$ is called the \navy{$k$-step transition matrix}
corresponding to $P$.  The $k$-step transition matrix has the following
interpretation: If $(X_t)$ is $P$-Markov, then, for any $t, k \in \NN$ and $x,
y \in S$, 
\begin{equation}
    \label{eq:pik}
    P^k(x, y) = \PP\{X_{t + k} = y \given X_t = x\}.
\end{equation}
In other words, $P^k$ provides the $k$-step transition probabilities for the
$P$-Markov chain $(X_t)$, as suggested by its name.

This claim can be verified by induction.  Fix $t \in \NN$ and $x, y \in S$.
The claim is true by definition when $k=1$.  Suppose the claim is also true at
$k$ and now consider the case $k+1$.  By the law of total probability, we have
\begin{equation*}
    \PP\{X_{t+k+1} = y \given X_t = x \}
    = \sum_z \PP\{X_{t+k+1} = y \given X_{t+k} = z\} \PP\{X_{t+k} = z \given X_t = x\}.
\end{equation*}
The induction hypothesis allows us to use~\eqref{eq:pik}, so 
the last equation becomes
\begin{equation*}
    \PP\{X_{t+k+1} = y \given X_t = x \}
    = \sum_z P(z, y) P^k(x, z)  
    = P^{k+1} (x, y).
\end{equation*}
This completes our proof by induction.

A useful identity for the higher order Markov matrices is
\begin{equation}
    \label{eq:chapkol}
    P^{j+k}(x,y) = \sum_z P^k(x,z) P^j(z,y)
    \qquad ((x, y) \in S \times S)
\end{equation}
which holds for any $j, k$ in $\NN$.  This is called the
\navy{Chapman--Kolmogorov equation}\index{Chapman--Kolmogorov equation}. Note
that
\begin{itemize}
    \item \eqref{eq:updatepi} is a special case of \eqref{eq:chapkol} and
    \item \eqref{eq:chapkol} is a special case of \eqref{eq:accip} on
        page~\pageref{eq:accip}, written with different notation.
\end{itemize}

To provide probabilistic intuition for the validity of the Chapman--Kolmogorov
equation, let $X_0 = x$ and let $y \in S$ be given.  Using the law of total
probability again, we have
\begin{equation*}
    \PP\{X_{j+k} = y \given X_0 = x\}
     = \sum_z \PP\{X_{j+k} = y \given X_0 = x, \, X_k = z\} 
        \PP\{X_k = z \given X_0 = x\}
\end{equation*}
By Markov property~\eqref{eq:mcdef}, the future and past are independent given
the present, so 
\begin{equation*}
    \sum_z \PP\{X_{j+k} = y \given X_0 = x, \, X_k = z\}
     = \sum_z \PP\{X_{j+k} = y \given X_k = z\} .
\end{equation*}
As a result of this fact and~\eqref{eq:pik}, the equation before this one can
be rewritten
as~\eqref{eq:chapkol}.

\subsection{Distribution Dynamics}\label{ss:dd}

Let $\mM$ be a finite Markov model with state space $S$ and transition matrix
$P$.  Let $(X_t)$ be $P$-Markov and, for each $t \geq 0$, let $\psi_t \in \dD(S)$ be defined by 
\begin{equation*}
    \psi_t
    := \PP\{X_t = \cdot\}
    = \text{ the distribution of $X_t$}.    
\end{equation*}
The vector $\psi_t$ is called the \navy{marginal distribution}\index{Marginal
distribution} of $X_t$.  While $(X_t)$ is random, the sequence $(\psi_t)$ is
deterministic.  In this section we investigate its dynamics.

\subsubsection{Updating Marginal Distributions}\label{sss:lmarg}

The key idea for this section is that there is a simple link between
successive marginal distributions: by the law of total probability, we have
\begin{equation*}
  \PP \{X_{t+1} = y \}
     = \sum_x \PP \{ X_{t+1} = y \, | \, X_t = x \}
                 \cdot \PP \{ X_t = x \},
\end{equation*}
which can be rewritten as
\begin{equation}
    \label{eq:flbm}
      \psi_{t+1}(y) = \sum_x P(x,y) \psi_t(x)
      \qquad (y \in S). 
\end{equation}
This fundamental expression tells how to update marginal distributions
given the transition matrix $P$.

When each $\psi_t$ is interpreted as a row vector, we can write~\eqref{eq:flbm} as
\begin{equation}
    \label{eq:bauprl}
    \psi_{t+1} = \psi_t P.
\end{equation}
(Henceforth, in expressions involving matrix algebra, distributions are
row vectors unless otherwise stated).  Some authors refer to~\eqref{eq:bauprl}
as the \navy{forward equation} associated with $P$, while $\psi \mapsto \psi
P$ is called the \navy{forward operator}, by analogy with the
Kolmogorov forward equation from continuous time.

Think of \eqref{eq:bauprl} as a difference equation in distribution space.
Iterating backwards, 
\begin{equation*}
    \psi_t 
    = \psi_{t-1} P 
    = (\psi_{t-2} P) P
    = \psi_{t-2} P^2
    = \cdots
    = \psi_0 P^t. 
\end{equation*}
Given any $\psi_0 \in \dD(S)$ and $t \in \NN$, we have
\begin{equation*}
    \psi_0 P^t 
    = \text{ the distribution of $X_t$ given $X_0 \eqdist \psi_0$}.
\end{equation*}

\begin{example}\label{eg:quahpred}
    As an illustration, let's take the matrix $P_Q$ estimated by
    \cite{quah1993empirical} using 1960--1984 data and use $P_Q^t$ to update the
    1985 distribution $t = 2019 - 1985 = 34$ times, in order to obtain a
         prediction for the cross-country income distribution in 2019.
    Figure~\ref{f:quah_gdppc_prediction} shows how this prediction fares
    compared to the realized 2019 distribution, calculated from World Bank GDP
    data.\footnote{Example~\ref{eg:quahpred} is intended as an 
    illustration of the mechanics of updating distributions. While the
    methodology in \cite{quah1993empirical} is thought provoking, it struggles to make plausible
    long run predictions about something as complex as the cross-country income
    distribution.  Indeed, the Kullback--Leibler deviation between the predicted and
    realized 2019 distributions is actually larger than the Kullback--Leibler
    deviation between the 1985 and 2019 distributions.  Evidently, a naive
    estimate ``nothing will change''  model predicts better than Quah's.
    Most of the focus in this text 
        is on short term forecasts and scenarios, where the system is
        approximately stationary after suitable transformations, rather than heroic long term
        predictions where nonstationary change is hard to quantify.}
\end{example}

\begin{figure}
   \centering
   \scalebox{0.7}{\includegraphics{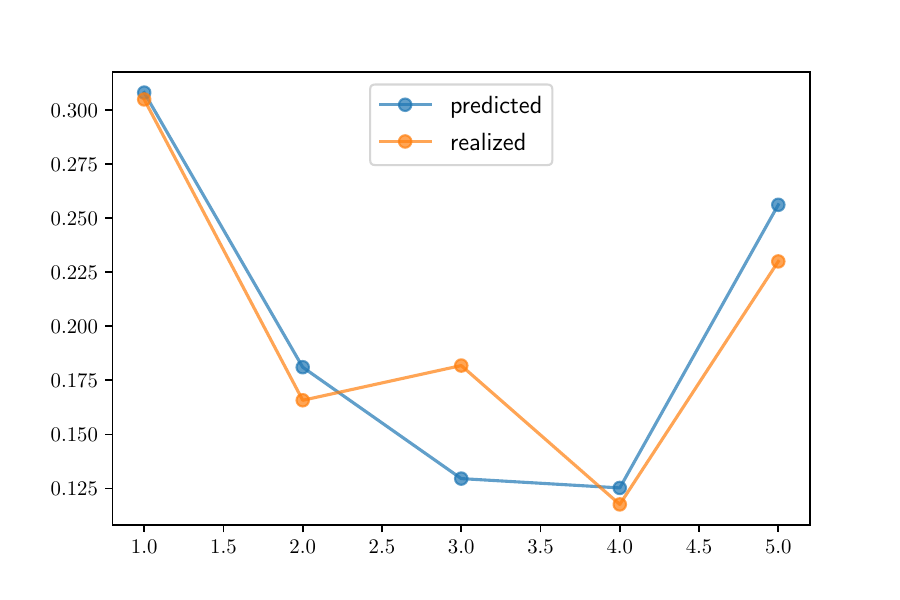}}
   \caption{\label{f:quah_gdppc_prediction} Predicted vs realized cross-country income distributions for 2019}
\end{figure}

\subsubsection{Trajectories in the Long Run}\label{sss:trajlr}

One of the main sub-fields of Markov chain analysis is asymptotics of
distribution sequences.  This topic turns out to be important for network
theory as well.  In \S\ref{ss:stabil} we will investigate asymptotics in
depth.  In this section we build some intuition via simulations.

Figure~\ref{f:simplex_2} shows the trajectory $(\psi P_a^t)$ when $S = \{1, 2,
3\}$, $\psi = (0, 0, 1)$, and $P_a$ is the transition matrix displayed in
\eqref{eq:pia}.  The blue triangle is the unit simplex in $\RR^3_+$, consisting of
all row vectors $\psi \in \RR^3$ such that $\psi \geq 0$ and $\psi \1 = 1$.
The unit simplex can be identified with $\dD(S)$ when $S = \{1, 2, 3\}$. Red
dots in the figure are distributions in the sequence $(\psi P_a^t)$ for $t=0,
\ldots, 20$.
Figure~\ref{f:simplex_3} shows distribution dynamics for
$P_a$ that start from initial condition $\psi = (0, 1/2, 1/2)$.  

It seems that both of the sequences are converging.  This turns out to be the
case---the black dot in the figures is the limit of both sequences and,
moreover, this point is a stationary distribution for $P_a$, as defined in
\S\ref{sss:stochmat}.  In fact we can---and will---also show that no other
stationary distribution exists, and that $\psi P_a^t$ converges to the
stationary distribution regardless of the choice of $\psi \in \dD(S)$.
This is due to certain properties of $P_a$, related to connectivity and
aperiodicity.

\begin{figure}
   \begin{center}
         \scalebox{0.45}{\includegraphics[trim = 0mm 15mm 0mm 0mm, clip]{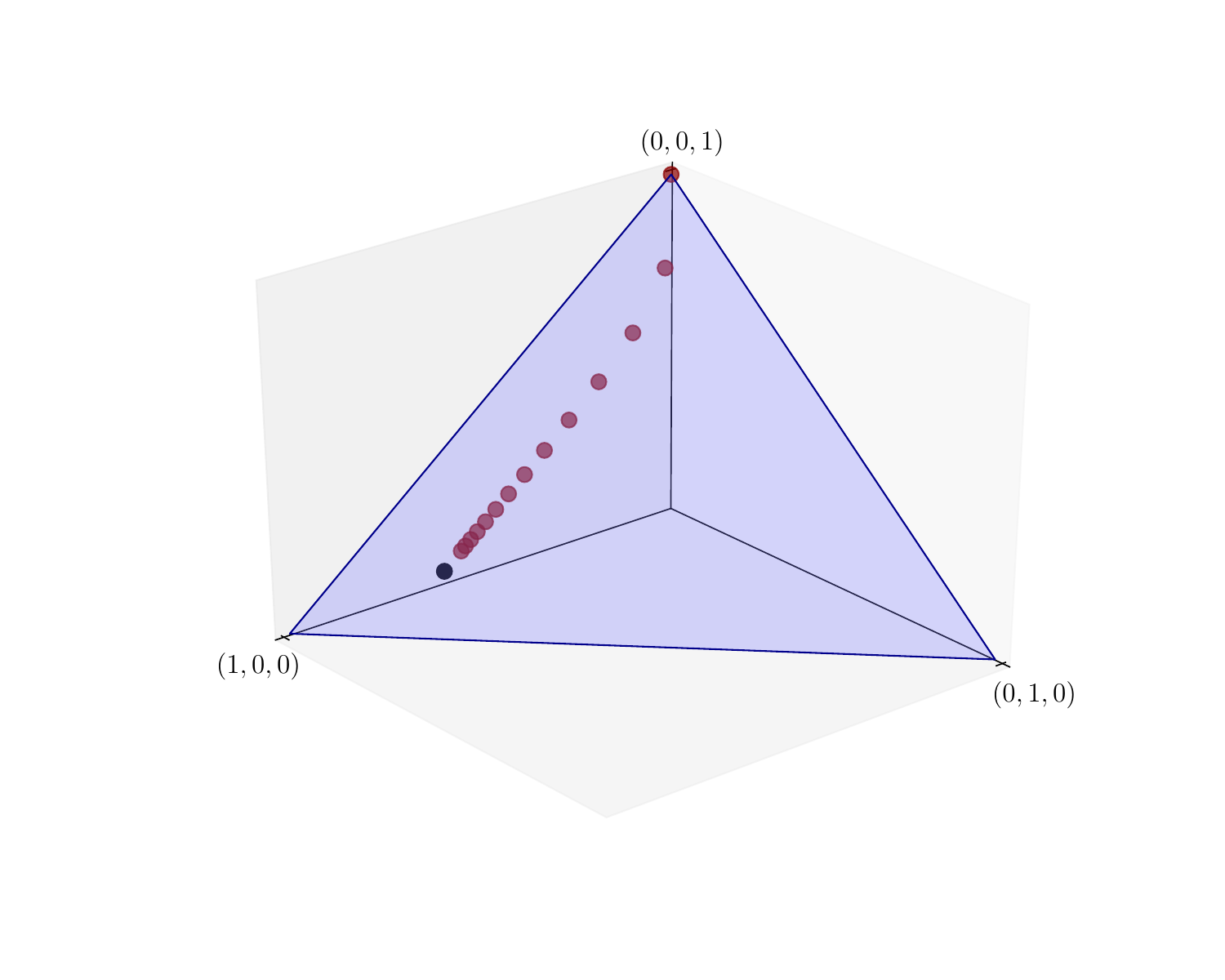}}
   \end{center}
   \caption{\label{f:simplex_2} A trajectory from $\psi_0 = (0, 0, 1)$}
\end{figure}

\begin{figure}
   \begin{center}
       \scalebox{0.45}{\includegraphics[trim = 0mm 15mm 0mm 0mm, clip]{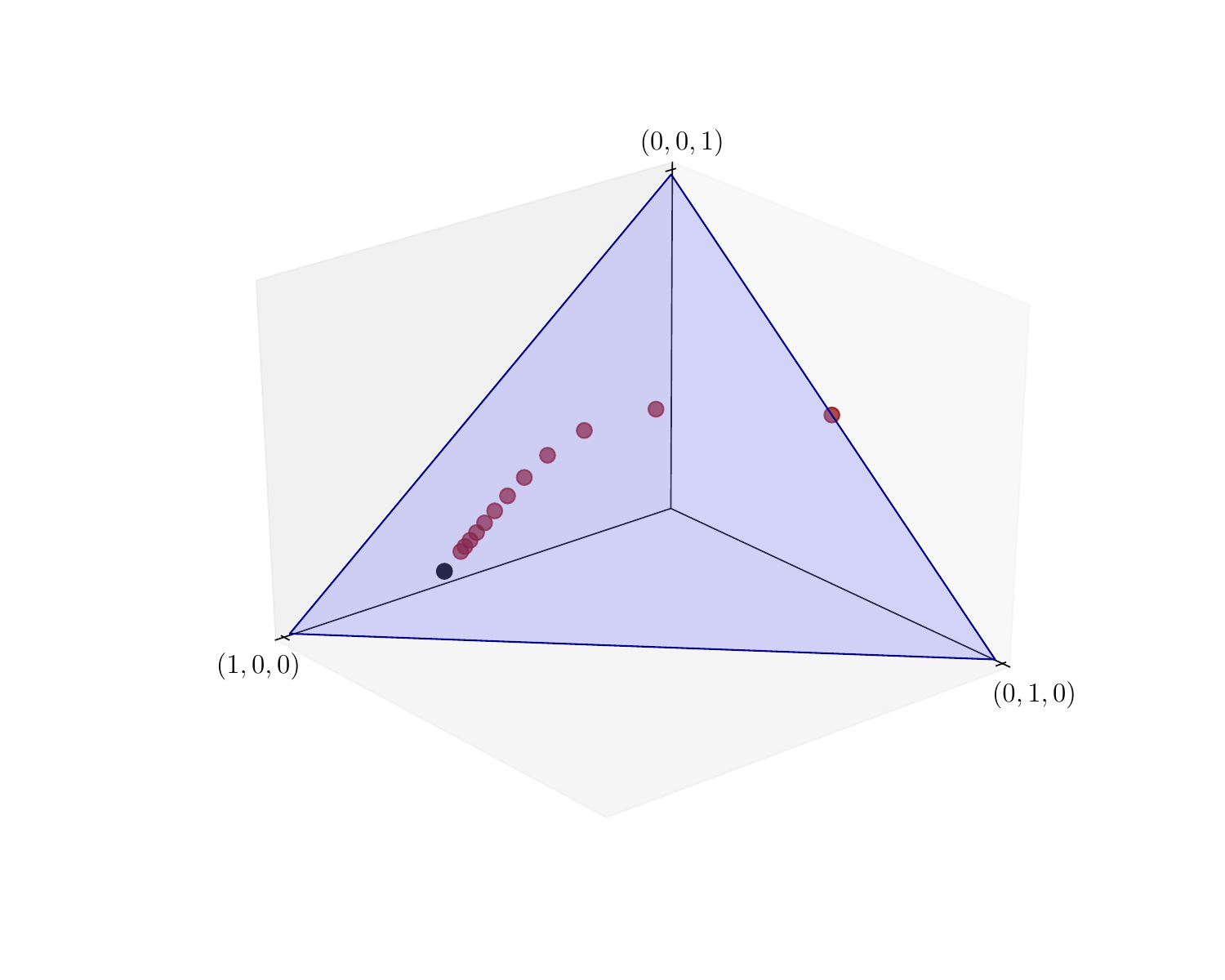}}
   \end{center}
   \caption{\label{f:simplex_3} A trajectory from $\psi_0 = (0, 1/2, 1/2)$}
\end{figure}


Figure~\ref{f:benhabib_mobility_dists} provides  another view of a distribution
sequence, this time generated from the matrix $P_B$ in \eqref{eq:piben}.  The initial
condition $\psi_0$ was uniform.  Each distribution shown was
calculated as $\psi P_Q^t$ for different values of $t$. The distribution
across wealth classes converges rapidly to what appears to be a long run
limit. Below we confirm that this is so, and that the limit is independent of
the initial condition.  The rapid rate of convergence is due to the high level
of mixing in the transition dynamics.

\begin{figure}
   \centering
   \scalebox{.6}{\includegraphics{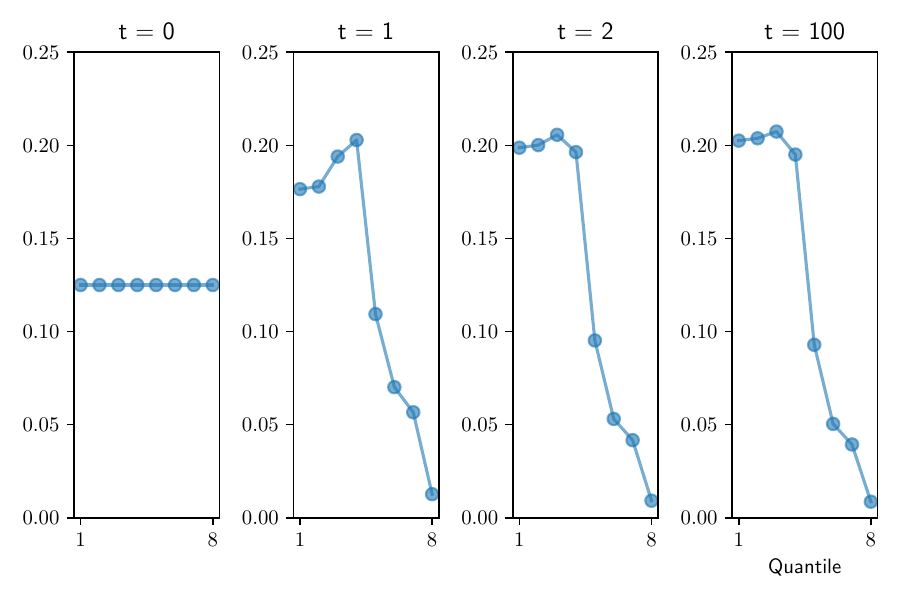}}
   \caption{\label{f:benhabib_mobility_dists} Distribution projections from $P_B$}
\end{figure}

\subsection{Stationarity}\label{ss:station}

In this section we examine stationary distributions and their properties.
As we will see, stationary distributions can be regarded as 
steady states for the evolution of the sequence of marginal distributions.
(Later, when we study ergodicity, stationary distributions will acquire another
important interpretation.)

\subsubsection{Stationary Distributions}\label{sss:statdist}\index{Stationary distribution}

Recall from \S\ref{sss:stochmat} that if $P$ is a stochastic matrix and 
$\psi \in \RR^n_+$ is a row vector such that $\psi \1 = 1$ and $\psi P =
\psi$, then $\psi$ is called stationary for $P$.  So now let $\mM$ be a finite Markov model with state space $S$ and transition matrix $P$.
Translating to the notation of Markov chains, a distribution $\psi^* \in
\dD(S)$ is \navy{stationary}\index{Stationary} for $\mM$ if 
\begin{equation*}
    \psi^*(y) = \sum_{x \in S} P(x,y) \psi^*(x)
    \quad \text{for all } 
    y \in S.
\end{equation*}
We also write this expression as $\psi^* = \psi^* P$ and understand $\psi^*$
as a fixed point of the map $\psi \mapsto \psi P$ that updates distributions
(cf., Equation \eqref{eq:bauprl}).

If $X_t \eqdist \psi^*$, then, for any $j \in \NN$, the fixed point
property implies that $X_{t+j} \eqdist \psi^* P^j = \psi^*$.  
Hence
\begin{equation*}
    X_t \eqdist \psi^* \implies  X_{t+j} \eqdist \psi^* \; \text{ for all } j \in \NN.
\end{equation*}

\begin{example}\label{eg:empl}
    Suppose workers are hired at rate $\alpha$ and fired at rate
    $\beta$, transitioning between unemployment and employment according to 
    \begin{equation}\label{eq:pwmat}
        P_w = 
        \begin{pmatrix}
            1-\alpha & \alpha \\
            \beta & 1-\beta 
        \end{pmatrix}.
    \end{equation}
    In \S\ref{sss:wdii} we showed $\alpha + \beta > 0$ implies $\psi = (\beta
    / (\alpha + \beta), \alpha / (\alpha + \beta))$ is a dominant left
    eigenvector for $P_w$.  Since $r(P_w)=1$, the Perron--Frobenius Theorem
    tells us that $\psi$ is
    stationary for $P_w$.  One implication is that, if the distribution of
    workers across the two employment states is given by $\psi$ and updates
    obey the dynamics encoded in $P_w$, then no further change is observed in
    the unemployment rate.
\end{example}

\begin{example}
    The black dot in each of Figures~\ref{f:simplex_2}--\ref{f:simplex_3},
    indicating the limit of the marginal sequences, is a stationary
    distribution for the Markov matrix $P_a$ displayed in~\eqref{eq:pia},
    page~\pageref{eq:pia}.  We discuss its calculation in \S\ref{sss:mccomp}.
\end{example}

\begin{Exercise}
    Let $\mM$ be a finite Markov model with state space $|S|=n$ and transition
    matrix $P$.  Let $\psi
    \equiv 1/n$ be the uniform distribution on $S$.  Prove that $\psi$ is
    stationary for $P$ if and only if $P$ is \navy{doubly stochastic} (i.e.,
    has unit column sums as well as unit row sums).
\end{Exercise}

In Figures~\ref{f:simplex_2}--\ref{f:simplex_3} all trajectories converge
towards the stationary distribution.  Not all stationary distributions have
this ``attractor'' property, and in general there can be many stationary
distributions.  The next example illustrates.

\begin{Exercise}\label{ex:msdi}
    Let $\mM = (S, E, p)$ be a finite Markov model with $(x, y) \in E$ if and
    only if $x=y$.
    Describe the implied weight function and corresponding transition matrix.  Show
    that every distribution in $\dD(S)$ is stationary for $\mM$.
\end{Exercise}

\begin{Answer}
    Let $\mM$ be as described. In this setting, the requirement that the
    transition matrix $P$ is stochastic implies that $p(x,x) = 1$ and $P(x, y)
    = \1\{x = y\}$. Thus, $P = I$, the $n \times n$ identity matrix. Every
    distribution is stationary because $\psi I = \psi$ for all $\psi \in
    \dD(S)$.
\end{Answer}

\subsubsection{Existence and Uniqueness}\label{sss:statexist}

From the Perron--Frobenius Theorem we easily obtain the following fundamental result.

\begin{theorem}[Existence and Uniqueness of Stationary Distributions]\label{t:kb}
    Every finite Markov model $\mM = (S, E, p)$ has at least one stationary
    distribution $\psi^*$ in $\dD(S)$.  If the digraph $(S, E)$ is strongly connected, 
    then $\psi^*$ is unique and everywhere positive on $S$.
\end{theorem}

\begin{proof}
    Let $\mM$ be a finite Markov model.
    Since the corresponding adjacency matrix $P$ is stochastic, existence follows from
    Exercise~\ref{ex:sm_sr1} in \S\ref{sss:stochmat}.  By
    Theorem~\ref{t:sconir} on page~\pageref{t:sconir}, strong connectedness of
    $\mM$ implies irreducibility of $P$.  When $P$ is irreducible,
    uniqueness and everywhere positivity of the stationary distribution
    follows from the Perron--Frobenius Theorem (page~\pageref{t:pf}).
\end{proof}

The basic idea behind the uniqueness part of Theorem~\ref{t:kb} is as
follows:  Suppose to the contrary that $\mM$ is strongly connected and yet two
distinct stationary distributions $\psi^*$ and $\psi^{**}$ exist in $\pP(S)$.
Since a $P$-Markov chain started at $\psi^*$ always has marginal distribution
$\psi^*$ and likewise for $\psi^{**}$, different initial conditions lead to
different long run outcomes.  This contradicts strong connectedness in the
following way: strong connectedness implies that both chains traverse the
whole state space.  Moreover, being Markov chains, they forget the past once
they arrive at any state, so the starting draws should eventually be
irrelevant.  

(Actually, the story in the last paragraph is incomplete.  Initial
conditions \emph{can} determine long run outcomes under strong connectedness
in one sense: a ``periodic'' Markov model can traverse the whole space but
\emph{only at specific times} that depend on the starting location.  If we
rule out such periodicity, we get results that are even stronger than
Theorem~\ref{t:kb}.  This topic is treated in \S\ref{ss:stabil}.)




\subsubsection{Brouwer's Theorem}\label{ss:brouwer}

Another way to prove the existence claim in Theorem~\ref{t:kb} is via the
famous fixed point theorem of L.\ E.\ J.\ Brouwer (1881--1966).  

\begin{theorem}[Brouwer]\label{t:brouwer}
    If $C$ is a convex compact subset of $\RR^n$
    and $G$ is a continuous self-map on $C$, then $G$ has at least one
    fixed point in $C$.
\end{theorem}

The proof of Theorem~\ref{t:brouwer} in one dimension is not difficult,
while the proof in higher dimensions is challenging. See, for example,
\cite{aliprantis1999border}.

\begin{Exercise}
    Prove Brouwer's fixed point theorem for the case $C=[0,1]$ by
    applying the intermediate value theorem.
\end{Exercise}

\begin{Answer}
    Let $G$ be a continuous function from $[0,1]$ to itself and set $f(x) := Gx - x$.
    Since $G$ is a self-map on $[0,1]$ we have $f(0) = G0 \geq 0$ and $f(1) =
    G1 - 1 \leq 0$, so $f$ is a continuous function on $[0,1]$ satisfying $f(0)
    \geq 0$ and $f(1) \leq 0$.  Existence of an $x$ satisfying $f(x)=0$
    follows from the intermediate value theorem.  The same $x$ is a fixed
    point of $G$.
\end{Answer}

One advantage of Brouwer's fixed point theorem is that its conditions are
quite weak.  One disadvantage is that it provides only existence, without
uniqueness or stability.  Figure~\ref{f:three_fixed_points}
provides an example of how multiple fixed points can coexist under the
conditions of the theorem.\footnote{There are many useful extensions to Brouwer's theorem, such as one for
set-valued functions due to~\cite{kakutani1941generalization}. These
results have many applications in economics (see,
for example, \cite{nash1950equilibrium}).
Another paper, due to~\cite{schawder1930fixpunktsatz}, extends Brouwer's result to
infinite dimensional spaces.}

\begin{Exercise}
    Prove the first part of Theorem~\ref{t:kb} (existence of a stationary
    distribution) using Brouwer's fixed point theorem. 
\end{Exercise}

\begin{Answer}
    Let $P$ be a Markov matrix on
    finite set $S$. Suppose in
    particular that $S$ has $d$ elements, so we can identify functions in
    $\RR^S$ with vectors in $\RR^d$ and $\dD(S)$ with the unit simplex
    in $\RR^d$. The set $\dD(S)$ is a closed, bounded and convex subset of
    $\RR^d$ that $P$ maps into itself.  As a linear matrix operation, the map
    $\psi \mapsto \psi P$ is continuous.  Existence of a fixed point now
    follows from Brouwer's fixed point theorem (p.~\pageref{t:brouwer}).
\end{Answer}

\subsubsection{Computation}\label{sss:mccomp}

Let's consider how to compute stationary distributions.  While
$\psi^* P = \psi^*$ is a finite set of linear equations that we might hope to solve
directly for $\psi^*$, there are problems with this idea. For example,
$\psi^* = 0$ is a solution that fails to lie in $\dD(S)$.

To restrict the solution to $\dD(S)$ we proceed as follows: Suppose
$|S|=n$ and note that row vector $\psi \in \dD(S)$ is
stationary if and only if $\psi(I - P) = 0$, where $I$ is the $n \times n$
identity matrix.  Let $\1_n$ be the $1 \times n$ row vector $(1,\ldots,1)$.
Let $\1_{n \times n}$ be the $n \times n$ matrix of ones.  

\begin{Exercise}\label{ex:statiffc}
    Consider the linear system
    \begin{equation}
        \label{eq:hsbi}
        \1_n = \psi (I - P + \1_{n \times n})
    \end{equation}
    where $\psi$ is a row vector.  Show that
    \begin{enumerate}
        \item every solution $\psi$ of \eqref{eq:hsbi} lies in $\dD(S)$ and
        \item $\psi$ is stationary for $P$ if and only if \eqref{eq:hsbi}
            holds.
    \end{enumerate}
\end{Exercise}

Solving the linear system~\eqref{eq:hsbi} produces a stationary distribution
in any setting where the stationary distribution is unique.  When this is not the case,
problems can arise.

\begin{Exercise}
    Give a counterexample to the claim that $(I - P + \1_{n \times n})$ is
    always nonsingular when $P$ is a stochastic matrix.
\end{Exercise}

There are also graph-theoretic algorithms for computing all of the stationary
distributions of an arbitrary stochastic matrix.  (The Python and Julia
libraries \texttt{QuantEcon.py} and \texttt{QuantEcon.jl} have efficient
implementations of this type.)


\section{Asymptotics}

In this section we turn to long run properties of Markov chains, including
ergodicity.  With these ideas in hand, we will then investigate the evolution
of information over social networks.

\subsection{Ergodicity}\label{ss:ergo} 

Let's begin with the fascinating and important topic of ergodicity. The simplest
way to understand ergodicity is as an extension of the law of large numbers
from {\sc iid} sequences to more general settings, where the {\sc iid}
property holds only in a limiting sense.

\subsubsection{Harmonic Functions}\label{sss:hfuncs}

Let $\mM$ be a finite Markov model with state space $S$ and transition
matrix $P$. Fix $h \in \RR^S$ and let 
\begin{equation}\label{eq:pexp}
    Ph(x) := \sum_{y \in S} P(x, y) h(y)
    \qquad (x \in S).
\end{equation}
If $h$ is understood as a column vector in $\RR^{|S|}$, then $Ph(x)$ is just
element $x$ of the vector $Ph$. The map $h \mapsto Ph$ is sometimes called the
``conditional expectation operator,'' since, given that $P(x, y) =
\PP\{X_{t+1}=y \given X_t = x\}$, we have
\begin{equation*}
    \sum_{y \in S} P(x, y) h(y)  
    = \EE [ h(X_{t+1}) \given X_t = x ]
\end{equation*}

A function $h \in \RR^S$ is called \navy{$P$-harmonic}\index{Harmonic} if $P h
= h$ pointwise on $S$.  Thus, $P$-harmonic functions are fixed
points of the conditional expectation operator $h \mapsto Ph$.

If $h$ is $P$-harmonic and $(X_t)$ is $P$-Markov, then
\begin{equation}\label{eq:harmogale}
    \EE \, [ h(X_{t+1}) \given X_t ]
    = (P h)(X_t)
    = h(X_t).
\end{equation}
(A stochastic process with this property---i.e., that the current value is the best
predictor of next period's value---is called a \emph{martingale}\index{Martingale}.
Martingales are one of the foundational concepts in probability, statistics
and finance.)

\begin{example}\label{eg:abshar}
    Let $\mM$ be a finite Markov model with state space $S$ and transition
    matrix $P$. If $A \subset S$ and $A^c := S \setminus A$ are both absorbing
    for $\mM$, then $\1_A$ and $\1_{A^c}$ are both $P$-harmonic.  To see this,
    we apply Exercise~\ref{ex:absorb} to obtain
    \begin{equation*}
        (P \1_A)(x) 
        = \sum_{y \in S} P(x, y) \1_A(y)
        = \sum_{y \in A} P(x, y) 
        = \begin{cases}
            1 & \text{ if } x \in A
            \\
            0 & \text{ if } x \in A^c.
        \end{cases}
    \end{equation*}
    In other words, $P \1_A = \1_A$. A similar argument gives $P \1_{A^c} =
    \1_{A^c}$.
\end{example}

\begin{Exercise}
    Let $\mM$ be a finite Markov model with state space $S$ and transition
    matrix $P$. Show that every constant function in
    $\RR^S$ is $P$-harmonic. 
\end{Exercise}

\subsubsection{Definition and Implications}\label{sss:ergdef}

Let $\mM$ be a finite Markov model with state space $S$ and transition matrix
$P$. We know that every constant function in $\RR^S$ is $P$-harmonic.  We call
$\mM$ \navy{ergodic}\index{Ergodic} if the
only $P$-harmonic functions in $\RR^S$ are the constant functions.

\begin{example}\label{eg:mciid0}
    If $P(x, y) = \phi(y)$ for all $x, y$ in $S$, where $\phi$ is some fixed
    distribution on $S$, then $P$ generates an {\sc
    iid} Markov chain, since the next period draw has no dependence on the
    current state.  Any model $\mM$ with an adjacency matrix of this type
    is ergodic. Indeed, if $P$ has this property and $h
    \in \RR^S$ is $P$-harmonic, then, for any $x \in S$, we have
    \begin{equation*}
        h(x) = (Ph)(x) = \sum_y P(x, y) h(y) = \sum_y \phi(y) h(y)
    \end{equation*}
    Hence, $h$ is constant.  This proves that any $P$-harmonic function is
    a constant function.
\end{example}

\begin{example}\label{eg:empl2}
    A laborer is either unemployed (state $1$) or employed (state $2$).  In
    state $1$ he is hired with probability $\alpha$.  In state $2$ he is fired
    with probability $\beta$.  The corresponding Markov model $\mM$ has state and
    transition matrix given by 
    \begin{equation*}
        S = \{1, 2\}
        \quad \text{and} \quad
        P_w =
        \begin{pmatrix}
            1-\alpha & \alpha \\
            \beta & 1-\beta 
        \end{pmatrix}.
    \end{equation*}
    The statement $P_w h = h$ becomes
    \begin{equation*}
        \begin{pmatrix}
            1-\alpha & \alpha \\
            \beta & 1-\beta 
        \end{pmatrix}
        \begin{pmatrix}
            h(1) \\
            h(2)
        \end{pmatrix}
        =
        \begin{pmatrix}
            h(1) \\
            h(2)
        \end{pmatrix}
    \end{equation*}
    The first row is $(1-\alpha) h(1) + \alpha h(2) = h(1)$, or $\alpha h(1) =
    \alpha h(2)$.  Thus, $\mM$ is ergodic whenever $\alpha > 0$.
    By a similar argument, $\mM$ is ergodic whenever $\beta > 0$.
\end{example}

\begin{example}\label{eg:abpo}
    Let $\mM$ be any finite Markov model with state space $S$ and transition
    matrix $P$.  It is immediate from
    Example~\ref{eg:abshar} that if $S$ can be partitioned into two
    nonempty absorbing sets, then $\mM$ is not ergodic.  Hence, the poverty
    trap model in Figure~\ref{f:poverty_trap_2} is not ergodic.
    Similarly, if $\alpha=\beta=0$ in the matrix $P_w$ discussed in
    Example~\ref{eg:empl2}, then $P_w = I$ and each state is a disjoint
    absorbing set.  Ergodicity fails.
\end{example}

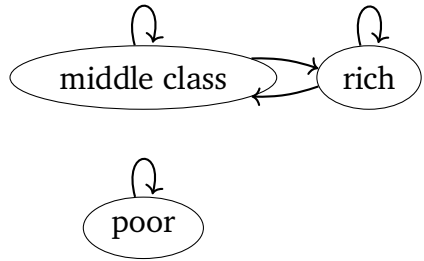
\begin{figure}
   \centering
   \input{tikz/poverty_trap_2.tex}  
   \caption{\label{f:poverty_trap_2} A poverty trap}
\end{figure}

The examples discussed above suggest that, for a finite Markov model $\mM$,
ergodicity depends on the connectivity properties of the digraph. The next
result confirms this.

\begin{proposition}\label{p:iimerg0}
    If $\mM$ is strongly connected, then $\mM$ is ergodic.
\end{proposition}

\begin{proof}
    Let $\mM$ be a strongly connected Markov model with state space $S$ and
    transition matrix $P$.  Let $h$ be $P$-harmonic and let $x \in
    S$ be the maximizer of $h$ on $S$.  Let $m = h(x)$.  Suppose there exists
    a $y$ in $S$ with $h(y) < m$.  Since $\mM$ is strongly connected, we
    can choose a $k \in \NN$ such that $P^k(x, y) > 0$.  Since $h$ is
    $P$-harmonic, we have $h = P^k h$, and hence
    \begin{equation*}
        m 
        = h(x) 
        = \sum_z P^k(x, z) h(z) 
        = P^k(x, y) h(y) + \sum_{z \not= y} P^k(x, z) h(z) 
        < m.
    \end{equation*}
    This contradiction shows that $h$ is constant.  Hence $\mM$ is ergodic.
\end{proof}

The implication in Proposition~\ref{p:iimerg0} is one way, as the next
exercise asks you to confirm.

\begin{Exercise}
    Provide an example of a finite Markov model that is ergodic but not
    strongly connected.  
\end{Exercise}

\begin{Answer}
    Let $\mM$ be a finite Markov model with states $S = \{1, 2\}$ and edge set
    $E = \{(1,2)\}$.  Thus, the chain immediately moves to state 2 and stays
    there forever.  The corresponding transition matrix is
    \begin{equation*}
        P 
        =
        \begin{pmatrix}
            0 & 1 \\
            0 & 1
        \end{pmatrix}
    \end{equation*}
    If $h$ is $P$-harmonic, then $h(x) = Ph(x) = h(2)$.  Hence $h$ is
    constant.  This shows that $P$ is ergodic.  At the same time $(S, E)$ is
    not strongly connected.
\end{Answer}


\subsubsection{Implications of Ergodicity}

One of the most important results in probability theory is the \navy{law of large
numbers}\index{Law of large numbers} (LLN).  In the finite state setting, the classical version of this
theorem states that
\begin{equation*}
       \PP \left\{
           \lim_{k \to \infty} \frac{1}{k} \sum_{t=0}^{k-1} h(X_t)
                   = \sum_{x \in S} h(x) \phi(x)
               \right\} = 1
\end{equation*}
when $(X_t)_{t \geq 0}$ is an {\sc iid} sequence of random variables with
common distribution $\phi \in \dD(S)$ and $h$ is an arbitrary element of
$\RR^S$.  

This version of the LLN is classical in the sense that the {\sc iid}
assumption is imposed.  It turns out that the {\sc iid} assumption can be
weakened to allow for a degree of dependence between observations, which leads
us to ask whether or not the LLN holds for Markov chains as well.

The answer to this question is yes, provided that dependence between
observations dies out fast enough.  One candidate for this condition is
ergodicity.  In fact, it turns out that ergodicity is the exact necessary and
sufficient condition required to extend the law of large numbers to Markov
chains.  The next theorem gives details.

\begin{theorem}\label{t:merg0}
    Let $\mM$ be any finite Markov model with state space $S$ and transition
    matrix $P$.
    The following statements are equivalent:
    \begin{enumerate}
        \item $\mM$ is ergodic.
        \item $\mM$ has a unique stationary distribution $\psi^*$ and,
            for any $P$-Markov chain $(X_t)$ and any $h \in \RR^S$,
            \begin{equation}\label{eq:erglln}
               \PP \left\{
                   \lim_{k \to \infty} \frac{1}{k} \sum_{t=0}^{k-1} h(X_t)
                   = \sum_{x \in S} h(x) \psi^*(x)
               \right\} = 1.
            \end{equation}
    \end{enumerate}
\end{theorem}

The proof of Theorem~\ref{t:merg0} is given in Chapter~17 of
\cite{meyn2009markov}.  We skip the proof but provide intuition through the
following examples.

\begin{example}\label{eg:mciid}
    Recall the {\sc iid} case from Example~\ref{eg:mciid0}.  We showed that
    $\mM$ is ergodic.  By Theorem~\ref{t:merg0}, the convergence
    in~\eqref{eq:erglln} holds with $\psi^* = \phi$. This is consistent with
    the LLN for {\sc iid} sequences.
\end{example}

\begin{example}\label{eg:failerg}
    Let $\mM$ be a finite Markov model with $S = \{1, 2\}$ and $P = I$, the identity matrix.  
    Markov chains generated by $P$ are constant.  Since every $h
    \in \RR^2$ satisfies $Ph=h$, we see that $\mM$ is not ergodic.  This means
    that the LLN result in~\eqref{eq:erglln} fails.
    But this is exactly what we would expect, since a constant chain $(X_t)$
    implies $\frac{1}{k} \sum_{t=0}^{k-1} X_t = X_0$ for all $k$.  In particular,
    if $X_0$ is drawn from a nondegenerate distribution, then the sample mean
    does not converge to any constant value.
\end{example}

\begin{example}\label{eg:failerg2}
    Consider again the poverty trap model in
    Figure~\ref{f:poverty_trap_2}.  Say that $h(x)$ is earnings in state $x$, and
    that $h(\text{poor}) = 1$,  $h(\text{middle}) = 2$ and $h(\text{rich}) = 3$.
    Households that start with $X_0 = \text{poor}$ will always be poor,
    so $\frac{1}{k} \sum_{t=0}^{k-1} h(X_t) = 1$ for all $k$.  Households that
    start with $X_0$ in $\{ \text{middle, rich} \}$ remain in this absorbing set
    forever, so $\frac{1}{k} \sum_{t=0}^{k-1} h(X_t) \geq 2$ for all $k$.
    In particular, the limit of the sum depends on the initial condition.
    This violates part (ii) of Theorem~\ref{t:merg0}, which states that the limit
    is independent of the distribution of $X_0$.
\end{example}

\subsubsection{Reinterpreting the Stationary Distribution}

Ergodicity has many useful implications.  One is a new
\emph{interpretation} for the stationary distribution. To illustrate this,
let $\mM$ be an ergodic Markov model with state space $S$ and transition matrix
$P$.  Let $(X_t)$ be a $P$-chain and let
\begin{equation*}
    \hat \psi_k (y) := \frac{1}{k} \sum_{t = 1}^k \1\{X_t = y\}
    \qquad (y \in S).
\end{equation*}
The value $\hat \psi_k (y)$ measures the fraction of time that the $P$-chain
spends in state $y$ over the time interval $1, \ldots, k$. Under ergodicity,
for fixed $y \in S$, we can use~\eqref{eq:erglln} with $h(x) = \1\{ x = y \}$
to obtain
\begin{equation}
    \label{eq:ergllnprob}
    \hat \psi_k (y) 
    \to \sum_{x \in S} \1\{x = y\} \psi^*(x) 
    = \psi^*(y).
\end{equation}
Turning \eqref{eq:ergllnprob} around, we see that
\begin{equation}\label{eq:erginterp}
    \psi^*(y) 
    \approx \text{ the fraction of time that any $P$-chain $(X_t)$ spends in state $y$}.
\end{equation}

Figure~\ref{f:benhabib_ergodicity_1} illustrates this idea for a simulated
Markov chain $(X_t)$ generated from the matrix $P_B$ introduced
in~\eqref{eq:piben}. The figure compares $\hat \psi_k$ and $\psi^*$ for
different values of $k$. As $k \to \infty$, the convergence claimed
in~\eqref{eq:ergllnprob} occurs.

\begin{figure}
   \centering
   \scalebox{0.6}{\includegraphics{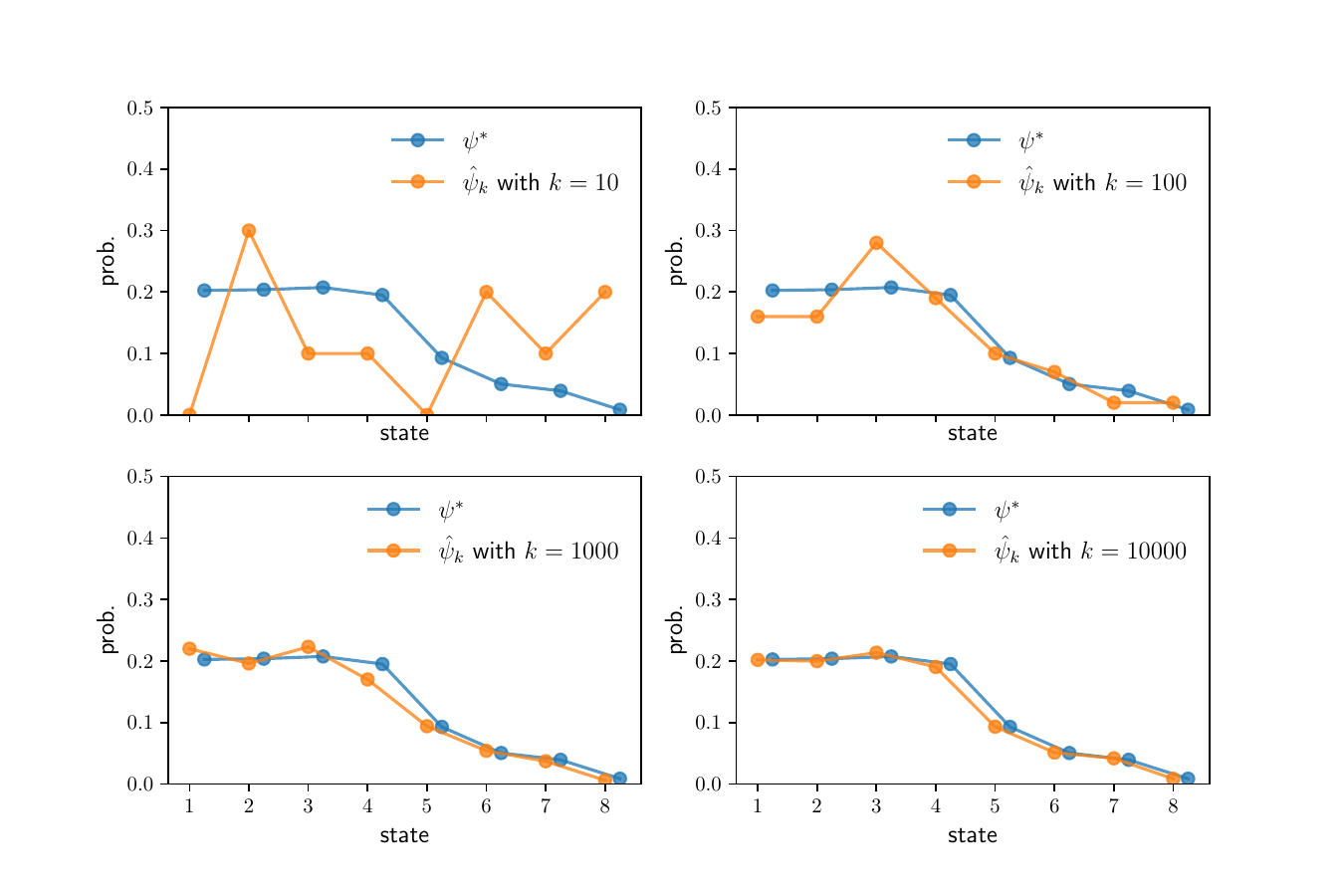}}
   \caption{\label{f:benhabib_ergodicity_1} Convergence of the empirical distribution to $\psi^*$}
\end{figure}

Of course, we must remember that \eqref{eq:erginterp} is only valid under
ergodicity.  For example, if $P = I$, the identity, then every distribution is
stationary, every $P$-Markov chain is constant, and \eqref{eq:erginterp} does
not hold.

Notice that, in view of Theorem~\ref{t:merg0}, the convergence
in~\eqref{eq:erglln} occurs for \emph{any} initial condition.  This gives us,
under ergodicity, a way of computing the stationary distribution via
simulation (and~\eqref{eq:ergllnprob}).  While this method is slower than
algebraic techniques (see, e.g., \S\ref{sss:mccomp}) for small problems, it
can be the only feasible option when $|S|$ is large.

\subsection{Aperiodicity and Stability}\label{ss:stabil}

In \S\ref{ss:ergo} we discussed sample path properties of Markov chains,
finding that strong connectivity is sufficient for stability of time series
averages.  Next we turn to marginal distributions of the chain and their long
run properties, which we examined informally in \S\ref{sss:trajlr}. We will
see that, to guarantee convergence of these sequences, we also need a
condition that governs periodicity.

\subsubsection{Convergence of Marginals}

Let $\mM$ be a finite Markov model with state space $S$ and transition matrix
$P$.  We call $\mM$ \navy{globally stable}\index{Globally stable} if there is
only one stationary distribution $\psi^*$ in $\dD(S)$ and, moreover, 
\begin{equation}\label{eq:finitemcgs}
    \lim_{t \to \infty} \psi P^t = \psi^* 
    \quad \text{for all } \psi \in \dD(S).
\end{equation}
In other words, marginal distributions of every $P$-Markov chain converge
to the unique stationary distribution of the model.

There is a useful connection between global stability and absorbing sets.
Intuitively, if $\mM$ is globally stable and has an absorbing set $A$ that can
be reached from any $x \in S$, then any Markov chain generated by these
dynamics will eventually arrive in $A$ and never leave.  Hence the stationary
distribution must put all of its mass on $A$.  The next exercise asks you to
confirm this.

\begin{Exercise}\label{ex:gsaa}
    Let $\mM$ be a finite Markov model with state space $S$ and transition matrix
    $P$. Let $A$ be a nonempty absorbing subset of $S$.  Assume that
    \begin{enumerate}
        \item[(a)] $\mM$ is globally stable with unique stationary distribution
            $\psi^*$ and
        \item[(b)] for each $x \in A^c := S \setminus A$, there exists an $a \in A$
            such that $a$ is accessible from $x$, 
    \end{enumerate}
    Show that, under the stated conditions, there exists an $\epsilon > 0$ such
    that $\sum_{y \in A^c} P^n(x, y) \} \leq 1 - \epsilon$ for all $x \in S$, where $n=|S|$.
    Using this fact, show in addition that $\sum_{y \in A} \psi^*(y) = 1$.  
\end{Exercise}

\begin{Answer}
    Let $\mM$ and $A$ have the stated properties. Fix $x \in S$.  Consider a
    $P$-chain $(X_t)$ that starts at $x$.  By property (b), there exists an $a \in
    A$ and a $k \leq n := |S|$ such that $\epsilon_x := \PP\{X_k = a\}$ is
    strictly positive.  Once $(X_t)$ enters $A$ it never leaves, so $X_k = a$
    implies $X_n \in A$.  Hence $\PP\{X_n \in A\} \geq \epsilon_x$.  With
    $\epsilon := \min_{x \in S} \epsilon_x > 0$, we then have $\PP\{X_n \in A\} \geq \epsilon$
    for any initial condition $x$. Another way to state this is $\sum_{y \in A^c}
    P^n(x, y) \} \leq 1 - \epsilon$ for all $x \in S$, which proves the first
    claim.

    Regarding the second claim, fix $\psi \in \dD(S)$ and let $\psi_t = \psi P^t$
    for all $t$.  Observe that, for fixed $m \in \NN$, we have $\psi_{(m+1)n} = P^n
    \psi_{mn}$.  As a result, for $y \in A^c$, we have
    \begin{equation*}
        \psi_{(m+1)n}(y) 
        = \sum_{x \in S} P^n(x,y) \psi_{mn}(x)
        = \sum_{x \in A^c} P^n(x,y) \psi_{mn}(x).
    \end{equation*}
    Summing over $y$ gives
    \begin{equation*}
        \sum_{y \in A^c} \psi_{(m+1)n}(y) 
        = \sum_{y \in A^c}\sum_{x \in A^c} P^n(x,y) \psi_{mn}(x)
        = \sum_{x \in A^c} 
            \left[ \sum_{y \in A^c} P^n(x,y) \right] \psi_{mn}(x).
    \end{equation*}
    Let $\eta_t := \sum_{x \in A^c} \psi_t(x)$ be the amount of
    probability mass on $A^c$ at time $t$.
    Using the first claim and the definition of $\eta_t$ now gives
    $\eta_{(m+1)n} \leq (1-\epsilon) \eta_{mn}$.  Hence $\eta_{mn} \to 0$ as $m
    \to \infty$.  At the same time, $\psi_{mn} \to \psi^*$ as $m \to \infty$,
    so 
    \begin{equation*}
        \sum_{x \in A^c} \psi^*(x)
        = \lim_{m \to \infty} \sum_{x \in A^c} \psi_{mn}(x)
        = \lim_{m \to \infty} \eta_{mn} = 0.
    \end{equation*}
    The second claim is now verified.
\end{Answer}

\subsubsection{A Key Theorem}

From Theorem~\ref{t:kb} we saw that strong connectedness is sufficient for
uniqueness of the stationary distribution. One might hope that strong
connectedness is enough for global stability too, but this is not true.  To
see why, suppose for example that $S = \{0,1\}$ and $E = \{(0, 1), (1, 0)\}$.
The transition matrix is 
\begin{equation}
    \label{eq:perp}
     P_d =
     \begin{pmatrix}
       0 & 1 \\
       1 & 0
     \end{pmatrix}.
\end{equation}
While this model is strongly connected, global
stability fails.  Indeed, if $(X_t)$ is $P_d$-Markov and starts at $0$,
then $(X_t)$ will visit state 1 on odd dates and state 0 on even dates.
That is, $P^t_d \, \delta_0 = \delta_{t \text{ mod } 2}$.  This sequence
does not converge.

The issue with $P_d$ is that, even though the chain traverses the whole state
space, the distribution of $X_t$ will affect that of $X_{t+j}$ for all $j$ due
to periodicity.  This causes stability to fail. If, however, we rule out
periodicity, then we have enough for stability to hold.  This line of
reasoning leads to the following famous theorem.

\begin{theorem}\label{t:crmst}
    Let $\mM$ be a finite Markov model.  If $\mM$ is strongly connected and
    aperiodic, then $\mM$ is globally stable.
\end{theorem}

\begin{proof}
    Let $\mM$ be a finite Markov model with state $S$ and transition matrix
    $P$.  Suppose $\mM$ is aperiodic and strongly connected.  Let $\psi^*$ be
    the unique stationary distribution of $\mM$.   By
    Theorem~\ref{t:scaperpr}, $P$ is primitive.  Hence we can apply the last
    part of the Perron--Frobenius Theorem (see page~\pageref{t:pf}). Using
    $r(P)=1$, this result tells us that $P^t \to e_r e_\ell^\top$ as $t \to
    \infty$, where $e_r$ and $e_\ell$ are the unique right and left
    eigenvectors of $P$ corresponding to the eigenvalue $r(P) = 1$ that also
    satisfy the normalization $\inner{e_\ell, e_r}=1$.  

    Now observe that $\psi^*$ obeys $\psi^* = \psi^* P$ and, in addition, $P
    \1 = \1$.  Hence $\psi^*$ and $\1$ are left and right eigenvectors
    corresponding to $r(P)=1$.  Moreover, $\inner{\psi^*, \1} = 1$.  Hence
    $e_r = \1$ and $e_\ell^\top = \psi^*$.

    Combining these facts leads to 
    \begin{equation}\label{eq:fppmm}
        \lim_{t \to \infty}
        P^t = P^* 
        \quad \text{where} \quad P^* := \1 \psi^* .
    \end{equation}
    If we pick any $\psi \in \dD(S)$, then, by~\eqref{eq:fppmm}, we get
    $\psi P^t \to \inner{\psi, \1} \psi^* = \psi^*$.  Hence global
    stability holds.
\end{proof}

\begin{example}
    The Markov models $P_Q$ and $P_B$ in \S\ref{sss:reps} are both aperiodic
    by the results in \S\ref{sss:aperiodicgraph}.  Being strongly connected,
    they are also globally stable.
\end{example}

The aperiodicity condition in Theorem~\ref{t:crmst} is, in general, not
stringent.  On the other hand, the strong connectedness requirement is
arguably quite strict.  Weaker conditions for global stability
are available, as we shall see in~\S\ref{ss:mdc}.

\subsubsection{Rates of Convergence: Spectral Gap}\label{sss:specgap}

While global stability is a very useful property, the implications are
qualitative rather than quantitative.  In practice, we usually want to know
something about the \emph{rate} of convergence.  There are several ways of
looking at this.   We cover the two most common in this section and the next.

As a first step, fixing a Markov model $\mM$ with state space $S$ and
transition matrix $P$, we use \eqref{eq:aspecrep} on page~\pageref{eq:aspecrep}
to write $P^t$ as 
\begin{equation}\label{eq:decmp}
    P^t = \sum_{i=1}^{n-1} \lambda_i^t e_i \epsilon_i^\top + \1 \psi^* ,
\end{equation}
where each $\lambda_i$ is an eigenvalue of $P$ and $e_i$ and $\epsilon_i$ are
the right and left eigenvectors corresponding to $\lambda_i$.  We have also
ordered the eigenvalues from smallest to largest, and used the
Perron--Frobenius Theorem to infer that $\lambda_n = r(P) = 1$, as well as the
arguments in the proof of Theorem~\ref{t:crmst} that showed $e_n = \1$ and
$\epsilon_n^\top = \psi^*$.

Premultiplying $P^t$ by arbitrary $\psi \in \dD(S)$ and rearranging now gives
\begin{equation}
    \psi P^t - \psi^* = \sum_{i=1}^{n-1} \lambda_i^t \psi e_i \epsilon_i^\top
\end{equation}
Recall that eigenvalues are ordered from smallest to largest and. Moreover, by the
Perron--Frobenius Theorem, $\lambda_i < 1$ for all $i < n$ when $P$ is
primitive (i.e., $\mM$ is strongly connected and aperiodic).  Hence, after taking
the Euclidean norm deviation, we obtain
\begin{equation}
    \| \psi P^t - \psi^* \| = O( \eta^t )
    \quad \text{where} \quad
    \eta := |\lambda_{n-1}| < 1.
\end{equation}
Thus, the rate of convergence is governed by the modulus of the
second largest eigenvalue.

The difference between the largest and second largest eigenvalue of a
nonnegative matrix is often called the \navy{spectral gap}\index{Spectral
gap}.  For this reason, we can also say that, for primitive stochastic matrices,
the rate of convergence is determined by the (nonzero) spectral gap.

\begin{example}
    When studying the worker model with hiring rate $\alpha$ and firing rate
    $\beta$ in \S\ref{sss:wdi}, we found that the eigenvalues of the
    transition matrix $P_w$ are $\lambda_1 = 1 - \alpha - \beta$ and
    $\lambda_2 = 1$.  Hence the spectral gap is $\alpha + \beta$ and the rate
    of convergence is $O((1-\alpha-\beta)^t)$.  High hiring and firing rates
    both produce faster convergence.  In essence, this is because higher
    hiring and firing rates mean workers do not stay in any state for long, so
    initial conditions die out faster.
\end{example}

\subsection{The Markov--Dobrushin Coefficient}\label{ss:mdc}

The rate of convergence of $\psi P^t$ to $\psi^*$ given in \S\ref{sss:specgap}
restricts the Euclidean distance between these vectors as a function of $t$.
There are, however, other metrics we could use in studying rates of
convergence, and sometimes these other metrics give more convenient results.
In fact, as we show in this section, a good choice of metric leads us to a
more general stability result than the (better known) aperiodicity-based
result in \S\ref{ss:stabil}.

\subsubsection{An Alternative Metric}

For the purposes of this section, 
For $\phi, \psi \in \dD(S)$, we set
\begin{equation*}
    \rho(\phi, \psi) 
    := \|\phi - \psi\|_1
    := \sum_{x \in S} | \phi(x) - \psi(x) |,
\end{equation*}
which is just the $\ell_1$ deviation between $\phi$ and $\psi$ (see
\S\ref{sss:norms}).

\begin{Exercise}\label{ex:borho}
    Show that, for any $\phi, \psi \in \dD(S)$, we have 
    \begin{enumerate}
        \item $\rho(\phi, \psi) \leq 2$.
        \item $\rho(\phi P, \psi P) \leq \rho(\phi, \psi)$ for any stochastic
            matrix $P$
    \end{enumerate}
\end{Exercise}

\begin{Answer}
    Regarding part (i), the triangle inequality, combined with the assumption
    that $\phi, \psi \in \dD(S)$, gives the bound 
    $$
        \sum_x | \phi(x) - \psi(x) | \leq
        \sum_x |\phi(x)| + \sum_x |\psi(x)| \leq 2.
    $$
    Regarding part (ii), if $P$ is a stochastic matrix, then
    \begin{equation*}
        \rho(\phi P, \psi P)
        = \sum_y
            \left| 
                \sum_x P(x, y) \phi(x) - \sum_x P(x, y) \psi(x) 
            \right|
        \leq \sum_y \sum_x
            P(x, y)
            \left| 
                \phi(x) - \psi(x) 
            \right|.
    \end{equation*}
    Swapping the order of summation and using $\sum_y P(x, y) = 1$ proves the
    claim.
\end{Answer}

Property (ii) is called the \navy{nonexpansiveness property} of stochastic
matrices under the $\ell_1$ deviation.  As we will see, we can tighten this
bound when $P$ satisfies certain properties.

As a first step we note that, for the $\ell_1$ deviation, given any stochastic
matrix $P$, we have
\begin{equation}\label{eq:bfmd}
    \rho(\phi P, \psi P)
    \leq (1 - \alpha(P)) \rho(\phi, \psi),
\end{equation}
where 
\begin{equation}\label{eq:mdc}
    \alpha(P) 
    := \min
        \left\{
            \sum_{y \in S} [P(x,y) \wedge P(x', y)]
            \; :\;
            (x, x') \in S \times S
        \right\}.
\end{equation}
Here $a \wedge b := \min\{a, b\}$ when $a, b \in \RR$.  We call $\alpha(P)$
the \navy{Markov--Dobrushin coefficient}\index{Markov--Dobrushin
coefficient} of $P$, although other names are also used in the literature. A
proof of the bound in \eqref{eq:bfmd} can be found in
\cite{stachurski2022economic} or
\cite{seneta2006markov}.\footnote{\cite{seneta2006markov} also discusses the history of 
Andrey Markov's work, which originated the kinds of contraction based
arguments used in this section.}

\begin{Exercise}
    Consider the stochastic matrices
    \begin{equation*}
        P_w =
        \begin{pmatrix}
            1-\alpha & \alpha \\
            \beta & 1-\beta 
        \end{pmatrix}.
    \end{equation*}
    Show that $\alpha(P_w) = 0$ if and only if $\alpha=\beta=0$ or
    $\alpha=\beta=1$.
\end{Exercise}

How should the Markov--Dobrushin coefficient be understood?  The coefficient
is large when the rows of $P$ are relatively similar.  For example, if rows
$P(x, \cdot)$ and $P(x', \cdot)$ are identical, the $\sum_{y \in S} [P(x,y)
\wedge P(x', y)] = 1$.  Similarity of rows is related to stability.  The next
exercise helps to illustrate.

\begin{Exercise}
    Let $P$ be such that all rows are identical and equal to $\phi \in
    \dD(S)$.  Prove that global stability holds in one step, in the sense that
    $\phi$ is the unique stationary distribution and $\psi P = \phi$ for all
    $\psi \in \dD(S)$.
\end{Exercise}

\begin{Exercise}
    Using \eqref{eq:bfmd}, show that, for any $\phi, \psi \in \dD(S)$, we have
    \begin{equation}\label{eq:bfmdt}
        \rho(\phi P^t, \psi P^t)
        \leq (1 - \alpha(P))^t \rho(\phi, \psi)
        \quad \text{for all } \;
        t \in \NN.
    \end{equation}
\end{Exercise}

\begin{Answer}
    Fix $\phi, \psi \in \dD(S)$. From \eqref{eq:bfmd} we know that
    \eqref{eq:bfmdt} is true when $t=1$.  Now suppose it is true at $t$.  Then,
    using the fact that \eqref{eq:bfmd} holds for any pair of distributions,
    \begin{equation*}
        \rho(\phi P^{t+1}, \psi P^{t+1})
            \leq (1 - \alpha(P)) \rho(\phi P^t, \psi P^t)
            \leq (1 - \alpha(P))^{t+1} \rho(\phi, \psi)
    \end{equation*}
    where the last step uses the induction hypothesis. Hence \eqref{eq:bfmdt} also
    holds at $t+1$, and, by induction, at all $t \in \NN$.
\end{Answer}

Since powers of stochastic matrices are again stochastic, and since
\eqref{eq:bfmd} is valid for any stochastic matrix, the bound in
\eqref{eq:bfmdt} can be generalized by replacing $P$ with $P^k$ for any given
$k \in \NN$, which gives (with $t$ replaced by $\tau$)
\begin{equation*}
    \rho(\phi P^{\tau k}, \psi P^{\tau k})
        \leq (1 - \alpha(P^k))^\tau \rho(\phi, \psi)
        \quad \text{for all } \;
        \tau \in \NN.
\end{equation*}
Now observe that, for any $t \in \NN$, we can write $t = \tau k + j$, where
$\tau \in \ZZ_+$ and $j \in \{0, \ldots, k-1\}$.  Fixing $t$ and choosing $j$ to make
this equality hold, we get
\begin{equation*}
    \rho(\phi P^t, \psi P^t)
    = \rho(\phi P^{\tau k + j}, \psi P^{\tau k + j})
    \leq \rho(\phi P^{\tau k}, \psi P^{\tau k})
        \leq (1 - \alpha(P^k))^\tau \rho(\phi, \psi)
\end{equation*}
where the second inequality is due to the nonexpansive property of stochastic
matrices (Exercise~\ref{ex:borho}).  Since $\tau$ is an integer satisfying
$\tau = \lfloor t/k \rfloor$, where $\lfloor \cdot
\rfloor$ is the floor function, we can now state the following.

\begin{theorem}\label{t:mbbk}
    Let $\mM$ be a finite Markov model with state space $S$ and transition
    matrix $P$.  For all $\phi, \psi \in \dD(S)$ and all $k, t \in \NN$, 
    we have
    \begin{equation}\label{eq:bfmdt2}
        \rho(\phi P^t, \psi P^t)
        \leq (1 - \alpha(P^k))^{\lfloor t/k \rfloor} \rho(\phi, \psi).
    \end{equation}
    In particular, if there exists a $k \in \NN$ such that $\alpha(P^k) > 0$,
    then $\mM$ is globally stable.
\end{theorem}
 
To see why the global stability implication stated in Theorem~\ref{t:mbbk}
holds, suppose there exists a $k \in \NN$ such that $\alpha(P^k)>0$.
Now substitute arbitrary $\psi \in \dD(S)$ and any stationary distribution
$\psi^*$ for $P$ into \eqref{eq:bfmdt2} to obtain
\begin{equation}\label{eq:mdt}
    \rho(\psi P^t, \psi^*)
    \leq (1 - \alpha(P^k))^{\lfloor t/k \rfloor} \rho(\psi, \psi^*)
    \leq 2 (1 - \alpha(P^k))^{\lfloor t/k \rfloor} 
\end{equation}
for all $t \in \NN$, where the second bound is due to Exercise~\ref{ex:borho}.

\begin{Exercise}
    In the preceding discussion, the distribution $\psi^*$ was taken to be any
    stationary distribution of $P$.  Using \eqref{eq:mdt}, prove that $P$ has
    only one stationary distribution when $\alpha(P^k) > 0$.
\end{Exercise}

One of the major advantages of Theorem~\ref{t:mbbk} is that strong
connectivity of $\mM$ is not required.  In the next section we will see an
example of a finite Markov model $\mM$ where (i) strong connectivity fails but
(ii) the conditions of Theorem~\ref{t:mbbk} are satisfied.\footnote{It can be
    shown that the condition $\alpha(P^k) > 0$ for some $k \in \NN$ is
    necessary as well as sufficient for global stability of a finite Markov
    model. Hence the conditions of Theorem~\ref{t:mbbk} are strictly weaker
    than strong connectedness plus aperiodicity, as used in
    Theorem~\ref{t:kb}.}

\begin{Exercise}
    Let $\mM$ be a finite Markov model with state space $S$ and transition
    matrix $P$.  Prove the existence of a $k \in \NN$ with $\alpha(P^k) > 0$
    whenever $\mM$ is strongly connected and aperiodic.
\end{Exercise}
     
\begin{Answer}
    If $\mM$ is strongly connected and aperiodic, then $P$ is primitive, in
    which case there exists a $k \in \NN$ with $P^k \gg 0$.  Clearly
    $\alpha(P^k) > 0$.
\end{Answer}

\subsubsection{Sufficient Conditions}

While the Markov--Dobrushin coefficient $\alpha(P^k)$ can be calculated for any
given $k$ on a computer by stepping through all pairs of rows in $P^k$, this
calculation is computationally intensive when $S$ is large.  
When applicable, the following lemma simplifies the problem  by
providing a sufficient condition for $\alpha(P^k) > 0$.

\begin{lemma}\label{l:kdw}
    Let $k$ be a natural number and let $\mM$ be a finite Markov model with
    state space $S$ and transition matrix $P$.  If there is a state $z \in
    S$ such that, for every $x \in S$, there exists a directed walk from $x$
    to $z$ of length $k$, then $\alpha(P^k) > 0$.
\end{lemma}

\begin{proof}
    Let $k \in \NN$ and $z \in S$ be such that, for every $x \in S$, there
    exists a directed walk from $x$ to $z$ of length $k$.  By
    Proposition~\ref{p:accesspos}, we then have $r := \min_{x \in S} P^k(x, z) > 0$.
    Since, for any $x, x' \in S$,
    \begin{equation*}
            \sum_{y \in S} [ P^k(x,y) \wedge P^k(x', y) ]
            \geq P^k(x,z) \wedge P^k(x', z) \geq r > 0,
    \end{equation*}
    strict positivity of $\alpha(P^k)$ now follows from the
    definition of the Markov--Dobrushin coefficient.  
\end{proof}

\begin{example}
    Consider the digraph in Figure~\ref{f:io_reducible_r}.
    This digraph is not strongly connected because 4 is not accessible from
    anywhere.  However, there exists a directed walk from any vertex to vertex 1
    in $k=2$ steps.  For example, from 2 we can choose $2 \to 1$ and then $1 \to
    1$, from 1 we can choose $1 \to 1$ and then $1 \to 1$, etc. Hence, if
    Figure~\ref{f:io_reducible_r} is the digraph of a finite Markov model with
    transition matrix $P$, then $\alpha(P^2) > 0$.
\end{example}

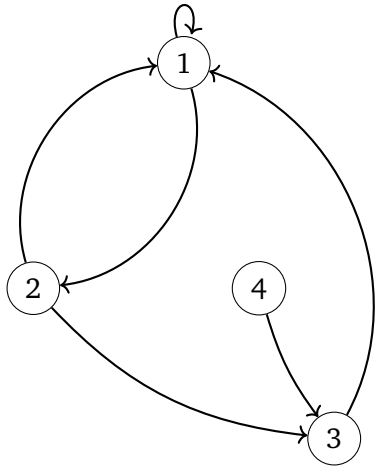
\begin{figure}
   \begin{center}
    \input{tikz/io_reducible_r.tex}
    \caption{\label{f:io_reducible_r} A digraph with walk of length 2 from any node to 1}
   \end{center}
\end{figure}

\begin{example}
    Consider the Markov dynamics suggested in Figure~\ref{f:poverty_trap} on
    page~\pageref{f:poverty_trap}.  Although there are no weights, we can see
    that \texttt{poor} is accessible from every state in one step, 
    so $\mM$ must be globally stable.  In addition, \texttt{poor} is, by
    itself, an absorbing set. Hence, by Exericse~\ref{ex:gsaa}, for any choice
    of weights compatible with these edges, the stationary distribution will
    concentrate all its mass on \texttt{poor}.
\end{example}

\begin{remark}
    As was pointed out in proof of Lemma~\ref{l:kdw}, under the conditions of that
    lemma we have $P^k(x, z) > 0$ for all $x \in S$.  This means that $\alpha(P^k)
    > 0$ whenever $P^k$ has a strictly positive column.
\end{remark}

\subsubsection{Application: PageRank}\label{sss:pagerank}

In \S\ref{ss:netcen} we discussed centrality measures for networks.
Centrality measures provide a ranking of vertices in the network according to
their ``centrality'' or ``importance.''  One of the most important
applications of ranking of vertices in a network is ranking the importance of
web pages on the internet. Historically, the most prominent example of a
ranking mechanism for the internet is PageRank, which transformed Google from
a minor start up to a technology behemoth.
In this section we provide a simple introduction to the original form of
PageRank and connect it to previously discussed measures of centrality.

Consider a finite collection of web pages $W$ and let $L$ be the hyperlinks
between them. We understand $(W, L)$ as a digraph $\gG$, where $W$ is the
vertices and $L$ is the edges.  Let $A$ be the associated adjacency matrix, so
that $A(u, v) = 1$ if there is a link from $u$ to $v$ and zero otherwise. We
set $n = |W|$, so that $A$ is $n \times n$.

To start our analysis, we consider the case where $\gG$ is strongly connected,
such as the small network in Figure~\ref{f:transitions}.  Furthermore, we introduce a second matrix
$P$ in which each row of $A$ has been normalized so that it sums to one.
For the network in Figure~\ref{f:transitions}, this means that
\begin{equation*}
    A = 
    \begin{pmatrix}
        0 & 1 & 1 & 1 \\
        0 & 0 & 1 & 0 \\
        0 & 0 & 0 & 1 \\
        1 & 0 & 0 & 0 
    \end{pmatrix}
    \quad \text{and} \quad
    P = 
    \begin{pmatrix}
        0 & 1/3 & 1/3 & 1/3 \\
        0 & 0 & 1 & 0 \\
        0 & 0 & 0 & 1 \\
        1 & 0 & 0 & 0 
    \end{pmatrix}.
\end{equation*}
Now consider an internet surfer who, once per minute, randomly clicks on one of
the $k$ outgoing links on a page, each link selected with uniform probability
$1/k$.  The idea of PageRank is to assign to each page $u \in W$ a value
$g(u)$ equal to the fraction of time that this surfer spends on page $u$ over
the long run.  Intuitively, a high value for $g(u)$ indicates a heavily
visited and hence important site.

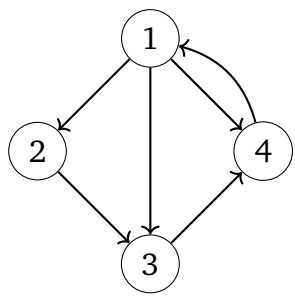
\begin{figure}
   \centering
    \input{tikz/transitions.tex}
    \caption{\label{f:transitions} A digraph with walk of length 2 from any node to 1}
\end{figure}

The vector $g$ is easy to compute, given our knowledge of Markov chains. Let
$\mM$ be the finite Markov model associated with the random surfer, with state
space $W$ and adjacency matrix given by $P$.  Since $\mM$ is strongly
connected, the ergodicity theorem (page~\pageref{t:merg0}) tells us that $P$
has a unique stationary distribution $\psi^*$, and that the fraction of time
the surfer spends at page $u$ is equal to the probability assigned to $u$
under the stationary distribution (see, in particular, \eqref{eq:erginterp}).
Hence $g = \psi^*$.

As $g$ is stationary and $r(P)=1$, we can write $g = (1/r(P)) g P$. Taking
transposes gives $g^\top = (1/r(P)) P^\top g^\top$.  Comparing with
\eqref{eq:eicena0} on page~\pageref{eq:eicena0}, we see that, for this simple
case, the PageRank vector $g$ is just the authority-based eigenvector
centrality measure of $\mM$.  Thus, PageRank gives high ranking to pages with
many inbound links, attaching high value to inbound links from pages that are
themselves highly ranked.

There are two problems with the preceding analysis.  First, we assumed that
the internet is strongly connected, which is clearly violated in practice (we
need only one page with no outbound links).  Second, internet users sometimes
select pages without using hyperlinks, by manually entering the URL.

The PageRank solution to this problem is to replace $P$ with the so-called
\navy{Google matrix}
\begin{equation*}
    G := \delta P + (1 - \delta) \frac{1}{n} \1,
\end{equation*}
where $\1$ is the $n \times n$ matrix of ones. 
The value $\delta \in (0, 1)$ is called the \navy{damping factor}.

\begin{Exercise}
    Prove that $G$ is a stochastic matrix for all $\delta \in (0, 1)$.
\end{Exercise}

The Markov dynamics embedded in the stochastic matrix $G$ can be understood as
follows: The surfer begins by flipping a coin with heads probability
$\delta$.  (See Exercise~\ref{ex:concond} and its solution for the connection
between convex combinations and coin flips.)  If the coin is heads then the
surfer randomly selects and follows one of the links on the current page. If
not then the surfer randomly selects and moves to any page on the internet.

For given $\delta$, the PageRank vector for this setting is adjusted to be the
stationary distribution of the Google matrix $G$.  

\begin{Exercise}\label{ex:goosc}
    Verify that the digraph associated with the transition probabilities in
    $G$ is always strongly connected (assuming, as above, that $\delta \in (0, 1)$).  
\end{Exercise}

\begin{Answer}
    All elements of $G$ are strictly positive, so a directed edge exists
    between every pair of pages $u, v \in W$.  This clearly implies strong
    connectedness.
\end{Answer}

As a result of Exercise~\ref{ex:goosc}, we can always interpret the stationary
of $G$ as telling us the fraction of time that the surfer spends on each page
in the long run.

\begin{Exercise}
    Use \eqref{eq:mdt} to obtain a rate of convergence of $\psi G^t$ to the
    adjusted PageRank vector $g^*$ (i.e., the unique stationary distribution $g^*$
    of $G$), where $\psi$ is an arbitrary initial distribution on $W$.  (Set
    $k=1$.)
\end{Exercise}

\begin{Answer}
	For any $u, u', v \in W$, we have $G(u, v) \wedge G(u', v) \geq 1-\delta$.  Hence
	$\alpha(G)  \geq 1-\delta$.  Therefore, by \eqref{eq:mdt}, we have
	\begin{equation*}
	    \rho(\psi G^t, g^*) \leq 2 \delta^t .
	\end{equation*}
\end{Answer}

\subsection{Information and Social Networks}

In recent years, the way that opinions spread across social networks has become a
major topic of concern in many countries around the globe.  A well-known mathematical model of this phenomenon is \navy{De Groot
learning}\index{De Groot learning}, which was originally proposed in
\cite{degroot1974reaching}.  This mechanism has linear properties that make it
relatively easy to analyze (although large and complex underlying networks can
cause significant challenges).

In De Groot learning, a group of agents, labeled from $1$ to $n$, is connected
by a social or information network of some type.  Connections are indicated by a
\navy{trust matrix}\index{Trust matrix} $T$, where, informally,
\begin{equation*}
    T(i, j) = \text{amount that $i$ trusts the opinion of $j$}.
\end{equation*}
In other words, $T(i,j)$ is large if agent $i$ puts a large positive weight on
the opinion of agent $j$.  The matrix $T$ is assumed to be stochastic.

We can view the trust matrix as an adjacency matrix for a weighed digraph
$\sS$ with vertex set $V := \natset{n}$ and edges
\begin{equation*}
    E = \{(i, j) \in V \times V \,:\, T(i,j) > 0\}.
\end{equation*}

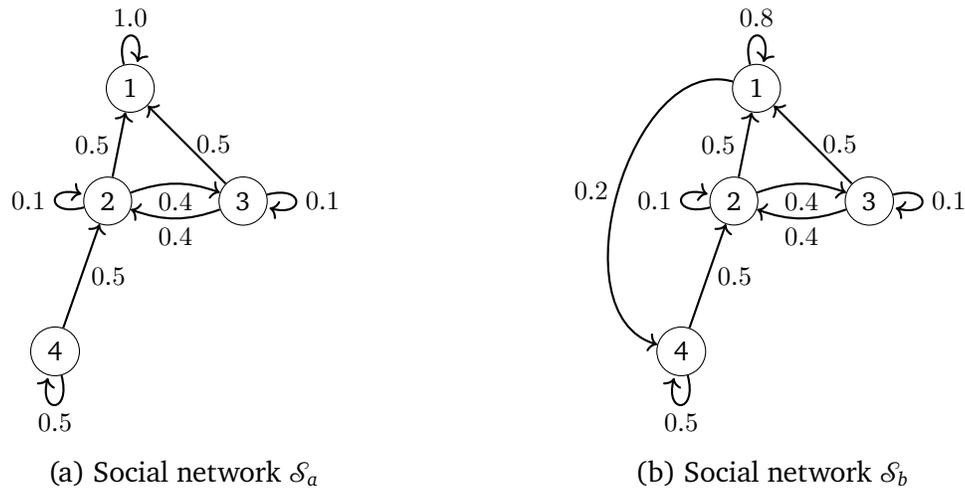
\begin{figure}
    \begin{subfigure}{.5\textwidth}
         \centering
         \input{tikz/degroot_1.tex}
         \caption{Social network $\sS_a$}
         \label{f:degroot_1}
    \end{subfigure}
    \begin{subfigure}{.5\textwidth}
         \centering
         \input{tikz/degroot_2.tex}
         \caption{Social network $\sS_b$}
         \label{f:degroot_2}
     \end{subfigure}
     \caption{Two social networks}
     \label{f:degroot}
\end{figure}

Figure~\ref{f:degroot} shows two social networks $\sS_a$ and $\sS_b$ with trust matrices given by 
\begin{equation*}
    T_a = 
    \begin{pmatrix}
        1 & 0 & 0 & 0 \\
        0.5 & 0.1 & 0.4 & 0 \\
        0.5 & 0.4 & 0.1 & 0 \\
        0 & 0.5 & 0 & 0.5 
    \end{pmatrix} 
    \quad \text{and} \quad
    T_b = 
    \begin{pmatrix}
        0.8 & 0 & 0 & 0.2 \\
        0.5 & 0.1 & 0.4 & 0 \\
        0.5 & 0.4 & 0.1 & 0 \\
        0 & 0.5 & 0 & 0.5 
    \end{pmatrix}
\end{equation*}
respectively.  In network A, agent 1 places no trust in anyone's opinion but
his own.  In network B, he places at least some trust in the opinion of agent
4. Below we show how these differences matter for the dynamics of beliefs.

\subsubsection{Learning}

At time zero, all agents have an initial subjective belief concerning the
validity of a given statement.  Belief takes values in $[0, 1]$, with 1
indicating complete (subjective) certainty that the statement is true.  Let
$b_0(i)$ be the belief of agent $i$ at time zero.  

An agent updates beliefs sequentially based on the beliefs of others, weighted by the amount
of trust placed in their opinion.  Specifically, agent $i$ updates her belief
after one unit of time to $\sum_{j =1}^n T(i, j) b_0(j)$.   More generally,
at time $t+1$, beliefs update to
\begin{equation}
    b_{t+1}(i) = \sum_{j =1}^n T(i, j) b_t(j)
    \qquad (i \in V).
\end{equation}
In matrix notation this is $b_{t+1} = T b_t$, where each $b_t$ is treated as a
column vector.

(Notice that this update rule is similar but not identical to the marginal
distribution updating rule for Markov chains (the forward equation) discussed
on page~\pageref{sss:lmarg}.  Here we are postmultiplying by a column vector
rather than premultiplying be a row vector.)

\begin{Exercise}
    If some subgroup of agents $U \subset V$ is an absorbing set
    for the digraph $\sS$, indicating that members of this group place no
    trust in outsiders, then the initial beliefs $\{b_0(i)\}_{i \in U^c}$ of
    the outsiders (members of $U^c = V \setminus U$) have no influence on the
    beliefs of the insiders (members of $U$) at any point in time.  Prove
    that this is true.
\end{Exercise}


\subsubsection{Consensus}

A social network $\sS$ is said to lead to \navy{consensus}\index{Consensus} if $|b_t(i) -
b_t(j)| \to 0$ as $t \to \infty$ for all $i, j \in V$. Consensus
implies that all agents eventually share the same belief.  An important
question is, what conditions on the network lead to a consensus outcome?

\begin{proposition}\label{p:consen}
    If there exists a $k \in \NN$ such that $\alpha(T^k) > 0$, then consensus
    is obtained.  In particular,
    \begin{equation}
        |b_t(i) - b_t(j)|
        \leq 2 (1 - \alpha(T^k))^{\lfloor t/k \rfloor}
        \quad \text{for all $t \in \NN$ and $i, j \in V$}.
    \end{equation}
\end{proposition}

\begin{proof}
    Fix $i,j \in V$ and $t \in \NN$. Let $b = b_0$.  For any $\phi, \psi \in
    \dD(V)$, an application of the triangle inequality gives
    \begin{equation*}
        | \phi T^t b - \psi T^t b|
        = \left|
                \sum_j (\phi T^t)(j) b(j)
                - \sum_j (\psi T^t)(j) b(j)
            \right| 
        \leq \sum_j
            \left|
                (\phi T^t)(j) 
                - (\psi T^t)(j) 
            \right| ,
    \end{equation*}
    where we have used the fact that $|b(j)| = b(j) \leq 1$.
    Applying the definition of the $\ell_1$ deviation and \eqref{eq:bfmdt2}, 
    we obtain the bound
    \begin{equation}\label{eq:tbbond}
        | \phi T^t b - \psi T^t b|
        \leq
        (1 - \alpha(T^k))^{\lfloor t/k \rfloor} \rho(\phi, \psi) 
        \leq 2  (1 - \alpha(T^k))^{\lfloor t/k \rfloor} .
    \end{equation}
    Since this bound is valid for any choice of $\phi, \psi \in \dD(V)$, we
    can specialize to $\phi = \delta_i$ and $\psi =
    \delta_j$ and, applying $T^t b = b_t$, get
    \begin{equation*}
        |b_t(i) - b_t(j)| 
        = | \delta_i T^t b - \delta_j T^t b| 
        \leq 2  (1 - \alpha(T^k))^{\lfloor t/k \rfloor} .
        \qedhere
    \end{equation*}
\end{proof}

Proposition~\ref{p:consen} can be applied to the two social networks $\sS_a$
and $\sS_b$ in Figure~\ref{f:degroot}.  For example, in network $\sS_a$, for
every node $i$, there exists a walk of length 2 from $i$ to node $1$.
Hence, by Lemma~\ref{l:kdw}, we have $\alpha(T^2) > 0$.  In network $\sS_b$
the same is true.  

\begin{Exercise}
    Let $\sS$ be a social network with trust (and adjacency) matrix $T$.  Use
    Proposition~\ref{p:consen} to show that $\sS$ leads to consensus whenever
    $\sS$ is strongly connected and aperiodic.
\end{Exercise}

\begin{Answer}
    If $\sS$ is strongly connected and aperiodic, then the adjacency matrix $T$ is
    primitive, so there exists a $k \in \NN$ such that $T^k \gg 0$.  Hence
    $\alpha(T^k) > 0$, and Proposition~\ref{p:consen} applies.
\end{Answer}

\subsubsection{Influence of Authorities}

Now let's consider what beliefs converge to when consensus emerges. In
particular, we are interested in discovering whose opinions are most
influential under De Groot learning, for a given trust matrix.

To answer the question, let $\sS$ be a given social network with trust matrix
$T$.   Suppose that $\alpha(T^k) > 0$ for some $k \in \NN$.  By
Theorem~\ref{t:mbbk}, the network $\sS$ is globally stable when viewed as a
finite Markov model.  Let $\psi^*$ be the unique stationary distribution, so
that $\psi^* = \psi^* T$.

Applying \eqref{eq:tbbond} with $\phi = \delta_i$ and $\psi = \psi^*$ yields
\begin{equation*}
    \left| 
        b_t(i) - b^* 
    \right|
    \leq
    (1 - \alpha(T^k))^{\lfloor t/k \rfloor} \rho(\phi, \psi) 
    \leq 2  (1 - \alpha(T^k))^{\lfloor t/k \rfloor} 
\end{equation*}
where 
\begin{equation*}
    b^* := \sum_{j \in V} \psi^*(j)b_0(j) .
\end{equation*}
We conclude that the belief of every agent converges geometrically to $b^*$,
which is a weighted average of the initial beliefs of all agents. In
particular, those agents with high weighting under the stationary distribution
have a large influence on these equilibrium beliefs.

We can interpret this through notions of centrality.  Since $r(T)=1$, we have
$(\psi^*)^\top = (1/r(T)) T^\top (\psi^*)^\top$, so $\psi^*$ is the
authority-based eigenvector centrality measure on $\sS$.  Thus, the influence
of each agent on long run beliefs is proportional to their authority-based
eigenvector centrality.  This makes sense because such agents are
highly trusted by many agents who are themselves highly
trusted.

\begin{Exercise}
    For network $\sS_a$ in Figure~\ref{f:degroot}, show that $b^* = b_0(1)$.
    That is, all agents' beliefs converge to the belief of agent 1.
\end{Exercise}

\begin{Exercise}
    Using a computer, show that the stationary distribution for $\sS_b$,
    rounded to two decimal places, is $\psi^* = (0.56, 0.15, 0.07, 0.22)$.  
    (Notice how the relatively slight change in network structure from $\sS_a$
    to $\sS_b$ substantially reduces the influence of agent 1.)
\end{Exercise}

\section{Chapter Notes}

High quality treatments of finite-state Markov dynamics include
\cite{norris1998markov}, \cite{privault2013understanding} and
\cite{haggstrom2002finite}.   For the general state case see
\cite{meyn2009markov}.

A review of De Groot learning is available in \cite{jackson2010social}.  Some
interesting extensions related to the ``wisdom of crowds'' phenomenon are
provided in \cite{golub2010naive}.  \cite{acemoglu2021misinformation} study
misinformation and echo chambers in information networks.
\cite{board2021learning} analyze learning dynamics in continuous time on large
social networks.  \cite{shiller2020narrative} provides an interesting
discussion of how ideas spread across social networks and shape economic
outcomes.

%% file: tikz/quah_graph.tex
\begin{tikzpicture}
  \node[ellipse, draw] (1) at (0, 0) {1};
  \node[ellipse, draw] (2) at (2, 0) {2};
  \node[ellipse, draw] (3) at (4, 0) {3};
  \node[ellipse, draw] (4) at (6, 0) {4};
  \node[ellipse, draw] (5) at (8, 0) {5};
  \draw[->, thick, black]
  (1) edge [bend left=20, above] node {$0.03$} (2)
  (2) edge [bend left=20, below] node {$0.05$} (1)
  (2) edge [bend left=20, above] node {$0.03$} (3)
  (3) edge [bend left=20, below] node {$0.04$} (2)
  (3) edge [bend left=20, above] node {$0.04$} (4)
  (4) edge [bend left=20, below] node {$0.04$} (3)
  (4) edge [bend left=20, above] node {$0.02$} (5)
  (5) edge [bend left=20, below] node {$0.01$} (4)
  (1) edge [loop above] node {$0.97$} (1)
  (2) edge [loop above] node {$0.92$}  (2) 
  (3) edge [loop above] node {$0.92$}  (3)
  (4) edge [loop above] node {$0.94$}  (4)
  (5) edge [loop above] node {$0.99$}  (5);
\end{tikzpicture}

%% file: tikz/quah_graph_2.tex
\begin{tikzpicture}
  \node[ellipse, draw] (1) at (0, 0) {1};
  \node[ellipse, draw] (2) at (2, 0) {2};
  \node[ellipse, draw] (3) at (4, 0) {3};
  \node[ellipse, draw] (4) at (6, 0) {4};
  \node[ellipse, draw] (5) at (8, 0) {5};
  \draw[->, thick, black]
  (1) edge [bend left=20, above] node {$0.03$} (2)
  (2) edge [bend left=20, below] node {$0.05$} (1)
  (2) edge [bend left=20, above] node {$0.06$} (3)
  (3) edge [bend left=20, below] node {$0.05$} (2)
  (3) edge [bend left=20, above] node {$0.05$} (4)
  (4) edge [bend left=20, below] node {$0.06$} (3)
  (4) edge [bend left=20, above] node {$0.04$} (5)
  (5) edge [bend left=20, below] node {$0.01$} (4)
  (1) edge [loop above] node {$0.97$} (1)
  (2) edge [loop above] node {$0.89$}  (2) 
  (3) edge [loop above] node {$0.90$}  (3)
  (4) edge [loop above] node {$0.90$}  (4)
  (5) edge [loop above] node {$0.99$}  (5);
\end{tikzpicture}

%% file: tikz/poverty_trap_2.tex
\begin{tikzpicture}
  \node[ellipse, draw] (0) at (0, 2) {middle class};
  \node[ellipse, draw] (1) at (3, 2) {rich};
  \node[ellipse, draw] (2) at (0, 0) {poor};
  \draw[->, thick, black]
  (0) edge [bend left=10, above] node {} (1)
  (1) edge [bend left=10, below] node {} (0)
  (0) edge [loop above] node {} (0);
  \draw[->, thick, black]
  (1) edge [loop above] node {} (1);
  \draw[->, thick, black]
  (2) edge [loop above] node {} (2);
\end{tikzpicture}

%% file: tikz/io_reducible_r.tex
{\small
\begin{tikzpicture}
  \node[circle, draw] (1) at (-1, 3) {1};
  \node[circle, draw] (2) at (-3, 0) {2};
  \node[circle, draw] (3) at (1, -2) {3};
  \node[circle, draw] (4) at (0, 0) {4};
  \draw[->, thick, black]
  (1) edge [bend left=50, right]  (2)
  (2) edge [bend left=50, left]  (1)
  (2) edge [bend right=20, below]  (3)
  (3) edge [bend right=50, right]  (1)
  (4) edge [bend right=10, left]  (3)
  (1) edge [loop above]  (1);
\end{tikzpicture}
}

%% file: tikz/transitions.tex
\begin{tikzpicture}
  \node[circle, draw] (1) at (0, 1.5) {1};
  \node[circle, draw] (2) at (-1.5, 0) {2};
  \node[circle, draw] (3) at (0, -1.5) {3};
  \node[circle, draw] (4) at (1.5, 0) {4};
  \draw[->, thick, black]
  (1) edge [bend left=0, above] node {} (2)
  (1) edge [bend left=0, above] node {} (3)
  (1) edge [bend left=0, below] node {} (4)
  (2) edge [bend left=0, below] node {} (3)
  (3) edge [bend left=0, below] node {} (4)
  (4) edge [bend right=30, above] node {} (1);
\end{tikzpicture}

%% file: tikz/degroot_1.tex
{\footnotesize
\begin{tikzpicture}
  \node[circle, draw] (1) at (0, 1.5) {1};
  \node[circle, draw] (2) at (-0.3, 0) {2};
  \node[circle, draw] (3) at (1.5, 0) {3};
  \node[circle, draw] (4) at (-1, -2) {4};
  \draw[->, thick, black]
  (4) edge [bend left=0, right] node {$0.5$} (2)
  (2) edge [bend left=20, below] node {$0.4$} (3)
  (3) edge [bend left=20, below] node {$0.4$} (2)
  (2) edge [bend left=0, left] node {$0.5$} (1)
  (3) edge [bend left=0, right] node {$0.5$} (1)
  (4) edge [loop below] node {$0.5$} (4)
  (3) edge [loop right] node {$0.1$} (3)
  (2) edge [loop left] node {$0.1$} (2)
  (1) edge [loop above] node {$1.0$} (1);
\end{tikzpicture}
}

%% file: tikz/degroot_2.tex
{\footnotesize
\begin{tikzpicture}
  \node[circle, draw] (1) at (0, 1.5) {1};
  \node[circle, draw] (2) at (-0.3, 0) {2};
  \node[circle, draw] (3) at (1.5, 0) {3};
  \node[circle, draw] (4) at (-1, -2) {4};
  \draw[->, thick, black]
  (1) edge [bend right=90, left] node {$0.2$} (4)
  (4) edge [bend left=0, right] node {$0.5$} (2)
  (2) edge [bend left=20, below] node {$0.4$} (3)
  (3) edge [bend left=20, below] node {$0.4$} (2)
  (2) edge [bend left=0, left] node {$0.5$} (1)
  (3) edge [bend left=0, right] node {$0.5$} (1)
  (4) edge [loop below] node {$0.5$} (4)
  (3) edge [loop right] node {$0.1$} (3)
  (2) edge [loop left] node {$0.1$} (2)
  (1) edge [loop above] node {$0.8$} (1);
\end{tikzpicture}
}

%% file: ch_fpms.tex
\chapter{Nonlinear Interactions}\label{c:fpms}

Much of what makes network analysis interesting is how ramifications of choices flow across 
networks.  In general, decisions made at a given node $i$
affect responses of neighboring nodes and, through them, neighbors of
neighboring nodes, and so on.  As these consequences flow across the network,
they in turn affect choices at $i$.  This is a tail-chasing
scenario, which can be unraveled through fixed point theory.

In some network settings, such as the input-output model in \S\ref{ss:mutmod},
interactions are linear and fixed point problems reduce a
system of linear equations.  In other settings, interactions are inherently
nonlinear and, as a result, we need more sophisticated fixed point theory.

This chapter is dedicated to the study of networks with nonlinear
interactions.  We begin with relevant fixed point theory and then apply it to
a sequence of problems that arise in analysis of economic networks, including
production models with supply constraints and financial networks.

\section{Fixed Point Theory}

Let $S$ be any set.  Recall from \S\ref{sss:fpfd} that, given a self-map $G$
on $S$, a point $x \in S$ is called a fixed point\index{Fixed point} of $G$ if
$Gx = x$.  (A self-map\index{Self-map} on $S$ is a function $G$ from $S$ to itself.  When
working with self-maps it is common to abbreviate $G(x)$ to $Gx$.)  
In this chapter, we will say that $G$ is \navy{globally stable} on $S$ if $G$
has a unique fixed point $x^* \in S$ and $G^k x \to x^*$ as $k \to \infty$ for
all $x \in S$.  In other words, under this property, the fixed point is not
only unique but also globally attracting under iteration of $G$.

We have already discussed fixed points, indirectly or directly, in multiple contexts:
\begin{itemize}
    \item In Chapter~\ref{c:prod} we studied the equation 
        $x = A x + d$, where $x$ is an output vector, $A$ is a matrix of
        coefficients and $d$ is a demand vector.  A solution $x$ to this equation
        can also be thought of as a fixed point of the affine map $Fx = Ax + d$.
    \item In Chapter~\ref{c:mcs} we learned that a stationary distribution of
        a finite Markov model with state space $S$ and adjacency matrix $P$ is
        a $\psi \in \dD(S)$ with $\psi = \psi P$.  In other words, $\psi$ is a
        fixed point of $\psi \mapsto \psi P$ in $\dD(S)$.
    \item In Chapter~\ref{c:ofd} we studied the Bellman equation $q(x) =
        \min_{y \in \oO(x)} \{ c(x, y) + q(y) \}$ and introduced an
        operator the Bellman operator $T$, with the property that its fixed points 
        exactly coincide with solutions to the Bellman equation.
\end{itemize}

In each case, when we introduced these fixed point problems, we immediately
needed to consider questions of existence and uniqueness of fixed points.  Now we address
these same issues more systematically in an abstract setting.

\subsection{Contraction Mappings}\label{ss:bcmt}

In Chapter~\ref{c:prod} we studied solutions of the
system $x = Ax + b$ that are fixed points of the affine map $Fx = Ax + b$ on
$\RR^n$ studied in Example~\ref{eg:bcam}. The Neumann series lemma
on page~\pageref{t:nsl} is, in
essence, a statement about existence and uniqueness of fixed points of this
map. Here we investigate another fixed point theorem, due to Stefan Banach
(1892--1945), that can be thought of as extending the Neumann series lemma to
nonlinear systems.

\subsubsection{Contractions}

Let $S$ be a nonempty subset of $\RR^n$.  A self-map $F$ on $S$ is called
\navy{contracting}\index{Contracting} or a \navy{contraction of modulus
$\lambda$} on $S$ if there exists a $\lambda < 1$ and a norm $\| \cdot \|$ on
$\RR^n$ such that
\begin{equation}
    \label{eq:uc}
    \| Fu - Fv \| \leq \lambda \| u - v \| \quad \text{for all} \quad u, v \in S.
\end{equation}

\begin{Exercise}\label{ex:sciufp}
    Let $F$ be a contraction of modulus $\lambda$ on $S$.
    Show that 
    \begin{enumerate}
        \item $F$ is continuous on $S$ and
        \item $F$ has at most one fixed point on $S$. 
    \end{enumerate}
\end{Exercise}

\begin{example}\label{eg:bcam}
    Let $S = \RR^n$, paired with the Euclidean norm $\| \cdot \|$.  Let
    $Fx = Ax + b$, where $A \in \matset{n}{n}$ and $b \in \RR^n$.  If $\| A \|
    < 1$, where $\| \cdot \|$ is the operator norm on $\matset{n}{n}$, then 
    $F$ is a contraction of modulus $\| A \|$, since, for any $x, y \in S$,
    \begin{equation*}
        \| Ax + b - Ay - b \|
        = \| A(x - y) \| \leq \|A \| \| x - y \|.
    \end{equation*}
\end{example}

The next example uses a similar idea but is based on a different norm.

\begin{example}
    In~\eqref{eq:msirho} we studied the system $\rho = A^\top
    \rho - \epsilon$, where $A = (a_{ij})$ is a matrix of 
    coefficients satisfying $\sum_i a_{ij} = 1 - \alpha$ for some $\alpha \in
    (0,1)$, the vector $\epsilon$ is given and $\rho$ is unknown.
    Solutions can be viewed as a fixed points of the map $F \colon \RR^n
    \to \RR^n$ defined by $F p = A^\top p - \epsilon$.  Under the $\| \cdot
    \|_\infty$ norm, $F$ is a contraction of modulus $1-\alpha$ on $\RR^n$.
    Indeed,  for any $p, q \in \RR^n$,
    \begin{equation*}
        \| Fp - Fq \|_\infty
        = \max_j \left| \sum_{i=1}^n a_{ij} (p_i - q_i) \right|
        \leq \max_j \sum_{i=1}^n  a_{ij} \left|  p_i - q_i \right|.
    \end{equation*}
    Since $|p_i - q_i| \leq \| p - q \|_\infty$, we obtain
    \begin{equation*}
        \| Fp - Fq \|_\infty
        \leq \max_j \sum_{i=1}^n  a_{ij} \| p - q \|_\infty
        = (1- \alpha) \| p - q \|_\infty.
    \end{equation*}
\end{example}

Consider again Example~\ref{eg:bcam}. For the affine map $Fx = Ax + b$, the
condition $\| A \| < 1$ used to obtain contraction is stronger than the
condition $r(A) < 1$ used to obtain a unique fixed point in the Neumann series
lemma (see Exercise~\ref{ex:srsn} on page~\pageref{ex:srsn}).  Furthermore, the
Neumann series lemma provides a geometric series representation of the fixed
point. On the other hand, as we now show, the contraction property can be used
to obtain unique fixed points when the map in question is not affine.

\subsubsection{Banach's Theorem}\label{sss:bcmts}

The fundamental importance of contractions stems from the following theorem.

\begin{theorem}[Banach's contraction mapping theorem]\label{t:bfpt} If $S$ is
    closed in $\RR^n$ and $F$ is a contraction of modulus $\lambda$ on $S$,
    then $F$ has a unique fixed point $u^*$ in $S$ and 
    \begin{equation}\label{eq:banachrate}
        \| F^n u - u^* \| \leq \lambda^n \| u - u^* \|
        \quad \text{for all } n \in \NN \text{ and } u \in S.
    \end{equation}
    In particular, $F$ is globally stable on $S$.
\end{theorem}

We complete a proof of Theorem~\ref{t:bfpt} in stages.

\begin{Exercise}\label{ex:bctqb}
    Let $S$ and $F$ have the properties stated in Theorem~\ref{t:bfpt}.  Fix
    $u_0 \in S$ and let $u_m := F^m u_0$.  Show that
    \begin{equation*}
        \| u_m - u_k \| \leq \sum_{i=m}^{k-1} \lambda^i \| u_0 - u_1 \|
    \end{equation*}
    holds for all $m, k \in \NN$ with $ m < k$.
\end{Exercise}

\begin{Exercise}\label{ex:bctic}
    Using the results in Exercise~\ref{ex:bctqb}, prove that $(u_m)$ is a
    Cauchy sequence in $\RR^n$ (see \S\ref{sss:ocon} for notes on the Cauchy
    property). 
\end{Exercise}

\begin{Answer}
    From the bound in Exercise~\ref{ex:bctqb}, we obtain
    \begin{equation*}
        \| u_m - u_k \| \leq \frac{\lambda^m}{1 - \lambda} \lambda^i \| u_0 - u_1 \|
         \qquad (m,k \in \NN \text{ with } m < k).
    \end{equation*}
    Hence $(u_m)$ is Cauchy, as claimed.
\end{Answer}

\begin{Exercise}
    Using Exercise~\ref{ex:bctic}, argue that $(u_m) $ hence has a limit $u^*
    \in \RR^n$.  Prove that $u^* \in S$.
\end{Exercise}

\begin{proof}[Proof of Theorem~\ref{t:bfpt}]
    In the exercises we proved existence of a point $u^* \in S$ such
    that $F^m u \to u^*$.  The fact that $u^*$ is a fixed point of $F$ now follows
    from Lemma~\ref{l:clifp} on page~\pageref{l:clifp} and
    Exercise~\ref{ex:sciufp}.  Uniqueness is implied by
    Exercise~\ref{ex:sciufp}.  The bound \eqref{eq:banachrate} follows from
    iteration on the contraction inequality~\eqref{eq:uc} while setting
    $v=u^*$.
\end{proof}

\subsubsection{Eventual Contractions}

Let $S$ be a nonempty subset of $\RR^n$. A self-map $F$ on $S$ is called
\navy{eventually contracting}\index{Eventually contracting} if there exists a
$k \in \NN$ such that $F^k$ is a contraction on $S$.  Significantly, most of
the conclusions of Banach's theorem carry over to the case where $F$ is
eventually contracting.

\begin{theorem}\label{t:bfpt2}
    Let $F$ be a self-map on $S \subset \RR^n$.  If $S$ is closed and $F$ is eventually
    contracting, then $F$ is globally stable on $S$.
\end{theorem}

\begin{Exercise}
    Prove Theorem~\ref{t:bfpt2}.\footnote{Hint: Theorem~\ref{t:bfpt} is
    self-improving: it implies this seemingly stronger result.  The proof is
    not trivial but see if you can get it started.  You might like to note
    that $F^k$ has a unique fixed point $u^*$ in $S$.  (Why?)
    Now consider the fact that $\| F u^*- u^* \| = \| F F^{nk} u^*- u^*\|$ for
    all $n \in \NN$.}
\end{Exercise}

\begin{Answer}
    Let $S$ be complete, let $F$ be a self-map on $S$ and let $F^k$ be a
    uniform contraction.  Let $u^*$ be the unique fixed point of $F^k$.  Fix
    $\epsilon > 0$.  We can choose $n$ such that $\|F^{nk} F u^* - u^* \| <
    \epsilon$.  But then 
    \begin{equation*}
        \| F u^*- u^* \| = \| F F^{nk} u^*- u^*\| = \|F^{nk} F u^* - u^* \|  < \epsilon.
    \end{equation*}
    Since $\epsilon$ was arbitrary we have $\| F u^*- u^*\| = 0$, implying
    that $u^*$ is a fixed point of $F$.  

    Regarding convergence, fix $u \in S$.  Given $n \in \NN$, there exist
    integers $j(n)$ and $i(n)$ such that $n = j(n) k + i(n)$, and $j(n) \to
    \infty$ as $n \to \infty$.  Hence
    \begin{equation*}
        \| F^n u - u^* \|
        = \| F^{j(n)k + i(n)} u - u^* \|
        = \| F^{j(n)k} F^{i(n)} u - u^* \| \to 0
        \qquad (n \to \infty),
    \end{equation*}
    by the assumptions on $F^k$.  Convergence implies uniqueness of the fixed
    point (why?).
\end{Answer}

There is a close connection between Theorem~\ref{t:bfpt2} and the Neumann
series lemma (NSL).  If $S=\RR^n$ and $F x = A x + b$ with $r(A) < 1$, then 
the NSL implies a unique fixed point.  We can also obtain this result from 
Theorem~\ref{t:bfpt2}, since, for any $k \in \NN$,
\begin{equation*}
    \| F^k x - F^k y \|
    = \| A^k x - A^k y \|
    = \| A^k (x - y) \| 
    \leq \|A^k \| \| x - y \|.
\end{equation*}
As $r(A) < 1$, we can choose $k$ such that $\| A^k \| < 1$
(see \S\ref{sss:mnspec}).  Hence $F$ is eventually contracting and
Theorem~\ref{t:bfpt2} applies. 

As mentioned above, contractions and eventual contractions have much wider
scope than the NSL, since they can also be applied in nonlinear settings. At
the same time, the NSL is preferred when its conditions hold, since it also
gives inverse and power series representations of the fixed point.

\subsubsection{A Condition for Eventual Contractions}

The result below provides a useful test for the eventual
contraction property. (In the statement, the absolute value of a vector is
defined pointwise, as in \S\ref{sss:poov}.)

\begin{proposition}\label{p:cec}
    Let $F$ be a self-map on $S \subset \RR^n$ such that, for some $n \times
    n$ matrix $A$,
    \begin{equation*}
        |F x - F y| \leq A | x - y |
        \quad \text{for all } \, x, y \in S.
    \end{equation*}
    If, in addition, $A \geq 0$ and $r(A) <
    1$, then $F$ is eventually contracting on $S$ with respect to the Euclidean
    norm.
\end{proposition}

\begin{proof}
    Our first claim is that, under the conditions of the proposition, 
    \begin{equation}\label{eq:fk}
        |F^k x - F^k y| \leq A^k | x - y |
        \quad \text{for all $k \in \NN$ and $x, y \in S$}.  
    \end{equation}
    This is true at $k=1$ by
    assumption.  If it is true at $k-1$, then
    \begin{equation}
        |F^k x - F^k y| 
        \leq A | F^{k-1} x - F^{k-1} y |
        \leq A A^{k-1} | x -  y |,
    \end{equation}
    where the second inequality uses the induction hypothesis and $A \geq 0$
    (so that $u \leq v$ implies $Au \leq Av$).  Hence \eqref{eq:fk} holds.

    It follows from the definition of the Euclidean norm that $\| |u| \|
    = \| u\|$ for any vector $u$.  Also, for the same norm, 
    $|u| \leq |v|$ implies $\| u \| \leq \| v\|$.  Hence, 
        for all $k \in \NN$ and $x, y \in S$,  
    \begin{equation*}
        \| F^k x - F^k y \| 
        \leq \| A^k |x - y| \|
        \leq \| A^k \|_o \| x - y \|.
    \end{equation*}
    In the second inequality, we used $\| \cdot \|_o$ for the operator norm,
    combined with the fact that $\| A u \| \leq \|A\|_o \|u\|$ always holds, as
    discussed in \S\ref{sss:onorm}.

    By Gelfand's lemma (see in particular Exercise~\ref{ex:igfo}
    on page~\pageref{ex:igfo}), we obtain existence of a $\lambda \in (0,1)$
    and a $k \in \NN$ with $\|A^k\|_o \leq \lambda < 1$.  Hence, for this $k$,
    \begin{equation*}
        \| F^k x - F^k y \| \leq \lambda \|x - y\|.
    \end{equation*}
    Since $\lambda$ does not depend on $x$ or $y$, we have shown that $F$ is
    an eventual contraction on $S$ with respect to the Euclidean norm.
\end{proof}

\subsection{Shortest Paths Revisited} 

Consider again the shortest path problem introduced in \S\ref{ss:shp}. One
modification that sometimes appears in applications is the addition of
discounting during travel between vertices.  For example, if the vertices are
international ports and travel takes place by sea, then port-to-port travel
time is measured in weeks or even months. It is natural to apply time
discounting to future costs associated with that travel, to implement the idea
that paying a given dollar amount in the future is preferable to paying it
now.

\begin{Exercise}
    Suppose it is possible to borrow and lend risk-free at a positive interest
    rate $r$.  Explain why it is always preferable to have \$100 now than \$100 in
    a year's time in this setting.
\end{Exercise}

\begin{Answer}
    If the risk-free real interest rate $r$ is positive, then \$100 received
    now can be converted with probability one into $(1+r) 100 > 100$ dollars in one year.
\end{Answer}

Recall that, without discounting, the Bellman equation for the shortest path
problem takes the form $q(x) = \min_{y \in \oO(x)} \{ c(x, y) + q(y) \}$
for all $x \in V$, where $c$ is the cost function $V$ is the set of vertices
and $q$ is a candidate for the cost-to-go function.   We showed that the
minimum cost-to-go function $q^*$ satisfies the Bellman equation and is the
unique fixed point of the Bellman operator.  

The Bellman equation neatly divides the problem into current costs, embedded
in the term $c(x, y)$, and future costs embedded in $q(y)$.  To add
discounting, we need only discount $q(y)$.  We do this by multiplying it by a
\navy{discount factor} $\beta \in (0, 1)$.  The Bellman equation is then $q(x)
= \min_{y \in \oO(x)} \{ c(x, y) + \beta q(y) \}$ for all $x \in V$ and the
Bellman operator is 
\begin{equation}\label{eq:bespb}
     Tq(x) = \min_{y \in \oO(x)} \{ c(x, y) + \beta q(y) \}
     \qquad (x \in V).
\end{equation}

In \S\ref{ss:belm}, without discounting, we had to work hard to show that the
Bellman operator has a unique fixed point in $U$, the set of all $q \colon
V \to \RR_+$ with $q(d)=0$.  With discounting the proof is easier, since
we can leverage the Banach contraction mapping theorem.

In what follows, we identify the vertices in $V$ with integers $1, \ldots, n$,
where $d$ is identified with $n$.  We then understand $U$ as all nonnegative
vectors $q$ in $\RR^n$ with $q(n) = 0$.  (We continue to write $q(x)$ for the
$x$-th element of the vector $q$, but now $x$ is in $\natset{n}$.)

\begin{Exercise}
    Prove that $U$ is a closed subset of $\RR^n$.
\end{Exercise}

\begin{Answer}
    Take $q_k \to q$ where $(q_k)$ is a sequence of $n$-vectors contained in $U$.
    By Exercise~\ref{ex:clic} on page~\pageref{ex:clic}, since $q_k \geq 0$
    for all $k$, we must have $q \geq 0$. It remains only to show that
    $q(n)=0$.  As $q_k \in U$ for all $k$, we have $q_k(n)=0$ for all $k$.
    By Lemma~\ref{l:eqconvec}, we also have $q_k(n) \to q(n)$.  Hence
    $q(n)=0$.
\end{Answer}

\begin{Exercise}
    Prove that $T$ is order-preserving on $U$ with respect to the pointwise order.
\end{Exercise}

\begin{Answer}
    We need to show that if $p, q \in U$ and $p \leq q$, then $Tp \leq Tq$.
    This follows easily from the definition of $T$ in \eqref{eq:bespb}.
\end{Answer}

\begin{Exercise}\label{ex:tqa}
    Prove that, for any $q \in U$ and $\alpha \in \RR_+$, we have $T(q +
    \alpha \1) = Tq + \beta \alpha \1$.
\end{Exercise}

\begin{Answer}
    Fix $q \in U$ $\alpha \in \RR_+$ and $x \in V$.  By definition, 
    \begin{equation*}
         T(q + \alpha \1) (x) 
         = \min_{y \in \oO(x)} \{ c(x, y) + \beta q(y) + \alpha \beta \}
         = Tq(x) + \alpha \beta.
    \end{equation*}
    Hence $T(q + \alpha \1) = Tq + \beta \alpha \1$ as claimed.
\end{Answer}

Now let $\| \cdot \|_\infty$ be the supremum norm on $\RR^n$ (see
\S\ref{ss:vecnorms}).  We claim that $T$ is a contraction on $U$ of modulus
$\beta$.  To see that this is so, fix $p, q \in U$ and observe that, pointwise,
\begin{equation*}
    Tq 
    = T(p + q - p)
    \leq T(p + \|q - p\|_\infty \1)
    \leq Tp + \beta \|q - p\|_\infty \1,
\end{equation*}
where the first inequality is by the order-preserving of $T$ and the second follows
from Exercise~\ref{ex:tqa}.  Hence
\begin{equation*}
    Tq - Tp \leq  \beta \|q - p\|_\infty \1.
\end{equation*}
Reversing the roles of $p$ and $q$ gives the reverse inequality.  Hence
\begin{equation*}
    |Tq (x) - Tp(x) | \leq  \beta \|q - p\|_\infty 
\end{equation*}
for all $x \in \natset{n}$.  Taking the maximum on the left hand side yields
$\|Tq - Tp \|_\infty \leq  \beta \|q - p\|_\infty$, which shows that $T$ is a
contraction of modulus $\beta$.  Hence Banach's theorem applies and a unique
fixed point exists.

\subsection{Supply Constraints} 

While the input-output model from \S\ref{sss:ioeq} has many useful
applications, its linear structure can be a liability.  One
natural objection to linearity is supply constraints: if sector $j$ doubles
its orders from sector $i$, we cannot always expect that sector $i$ will
quickly meet this jump in demand.

In this sector we investigate the impact of supply constraints on equilibrium.
These constraints introduce nonlinear relationships between nodes that
affect equilibria and make analysis more challenging.

\subsubsection{Production with Constraints}\label{sss:pwconst}

We recall from \S\ref{sss:ioeq} that $d_i$ is final demand for good $i$, $x_i$ is
total sales of sector $i$, $z_{ij} $ is inter-industry sales from sector $i$
to sector $j$, and $a_{ij} = z_{ij}/x_j $ is dollar value of inputs from $i$
per dollar output from $j$. 

Departing from our previous formulation of equilibrium in the input-output
model, suppose that, in the short run, the total output value of sector $i$ is
constrained by positive constant $\bar x(i)$.   Holding prices fixed (in the
short run), this means that sector $i$ has a capacity constraint in terms of
unit output.  For the purposes of our model, the vector of capacity
constraints $\bar x := (\bar x(i))_{i=1}^n$ can be any vector in $\RR^n_+$.

For each sector $i$, we modify the accounting identity \eqref{eq:sales} from
page~\pageref{eq:sales} to 
\begin{equation}\label{eq:csales}
    x_i = \min 
    \left\{ \sum_{j=1}^n o_{ij} + d_i, \; \bar x(i) \right\},
\end{equation}
where $o_{ij}$ is the value of orders from sector $i$ made by sector $j$.
Thus, if the capacity constraint in sector $i$ is not binding, then 
output is the sum of orders from other sectors and orders from final
consumers.  If $\bar x(i)$ is less than this sum, however, then sector $i$
produces to capacity $\bar x(i)$.

An equilibrium for this model is one where all orders are
met, subject to capacity constraints.  The fact that orders are met means that
$o_{ij} = z_{ij} = a_{ij} x_j$.  Substituting this equality into
\eqref{eq:csales} and rewriting as a vector equality, \eqref{eq:csales} can
equivalently be formulated as 
\begin{equation}\label{eq:cvsales}
    x = G x
    \quad \text{where} \quad
    G x := (A x + d) \wedge \bar x.
\end{equation}

The following exercise is key to solving the fixed point problem
\eqref{eq:cvsales}.

\begin{Exercise}\label{ex:gkgk}
    Prove that, for any $x, y \in \RR^n_+$ and $k \in \NN$, we have
    \begin{equation}\label{eq:gky}
        |G x - G y| \leq A |x - y|.
    \end{equation}
\end{Exercise}

\begin{Answer}
    Fix $x, y \in \RR^n_+$ and $k \in \NN$.  By the inequalities in
    \S\ref{sss:sets}, applied pointwise to vectors, we have
    \begin{equation*}
        |G x - G y|
        = | (A  x + d) \wedge \bar x - (A G y + d) \wedge \bar x|
        \leq | A x  + d - (A  y + d)|.
    \end{equation*}
    This proves the claim because, by Exercise~\ref{ex:bmk},
    \begin{equation*}
        | A x  + d - (A y + d)|
        = | A ( x-y)|
        \leq A |x-y|.
    \end{equation*}
\end{Answer}

We are now ready to prove existence of a unique fixed point under the
assumption that every sector has positive value added. 

\begin{proposition}\label{p:gac}
    If Assumption~\ref{a:pva} holds, then $G$ is globally stable in $\RR^n_+$.
    In particular, the constrained production model has
    a unique equilibrium $x^* \in \RR^n_+$.
\end{proposition}

\begin{proof}
    As shown in Exercise~\ref{ex:eara}, Assumption~\ref{a:pva} yields $r(A) <
    1$.  Moreover, $A \geq 0$.  Hence, by Exercise~\ref{ex:gkgk} and Proposition~\ref{p:cec}, $G$ is
    eventually contracting on $\RR_+$.  In consequence, a unique 
    equilibrium exists.
\end{proof}

\begin{remark}
    In Proposition~\ref{p:gac}, the weaker conditions on production discussed in
    \S\ref{sss:srsm} can be used in place of Assumption~\ref{a:pva}, which
    requires positive value added in every sector.  As explained in
    \S\ref{sss:srsm}, for $r(A) < 1$ it is enough that value added is
    nonnegative in each sector and, in addition, every sector has an upstream
    supplier with positive value added.
\end{remark}

\subsection{Fixed Points and Monotonicity}

Banach's fixed point theorem and its extensions are foundations of many
central results in pure and applied mathematics.  For our purposes, however,
we need to search a little further, since not all mappings generated by
network models have the contraction property.  In this section, we investigate
two fixed point results that drop contractivity in favor of monotonicity.

\subsubsection{Existence}

Without contractivity, one needs to work harder to obtain even existence of
fixed points, let alone uniqueness and convergence.  This is especially true
if the map in question fails to be continuous.  If, however, the map is order
preserving, then existence can often be obtained via some variation on the
Knaster--Tarski fixed point theorem.

Here we present a version of this existence result that is optimized to our
setting, while avoiding unnecessary excursions into order theory.
In stating the theorem, we recall that a \navy{closed order interval} in $\RR^n$
is a set of the form
\begin{equation*}
    [a, b] := \setntn{x \in \RR^n}{a \leq x \leq b}
\end{equation*}
where $a$ and $b$ are vectors in $\RR^n$.  Also, we call $(x_k) \subset \RR^n$
\navy{increasing} (resp., \navy{decreasing}) if $x_k \leq x_{k+1}$ (resp.,
$x_k \geq x_{k+1}$) for all $k$.

\begin{Exercise}\label{ex:opab}
    Let $[a, b]$ be a closed order interval in $\RR^n$ and let $G$ be an
    order-preserving self-map on $[a,b]$. Prove the following:
    \begin{enumerate}
        \item $(G^k a)$ is increasing and $(G^k b)$ is decreasing.
        \item If $x$ is a fixed point of $G$ in $[a,b]$, then $G^k a \leq x
            \leq G^k b$ for all $k \in \NN$.
    \end{enumerate}
\end{Exercise}

\begin{Answer}
    Since $G$ is a self-map on $[a,b]$, we have $Ga \in [a, b]$ and hence $a
    \leq G a$.  As $G$ is order preserving, applying $G$ to this inequality
    yields $G a \leq G^2 a$.  Continuing in this way (or using induction)
    proves that $(G^k a)$ is increasing.  The proof for $(G^k b)$ is similar.

    If $Gx = x$ for some $x \in [a,b]$, then, since $a \leq x$, we have $Ga
    \leq Gx = x$.  Iterating on this inequality gives $G^k a \leq x$ for all
    $k$.
\end{Answer}

For a self-map $G$ on $S \subset \RR^n$, we say that $x^*$ is a \navy{least
fixed point} (resp., \navy{greatest fixed point}) of $G$ on $S$ if $x^*$ is a
fixed point of $G$ in $S$ and $x^* \leq x$ (resp., $x \leq x^*$) for every
fixed point $x$ of $G$ in $S$.  Finally, we say that $G$ is 
\begin{itemize}
    \item \navy{continuous from below} if $x_k \uparrow x$ in $S$
        implies $G x_k \uparrow G x$ in $S$.
    \item \navy{continuous from above} if $x_k \downarrow x$ in $S$
        implies $G x_k \downarrow G x$ in $S$.
\end{itemize}
Here $x_k \uparrow x$ means that $(x_k)$ is increasing and $x_k \to x$. The
definition of $x_k \downarrow x$ is analogous.  In the next theorem, $S := [a,
b]$ is a closed order interval in $\RR^n$ and $G$ is a self-map on $S$.  

\begin{theorem}\label{t:ktk}
    If $G$ is order-preserving on $S$, then $G$ has 
    a least fixed point $x^*$ and a greatest fixed point $x^{**}$ in $S$.
    Moreover, 
    \begin{enumerate}
        \item if $G$ is continuous from below, then $G^k a \uparrow x^*$
            and
        \item if $G$ is continuous from above, then $G^k b \downarrow
            x^{**}$.
    \end{enumerate}
\end{theorem}

\begin{remark}
    As alluded to above, fixed point results for order preserving maps can be
    obtained in more general settings than the ones used in
    Theorem~\ref{t:ktk} (see, e.g., \cite{davey2002introduction}, Theorem~2.35).  Theorem~\ref{t:ktk}
    is sufficient for our purposes, given our focus on finite networks.  
\end{remark}

\begin{proof}[Proof of Theorem~\ref{t:ktk}]
    Under the stated conditions, existence of least and greatest fixed points
    $x^*$ and $x^{**}$ follow from the Knaster--Tarski fixed point theorem.
    (This holds because $[a, b]$ is a complete lattice.  For a definition of complete
    lattices and a proof of the Knaster--Tarski theorem, see, e.g.,
    \cite{davey2002introduction}.)

    Regarding claim (i), suppose that
    $G$ is continuous from below and consider the sequence $(x_k) := (G^k a)_{k
        \geq 1}$.  Since $G$ is order-preserving (and applying
        Exercise~\ref{ex:opab}), this sequence is increasing and bounded above
        by $x^*$.  Since bounded monotone sequences in $\RR$ converge, each individual
        component of the vector sequence $x^k$ converges in $\RR$.  Hence, by
        Lemma~\ref{l:eqconvec}, the vector sequence $x^k$ converges in $\RR^n$ to some
        $\bar x \in [a, x^*]$.  Finally, by continuity from below, we have
    \begin{equation*}
        \bar x 
        = \lim_k G^k a
        = \lim_k G^{k+1} a
        = G \lim_k G^k a
        = G \bar x,
    \end{equation*}
    so that $\bar x$ is a fixed point.

    We have now shown that $(G^k a)$ converges up to a fixed point $\bar x$ of
    $G$ satisfying $\bar x \leq x^*$.  Since $x^*$ is the least fixed point of
    $G$ in $S$, we also have $x^* \leq \bar x$.  Hence $\bar x = x^*$.

    The proof of claim (ii) is similar to that of claim (i) and hence omitted.
\end{proof}

\begin{remark}
    In the preceding theorem, $x^*$ and $x^{**}$ can be equal, in which
    case $G$ has only one fixed point in $S$.
\end{remark}

\begin{Exercise}\label{ex:cppe}
    Consider the map $G x = (A x + d) \wedge \bar x$ from the constrained
    production model.  In \S\ref{sss:pwconst}, we showed that $G$ has a
    unique fixed point in $\RR^n_+$ when $r(A) < 1$.
    Show now that $G$ has at least one fixed point in $\RR^n_+$, even when
    $r(A)< 1$ fails.  (Continue to assume that $A \geq 0$, $d \geq 0$ and
    $\bar x \geq 0$.)
\end{Exercise}

\begin{Answer}
    Clearly $G$ is a self-map on $S := [0, \bar x]$.  Since $A \geq 0$, we have $G x
    \leq G y$ for all $x, y \in S$.  From this it follows easily that $G$ is
    order-preserving.  Theorem~\ref{t:ktk} now guarantees
    existence of at least one fixed point.
\end{Answer}

\subsubsection{Du's Theorem}

Theorem~\ref{t:ktk} is useful because of its relatively weak assumptions. At
the same time, it fails to deliver uniqueness.  Hence its conclusions are
considerably weaker than the results we obtained from contractivity
assumptions in \S\ref{ss:bcmt}.

In order to recover uniqueness without imposing contractivity, we now consider
order-preserving maps that have additional shape properties.  In doing so, we
use the definition of concave and convex functions in \S\ref{sss:conconfun}.

\begin{theorem}[Du]\label{t:du}
    Let $G$ be an order-preserving self-map on order interval $S = [a, b]
    \subset \RR^n$.  In this setting, if either
    \begin{enumerate}
        \item $G$ is concave and $G a \gg a$ or
        \item $G$ is convex and $G b \ll b$,
    \end{enumerate}
    then $G$ is globally stable on $S$.
\end{theorem}

A proof of Theorem~\ref{t:du} was obtained in an more abstract setting in
\cite{du1990fixed}.  Interested readers can consult that article for a
proof.

To illustrate how these results can be applied,
consider the constrained production model without assuming positive
value added, so that $r(A) < 1$ is not enforced.  
In Exercise~\ref{ex:cppe} we obtained existence.  With Theorem~\ref{t:du} in
hand, we can also show uniqueness whenever $d \gg 0$ and $\bar x \gg 0$.

Indeed, we have already seen that $G$ is a self-map on $S := [0, \bar x]$ and, 
when this last condition holds, we have $G 0 = d \wedge \bar x \gg 0$.
Hence the conclusions of Theorem~\ref{t:du} will hold if we can establish that
$G$ is concave.

\begin{Exercise}\label{ex:gic}
    Prove that $G$ is concave on $S$.  [Hint: Review \S\ref{sss:conconfun}.]
\end{Exercise}

\begin{Answer}
    By Exercise~\ref{ex:prescon}, the minimum of two concave functions is concave.
    Since $Fx = \bar x$ and $H x = Ax + b$ are both concave, the claim holds.
\end{Answer}

Here is a small extension of Du's theorem that will prove useful soon:

\begin{corollary}\label{c:duext}
    Let $G$ be an order-preserving self-map on $S = [a, b]$.
    If $G$ is concave and there exists an $\ell \in \NN$ such that $G^\ell a \gg
    a$, then $G$ is globally stable on $S$.
\end{corollary}

\begin{proof}
    Assume the conditions of Corollary~\ref{c:duext}.
    Since compositions of increasing concave operators are increasing and
    concave, Theorem~\ref{t:du} implies that $G^\ell$ is globally stable on $[a, b]$.
    Denote its fixed point by $\bar v$.
    Since $\{ G^m a \}_{m \in \NN}$ is increasing and since the subsequence
    $\{G^{m\ell} a\}_{m \in \NN}$ converges up to $\bar v$ as $m \to \infty$, we must have
    $G^m a \to \bar v$.  A similar argument gives $G^m b \to \bar
    v$.  For any $v \in [a, b]$ we have $G^m a \leq G^m v \leq G^m b$, so $G^m v
    \to \bar v$ as $m \to \infty$.

    The last step is to show that $\bar v$ is the unique fixed point of $G$.
    From Theorem~\ref{t:ktk}, we know that at least one fixed point
    exists.  Now suppose $v \in [a, b]$ is such a point.  Then $v = G^m v$ for
    all $m$.  At the same time, $G^m v \to \bar v$ by the results just
    established.  Hence $v = \bar v$.  The proof is now complete.
\end{proof}

\section{Financial Networks}\label{s:finnet}

Given the long history of crises in financial markets around the globe,
economists and business analysts have developed many tools for assessing the
credit-worthiness of banks and other financial institutions. After the major
financial crises of 2007-2008, originating in the subprime market in the US
and the sudden collapse of Lehmann Brothers, it became clear that the
financial health of individual institutions cannot be assessed in isolation.
Rather, it is essential to analyze solvency and credit-worthiness in terms of
the entire network of claims and liabilities within a highly interconnected
financial system.   In this section, we review financial crises and apply network
analysis to study how they evolve.

\subsection{Contagion}

Some financial crises have obvious causes external to the banking sector. A
prominent example is the hyperinflation that occurred in Weimar Germany around
1921--1923, which was driven by mass printing of bank notes to meet war
reparations imposed under the Treaty of Versailles.  Here the monetary
authority played the central role, while the actions of private banks were
more passive.

Other crises seem to form within the financial sector itself, driven by
interactions between banks, hedge funds, and asset markets.  In many cases,
the crisis follows a boom, where asset prices rise and economic growth is
strong.   Typically, the seeds of the crisis are laid during this boom phase,
when banks extend loans and firms raise capital on the basis of progressively
more speculative business plans.  At some point it becomes clear to investors
that these businesses will fail to meet expectations, leading to a rush for the
exit.  

The last phase of this cycle is painful for the financial sector, since
rapidly falling asset values force banks and other financial institutions to
generate short-term capital by liquidating long-term loans, typically with
large losses, as well as selling assets in the face of falling prices,
hoarding cash and refusing to roll over or extend short term loans to other
institutions in the financial sector.  The financial crisis of 2007--2008
provides a textbook example of these dynamics.

One key aspect of the financial crisis of 2007--2008, as well as other similar
crises, is \navy{contagion}, which refers to the way that financial stress
spreads across a network of financial institutions.  If one institution becomes
stressed, that stress will often spread to investors or counterparties to
which this institution is indebted.  The result of this process is not easy to
predict, since, like equilibrium in the production networks studied in
Chapter~\ref{c:prod}, there is a tail chasing problem: stress spreads from
institution $A$ to institutions $B$, $C$ and $D$, which may in turn increase
stress on $A$, and so on.

In this section we study financial contagion, beginning with a now-standard
model of default cascades.

\subsection{Default Cascades}\label{ss:defcas}

Default cascades are a form of financial contagion, in which default by a node
in a network leads to default by some of its counterparties, which then
spreads across the network.  Below we present a model of default cascades
and analyze its equilibria.

\subsubsection{Network Valuation}

Consider a financial network $\gG = (V, E, w)$, where $V = \natset{n}$ is a
list of $n$ financial institutions called \emph{banks}, with an edge $(i, j)
\in E$ indicating that $j$ has extended credit to $i$.  The size of that loan
is $w(i, j)$.  Thus, an edge points in the direction of a liability, as in
Figure~\ref{f:network_liabfin_trans} on
page~\pageref{f:network_liabfin_trans}: edge $(i, j)$ indicates a liability
for $i$ and an asset for $j$.  As in \S\ref{ss:uwdg},
the set of all direct predecessors of $i \in V$ will be
written as $\iI(i)$, while the set of all direct successors will be denoted
$\oO(i)$.

Banks in the network have both internal and external liabilities, as well as
internal and external assets.  Internal (i.e., interbank) liabilities and
assets are given by the weight function $w$, in the sense that $w(i,j)$ is a
liability for $i$, equal to the size of its loan from $j$, and also an asset
for $j$.  Positive weights indicate the presence of counterparty risk: when
$j$ holds an asset of book value $w(i, j)$, whether or not the loan is repaid
in full depends on the stress placed on bank $i$ and rules that govern
repayment in the event of insolvency.

We use the following notation for the primitives of the model: $x_i := \sum_{j
\in \oO(i)} w(i, j)$ is total interbank liabilities of bank $i$,
\begin{equation}\label{eq:defpi}
    \Pi_{ij} := 
    \begin{cases}
        w(i, j) / x_i & \text{ if } x_i > 0 \\
        0 & \text{ otherwise}
    \end{cases}
\end{equation}
is the matrix of relative interbank liabilities, $a_i$ is external assets held
by bank $i$ and $d_i$ is external liabilities.

When considering the interbank assets of bank $j$, we need to distinguish
between the book value $\sum_{i \in \iI(j)} w(i, j)$ of its claims on other
banks and the realized value in the face of partial or complete default by its
counterparties $\iI(j)$.  To this end, we introduce a clearing vector $p
\in \RR^n_+$, which is a list of proposed payments by each bank in the
network. In particular, $p_i$ is total payments made by bank $i$ to its
counterparties within the banking sector.  Under the choice of a particular
clearing vector, the actual payments received by bank $j$ on its internal loan
portfolio are $\sum_{i \in V} p_i \Pi_{ij}$.

The last statement is an assumption about the legal framework for the
banking sector.  It means that the actual payment $p_i \Pi_{ij}$ from $i$ to
$j$ is proportional to the amount that $i$ owes $j$, relative to its total
interbank obligations.  The idea is that all counterparties in the banking
sector have equal seniority, so that residual funds are spread across
claimants according to the relative size of the claims.

Let $\hat p_j$ be the amount of funds bank $j$ makes available to repay all of its
debts, both interbank and external.  This quantity is 
\begin{equation}\label{eq:qj}
   \hat p_j =
   \min
   \left\{
       a_j + \sum_{i \in V} p_i \Pi_{ij}, \; d_j + x_j
   \right\}.
\end{equation}
The right hand term inside the $\min$ operator is total debts of bank $j$.
The left hand side is the amount on hand to repay those debts, including
external assets and repayments by other banks.  The bank repays up
to---but not beyond---its ability to pay.

External liabilities are assumed to be senior to interbank liabilities, which
means that for bank $j$ we also have
\begin{equation}\label{eq:pj}
    p_j = \max\{\hat p_j - d_j, \; 0\}.
\end{equation}
Thus, interbank payments by $j$ are a remainder after external debts are
settled.  If these debts exceed the bank's ability to pay, the bank becomes
insolvent and pays nothing to internal creditors.  This is a form of limited
liability.

Combining \eqref{eq:qj} and \eqref{eq:pj} and rearranging slightly yields
\begin{equation*}
    p_j = \max
    \left\{
       \min
       \left\{
           a_j - d_j + \sum_{i \in V} p_i \Pi_{ij}, \; x_j
       \right\}, \; 0
    \right\}
\end{equation*}
Now let's take $p$, $a$, $d$ and $x$ as row vectors in $\RR^n$ and write this
collection of equations, indexed of $j$, in vector form.  With $\max$ and
$\min$ taken pointwise, and using the symbols $\vee$ and $\wedge$ for max and
min, we get 
\begin{equation}
    p = ((a - d + p\Pi) \wedge x) \vee 0.
\end{equation}
A solution to this equation is called an \navy{equilibrium clearing vector}
for the banking system.

\begin{remark}
    An equilibrium clearing vector captures impacts of contagion within
    the specified banking system, in
    the sense that it traces out the full network effects of interbank lending
    within the model.  We shall study this equilibrium, while also
    recognizing that the model
    is restrictive in the sense that it assumes a specific form for seniority
    and implicitly rules out some kinds of
    nonlinear phenomena.  We return to this
    theme in \S\ref{ss:echs}.
\end{remark}

\subsubsection{Existence and Uniqueness of Fixed Points}

In order to analyze existence and uniqueness of equilibria, we introduce the
operator $T \colon \RR^n \to \RR^n$ defined by 
\begin{equation}
    Tp = ((e + p\Pi) \wedge x) \vee 0,
\end{equation}
where $e := a - d$ represents net external assets.  Evidently $p
\in \RR^n_+$ is an equilibrium clearing vector if and only if it is a fixed
point of $T$.

\begin{Exercise}
    Prove that the operator $T$ is continuous on $\RR^n$.
\end{Exercise}

\begin{Answer}
    In essence, this holds because compositions of continuous functions are
    continuous.  Nonetheless, here is a more explicit proof:
    Recalling the inequalities in \S\ref{sss:sets}, applied pointwise to vectors, 
    we have, for any $p, q \in \RR^n$,
    \begin{equation*}
        |Tp - Tq|
         \leq | (e + p\Pi) \wedge x) - (e + q\Pi) \wedge x) |
         \leq | (p - q) \Pi |.
    \end{equation*}
    Since, for the Euclidean norm, $\| |u| \| = \| u\|$ and $u, v \geq 0$ with $u
    \leq v$ implies $\| u \| \leq \| v\|$, we then have $ \|Tp - Tq\| \leq \| (p -
    q) \Pi \| \leq \| p - q \| \| \Pi \|_o$.  It follows easily that if $\| p_n -
    p \| \to 0$, then $ \|Tp_n - Tp\| \to 0$ also holds.
\end{Answer}

Using this operator, establishing existence of at least one equilibrium
clearing vector is not problematic for this model, regardless of the values of
the primitives and configuration of the network:

\begin{Exercise}
    Show that the banking model described above always has at least one
    equilibrium clearing vector.  What else can you say about equilibria in
    this general case?
\end{Exercise}

\begin{Answer}
    Since $\Pi \geq 0$, we always have $p \Pi \leq q \Pi$ whenever $p \leq q$.  As
    a result, $p \mapsto a - d + p \Pi$ is order-preserving and hence so is $T$.
    Moreover, for any $p \in S := [0, x]$, we have $Tp \in S$.  Hence, by
    Theorem~\ref{t:ktk}, $T$ has a fixed point in $S$.

    What else can we say about equilibria in this setting?
    By the same theorem, $T$ has a least fixed point $p^*$ and a greatest fixed
    point $p^{**}$ in $S$.  Moreover, since $T$ is continuous, $T^k 0 \uparrow
    p^*$ and $T^k x \downarrow p^{**}$.
\end{Answer}

While existence is automatic in this model, uniqueness is not:

\begin{example}
    If $n=2$, $e=(0, 0)$ and $(1, 1)$ and $\Pi = \begin{pmatrix} 0 & 1 \\ 1 & 0
    \end{pmatrix}$, then $Tp = p$ is equivalent to
    \begin{equation*}
        \begin{pmatrix}
           p_1 \\
           p_2
        \end{pmatrix}
         = 
            \begin{pmatrix}
                p_2 \\
                p_1
            \end{pmatrix}
            \wedge
            \begin{pmatrix}
                1 \\
                1
            \end{pmatrix}.
    \end{equation*}
    Both $p = (1, 1)$ and $p=(0, 0)$ solve this equation.
\end{example}

There are several approaches to proving uniqueness of fixed points of $T$.
Here is one:

\begin{Exercise}\label{ex:ufptr}
    Prove that $T$ is globally stable on $S := [0, x]$ whenever $r(\Pi) <
    1$.\footnote{Here's a hint, if you get stuck: show that $r(\Pi) < 1$
        implies $T$ is an eventual contraction on $[0, x]$ using one of the
        propositions presented in
        this chapter.}
\end{Exercise}

\begin{Answer}
    Fix $p, q \in S$.  Using the inequalities for min and max in \S\ref{sss:sets},
    applied pointwise, we have
    \begin{equation*}
        |Tp - Tq| 
        \leq |e + p \Pi - (e + q\Pi)|
        = |(p-q) \Pi|
        \leq |p-q| \Pi
    \end{equation*}
    After transposing both sides of this equation, we see that, when $r(\Pi)<1$, the conditions of 
    of Proposition~\ref{p:cec} hold.  The result follows.  (If you prefer, instead
    of taking transposes, just use the proof of Proposition~\ref{p:cec} directly,
    modified slightly to use the fact that we are operating on row vectors.)
\end{Answer}

This leads us to the following result:

\begin{proposition}\label{p:fnw}
    Let $\gG$ be a financial network.  If, for each bank $i \in V$, there
    exists a bank $j \in V$ with $i \to j$ and such that $j$ has zero
    interbank liabilities, then $T$ is globally stable on $S$ and $\gG$
    has a unique equilibrium clearing vector.
\end{proposition}

\begin{proof}
    By construction, the matrix $\Pi$ is substochastic.  Suppose that $\Pi$
    is also weakly chained substochastic.  Then, by Proposition~\ref{p:wcs},
    we have  $r(\Pi)<1$ and, by Exercise~\ref{ex:ufptr}, $T$ is globally
    stable on $S$.  Hence the proof will be complete if we can show
    that, under the stated conditions, $\Pi$ is weakly chained substochastic.

    To see that this is so, let $\gG$ be a financial network that
    satisfies the conditions of Proposition~\ref{p:fnw}. Now fix $i \in V$.
    We know that there exists a bank $j \in V$ with $i \to j$ and $j$ has no
    interbank liabilities.  If $w(i,j) > 0$, then $\Pi_{ij} > 0$, so $i \to j$
    under digraph $\gG$ implies $i \to j$ in the digraph induced by the
    substochastic matrix $\Pi$.  Also, since $j$ has no interbank liabilities,
    we have $x_j = 0$, and hence row $j$ of $\Pi$ is identically zero.  In
    particular, $\sum_k \Pi_{jk} = 0$.  Hence $\Pi$ is weakly chained
    substochastic, as required.
\end{proof}

\begin{remark}
    The proof of Proposition~\ref{p:fnw} also shows that $T$ is eventually
    contracting, so we can compute the unique fixed point by taking the limit
    of $T^k p$ for any choice of $p \in S$.
\end{remark}

There are several possible assumptions about the structure of the bank
network $\gG$ that imply the conditions of Proposition~\ref{p:fnw}.  For
example, it would be enough that $\gG$ is strongly connected and has at least
one bank with zero interbank liabilities.  Below we investigate another
sufficient condition, related to cyclicality.

A digraph $\dD$  is called a \navy{directed acyclic graph} if $\dD$ contains
no cycles.  

\begin{Exercise}\label{ex:dag}
    Prove the following:  If $\dD$ is a directed acyclic graph, then, for any
    node $i$ in $\dD$, there exists a node $j$ such that $i \to j$ and $\oO(j)=0$.
\end{Exercise}
    
\begin{Answer}
    Let $\dD$ be a directed acyclic graph and fix $i$ in $\dD$.
    Suppose to the contrary that every node reachable from $i$ has positive
    out-degree.  In this case, we can construct a walk from $i$ of arbitrary
    length.  But $\dD$ has only finitely many nodes, so any such walk must
    eventually cycle.  Contradiction.
\end{Answer}

\begin{Exercise}
    Let $\gG$ be a financial network.  Show that $\gG$ has a unique
    equilibrium clearing vector whenever $\gG$ is a directed acyclic graph.
\end{Exercise}

\begin{Answer}
    Let $\gG = (V, E, w)$ be a financial network and a directed acyclic graph.  By
    Proposition~\ref{p:fnw}, it suffices to show that, for each bank $i \in V$,
    there exists a bank $j \in V$ with $i \to j$ and such that $j$ has zero
    interbank liabilities.  To this end, fix $i \in V$.  By the directed acyclic
    graph property, we know that there exists a $j \in V$ with $i \to j$ and
    $\oO(j)=0$.  But if $\oO(j)=0$, then $j$ has no interbank liabilities.  The claim
    follows.
\end{Answer}

\subsubsection{Nonnegative External Equity}\label{sss:nee}

In this section we investigate the special case in which 
\begin{enumerate}
    \item[(E1)] each bank has nonzero interbank debt, so that $\Pi$ is a stochastic
        matrix, and
    \item[(E2)] external net assets are nonnegative, in the sense that $e = a -
        d \geq 0$.
\end{enumerate}
In view of (E1), we cannot hope to use Proposition~\ref{p:fnw}, since 
we always have $r(\Pi)=1$ when $\Pi$ is stochastic.  Nonetheless, we can
still obtain global stability under certain restrictions on $e$ and the topology of
the network.  Here is a relatively straightforward example, which we will
later try to refine:

\begin{Exercise}\label{ex:fnee}
    Let $\gG$ be a financial network such that (E1)--(E2) hold.  Using Du's
    theorem (page~\pageref{t:du}), prove that $T$ has a unique fixed point in
    $S := [0, x]$ whenever $e \gg 0$.  
\end{Exercise}

\begin{Answer}
    Let $\gG$ be a financial network such that (E1)--(E2) hold.  Since $e \gg
    0$, we have 
    \begin{equation*}
        p \in S 
        \quad \implies \quad
        Tp := ((e + p\Pi) \wedge x) \vee 0 = (e + p\Pi) \wedge x .
    \end{equation*}
    By an argument identical to that employed for Exercise~\ref{ex:gic} on
    page~\pageref{ex:gic}, $T$ is a concave operator on $S = [0, x]$.  Evidently
    $T$ is order preserving.  Finally, by (E1), we have $x_i > 0$ for all $i$, so
    $x \gg 0$ and hence $T0 = e \wedge x \gg 0$.  It now follows directly from
    Du's theorem that $T$ has a unique fixed point in $S$.
\end{Answer}

The condition $e \gg 0$ is rather strict.  Fortunately, it turns out that we
can obtain global stability under significantly 
weaker conditions.  To this end, we will say that node $j$ in a
financial system $\gG$, is \emph{cash accessible} if there exists an $i \in V$
such that $i \to j$ and $e(i) > 0$.  In other words, $j$ is downstream in the
liability chain from at least one bank with positive net assets outside of the
banking sector.

\begin{Exercise}\label{ex:enca}
    Prove the following result: If (E1)--(E2) hold and every node in $\gG$
    is cash-accessible, then $T^k 0 \gg 0$ for some $k \in \NN$.  [This is a
    relatively challenging exercise.]  
\end{Exercise}

\begin{Answer}
    Let $\gG$ be such that every node is cash-accessible.  Set
    \begin{equation*}
        \delta 
        := \frac{1}{n^2} 
        \cdot \min 
        \left\{
            \setntn{x_i}{i \in V} \cup \setntn{e_i}{i \in V \, \st e_i > 0}
        \right\}.
    \end{equation*}
    Let $\hat e$ be defined by $\hat e_i = 1$ if $e_i > 0$ and zero
    otherwise.  We claim that, for all $m \leq n$, 
    \begin{equation}\label{eq:phikprop}
        T^m 0 
        \geq \delta (\hat e + \hat e \Pi + \cdots + \hat e \Pi^{m-1}).
    \end{equation}
    This holds at $m=1$ because $T 0 = e
    \wedge x \geq \delta \hat e$.  Now suppose \eqref{eq:phikprop} holds at some $m \leq n
    - 1$.  Then, since $T$ is order-preserving, we obtain
    \begin{align*}
        T^{m+1} 0 
        & \geq (\delta (\hat e + \hat e \Pi + \cdots + \hat e \Pi^{m-1}) \Pi +
        e) \wedge x
        \\
        & \geq (\delta (\hat e + \hat e \Pi + \cdots + \hat e \Pi^m) ) \wedge
        x
    \end{align*}
    Since $\hat e + \hat e \Pi + \cdots + \hat e \Pi^m \leq n^2 \1$, where
    $\1$ is a vector of ones, and since $(\delta n^2 \1) \leq x$
    by the definition of $\delta$, we have 
    $T^{m+1} 0 \geq \delta (\hat e + \hat e \Pi + \cdots + \hat e \Pi^m)$.
    This argument confirms that \eqref{eq:phikprop} holds for all $m \leq n$.

    We now claim that $T^n 0 \gg 0$.  In view of~\eqref{eq:phikprop}, it suffices
    to show that, for any $j \in V$, there exists a $k < n$ with 
        $(\hat e \Pi^k) (j) = \sum_{i \in V} \hat e_i \Pi^k_{ij} > 0$.
    Since every node in $S$ is cash accessible, we know there exists an $i
    \in V$ with $e_i > 0$ and $j$ is accessible from $i$.  For this $i$ we can 
    choose $k \in \NN$ with $k < n$ and $\Pi^k_{ij} =  \hat e_i \Pi^k_{ij} > 0$.  
    We conclude that $T^n 0 \gg 0$, as claimed.
\end{Answer}

With the result from Exercise~\ref{ex:enca} in hand, the next lemma is easy to
establish.

\begin{lemma}\label{l:cfpc}
    If (E1)--(E2) hold and every node in $\gG$ is cash accessible, then $T$
    is globally stable and $\gG$ has a unique clearing vector $p^*\gg 0$.
\end{lemma}

\begin{proof}
    Let $\gG$ be as described. By Corollary~\ref{c:duext},
    it suffices to show that $T$ is an order-preserving concave self-map on
    $[0,x]$ with $T^k 0 \gg 0$ for some $k \in \NN$.
    The solution to Exercise~\ref{ex:fnee} shows that 
    $T$ is order-preserving and concave.  
    The existence of a $k \in \NN$ with $T^k 0 \gg 0$ was verified in
    Exercise~\ref{ex:enca}.
\end{proof}

Stronger results are available, with a small amount of effort.  In fact, under
(E1)--(E2), there is a strong sense in which uniqueness of the fixed point is
obtained without any further assumptions, provided that we rule out an
ambiguity related to what happens when $e=0$.  That ambiguity is discussed in
the next exercise and further references can be found in \S\ref{s:cnfp}.

\begin{Exercise}\label{ex:amben}
    Conditions (E1) and (E2) cannot by themselves pin down outcomes for the
    extreme case where every firm in the network has zero net external assets.
    Illustrate this with an example.\footnote{Hint: every stochastic matrix has at
    least one stationary distribution.}
\end{Exercise}

\begin{Answer}
    Let $\psi$ be a stationary distribution for $\Pi$.  Suppose $\lambda$ is a
    constant in $[0,1]$, $p = \lambda \psi$ and $x = \psi$. Then $e=0$ implies $T
    p = (e + \lambda \psi \Pi) \wedge x = (\lambda \psi) \wedge \psi = \lambda
    \psi = p$.  Since $\lambda$ was arbitrary in $[0,1]$, there is a continuum of
    equilibria.
\end{Answer}

Although Exercise~\ref{ex:amben} suggests ambiguity about outcomes when $e =
0$, it is natural to adopt the convention that the equilibrium clearing vector
$p^*$ obeys $p^*=0$ whenever $e=0$.  If the entire banking sector has no zero
net assets, then no positive payment sequence can be initiated (without
outside capital).

In the next exercise, we say that $U \subset V$ is accessible from $i \in V$
if there exists a $j \in U$ such that $j$ is accessible from $i$.

\begin{Exercise}
    Let $P$ be the set of all nodes in $V$ that are cash accessible.   Let $A$
    be all $i$ in $P^c$ such that $P$ is accessible from $i$.  Let $N$ be all
    $i$ in $P^c$ such that $P$ is not accessible from $i$.  Note that $V = P
    \cup A \cup N$ and that these sets are disjoint.  Show that $N$ and $P$
    are both absorbing sets.
\end{Exercise}

\begin{Answer}
    The set $N$ is an absorbing set, since, by definition, $P$ is not
    accessible from $N$, and  $A$ cannot be accessible because otherwise $P$
    would also be accessible.  The set $P$ is also absorbing because, if $j
    \in P^c$ is accessible from some $i \in P$, then $j$ is cash accessible.
    But then $j \in P$, which is a contradiction.
\end{Answer}

\subsection{Equity-Cross Holdings}\label{ss:echs}

In this section we analyze the model of default cascades constructed by
\citep{elliott2014financial}.  The model differs from the one studied in
\S\ref{ss:defcas} in several ways.  One is that financial institutions are
linked by share cross-holdings: firm $i$ owns a fraction $c_{ij}$ of firm $j$
for $i, j \in V := \natset{n}$.  This implies that failure of firm $j$ reduces
the market value of firm $i$, which in turn reduces the market value of other
firms, and so on.  

The second---and ultimately more significant---difference is the introduction
of failure costs that add significant nonlinearities to the model.
Failure costs reinforce the impact of each firm failure, leading to greater
shock propagation across the network.  This feature ties into the intuitive
idea that, when many firms are financially stressed, a failure by one firm can
trigger a wave of bankruptcies.

\subsubsection{Book and Market Value}

We now describe the features of the model.
Let $C = (c_{ij})_{i, j \in V}$ be the matrix of fractional cross-holdings,
as mentioned above, with $0 \leq c_{ij} \leq 1$ for all $i,j$.

\begin{assumption}\label{a:cv}
    The matrix of cross-holdings satisfies $\sum_k c_{ki} < 1$ for all $i \in V$.
\end{assumption}

Assumption~\ref{a:cv} implies that firms are not solely owned by other firms
in the network: investors outside the
network hold at least some fraction of firm $i$ for all $i$.

The \navy{book value}\index{Book value} of firm $i$ is given by 
\begin{equation}\label{eq:bvwo}
    b_i = e_i + \sum_j c_{ij} b_j
    \qquad (i \in V).
\end{equation}
Here the first term $e_i \geq 0$ is external assets of firm $i$ and the second
represents the value of firm $i$'s cross-holdings.  In vector form, the last
equation becomes $b = e + C b$.

\begin{Exercise}
    Let $I$ be the $n \times n$ identity.  Prove that, under
    Assumption~\ref{a:cv}, the matrix $I - C$ is invertible and 
    the equation $b = e + C b$ has the unique solution
    \begin{equation}\label{eq:bv}
        b = (I - C)^{-1} e.
    \end{equation}
    Prove also that $b \geq 0$.
\end{Exercise}

\begin{Answer}
    Since $C \geq 0$, Lemma~\ref{l:rscsbounds}, we have $r(C) \leq \max_j
    \csum_j(C)$.  In view of Assumption~\ref{a:cv}, this maximum is strictly less
    than one.  Hence $r(C) < 1$ and, by the Neumann series lemma $I-C$ is
    invertible and $(I-C)^{-1} = \sum_{k \geq 0} C^k$.  The last equality implies
    $b \geq 0$ when $b = (I - C)^{-1} e$.
\end{Answer}

\begin{Exercise}
    Provide a weaker condition on $C$ such that $I - C$ is invertible.
\end{Exercise}

\begin{Answer}
    Consider the condition $\sum_{i \in V} c_{ij} \leq 1$ for all $j \in V$, with
    strict inequality for at least one $j$.  Since $C \geq 0$, this implies that
    $C^\top$ is weakly chained substochastic, by Exercise~\ref{ex:awcs}.
    Hence $r(C) = r(C^\top)<1$, by Proposition~\ref{p:wcs}.  Hence $I-C$ is
    invertible, by the Neumann series lemma.
\end{Answer}

There is a widely used argument that cross-holdings artificially inflate the
value of firms, in the sense that the sum of book values of firms exceeds
$\sum_i e_i$, the sum of underlying equity values.   The next exercise
illustrates:

\begin{Exercise}
    Show that $e \gg 0$ and $\min_{i,j} c_{ij} > 0$ implies $\sum_i b_i > \sum_i e_i$.
\end{Exercise}

\begin{Answer}
    We have
    \begin{equation*}
        \sum_i b_i 
        = \1' b
        = \1' e + \1' C e + \1' C^2 e + \cdots
        \geq \1' e + \1' C e.
    \end{equation*}
    Hence it suffices to show that $\1' C e > 0$.  This will be true if at least
    one column of $Ce$ has a nonzero entry.  Since $e \gg 0$, we require only that
    $c_{ij} > 0$ for some $i,j$, which is true by assumption.
\end{Answer}

Due to this artificial inflation, we distinguish between the book value $b_i$
of a firm and its \navy{market value} $\bar v_i$, which is defined as $r_i
b_i$ with $r_i := 1 - \sum_{k} c_{ki}$.  The value $r_i$, which gives the
share of firm $i$ held by outsider investors, is strictly
positive for all $i$ by Assumption~\ref{a:cv}.  With $R := \diag(r_1, \ldots, r_n)$, we
can write the vector of market values as $\bar v = R b$.  Substituting in
\eqref{eq:bv} gives
\begin{equation}\label{eq:bv2}
        \bar v := A e ,
        \quad \text{where} \quad
        A : =R (I - C)^{-1}.
\end{equation}

\subsubsection{Failure Costs}

So far the model is very straightforward, with the market value of firms being
linear in external assets $e$.  However, since bankruptcy proceedings are
expensive, it is reasonable to assume that firm failures are costly.  Moreover,
when the market value of a firm falls significantly, the firm will experience
difficulty raising short term funds, and will often need to cease revenue
generating activities and sell illiquid assets well below their potential
value.

We now introduce failure costs.  As before, $\bar v_i$ is market value without
failure costs, as determined in \eqref{eq:bv2}, while $v_i$ will represent
market value in the presence of failure costs.  Failure costs for firm $i$ are
modeled as a threshold function
\begin{equation*}
    f(v_i) = \beta \1\{v_i < \theta \bar v_i\}
    \qquad (i = 1, \ldots, n),
\end{equation*}
where $\theta \in (0, 1)$ and $\beta > 0$ are parameters.
Thus, costs are zero when $v_i$ is large and $-\beta$ when they fall below the
threshold $\theta \bar v_i$.  In particular, a discrete failure cost of
$-\beta$ is incurred when firm value falls far enough below the no-failure market
value $\bar v_i$.  The larger is $\theta$, more prone firms are to failure.

The book value of firm $i$ without failure costs was defined in \eqref{eq:bvwo}.
The book value of firm $i$ in the presence of failure costs is defined as
\begin{equation*}
    b_i = e_i + \sum_j c_{ij} b_j - f(v_i).
\end{equation*}
Written in vector form, with $f$ applied pointwise to the vector $v$,
we get $b = e + C b - f (v)$.  Solving for $b$ gives $b = (I-C)^{-1} (e
- f (v))$.  The corresponding market value is 
\begin{equation}\label{eq:sfv}
    v = R b = A (e - f(v)).
\end{equation}
Notice that, when no firms fail, we have $v_i = \bar v_i$, as expected.

\subsubsection{Equilibria}

Equation~\eqref{eq:sfv} is a nonlinear equation in $n$ unknowns. A vector $v
\in \RR^n$ solves \eqref{eq:sfv} if and only if it is a fixed point of the
operator $T \colon \RR^n \to \RR^n$ defined by
\begin{equation}\label{eq:tsfv}
    Tv = A (e - f(v)).
\end{equation}
In what follows we set $d := A(e - \beta \1)$ and $S := [d, \bar v]$.

\vspace{.3em}

\begin{proposition}\label{p:ec}
    $T$ is a self-map on $S$ with least fixed
    point $v^*$ and greatest fixed point $v^{**}$ in $S$.  Moreover, 
    \begin{enumerate}
        \item the sequence $(T^k d)_{k \in \NN}$ converges up to $v^*$ in a finite number of
            steps and
        \item the sequence $(T^k \bar v)_{k \in \NN}$ converges down to $v^{**}$ in a
            finite number of steps.
    \end{enumerate}
\end{proposition}

Note that $v^* = v^{**}$ is a possibility, in which case $T$ has a unique
fixed point in $S$.  

\begin{Exercise}\label{ex:pfo}
    Prove the first claim in Proposition~\ref{p:ec}.
\end{Exercise}

\begin{Answer}
    Observe that $v_i \leq v'_i$ implies $f(v_i) \geq f(v'_i)$.  As
    a result, the vector-valued map $v \mapsto - f(v)$ is order-preserving,
    and hence so is $T$.  Moreover, for $v \in [d, \bar v]$, we have 
    \begin{equation*}
        d = A(e - \beta \1) \leq Tv := A (e - f(v)) \leq A e = \bar v,
    \end{equation*}
    so $T$ is a self-map on $[d , \bar v]$.  It follows directly from
    Theorem~\ref{t:ktk} that $T$ has a least and greatest fixed point in $S$.
\end{Answer}

Why does $(T^k d)_{k \in \NN}$ converge up to $v^*$ in a finite number of
steps?  First, as you saw in the solution to Exercise~\ref{ex:pfo}, the map
$T$ is an order-preserving self-map on $[d, \bar v]$, so $Td \in [d, \bar v]$.
In particular, $d \leq Td$.  Iterating on this inequality and using the
order-preserving property gives $d \leq Td \leq T^2 d \leq \cdots$, so $(T^k
d)$ is indeed increasing.  Moreover, the range of $T$ is a finite set,
corresponding to all vectors of the form
\begin{equation*}
    u = A(e - \beta w),
\end{equation*}
where $w$ is an $n$-vector containing only zeros and ones.  Finiteness holds
because there are only finitely many binary sequences of length $n$.

\begin{Exercise}
    Given the above facts, prove that $(T^k d)_{k \in \NN}$ converges up to
    $v^*$ in a finite number of steps.
\end{Exercise}

\begin{Answer}
    First, $v_k := T^k d$ is increasing, as just discussed.  Second, this sequence
    can take only finitely many values, since $T$ has finite range.  As $(v_k)$ is
    increasing, it cannot cycle, so it must converge in finitely many steps.
    Let $v'$ be this limiting value and let $K$ be the number of steps required
    for $(v_k)$ to attain $v'$.  Since $v_k = v'$ for all $k \geq K$, we have
    \begin{equation*}
        T v' = T T^K v' = T^{K+1} v' = v',
    \end{equation*}
    so $v'$ is a fixed point of $T$ in $S$.  Moreover, if $v''$ is any other fixed
    point of $T$ in $S$, then $d \leq v''$ and hence, by the order-preserving
    property, $v_k = T^k d \leq T^k v'' = v''$ for all $k$.  Hence $v' \leq v''$.
    Thus, $v'$ is the least fixed point of $T$.
\end{Answer}

Similar logic can be applied to prove that $(T^k \bar v)_{k \in \NN}$
converges down to $v^{**}$ in a finite number of steps.

If we set $v^0 = \bar v$ and $v^{k+1} = T v^k$, we can consider the sequence
of valuations $(v^k)$ as a dynamic process, and the number of currently failing firms
$m^k := \sum_i \1\{v^k_i < \theta\}$ can be understood as tracking waves of
bankruptcies. Failures of firms in the first wave put stress on otherwise
healthy firms that have exposure to the failed firms, which in turn causes
further failures and so on.

\begin{Exercise}
    Prove that the sequence $(m_k)$ is monotone increasing.
\end{Exercise}

As discussed in Proposition~\ref{p:ec}, the sequence $(v^k)$, which can also
be written as $(T^k \bar v)$, is decreasing pointwise.  In other words, the
value of each firm is non-increasing. Hence, if $v^k_i <
\theta$ for some $k$, then $v^{k+j}_i < \theta$ for all $j \geq 0$.

Figure~\ref{f:fin_network_sims_1} illustrates a growing wave of failures
that can arise in a financial network.  Firms with lighter
colors have better balance sheets.  Black firms have failed.  The code for
generating this figure, along with details on parameters, can be found in the
code book.

\begin{figure}
   \centering 
   \scalebox{0.72}{\includegraphics[trim = 20mm 20mm 0mm 20mm, clip]{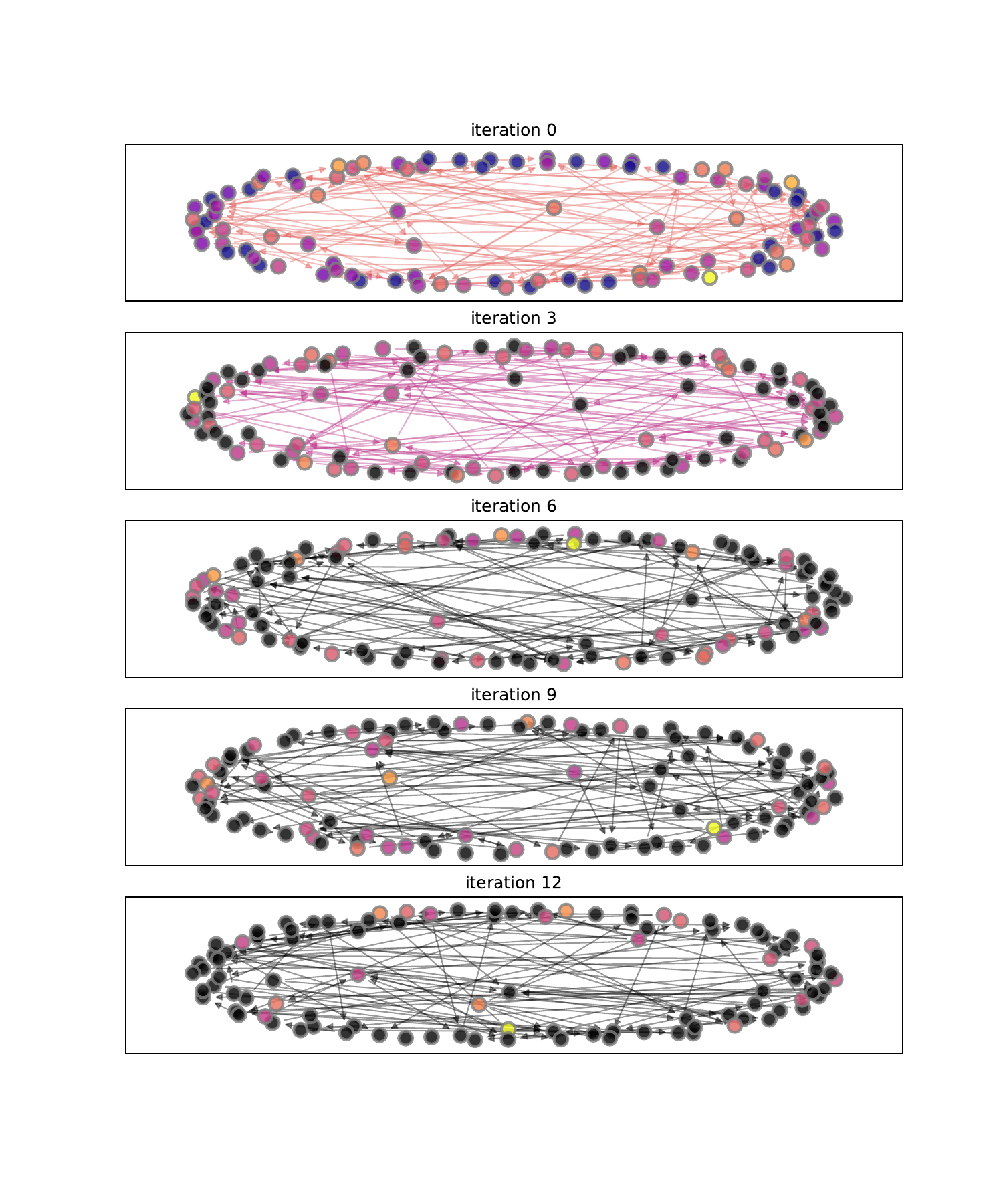}}
   \caption{\label{f:fin_network_sims_1} Waves of bankruptcies in a financial
   network}
\end{figure}

\section{Chapter Notes}\label{s:cnfp}

\cite{shin2010risk} gives an excellent overview of systemic risk and the
financial crisis of 2007--2008. 
\cite{battiston2012debtrank} study connectedness in financial networks and
introduce a measure of systemic impact called DebtRank.
\cite{bardoscia2015debtrank} provide a dynamic theory of instability related
to DebtRank.
\cite{demange2018contagion} provides a threat index for
contagion in financial networks related to Katz centrality.
\cite{bardoscia2019forward} analyze risks associated with
solvency contagion during crises.  \cite{jackson2019investment} study
investment in risky portfolios by banks in financial networks with debt and
equity interdependencies.  \cite{jackson2020credit} analyze optimal bailouts
in financial networks.  
\cite{amini2020clearing} offer an introduction to modeling of clearing
systems, analyzing equilibria for network payments and asset prices.
 \cite{jackson2021systemic}  provide a survey
of the relationship between financial networks and systemic risk.

The setting in \S\ref{sss:nee} was
investigated by \cite{eisenberg2001systemic}, one of the first papers on a
network approach to default cascades.  Additional details on stability
properties can be found in \cite{stachurski2022systemic}.  As already
mentioned, \S\ref{ss:echs} is based on \citep{elliott2014financial}, which
also includes an interesting discussion of how the level of integration across
a network affects equilibria.  \cite{klages2021optimal} discuss optimal
intervention in economic networks via influence maximization methods, using
\cite{elliott2014financial} as an example. \cite{acemoglu2021systemic} 
study how anticipation of future defaults can result in ``credit freezes."

A general discussion of contraction maps and related fixed point theory can be
found in \cite{goebel1990topics} and \cite{cheney2013analysis}.  For more on
fixed point methods for order preserving operators, see, for example,
\cite{guo2004partial}, \cite{zhang2012variational}, \cite{marinacci2019unique}
or \cite{deplano2020nonlinear},

%% file: appendix.tex
\chapter{Appendix}\label{c:ap}

\section{Math Review}\label{s:review}

This section provides a brief review of basic analysis and linear
algebra.  The material contained here should be covered in intermediate
mathematical economics courses or, if not, in math boot camp at the very start
of a graduate program.

(For those who want a slower treatment of the analysis section, we recommend
\cite{bartle2011introduction}, which is carefully constructed and beautifully
written.  High quality texts on linear algebra at the right level for this
course include \cite{jan}, \cite{meyer2000matrix}, \cite{aggarwal2020linear}
and \cite{cohen2021linear}.)

\subsection{Sets and Functions}\label{ss:sefun}

As a first step, let's clarify elementary terminology and notation.

\subsubsection{Sets}\label{sss:sets}

A \navy{set} is an arbitrary collection of objects.  Individual objects are
called \navy{elements} of the set.  We assume the reader is familiar with
basic set operations such as intersections and unions.  If $A$ is a finite
set, then $|A|$ is the number of elements in $A$.  Powers applied to sets
indicate \navy{Cartesian products}\index{Cartesian product}, so that 
\begin{equation*}
    A^2 := A \times A := \setntn{(a, a')}{a \in A,\; a' \in A},
    \;\; \text{etc.}
\end{equation*}
Throughout, $\wp(A)$ is the \navy{power set}\index{Power set} of $A$,
consisting of all subsets of $A$.  For example,
\begin{equation*}
    A = \{1, 2\}
    \quad \implies \quad
    \wp(A) = \{ \emptyset, \{1\}, \{2\}, A \}.
\end{equation*}

Let $\NN$ be the natural numbers, $\ZZ$ be the integers, $\QQ$ be the rational
numbers and $\RR$ be the reals\index{real numbers} (i.e., the union of the rational
and irrational numbers).  
For $x, y$ in $\RR$, we let
\begin{equation}
    \label{eq:vw}
    x \vee y := \max\{x, y\}
    \quad \text{and} \quad
    x \wedge y := \min\{x, y\}.
\end{equation}

Absolute value is $|x| := x \vee (-x)$.
For $n \in \NN$ we set $\natset{n} := \{1, \ldots, n\}$.

We make use of the following elementary facts:  For all $a, b, c \in \RR$,
\begin{itemize}
    \item $|a + b| \leq |a| + |b|$.
    \item $(a \wedge b) +  c = (a + c) \wedge (b + c)$ 
        and $(a \vee b) +  c = (a + c) \vee (b + c)$.
    \item $(a \vee b) \wedge c = (a \wedge c) \vee (b \wedge c)$
        and $(a \wedge b) \vee c = (a \vee c) \wedge (b \vee c)$.
    \item $|a \wedge c - b \wedge c | \leq |a - b|$.
    \item $|a \vee c - b \vee c | \leq |a - b|$.
\end{itemize}
The first item is called the \navy{triangle inequality}\index{Triangle
    inequality}. Also, if $a, b, c \in \RR_+$, then
\begin{equation}\label{eq:abc}
    (a + b) \wedge c  \leq (a \wedge c) + (b \wedge c).
\end{equation}

\begin{Exercise}\label{ex:awedg}
    Prove:  For all $a, b, c \in \RR_+$, we have
        $|a \wedge c - b \wedge c| \leq |a-b| \wedge c$.
\end{Exercise}

\begin{Answer}
    Fix $a, b \in \RR_+$ and $c \in \RR$. By \eqref{eq:abc}, we have
    \begin{equation*}
        a \wedge c 
        = (a - b + b) \wedge c
        \leq ( |a - b| + b) \wedge c
        \leq  |a - b| \wedge c + b \wedge c.
    \end{equation*}
    Thus, $a \wedge c - b \wedge c \leq |a-b| \wedge c$.  Reversing the roles of
    $a$ and $b$ gives $b \wedge c - a \wedge c \leq |a-b| \wedge c$.  This proves
    the claim in Exercise~\ref{ex:awedg}.
\end{Answer}

\subsubsection{Equivalence Classes}\label{sss:eqclass}

Let $S$ be any set.  A \navy{relation}\index{Relation} $\sim$ on $S$ is a nonempty subset of
$S \times S$.  It is customary to write $x \sim y$
rather than $(x, y) \in \, \sim$ to indicate that $(x,y)$ is in $\sim$.
A relation $\sim$ on $S$ is called an
\navy{equivalence relation} \index{Equivalence relation} if, for all $x, y, z
\in S$, we have
\begin{itemize}
    \item[] (\navy{reflexivity}\index{Reflexivity}) $x \sim x$,
    \item[] (\navy{symmetry}\index{Symmetry}) $x \sim y$ implies $y \sim x$ and
    \item[] (\navy{transitivity}\index{Transitivity}) $x \sim y$ and $y \sim z$ implies $x \sim z$.
\end{itemize}

Any equivalence relation on $S$ induces a partition of
$S$ into a collection of mutually disjoint subsets such that their union
exhausts $S$.  These subsets are called \navy{equivalence classes}. They can
be constructed by taking, for each $x \in S$, the set of all elements that are
equivalent to $x$.

\begin{example}\label{eg:eqrel}
    Let $S$ be the set of all people in the world. If $x \sim y$ indicates
    that $x$ and $y$ live in the same country, then $\sim$ is an equivalence
    relation.  (Check the axioms.) The equivalence classes are the population
    of each country.  The partition induced on $S$ is the set of these
    classes, which we can identify with the set of all countries in the world.
\end{example}

\subsubsection{Functions}\label{sss:functions}

A \navy{function}\index{Function} $f$ from set $A$ to set $B$ is a rule,
written $f \colon A \to B$ or $a \mapsto f(a)$, that associates each
element $a$ of $A$ with one and only one element $f(a)$ of $B$.  The set $A$ is
called the \navy{domain}\index{Domain} of $f$ and $B$ is called the
\navy{codomain}\index{codomain}.  The \navy{range}\index{Range} or
\navy{image} of $f$ is 
\begin{equation*}
    \range(f) := \setntn{b \in B}{b = f(a) \text{ for some } a \in A}.
\end{equation*}
A function $f \colon A \to B$ is called
\begin{itemize}
    \item \navy{one-to-one}\index{One-to-one function} if 
            $f(a) = f(a')$ implies $a = a'$,
    \item \navy{onto}\index{Onto function} if $\range(f) = B$, and
    \item a \navy{bijection}\index{Bijection} or \navy{one-to-one
        correspondence} if $f$ is both onto and one-to-one.
\end{itemize}

\begin{example}
    If $S$ is a nonempty set, then the \navy{identity map} on $S$ is the map
    $I \colon S \to S$ such that $I(x) = x$ for all $x \in S$.  The identity
    map is a bijection for any choice of $S$.
\end{example}

The left panel of Figure~\ref{f:func_types_1} shows a one-to-one function on
$(0,1)$.  This function is not onto, however. For example, there exists no $x
\in (0,1)$ with $f(x)=1/4$.  The right panel of Figure~\ref{f:func_types_1}
shows an onto function, with $\range(f) = (0,1)$.  This function is not
one-to-one, however.  For example, $f(1/4)=f(3/4)=3/4$.

\begin{figure}
    \centering
    \scalebox{0.62}{\includegraphics{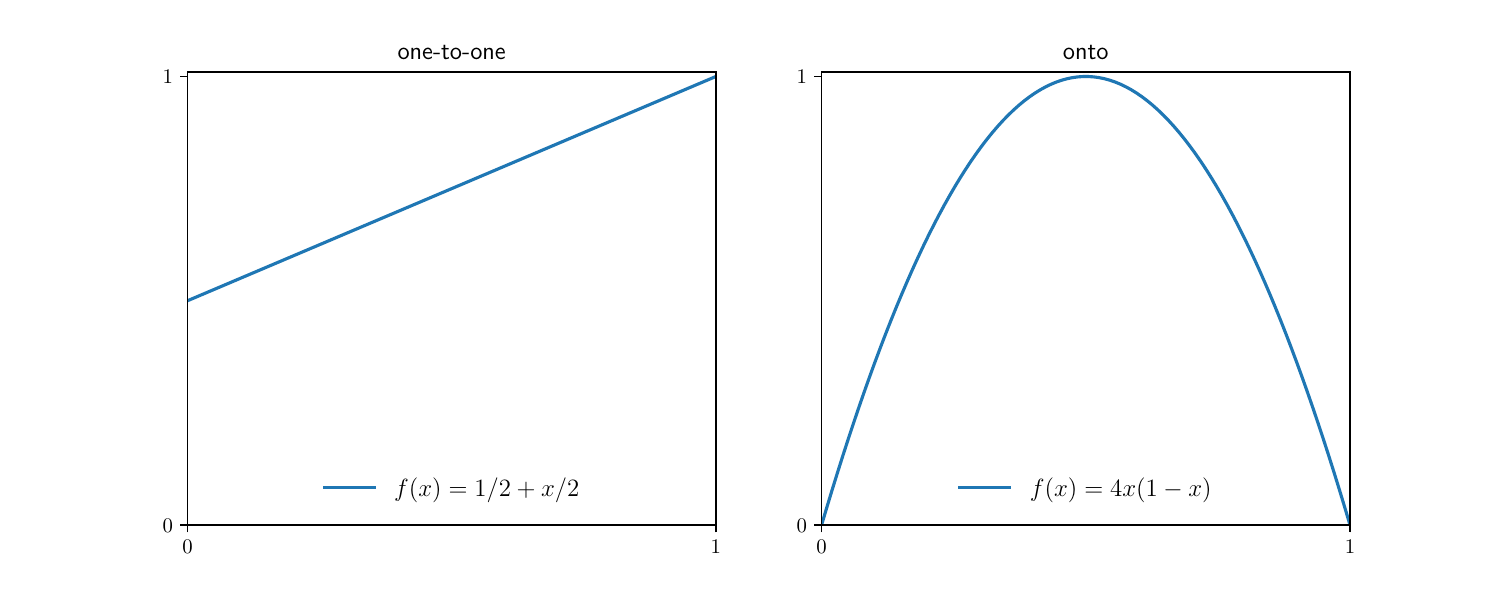}}
    \caption{\label{f:func_types_1} Different types of functions on $(0,1)$}
\end{figure}

The left panel of Figure~\ref{f:func_types_2} gives an example of a function
which is neither one-to-one nor onto.  The right panel of
Figure~\ref{f:func_types_2} gives an example of a bijection.

\begin{figure}
    \centering
    \scalebox{0.62}{\includegraphics{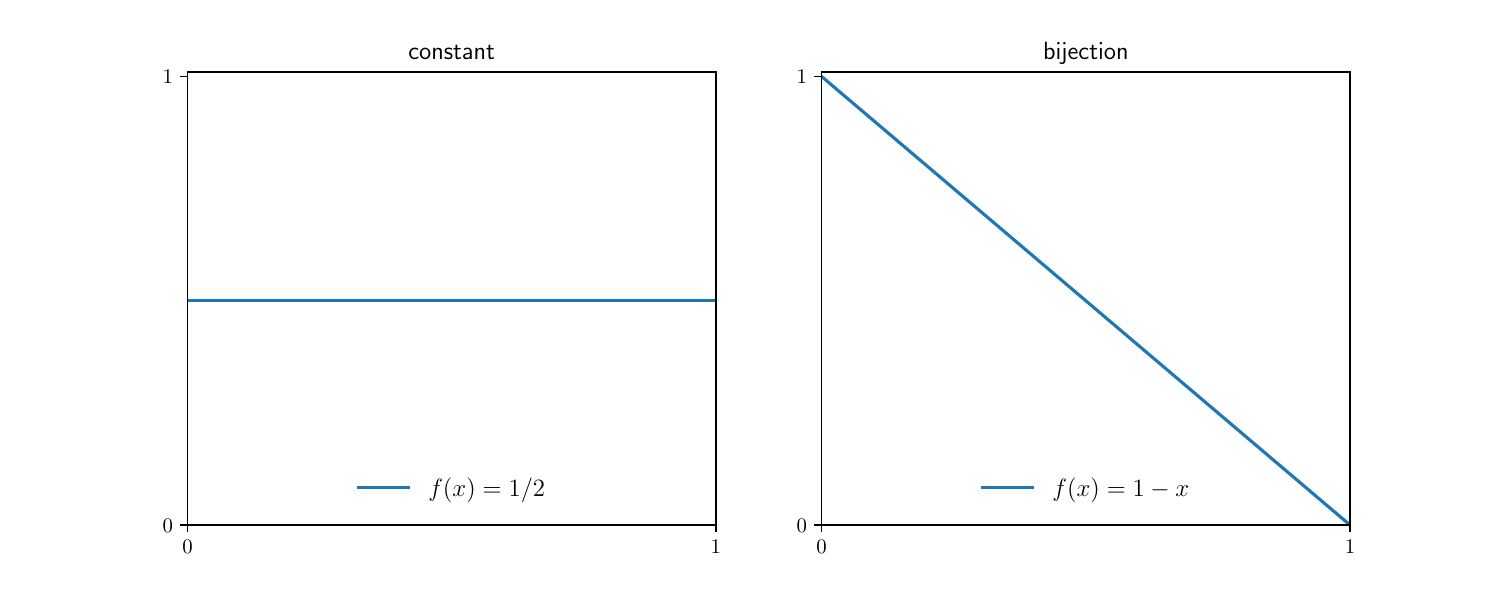}}
    \caption{\label{f:func_types_2} Some functions are bijections and some are not}
\end{figure}

\subsubsection{Inverse Functions}

One motivation for our interest in the basic properties of functions is that
we wish to solve inverse problems.  For an arbitrary nonempty set $S$ and
a function $f \colon S \to S$, an abstraction of an inverse problem is solving
$y=f(x)$ for $x \in S$.  

Prior to stating the next result, we recall that an \navy{inverse function}
for $f \colon S \to S$ is a function $g \colon S \to S$ such that $f \circ g =
g \circ f = I$, where $I$ is the identity map on $S$.  The inverse function of
$f$ is often written as $f^{-1}$.

\begin{lemma}\label{l:ibj}
    For $f \colon S \to S$, the following statements are equivalent:
    \begin{enumerate}
        \item $f$ is a bijection on $S$.
        \item For each $y \in S$, there exists a unique $x \in S$ such that
            $f(x)=y$.
        \item $f$ has an inverse on $S$.
    \end{enumerate}
\end{lemma}

\begin{proof}
    ((i) $\implies$ (ii)) Fix $y \in S$.  Since $f$ is onto, there exists an
    $x \in S$ such that $f(x)=y$.  Since $f$ is one-to-one, there is at most
    one such $x$.

    ((ii) $\implies$ (iii)) Let $g \colon S \to S$ map each $y \in S$ into the
    unique $x \in S$ such that $f(x)=y$.  By the definition of $g$, for fixed
    $x \in S$, we have $g(f(x)) = x$.  Moreover, for each $y \in S$, the point
    $g(y)$ is the point that $f$ maps to $y$, so $f(g(y))=y$.

    ((iii) $\implies$ (i)) Let $g$ be the inverse of $f$.  To see that $f$ must be
    onto, pick any $y \in S$.  Since $f \circ g$ is the identity, we have
    $f(g(y))=y$.  Hence there exists a point $g(y)$ in $S$ that is mapped into
    $y \in S$.  To see that $f$ is one-to-one, fix $x,y \in S$.  If $f(x)=f(y)$,
    then $g(f(x))=g(f(y))$.  But $g \circ f$ is the identity, so $x=y$.
\end{proof}

Here is a nice logical exercise that turns out to be useful when we solve linear
inverse problems.

\begin{Exercise}\label{ex:lioneone}
    Let $S$ and $T$ be nonempty sets.  For $f \colon S \to T$, a function $g
    \colon T \to S$ is called a \navy{left inverse} of $f$ if $g \circ f =
    I$, where $I$ is the identity on $S$.  Prove that $f$ is one-to-one if
    and only if $f$ has a left inverse.
\end{Exercise}

\begin{Answer}
    Suppose first that $f$ is one-to-one.  We construct a left inverse $g$ as
    follows.  For $y \in \range(f)$, let $g(y)$ be the unique $x$ such that
    $f(x)=y$.  (Uniqueness is by the one-to-one property.)  For $y \notin
    \range(f)$, let $g(y)=\bar x$, where $\bar x$ is any point in $S$.  The
    function $g$ is a left inverse of $f$ because, for any $x \in S$, the
    point $y=f(x)$ is in $\range(f)$, and $g(y)=x$.  Hence $g(f(x))=x$.

    Suppose next that $f$ has a left inverse $g$.  Suppose further that $x$ and $x'$ are
    points in $S$ with $f(x)=f(x')$.  Then $g(f(x))=g(f(x'))$.  Since $g$ is a
    left inverse, this yields $x=x'$.  Hence $f$ is one-to-one.
\end{Answer}

\subsubsection{Real-Valued Functions}

If $S$ is any set and $f \colon S \to \RR$, we call $f$ a \navy{real-valued
function}\index{Function!real-valued}.    The set of all real-valued functions
on $S$ is denoted $\RR^S$.    When $S$ has $n$ elements, $\RR^S$ is the same
set as $\RR^n$ expressed in different notation.  The next lemma clarifies.

\begin{lemma}\label{l:rxrn}
    If $|S| =n$, then
    \begin{equation}
        \label{eq:isomph}
        \RR^S \; \ni  \;
        h 
        = (h(x_1), \ldots, h(x_n))
        \quad \longleftrightarrow \quad
            \begin{pmatrix}
                h_1 \\
                \vdots \\
                h_n
            \end{pmatrix}
            \; \in \; 
        \RR^n
    \end{equation}
    is a one-to-one correspondence between $\RR^n$ and the function space
    $\RR^S$.
\end{lemma}

The lemma just states that a function $h$ can be identified by the set of
values that it takes on $S$, which is an $n$-tuple of real numbers.  We use
this identification routinely in what follows. 

The \navy{indicator function}\index{Indicator} of logical statement $P$ is
denoted $\1\{P\}$ and takes value 1 (resp., 0) if $P$ is true (resp., false).

\begin{example}
    If $x, y \in \RR$, then
    \begin{equation*}
        \1\{x \leq y\} =
        \begin{cases}
            1 & \text{ if } x \leq y \\
            0 & \text{ otherwise} .
        \end{cases}
    \end{equation*}
\end{example}
If $A \subset S$, where $S$ is any set, then $\1_A(x) := \1\{x \in A\}$ for
all $x \in S$.

A nonempty set $S$ is called \navy{countable}\index{Countable} if it is finite
or can be placed in one-to-one correspondence with the natural numbers $\NN$.
In the second case we can enumerate $S$ by writing it as $\{x_1, x_2, \ldots
\}$.  Any nonempty set $S$ that fails to be countable is called
\navy{uncountable}\index{Uncountable}.  For example, $\ZZ$ and $\QQ$ are
countable, whereas $\RR$ and every nontrivial interval in $\RR$ are
uncountable.

In general, if $f$ and $g$ are real-valued functions defined on some common
set $S$ and $\alpha$ is a scalar, then $f + g$, $\alpha f$, $fg$, etc.,
are functions on $S$ defined by 
\begin{equation}
    \label{eq:arpo}
    (f + g)(x) = f(x) + g(x),
    \quad
    (\alpha f)(x) = \alpha f(x),
    \quad
    \text{etc.}
\end{equation}
for each $x \in S$.  Similarly, $f \vee g$ and $f \wedge g$ are functions
on $S$ defined by
\begin{equation}\label{eq:fvg}
    (f \vee g)(x) = f(x) \vee g(x)
    \;\; \text{ and } \;\;
    (f \wedge g)(x) = f(x) \wedge g(x).
\end{equation}
Figure~\ref{f:infsup} illustrates.

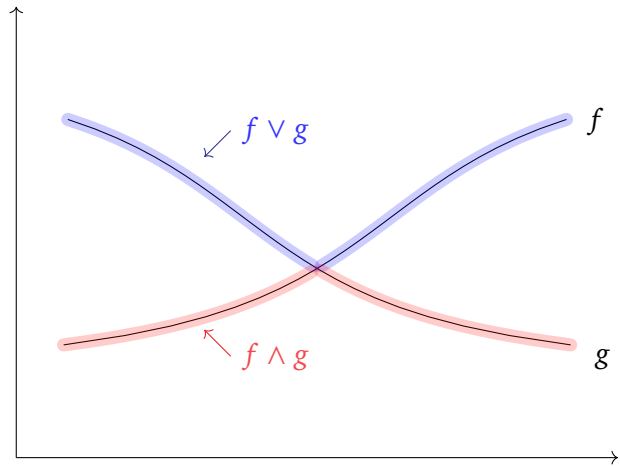
\begin{figure}
	\begin{center}
		\input{tikz/infsup.tex}
		\caption{\label{f:infsup} Functions $f \vee g$ and $f \wedge g$ when defined on a subset of $\RR$}
	\end{center}
\end{figure}

If $f \colon A \to B$ and $g \colon B \to C$, then \navy{$g \circ f$} is
called the \navy{composition}\index{Composition} of $f$ and $g$.  It is the
function mapping $a \in A$ to $g(f(a)) \in C$.

\subsubsection{Fixed Points}\label{sss:fpfd}

Let $S$ be any set.  A \navy{self-map}\index{Self-map} on $S$ is a function $G$ from $S$ to
itself.  When working with self-maps in arbitrary sets it is
common to write the image of $x$ under $G$ as $Gx$ rather than $G(x)$.  We
often follow this convention.  

Given a self-map $G$ on $S$, a point $x \in S$ is
called a \navy{fixed point}\index{Fixed point} of $G$ if $Gx = x$.  

\begin{example}
    Every point of arbitrary $S$ is fixed under the identity map $I \colon x \mapsto x$.
\end{example}

\begin{example}
    If $S = \NN$ and $Gx = x+1$, then $G$ has no fixed point.
\end{example}

Figure~\ref{f:three_fixed_points} shows another example, for a self-map $G$ on
$S = [0, 2]$.   Fixed points are numbers $x \in [0, 2]$ where $G$ meets the 45
degree line.  In this case there are three.

One of the most common techniques for solving systems of nonlinear equations
in applied mathematics---and quantitative economics---is to convert them into
fixed point problems and then apply fixed point theory.  We will see many
applications of this technique.

\begin{figure}
    \centering
    \scalebox{0.62}{\includegraphics{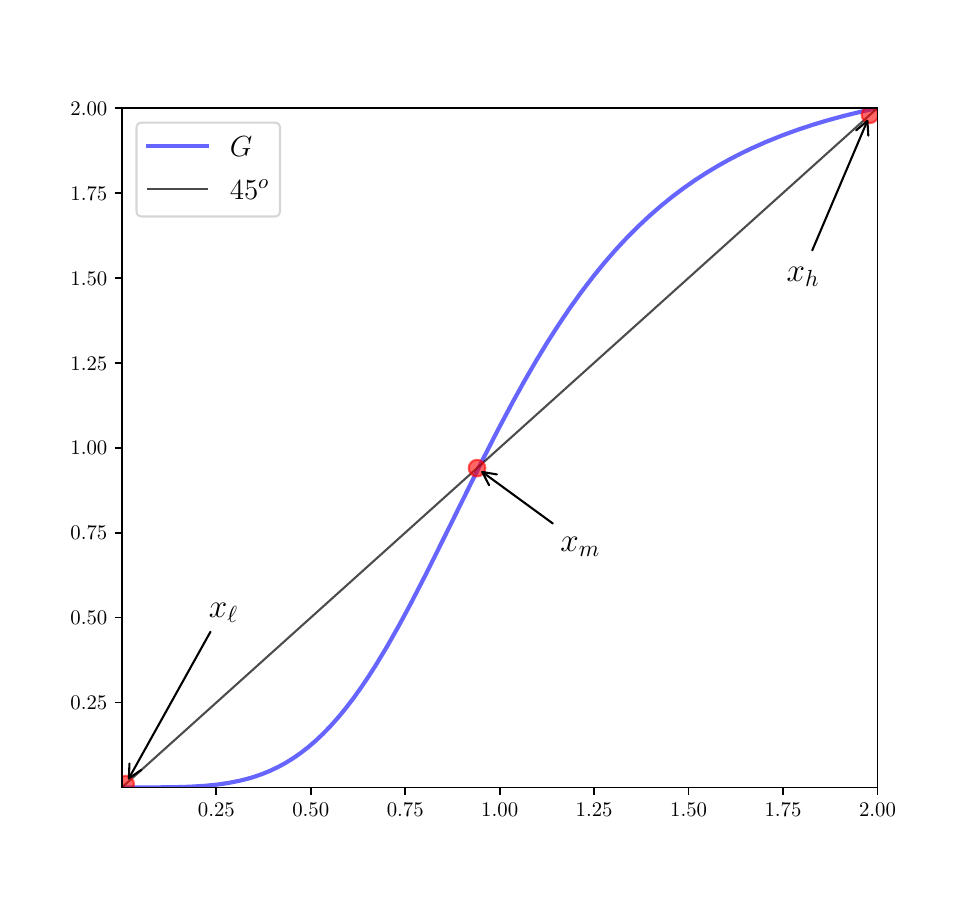}}
    \caption{\label{f:three_fixed_points} 
        Graph and fixed points of $G \colon x \mapsto 2.125/(1 + x^{-4})$ }
\end{figure}

\begin{Exercise}\label{ex:ufnflow}
    Let $S$ be any set and let $G$ be a self-map on $S$.  Suppose there exists
    an $\bar x \in S$ and an $m \in \NN$ such that $G^k x = \bar x$ for all $x
    \in S$ and $k \geq m$.  Prove that, under this condition, $\bar x$ is the
    unique fixed point of $G$ in $S$.
\end{Exercise}

\begin{Answer}
    Let $G$ and $S$ be as stated in the exercise. Regarding uniqueness,
    suppose that $G$ has two distinct fixed points $x$ and $y$ in $S$.  Since
    $G^m x = \bar x$ and $G^m y = \bar x$, we have $G^m x = G^m y$.  But $x$
    and $y$ are distinct fixed points, so $x = G^m x$ must be distinct from $y
    = G^m y$.  Contradiction.

    Regarding the claim that $\bar x$ is a fixed point, we recall that 
    $G^k x = \bar x$ for $k \geq m$.  Hence $G^m \bar x = \bar x$ and $G^{m+1}
    \bar x = \bar x$.  But then 
    \begin{equation*}
        G \bar x = G G^m \bar x = G^{m+1} \bar x = \bar x,
    \end{equation*}
    so $\bar x$ is a fixed point of $G$.
\end{Answer}

\subsubsection{Vectors}\label{sss:vectors}

An $n$-vector $x$ is a tuple of $n$ numbers $x = (x_1, \ldots, x_n)$ where
$x_i \in \RR$ for each $i \in \natset{n}$.  In general, $x$ is neither a row vector nor a
column vector---which coincides with the perspective of most 
scientific computing environments, where the basic vector structure is a
flat array.  When using matrix algebra, vectors are treated as
column vectors unless otherwise stated.

The two most fundamental vector operations are
vector addition and scalar multiplication.  These operations act
pointwise, so that, when $\alpha \in \RR$ and $x, y \in \RR^n$,  
\begin{equation*}
    x + y 
    = 
    \left(
    \begin{array}{c}
        x_1 \\
        x_2 \\
        \vdots \\
        x_n
    \end{array}
    \right)
    +
    \left(
    \begin{array}{c}
         y_1 \\
         y_2 \\
        \vdots \\
         y_n
    \end{array}
    \right)
    :=
    \left(
    \begin{array}{c}
        x_1 + y_1 \\
        x_2 + y_2 \\
        \vdots \\
        x_n + y_n
    \end{array}
    \right)
    \quad \text{and} \quad
    \alpha x 
    := 
    \left(
    \begin{array}{c}
        \alpha x_1 \\
        \alpha x_2 \\
        \vdots \\
        \alpha x_n
    \end{array}
    \right).
\end{equation*}

We let $\RR^n$ be the set of all $n$-vectors and $\matset{n}{k}$ be all $n
\times k$ matrices.  If $A$ is a matrix then $A^{\top}$ is its transpose.

\begin{itemize}
    \item The \navy{inner product}\index{Inner product} of $x, y \in \RR^n$ is defined as
            $\inner{ x, y} := \sum_{i=1}^n x_i y_i$.
    \item The \navy{Euclidean norm} of $x \in \RR^n$ is $\|x \| := \sqrt{\inner{x, x}}$.
\end{itemize}

The norm and inner product satisfy the \navy{triangle inequality} and
\navy{Cauchy--Schwarz inequality}, which state that, respectively,
\begin{equation*}
    \| x + y \| \leq \|x\| + \|y\|
    \quad \text{and} \quad
    |\inner{x, y}| \leq \|x \| \|y\|
    \quad \text{for all } 
    x,y \in \RR^n.
\end{equation*}

\subsubsection{Complex Numbers}\label{sss:complex}

We recall some elementary facts about $\CC$, the set of complex
numbers\index{Complex numbers}.  

Each element of $\CC$ can be understood as a point $(a, b) \in \RR^2$.  In
fact $\CC$ is just $\RR^2$ endowed with a special form of multiplication.  The
point $(a, b)$ is more often written as $a+ib$.  We elaborate below.

The first and second projections of $(a, b)$ are written as $\real (a, b) = a$ and $\imag (a, b) = b$
and called the \navy{real}\index{Real part} and
\navy{nonreal}\index{Nonreal} (or \navy{imaginary}\index{Imaginary}) part
respectively.  The symbol $i$ represents the point $(0, 1) \in \CC$.  As is
traditional, in the context of complex numbers, 
the complex number $(a, 0) \in \CC$ is often written more simply as 
$a$.  With addition and scalar multiplication defined pointwise, this means that, as expected, 
\begin{equation*}
    (a, b) 
    = (a, 0) + (0, b) 
    = (a, 0) + (0, 1) b
    = a + ib.
\end{equation*}
Let $z = (a, b)$.  The \navy{modulus}\index{Modulus} of $z$ is written
$|z|$ and defined as the Euclidean norm $(a^2+b^2)^{1/2}$ of the tuple $(a,b)$.
The two-dimensional Euclidean space is then endowed
with a new operation called \navy{multiplication}, which is defined by 
\begin{equation}\label{eq:complexmult}
    (a, b) (c, d) = (ac - bd, ad+bc).
\end{equation}
Note that, under this rule and our conventions, $i^2  = (0, 1)(0, 1) = (-1, 0) = -1$.

As in the real case, the exponential $\me^z$ is defined for $z \in \CC$ by
$\sum_{k \geq 0} z^k/(k!)$.  It can be shown that, under this extension, the
exponential function continues to enjoy the additive property $\me^{z_1 + z_2}
= \me^{z_1} \me^{z_2}$.  As a result, $\me^{a + ib} = \me^a \me^{ib}$.

Rather than providing its coordinates, another way to represent a vector $z =
(a, b) \in \RR^2$, and hence in $\CC$, is by providing a pair $(r, \phi)$ where $r >
0$ is understood as the length of the vector and $\phi \in [0, 2\pi)$ is the
angle.  This translates to Euclidean coordinates via
\begin{equation*}
    a + ib = (a, b) 
    = (r \cos \phi, r \sin \phi) = r(\cos \phi + i \sin \phi).
\end{equation*}

\begin{figure}
   \begin{center}
       \scalebox{0.45}{\includegraphics[trim = 0mm 0mm 0mm 15mm, clip]{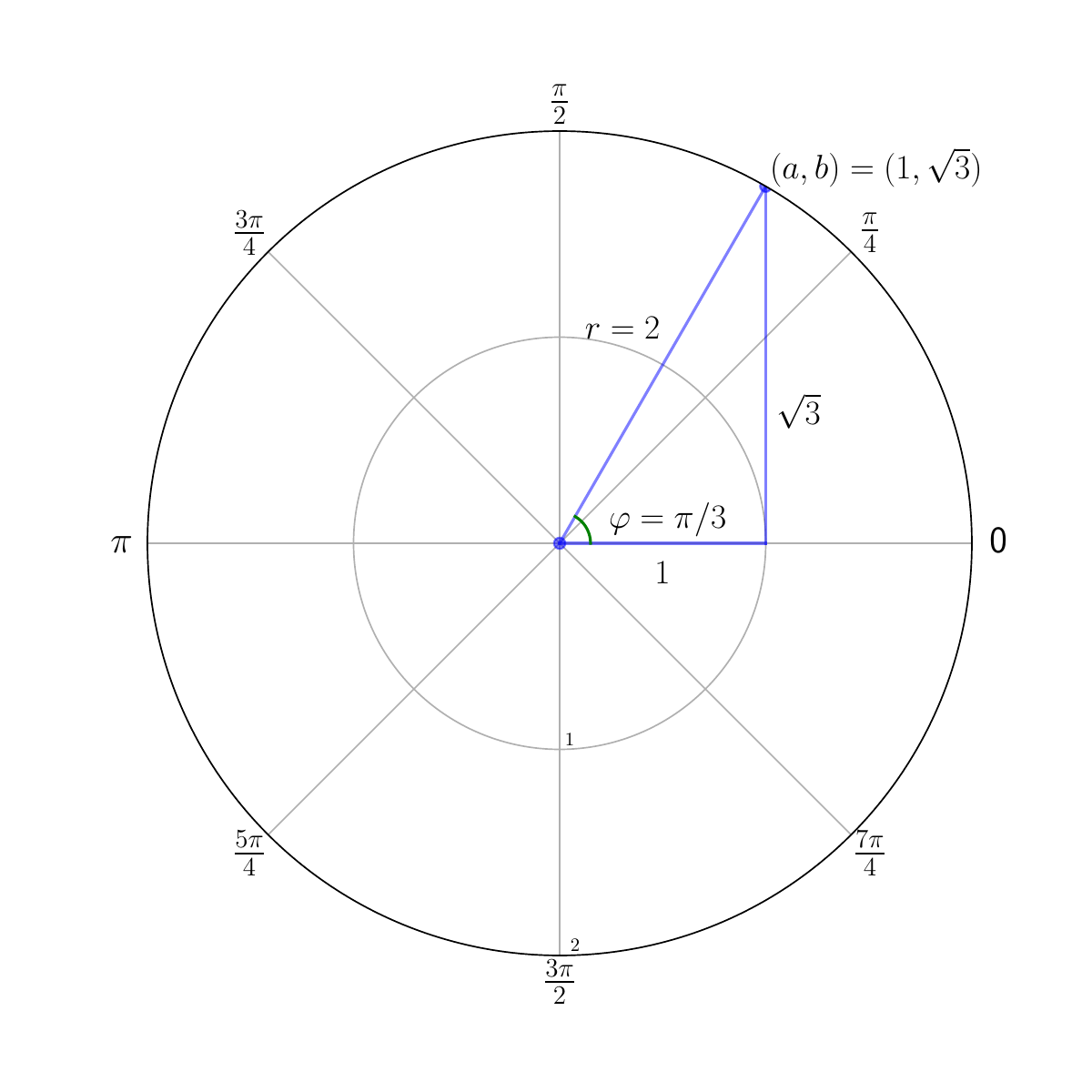}}
   \end{center}
   \caption{\label{f:complex_number} The complex number $(a, b) = r \me^{i \phi}$}
\end{figure}

The representation $(r, \phi)$ is called the \navy{polar
form}\index{Polar form} of the complex number.  By \navy{Euler's
formula}\index{Euler's formula} $\cos(\phi) + i \sin(\phi) = \me^{i \phi}$, 
we can also write
\begin{equation*}
    r(\cos \phi + i \sin \phi) = r \me^{i \phi}.
\end{equation*}
Figure~\ref{f:complex_number} translates 
$\CC \ni (1, \sqrt{3}) = 1 + i \sqrt{3}$  into polar coordinates $2 \me^{i (\pi/3)}$.

The advantage of these representations is clear when we multiply, since
the rule
\begin{equation}\label{eq:compexpr}
    r \me^{i \phi}
    \, s \me^{i \psi} = r  s \, \me^{i (\phi + \psi) }
\end{equation}
is easier to remember and apply than~\eqref{eq:complexmult}.  
Calculating the modulus is also easy, since, by the trigonometric formula
$\cos^2 \phi + \sin^2 \phi = 1$,
\begin{equation}\label{eq:modpol}
    | r \me^{i \phi}| 
    = | r(\cos(\phi) + i \sin(\phi)) | 
    = r \left(\cos^2 \phi + \sin^2 \phi \right)^{1/2} 
    = r.
\end{equation}

\begin{Exercise}
    Show that, for any $u, v \in \CC$ we have $|u v| = |u| |v|$.
\end{Exercise}

\begin{Answer}
    In polar form, we have $u = r \me^{i \phi}$ and $v = 
    s \me^{i \psi}$, so $uv = r  s \, \me^{i (\phi + \psi) }$.  
    Hence $|uv| = rs$ (see \eqref{eq:modpol}), which is equal to $|u| |v|$.
\end{Answer}

\subsection{Order}

Order structure is of great importance in economics---typically more so than
other fields such as physics or chemistry.  Here we review the basics.

\subsubsection{Partial Orders}\label{ss:posets}

It was mentioned in the preface that order-theoretic methods form a core part
of the text.  In this section we introduce some key concepts.

Let $P$ be a nonempty set.  A \navy{partial order}\index{Partial order} on $P$
is a relation $\preceq$ on $P \times P$ satisfying, for any $p, q, r$ in $P$,
\begin{multicols}{2}
    \begin{enumerate}
        \item[] $p \preceq p$, 
        \item[] $p \preceq q$ and $q \preceq p$ implies $p = q$ and
        \item[] $p \preceq q$ and $q \preceq r$ implies $p \preceq r$
        \item[] (reflexivity) 
        \item[] (antisymmetry) 
        \item[] (transitivity) 
    \end{enumerate}
\end{multicols}

When paired with a partial order $\preceq$, the set $P$ (or the pair $(P,
\preceq)$) is called a \navy{partially ordered set}\index{Partially ordered
set}.  

\begin{example}
    The usual order $\leq$ on $\RR$ is a partial order on $\RR$.
    Unlike other partial orders we consider, it 
    has the additional property that either $x \leq y$ or $y \leq x$ for
    every $x, y$ in $\RR$.  For this reason, $\leq$ is also called a
    \navy{total order} on $\RR$.
\end{example}

\begin{Exercise}
    Let $P$ be any set and consider the relation induced by equality, so that
    $p \preceq q$ if and only if $p = q$.   Show that this
    relation is a partial order on $P$.
\end{Exercise}

\begin{Exercise}
    Let $M$ be any set.  Show that $\subset$ is a partial order on $\wp(M)$,
    the set of all subsets of $M$. 
\end{Exercise}

\begin{example}[Pointwise order over functions]\label{eg:ppor}
    Let $S$ be any set. For $f, g$ in $\RR^S$, we write 
    \begin{equation*}
        f \leq g 
        \text{ if }
        f(x) \leq g(x) \text{ for all } x \in S.
    \end{equation*}
    This relation $\leq$ on $\RR^S$ is a partial order called the
    \navy{pointwise order}\index{Pointwise order} on $\RR^S$.
\end{example}

A subset $B$ of a partially ordered set $(P, \preceq)$ is called
\begin{itemize}
    \item \navy{increasing}\index{Increasing} if $x \in B$ and $x \preceq y$
        implies $y \in B$.
    \item \navy{decreasing}\index{Decreasing} if $x \in B$ and $y \preceq x$
        implies $y \in B$.
\end{itemize}

\begin{Exercise}
    Describe the set of increasing sets in $(\RR, \leq)$.
\end{Exercise}

\begin{example}[Pointwise order over vectors]
    For vectors $x = (x_1, \ldots, x_d)$ and $y = (y_1, \ldots, y_d)$, we write 
    \begin{itemize}
        \item $x \leq y$ if $x_i \leq y_i$ for all $i \in \natset{d}$ and
        \item $x \ll y$ if $x_i < y_i$ for all $i \in \natset{d}$.
    \end{itemize}
    The statements $x \geq y$ and $x \gg y$ are defined
    analogously.\footnote{The notation $x \leq y$ over vectors is standard,
        while $x \ll y$ is less so.  In some fields, $n \ll k$ is used as
        an abbreviation for ``$n$ is much smaller than $k$.''  Our usage
        lines up with most of the literature on partially ordered
    vector spaces.  See, e.g., \cite{zhang2012variational}.}
        The relation $\leq$ is a partial order on $\RR^n$, also called the \navy{pointwise
    order}\index{Pointwise order}.  (In fact, the present example is a special case
    of Example~\ref{eg:ppor} under the identification in Lemma~\ref{l:rxrn}
    (page~\pageref{l:rxrn}).) On the other hand, $\ll$ is not a partial order
    on $\RR^n$.  (Which axiom fails?)
\end{example}

\begin{Exercise}\label{ex:clic}
    Recall from Example~\ref{eg:clic} that limits in $\RR$ preserve weak
    inequalities. Prove that the same is true in $\RR^d$.  In particular, 
    show that, for vectors $a, b \in \RR^d$ and sequence $(x_n)$ in $\RR^d$,
    we have $a \leq x_n \leq b$ for all $n \in
    \NN$ and $x_n \to x$ implies $a \leq x \leq b$.
\end{Exercise}

\subsubsection{Pointwise Operations on Vectors}\label{sss:poov}

In this text, operations on real numbers such as $|\cdot|$ and $\vee$ are
applied to vectors pointwise.  For example, for vectors $a = (a_i)$ and $b =
(b_i)$ in $\RR^d$, we set
\begin{equation*}
    |a| = (|a_i|),
    \quad
    a \wedge b = (a_i \wedge b_i)
    \quad \text{and} \quad
    a \vee b = (a_i \vee b_i)
\end{equation*}
(The last two are special cases of \eqref{eq:fvg}.)

A small amount of thought will convince you that, with this convention plus
the pointwise order over vectors, the
order-theoretic inequalities and identities listed in \S\ref{sss:sets} also
hold for vectors in $\RR^d$.  (For example, $|a + b| \leq |a| + |b|$ for any
$a, b, c \in \RR^d$.)

\begin{Exercise}\label{ex:bmk}
    Prove:  If $B$ is $m \times k$ and $B \geq 0$, then $|B x| \leq B |x|$ for
    all $k \times 1$ column vectors $x$.
\end{Exercise}

\begin{Answer}
    Fix $B \in \matset{m}{k}$ with $b_{ij} \geq 0$ for all $i, j$.  Pick any $i
    \in \natset{m}$ and $x \in \RR^k$.  By the triangle inequality, we have
    $|\sum_j b_{ij} x_j| \leq \sum_j b_{ij} |x_j|$.  Stacking these inequalities
    yields $|B x| \leq B |x|$, as was to be shown.
\end{Answer}

\subsubsection{Monotonicity}\label{sss:mono}

Given two partially ordered sets $(P, \preceq)$ and $(Q,
\trianglelefteq)$, a function $G$ from $P$ to $Q$ is called
\navy{order-preserving}\index{Order-preserving} if 
\begin{equation}
    \label{eq:orpres}
    p, q \in P \text{ and } p \preceq q
    \quad \implies \quad
    Gp \trianglelefteq Gq.
\end{equation}

\begin{example}
    Let $\cC$ be all continuous functions from 
    $S = [a, b]$ to $\RR$ and let $\leq$ be the pointwise partial order on
    $\cC$.  Define 
    \begin{equation*}
        I \colon \cC \ni f \to \int_a^b f(x) dx \in \RR.
    \end{equation*}
    Since $f \leq g$ implies $\int_a^b f(x) dx \leq
    \int_a^b g(x) dx$, the integral map $I$ is order-preserving on $\cC$.
\end{example}

\begin{Exercise}
    Let $X$ be a random variable mapping $\Omega$ to finite $S$.  Define
    $\ell \colon \RR^S \to \RR$ by $\ell h = \EE h(X)$.  Show that $\ell$ is
    order-preserving when $\RR^S$ has the pointwise order.
\end{Exercise}

If $P = Q = \RR$ and $\preceq$ and $\trianglelefteq$ are both equal to $\leq$,
the standard order on $\RR$, then the order-preserving property reduces to the usual notion of
an \navy{increasing function}\index{Increasing function} (i.e., nondecreasing
function), and we will use the terms ``increasing'' and ``order-preserving''
interchangeably in this setting.\footnote{Other common terms for
    order-preserving in
the literature include ``monotone increasing,'' ``monotone''  and
``isotone.''} 

In addition, if $S = g$ maps $A \subset \RR$ into $\RR$, then we will call $g$
\begin{itemize}
    \item \navy{strictly increasing}\index{Strictly increasing} if $x < y$
        implies $g(x) < g(y)$, and
    \item \navy{strictly decreasing}\index{Strictly decreasing}
        if $x < y$ implies $g(x) > g(y)$.
\end{itemize}

\subsection{Convergence}\label{ss:convergence}

Let's now recall the basics of convergence and continuity.

Given any set $S$, an $S$-valued \navy{sequence}\index{Sequence} $(x_n) :=
(x_n)_{n \in \NN}$ is a function $n \mapsto x_n$ from $\NN$ to $S$.  If
$S=\RR$, we call $(x_n)$ a \navy{real-valued sequence}.  A
\navy{subsequence}\index{Subsequence} of $(x_n)_{n \in \NN}$  is a sequence of the form
$(x_{\sigma(n)})_{n \in \NN}$ 
where $\sigma$ is a strictly increasing function
    from $\NN$ to itself.   You can think of forming a subsequence from a
    sequence by deleting some of its elements---while still retaining
    infinitely many. 

In computer science and statistics, it is common to classify sequences
according to their asymptotic behavior.  Often this is done via \navy{big O
    notation}\index{Big O notation},  where, for a real-valued sequence
    $(x_n)$, we write \navy{$(x_n) = O(g_n)$} if there exists a nonnegative
    sequence $(g_n)$ and a constant $M < \infty$ such that $|x_n| \leq M g_n$
    for all $n \in \NN$.

\begin{Exercise}
    Let $x_n = -5n + n^2$ for all $n \in \NN$. Show that $(x_n) = O(n^2)$
    holds but $(x_n) = O(n)$ fails.
\end{Exercise}

\begin{Answer}
    Let $x_n= -5n + n^2$.  Then $|x_n| \leq 5n + n^2 \leq 6 n^2$.  Hence $(x_n) =
    O(n^2)$.  Regarding the second claim, suppose to the contrary that $(x_n) =
    O(n)$.  Then we can take an $M$ such that $|x_n| \leq M n$ for all $n \in \NN$. But
    then $x_n = -5n + n^2 \leq M n$ for all $n$.  Dividing by $n$ yields $n \leq 5 + M$
    for all $n$.  Contradiction.
\end{Answer}

\subsubsection{Metric Properties of the Real Line}\label{sss:ocon0}

The following definition is fundamental to what follows: a real-valued
sequence $(x_n)$ \navy{converges}\index{Convergence} to $x \in \RR$ and write
$x_n \to x$ if
\begin{equation*}
    \text{ 
        for each $\epsilon > 0$, there is an $N \in \NN$ such that
        $|x_n - x| < \epsilon$ whenever $n \geq N$.  
    }
\end{equation*}

\begin{example}
    If $x_n = 1-1/n$, then $x_n \to 1$.  Indeed, for any $\epsilon > 0$, the statement
    $|x_n - 1| < \epsilon$ is equivalent to $n > 1/\epsilon$.
    This clearly holds whenever $n$ is sufficiently large.
\end{example}

\begin{Exercise}
    Prove: If $a, b \in \RR$, $x_n \to a$ and $x_n \to b$, then $a=b$.
\end{Exercise}

Let's state some elementary limit laws that are used without comment
throughout.  (You can review the proofs in sources such as \cite{bartle2011introduction}).

\index{Bounded!sequence}
\index{Monotone!sequence}
A sequence $(x_n)$ is called \navy{bounded} if there is an $M \in \RR$ such
that $|x_n| \leq M$ for all $n \in \NN$. It is called 
\begin{itemize}
    \item \navy{monotone increasing} if $x_n \leq x_{n+1}$ for all $n \in \NN$, and   
    \item \navy{monotone decreasing} if $x_n \geq x_{n+1}$ for all $n$.
\end{itemize}
The sequence is called \navy{monotone} if it is either monotone increasing or 
decreasing.  The next theorem, concerning monotone sequences, is a deep result
about the structure of $\RR$.

\begin{theorem}
    \label{t:bmseq}
    A real-valued monotone sequence converges in $\RR$ if and only if it is
    bounded.
\end{theorem}

The intuitive meaning of the ``if'' part of Theorem~\ref{t:bmseq} is that
monotone bounded sequences always converge to some point in $\RR$ because
$\RR$ contains no gaps.  This statement is closely related to the
``completeness'' property of the real line, which is discussed 
in \cite{bartle2011introduction} and many other texts on real analysis.

\index{Series}
Next let's consider \navy{series}.  Given a sequence $(x_n)$ in $\RR$, we
set
\begin{equation*}
    \sum_{n \geq 1} x_n := \lim_{N \to \infty} \sum_{n=1}^N x_n
    \text{ whenever the limit exists in $\RR$}.   
\end{equation*}
More generally, given arbitrary countable $S$ and
$g \in \RR^S$, we write
$\sum_{x \in S} g(x) = M$ if there exists an enumeration $(x_n)_{n \in
    \NN}$ of $S$ such that $\sum_{n \geq 1} |g(x_n)|$ is
    finite and, in addition, $\sum_{n \geq 1} g(x_n) = M$.\footnote{This
        definition is not ambiguous because every possible enumeration leads
    to the same value when the absolute sum is finite (see, e.g., the
rearrangement theorem in \cite{bartle2011introduction}, \S9.1).}

\begin{Exercise}\label{ex:stricttr}
    Show that, if $S$ is countable, $g \in \RR^S$, and there exist $x', x''
    \in S$ such that $g(x') > 0$ and $g(x'') < 0$, then $\left| \sum_{x \in S}
    g(x) \right| < \sum_{x \in S} |g(x)|$.\footnote{Hint: Start with the case
        $|S|=2$.  Argue that the case with $n$ elements follows from this case
        and the ordinary (weak) triangle inequality $| \sum_{x \in S} g(x) |
    \leq \sum_{x \in S} |g(x)|$.}
\end{Exercise}

\subsubsection{Metric Properties of Euclidean Space}\label{sss:ocon}

Now we review the metric properties of $\RR^d$, for some $d \in \NN$, when
distance between vectors $x, y \in \RR^d$ is understood in terms of Euclidean
norm deviation $\| x - y \|$.  The notion of convergence for real-valued sequences
extends naturally to this setting: sequence $(x_n)$ in $\RR^d$ is said to
\navy{converge}\index{Convergence} to $x \in \RR^d$ if
\begin{equation*}
    \text{ 
        for each $\epsilon > 0$, there is an $N \in \NN$ such that
        $\|x_n - x\| < \epsilon$ whenever $n \geq N$.  
    }
\end{equation*}
In this case we write $x_n \to x$.  Figure~\ref{f:euclidean_convergence_1}
shows a sequence converging to the origin in $\RR^3$, with colder colors being
later in the sequence.

\begin{figure}
   \begin{center}
       \scalebox{0.9}{\includegraphics[trim = 0mm 20mm 0mm 10mm, clip]{
           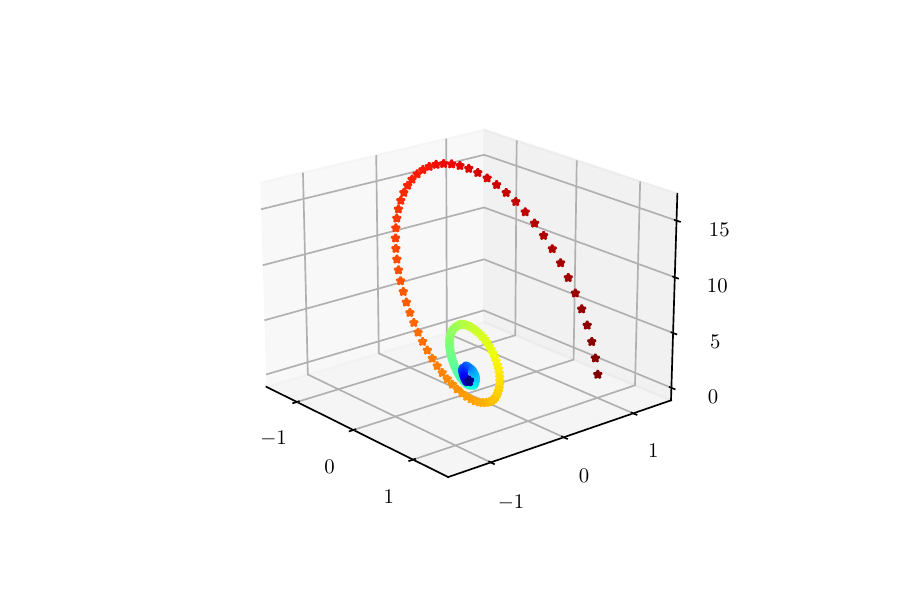}
       }
   \end{center}
   \caption{\label{f:euclidean_convergence_1} Convergence of a sequence to the
   origin in $\RR^3$}
\end{figure}

\begin{Exercise}
    Prove that limits in $\RR^d$ are unique.  In other words, show that, if $(x_n)$ is a
    sequence converging to $x \in \RR^d$ and $y \in \RR^d$, then $x=y$.
\end{Exercise}

Given any point $u \in \RR^d$ and $\epsilon > 0$, the
\navy{$\epsilon$-ball}\index{$\epsilon$-ball} around $u$ is the set
\begin{equation*}
    B_\epsilon(u) := \setntn{v \in \RR^d}{\| u - v \| < \epsilon} .
\end{equation*}
With this terminology, we can say that $(x_n)$ converges to $x \in \RR^d$ if
the sequence $(x_n)$ is eventually in any $\epsilon$-ball around $x$.

\begin{proposition}\label{p:rseqp}
    If $(x_n)$ and $(y_n)$ are sequences in $\RR^d$ with $x_n \to x$ and $y_n \to y$, then
    \begin{enumerate}
        \label{item:sll}
        \item $x_n + y_n \to x + y$ and $x_n y_n \to x y$
        \item $x_n \leq y_n$ for all $n \in \NN$ implies $x \leq y$
        \item $\alpha x_n \to \alpha x$ for any $\alpha \in \RR$ 
        \item $x_n \vee y_n \to x \vee y$ and $x_n \wedge y_n \to x \wedge y$.
    \end{enumerate}
\end{proposition}

A sequence $(x_n) \subset \RR^d$ is called \navy{Cauchy}\index{Cauchy
sequence} if, for all $\epsilon > 0$, there exists an $N \in \NN$ with $|x_n -
x_m| < \epsilon$ whenever $n, m \geq N$.  

\begin{Exercise}
    Let $d=1$ and suppose $x_n = 1/n$. Prove that $(x_n)$ is Cauchy.
\end{Exercise}

\begin{Exercise}
    Prove that every convergent sequence in $\RR^d$ is Cauchy.
\end{Exercise}

It is a fundamental result of analysis, stemming from axiomatic properties of
the reals, that the converse is also true:

\begin{theorem}\label{t:compreal}
    A sequence in $\RR^d$ converges to a point in $\RR^d$ if and only if it is Cauchy.
\end{theorem}

\subsubsection{Topology}

A point $u \in A \subset \RR^d$ is called \navy{interior}\index{Interior} to
$A$ if there exists an $\epsilon > 0$ such that $B_\epsilon(u) \subset A$.

\begin{Exercise}
    Let $d = 1$ so that $\| x - y \| = |x - y|$.   Show
    that $0.5$ is interior to $A := [0, 1)$ but $0$ is not.
    Show that $\QQ$, the set of rational numbers in $\RR$, contains no interior
    points.
\end{Exercise}

A subset $G$ of $\RR^d$ is called \navy{open}\index{Open set} in $\RR^d$ if
every $u \in G$ is interior to $G$.  
A subset $F$ of $\RR^d$ is called \navy{closed}\index{Closed set} if,
given any sequence $(x_n)$ satisfying $x_n \in F$ for all $n \in \NN$ and $x_n \to
x$ for some $x \in \RR^d$, the point $x$ is in $F$.  In other words, $F$ contains
the limit points of all convergent sequences that take values in $F$.

\begin{example}\label{eg:clic}
    Limits in $\RR$ preserve orders, so $a \leq x_n \leq b$ for all $n \in
    \NN$ and $x_n \to x$ implies $a \leq x \leq b$.  Thus, any closed interval
    $[a, b]$ in $\RR$ is closed in the standard (one dimensional Euclidean)
    metric.
\end{example}

\begin{Exercise}
    Prove that $G \subset \RR^d$ is open if and only if $G^c$ is closed.
\end{Exercise}

A subset $B$ of $\RR^d$ is called \navy{bounded}\index{Bounded} if there exists a
finite $M$ such that $\| b \| \leq M$ for all $b \in B$. A subset $K$ of
$\RR^d$ is called \navy{compact}\index{Compact} in $\RR^d$ if every sequence
in $K$ has a subsequence converging to some point in $K$. The 
Bolzano--Weierstrass theorem tells us that $K$ is compact if and only if $K$
is closed and bounded.

\subsubsection{Continuity in Vector Space}

If $A \subset \RR^d$, then $f \colon A \to \RR^k$ is called
\navy{continuous} at $x \in A$ if, for each sequence $(x_n) \subset A$ with
$x_n \to x$, we have $f(x_n) \to f(x)$ in $\RR^k$.  If $f$ is continuous at all $x \in
A$ then we call $f$ \navy{continuous on $A$}.

\begin{example}
    If $f(x)=x^2$ on $A = \RR$, then $f$ is continuous at all $x \in \RR$
    because,  by Proposition~\ref{p:rseqp}, $x_n \to x$ implies $x_n^2 = x_n
    \cdot x_n \to x \cdot x = x^2$.
\end{example}

More generally, every polynomial function is continuous on $\RR$.
The elementary functions $\sin$, $\cos$, $\exp$ and $\log$ are all continuous
on their domains.

\begin{Exercise}
    Prove: If $\alpha, \beta \in \RR$ and $f,g$ are continuous functions from
    $A \subset \RR^d$ to $\RR^k$, then so is $\alpha f+ \beta g$.
\end{Exercise}

\begin{Answer}
    If $x_n \to x$ in $\RR^d$ and $f$ and $g$ are continuous, then $f(x_n) \to f(x)$
    and $g(x_n) \to g(x)$ in $\RR^k$.  But then, by Proposition~\ref{p:rseqp},
    $\alpha f(x_n) + \beta g(x_n)$ converges to $\alpha f(x) + \beta g(x)$ in
    $\RR^k$, as was to be shown.
\end{Answer}

\begin{Exercise}
    Fix $a \in \RR^d$.  Prove that $f, g \colon \RR^d \to \RR^d$ defined by 
    $f(x) = x \wedge a$ and $g(x) = x \vee a$
    are both continuous functions on $\RR^d$.
\end{Exercise}

\begin{Answer}
    Here's the answer for $f$:  Take $(x_n)$ converging to $x$ in $\RR^d$.  Applying the
    inequalities in \ref{sss:sets} pointwise to vectors, we have
    \begin{equation*}
        0 \leq |f(x_n) - f(x)|
        = |x_n \wedge a - x \wedge a| \leq |x_n - x|.
    \end{equation*}
    Taking the Euclidean norm over these vectors and using $|u| \leq |v|$ implies
    $\| u \| \leq \| v\|$ yields $\| f(x_n) - f(x)\| \leq \|x_n - x\| \to 0$.
    Similar arguments can be applied to $g$.
\end{Answer}

The next lemma is helpful in locating fixed points.

\begin{lemma}\label{l:clifp}
    Let $F$ be a self-map on $S \subset \RR^d$.  If $F^m u \to u^*$ as $m \to
    \infty$ for some pair $u, u^* \in S$ and, in addition, $F$ is continuous
    at $u^*$, then $u^*$ is a fixed point of $F$.
\end{lemma}

\begin{proof}
    Assume the hypotheses of Lemma~\ref{l:clifp} and let $u_m := F^m u$ for
    all $m \in \NN$.  By continuity and $u_m \to u^*$ we have $F u_m \to F
    u^*$. But the sequence $(F u_m)$ is just $(u_m)$ with the first element
    omitted, so, given that $u_m \to u^*$, we must have $F u_m \to u^*$.
    Since limits are unique, it follows that $u^* = F u^*$.
\end{proof}

\subsection{Linear Algebra}\label{ss:fdvec} 

Next we review fundamental concepts and definitions from linear algebra.

\subsubsection{Subspaces and Independence}

A subset $E$ of $\RR^n$ is called a \navy{linear subspace}\index{Subspace} of
$\RR^n$ if 
\begin{equation*}
    x, y \in E \text{ and } \alpha, \beta \in \RR
    \; \implies \; \alpha x + \beta y \in E. 
\end{equation*}
In other words, $E$ is \navy{closed} under the operations of
addition and scalar multiplication; that is, (i) $\alpha \in \RR$ and $x \in
E$ implies $\alpha x \in E$ and (ii) $x, y \in E$ implies $x+y \in E$.

\begin{Exercise}\label{ex:whichli}
    Fix $c \in \RR^n$ and $C \in \matset{n}{k}$.  Show that
    \begin{itemize}
        \item $H := \setntn{x \in \RR^n}{\inner{c, x} = 0}$ and
        \item $\range C := \setntn{y \in \RR^n}{y = Cx \text{ for some } x \in \RR^k}$
    \end{itemize}
    are linear subspaces of $\RR^n$.  Show that $S := \setntn{x \in
    \RR^n}{\inner{c, x} \geq 0}$ is not.
\end{Exercise}

A \navy{linear combination}\index{Linear combination} of vectors $v_1,\ldots,
v_k$ in $\RR^n$ is a vector of the form 
\begin{equation*}
    \alpha_1 v_1 + \cdots +  \alpha_k v_k \text{ where } (\alpha_1,\ldots, \alpha_k) \in \RR^k.     
\end{equation*}
The set of all linear combinations of elements of $F \subset \RR^n$ is called
the \navy{span}\index{Span} of $F$ and written as $\Span F$.  

\begin{example}\label{eg:colspace}
    The set $\range C$ in Exercise~\ref{ex:whichli} is the span of the columns of the
    matrix $C$, viewed as a set of vectors in $\RR^n$.   The set $\range C$ is
    also called the \navy{column space}\index{Column space} of $C$.
\end{example}

\begin{Exercise}
    Let $F$ be a nonempty subset of $\RR^n$.  Prove that
    \begin{enumerate}
        \item $\Span F$ is a linear subspace of $\RR^n$ and
        \item $\Span F$ is the intersection of all linear subspaces $S \subset
            \RR^n$ with $S \supset F$.
    \end{enumerate}
\end{Exercise}

Figure~\ref{f:span1} shows the linear subspace spanned by the three vectors
\begin{equation}\label{eq:uvw}
    u =
    \begin{pmatrix}
        3 \\
        4 \\
        1
    \end{pmatrix},
    \quad
    v =
    \begin{pmatrix}
        3 \\
        -4 \\
       0.2 
    \end{pmatrix},
    \quad
    \text{and }
    w =
    \begin{pmatrix}
        -3.5 \\
        3 \\
       -0.4 
    \end{pmatrix}
    .
\end{equation}
The subspace $H$ in which these vectors lie is, in fact, the set of all $x \in \RR^3$ such
that $\inner{x, c} = 0$, with $c=(-0.2, -0.1, 1)$.  This plane is a two-dimensional object.  While we make this terminology precise in
\S\ref{sss:bvd}, the key idea is that 
\begin{itemize}
    \item at least two vectors are required to span $H$ and 
    \item any additional vectors will not increase the span.
\end{itemize}

\begin{figure}
   \centering
   \scalebox{0.65}{\includegraphics[trim = 35mm 35mm 0mm 35mm, clip]{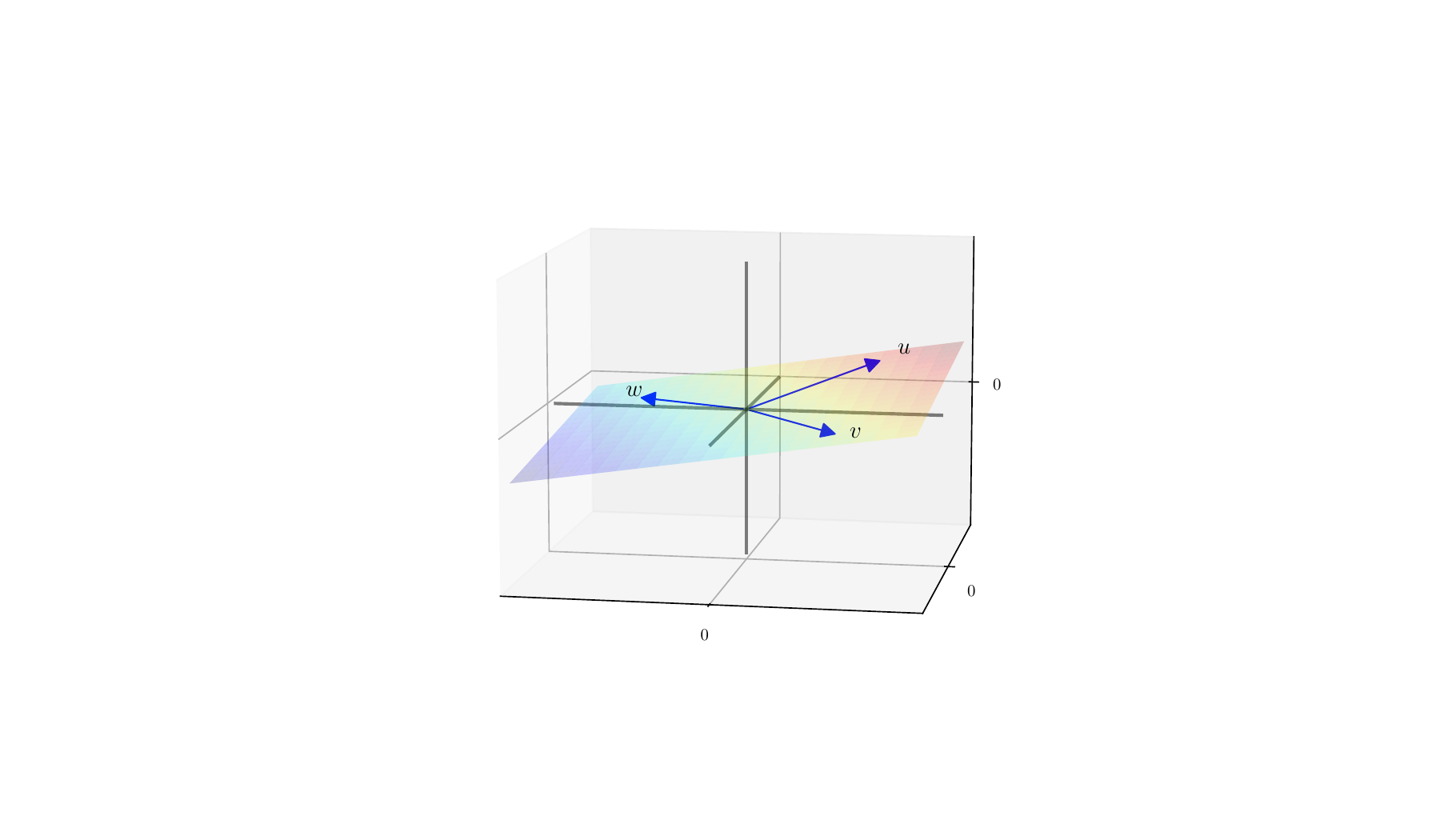}}
   \caption{\label{f:span1} The span of vectors $u, v, w$ in $\RR^3$}
\end{figure}

A finite set of vectors $F := \{v_1, \ldots, v_k\} \subset \RR^n$ is called \navy{linearly
independent}\index{Linearly independent} if, for real scalars $\alpha_1,
\ldots, \alpha_k$,
\begin{equation*}
    \alpha_1 v_1
    + \cdots +
    \alpha_k v_k
     = 0 
    \;  \implies  \;
    \alpha_1 = \cdots = \alpha_k = 0.
\end{equation*}
If $F$ is not linearly independent it is called \navy{linearly dependent}.  

\begin{Exercise}
    Show that $F$ is linearly dependent if and only if there exists a vector
    in $F$ that can be written as a linear combination of other vectors in
    $F$.
\end{Exercise}

\begin{Exercise}
    Prove the following:
    \begin{enumerate}
        \item Every subset of a linearly independent set in $\RR^n$ is linearly
            independent.\footnote{By the law of the excluded middle, the empty
            set must be linearly independent too.}
        \item Every finite superset of a linearly dependent set in $\RR^n$ is
            linearly dependent.
    \end{enumerate}
\end{Exercise}

\begin{Answer}
    Regarding the first claim, let $E = \{u_1, \ldots, u_k\}$ be linearly
    independent.  Suppose that $\{u_1, \ldots, u_m\}$ is linearly dependent for
    some $m < k$.  Then we can find a nonzero vector $(\alpha_1, \ldots,
    \alpha_m)$ such that $\sum_{i=1}^m \alpha_i u_i = 0$.  Setting $\alpha_i = 0$
    for $i$ in $\{m+1, \ldots, k\}$ yields $\sum_{i=1}^k \alpha_i u_i = 0$,
    contradicting independence of $E$.

    Regarding the second claim, suppose that $E \subset F$ and $E$ is linearly
    dependent.  If $F$ is linearly independent, then we have a violation of the
    first claim.  Hence $F$ is linearly dependent.
\end{Answer}

\begin{example}
    It is easy to check that the set $E := \{\delta_1, \ldots, \delta_n\}
    \subset \RR^n$ defined by
    \begin{equation*}
        \delta_1 := 
        \left(
        \begin{array}{c}
            1 \\
            0 \\
            \vdots \\
            0
        \end{array}
        \right),
        \quad 
        \delta_2 := 
        \left(
        \begin{array}{c}
            0 \\
            1 \\
            \vdots \\
            0
        \end{array}
        \right),
        \; 
        \cdots,
        \;
        \delta_n := 
        \left(
        \begin{array}{c}
            0 \\
            0 \\
            \vdots \\
            1
        \end{array}
        \right)
    \end{equation*}
    is linearly independent.  Its elements are called the \navy{canonical
    basis vectors}\index{Canonical basis vectors} of $\RR^n$.
\end{example}

The span of the canonical basis vectors in $\RR^n$ is
equal to all of $\RR^n$.  In particular, each $x \in \RR^n$ can be expressed
as 
\begin{equation}\label{eq:repcb}
    x = \sum_{i=1}^n \alpha_i  \delta_i
    \quad \text{where }
    \alpha_i := \inner{x, \delta_i}.
\end{equation}
On the other hand, we cannot omit an element of the basis $\{\delta_1,
\ldots, \delta_n\}$ 
and still span $\RR^n$.  The next theorem generalizes this idea.  

\begin{theorem}\label{t:indepspan}
    If $E = \{u_1, \ldots, u_n\}$ is a set of $n$ vectors in $\RR^n$, then
    $\Span E = \RR^n$ if and only if $E$ is linearly independent.
\end{theorem}

See, for example, \cite{jan}.  Theorem~\ref{t:indepspan} captures the notion that linear
independence of a set of vectors means linear diversity, which allows the span
to be large.

\begin{Exercise}\label{ex:nullspace}
    The \navy{null space}\index{Null space} or \navy{kernel}\index{Kernel} of
    a matrix $A \in \matset{n}{k}$ is the set
    \begin{equation*}
        \kernel A := \setntn{x \in \RR^k}{Ax = 0}.
    \end{equation*}
    Show that (i) $\kernel A$ is a linear subspace of $\RR^k$ and (ii)
    $\kernel A = \{0\}$, where $0$ is the origin in $\RR^k$, if and only if
    the columns of $A$ form a linearly independent subset of $\RR^n$.
\end{Exercise}

\begin{Answer}
    Regarding (i), if $Ax = Ay = 0$, then $A(\alpha x + \beta y) = \alpha Ax +
    \beta Ay = 0$, so $\kernel A$ is a linear subspace.  
    Regarding (ii), suppose that $\kernel A = \{0\}$.  This means that the only $x
    \in \RR^k$ satisfying $Ax = 0$ is the zero vector, which is equivalent to
    linear independence of the columns of $A$.
\end{Answer}

\subsubsection{Basis Vectors and Dimension}\label{sss:bvd}

Let $V$ be a linear subspace of $\RR^n$.  A set $B \subset V$ is called a
\navy{basis}\index{Basis} for $V$ if $B$ is linearly independent and 
$\Span B = V$.

The key idea behind a basis is this: If $S = \{v_1, \ldots, v_k\}$ spans some linear subspace $V$,
then each element of $V$ can be written as a linear combination of elements of
$S$.  If $S$ is, in addition, linearly independent (and hence a basis), then 
this representation is \emph{unique}: for each $u \in V$, there is exactly one
$(\alpha_1, \ldots, \alpha_k) \in \RR^k$ such that
\begin{equation*}
    u = \alpha_1 v_1 + \cdots + \alpha_k v_k.
\end{equation*}
Indeed, if $u = \beta_1 v_1 + \cdots + \beta_k v_k$ is another representation,
then, subtracting this equality from the last, we have
\begin{equation*}
    (\alpha_1 - \beta_1) v_1 + \cdots + (\alpha_k - \beta_k) v_k = 0.
\end{equation*}
Because $S$ is assumed to be linearly independent, this yields $\alpha_i =
\beta_i$ for all $i \in \natset{k}$.

Not surprisingly, given their name, the canonical basis vectors $E := \{\delta_1,
\ldots, \delta_n\}$ serve as a basis for the whole space $\RR^n$. The
representation \eqref{eq:repcb}, with coefficients $\alpha_i = \inner{x,
\delta_i}$,
is unique.

\begin{theorem}\label{t:eqbase}
    If $V$ is a nonzero linear subspace of $\RR^n$, then
    \begin{enumerate}
        \item $V$ has at least one basis and
        \item every basis of $V$ has the same number of elements.
    \end{enumerate}
\end{theorem}

Theorem~\ref{t:eqbase} is a relatively deep result.  See, for example,
\cite{jan}.  The common number of bases in (ii) is called the \navy{dimension}
of $V$ and written as $\dim V$.  

The ``nonzero'' qualification in Theorem~\ref{t:eqbase} is included for the
case $V = \{0\}$, which is a linear subspace that lacks any basis.  It is
sensible, and standard, to agree that the linear subspace $\{0\} \subset
\RR^n$ has dimension zero.

Dimensionality is one measure of the ``size'' of a linear subspace.  To
illustrate, consider the system $Ax = b$ where $A \in \matset{n}{k}$, $x \in
\RR^k$ and $b \in \RR^n$.  Is there an $x$ that solves this system?  This will
be more likely if the column space of $A$ is large (see Example~\ref{eg:colspace}).  A
large span will be obtained when the columns are linearly ``diverse.'' In
other words, our hope is that there exists a large subset of the columns of
$A$ that is linearly independent, which in turn will be true when the span of
$A$ is high-dimensional.

To quantify these ideas, we define the \navy{rank}\index{Rank} of $A$ 
as
\begin{equation*}
    \rank A 
    := \dim(\range A)
    = \text{ dimension of the column space of } A.
\end{equation*}

\begin{theorem}\label{t:rank}
    For any matrix $A$, the following quantities are equal:
    \begin{enumerate}
        \item $\rank A$
        \item the maximal number of linearly independent columns of $A$, and
        \item the maximal number of linearly independent rows of $A$.
    \end{enumerate}
\end{theorem}

See, for example, Chapter~2 of \cite{aggarwal2020linear} or Chapter~5 of \cite{jan}.  In view of
Exercise~\ref{ex:mtbili}, for $A \in \matset{n}{k}$, we have $\rank A \leq k$.
If $\rank A = k$, then $A$ is said to have \navy{full column rank}.

In general, a linear subspace $V$ contains within itself many other linear
subspaces.  (For example, a plane passing through the origin contains many
lines that pass through the origin, each one of which is a linear subspace.)
However, there is no proper subspace of $V$ (i.e., no linear subspace
contained in and distinct from $V$) with the same dimension as $V$.  The next
theorem records this fact.

\begin{theorem}\label{t:dimorder}
    If $U, V$ are linear subspaces of $\RR^n$, then $U \subset V$ implies 
       $\dim U \leq \dim V$ with equality if and only if $U = V$. 
\end{theorem}

\begin{Exercise}\label{ex:mtbili}
    Let $U$ be a linear subspace of $\RR^n$.
    Using the results given above, prove the following statements:
    \begin{enumerate}
        \item The only $n$-dimensional linear subspace of $\RR^n$ is $\RR^n$.
        \item If $A \subset U$ is finite and $|A| > \dim U$, then $A$ is linearly dependent.
    \end{enumerate}
\end{Exercise}

\begin{Answer}
    For part (i) just set $V = \RR^n$ in Theorem~\ref{t:dimorder}.
    Regarding part (ii), let $A$ and $U$ be as stated, with $|A| =: m > \dim U$.
    Suppose to the contrary that $A$ is linearly independent.  Then $A$ is a basis
    for $\Span A$ and, therefore, $\dim \Span A = m > \dim U$.  At the same time,
    since $A \subset U$ and $U$ is a linear subspace, we have $\Span A \subset U$.
    Hence, by Theorem~\ref{t:dimorder}, we have $m \leq \dim U$.  Contradiction.
\end{Answer}

Part (ii) is related to Figure~\ref{f:span1}.  The plane in that figure is 2-dimensional, 
as we confirm in \S\ref{ss:lmm} below. Any three vectors
lying in the plane are linearly dependent.

\subsubsection{Linear Maps Are Matrices}\label{ss:lmm}

We will see in the following chapters that many nonlinear dynamic models and
estimation problems can be expressed in terms of linear operations in
high-dimensional spaces.   We now state the definition of linear maps and
their connection to matrices.

A function $u \mapsto A u$ from $\RR^k$ to $\RR^n$ is called
\navy{linear}\index{Linear} if
\begin{equation*}
    A(\alpha u + \beta v)
    = \alpha A u + \beta A v
    \quad
    \text{ for all $\alpha, \beta \in \RR$ and all $u, v$ in $\RR^k$}.
\end{equation*}
In this context, $A$ is sometimes called an \navy{operator} rather than a
function, but the meaning is the same.

\begin{Exercise}
    Fix $s, t \in \RR$.  Show that $A \colon \RR \to \RR$ defined by $A u = s
    u + t$ is a linear function on $\RR$ if and only if $t=0$.
\end{Exercise}

\begin{Exercise}\label{ex:ilfil}
    Let $A \colon \RR^n \to \RR^n$ be a linear bijection.  By
    Lemma~\ref{l:ibj}, the operator $A$ has an inverse $A^{-1}$ mapping
    $\RR^n$ to itself.  Prove that $A^{-1}$ is linear.
\end{Exercise}

\begin{Answer}
    Here's a proof by contradiction:    Let $A$ be as stated. 
    Suppose to the contrary that $A^{-1}$ fails to be linear.   Then we can
    find $\alpha, \beta \in \RR$ and $x,y\in \RR^n$ such that $A^{-1} (\alpha
    x + \beta y)$ and $\alpha A^{-1}x + \beta A^{-1}y$ are distinct points.
    Since $A$ is a bijection, their images under $A$ are also distinct, so 
    \begin{equation*}
        \alpha x + \beta y \not= A (\alpha A^{-1}x + \beta A^{-1}y).
    \end{equation*}
    Linearity of $A$ leads to a contradiction.
\end{Answer}

One of the most striking and useful properties of linear operators is that the
one-to-one property and the onto property are equivalent when the domain and
codomain agree:

\begin{theorem}\label{t:ooo}
    Let $A \colon \RR^n \to \RR^n$ be linear.  The operator $A$ is a
    one-to-one function if and only if it is onto.
\end{theorem}

The proof can be found in \cite{jan} or \cite{stachurski2016primer}.
Figure~\ref{f:func_types_3} illustrates in the one-dimensional case.  The
linear map $f(x) = \alpha x$ is onto if and only if it is one-to-one, which
occurs if and only if $\alpha \not= 0$.

\begin{figure}
   \centering
   \scalebox{0.6}{\includegraphics{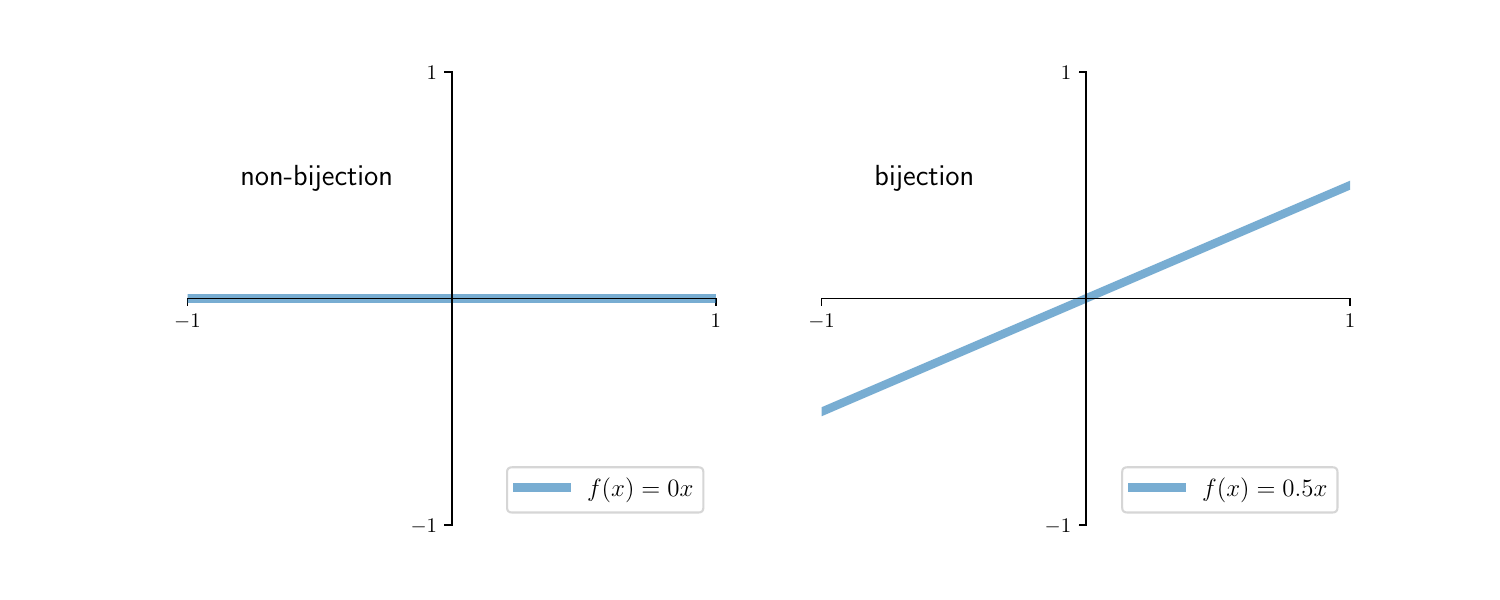}}
   \caption{\label{f:func_types_3} Equivalence of the onto and one-to-one properties}
\end{figure}

A matrix $A \in \matset{n}{k}$ is a linear function from $\RR^k$ to $\RR^n$
when identified with the mapping that sends a (column) vector $u$ in $\RR^k$
into the vector $Au$ in $\RR^n$.  In fact it is fundamental that, for
\emph{every} linear map $A \colon \RR^k \to \RR^n$, there exists a unique $M_A
\in \matset{n}{k}$ such that 
\begin{equation}
    \label{eq:matrep}
    A u = M_A u \quad \text{for all} \quad u \in \RR^k
\end{equation}
(see, e.g., \cite{kreyszig1978introductory}, \S2.9).  Thus, the set of linear
maps and the set of matrices are in one-to-one correspondence in the finite
dimensional setting.

\subsubsection{Linear Hyperplanes}\label{sss:linhy}

In Exercise~\ref{ex:whichli}, you confirmed that, for a given nonzero $c \in
\RR^n$, the subset of $\RR^n$ defined by $H := \setntn{x \in \RR^n}{\inner{c, x} = 0}$ is a linear
subspace of $\RR^n$.  Any set $H$ of this form is called a \navy{linear
hyperplane}\index{Linear hyperplane} in $\RR^n$.  The vector $c$ is called the
\navy{normal vector}\index{Normal vector} of the hyperplane.

\begin{theorem}[Rank-Nullity Theorem]\label{t:rnt}
    For each $A \in \matset{n}{k}$, we have
    \begin{equation*}
        \rank A + \dim (\kernel A) = k.
    \end{equation*}
\end{theorem}

\begin{example}
    If $A$ has linearly independent columns, then
    $\rank A =k$ and $A$ is said to have full column rank.
    Recall from Exercise~\ref{ex:nullspace} that this is precisely the setting
    where $\kernel A = \{0\}$.  Hence $ \dim (\kernel A) = 0$, and
    Theorem~\ref{t:rnt} holds.
\end{example}

In the discussion after Figure~\ref{f:span1}, we claimed that the linear
hyperplane has shown there is a two-dimensional subset of $\RR^3$.  The next
example generalizes this idea.

\begin{example}
    A linear hyperplane $H = \setntn{x \in \RR^n}{\inner{x, c} =
    0}$ with $c \not= 0$ has dimension $n-1$. 
    To see this, just set $A$ in Theorem~\ref{t:rnt} to $c$, treated as a $1
    \times n$ row vector.  Then $H = \kernel A$, and
    \begin{equation*}
        \dim H = \dim \kernel A = n - \rank A = n - 1.
    \end{equation*}
    (Here $\rank A = 1$ follows from Theorem~\ref{t:rank}.)
\end{example}

\subsubsection{Nonsingular Linear Systems}\label{sss:sls}

A crucial task in applied modeling is solving linear systems such as
$A x = b$, where $A$ is a matrix and $x$ and $b$ are conformable column
vectors.  The system can be underdetermined, overdetermined or exactly
determined (i.e., same number of equations as unknowns).  In this section we
consider the last case, where the theory is straightforward.

Let $A$ and $B$ be in $\matset{n}{n}$ and suppose that $A B = B A = I$.
Then $B$ is called the \navy{inverse}\index{Inverse} of $A$, written as
$A^{-1}$, while $A$ is said to be \navy{invertible}\index{Invertible} or
\navy{nonsingular}\index{Nonsingular}.

The next theorem states that, for square matrices, the property of having
either a left or a right inverse is equivalent to nonsingularity.

\begin{theorem}\label{t:lrilc}
    Given $A$ in $\matset{n}{n}$, the following statements are equivalent:
    \begin{enumerate}
        \item There exists an $L \in \matset{n}{n}$ such that $L A = I$.
        \item There exists an $R \in \matset{n}{n}$ such that $A R = I$.
    \end{enumerate}
    If one and hence both of these statements hold, then $A$ is nonsingular
    and $L = R = A^{-1}$.
\end{theorem}

\begin{proof}
    Fix $A$ in $\matset{n}{n}$ and suppose first that (i) holds.  This implies
    that $A$ has a left inverse.  In view of Exercise~\ref{ex:lioneone}, it
    must be that $x \mapsto Ax$ is one-to-one on $\RR^n$.   By
    Theorem~\ref{t:ooo}, this means that the same function is onto. Hence $x
    \mapsto Ax$ is a bijection, and therefore invertible (Lemma~\ref{l:ibj}).
    By Exercise~\ref{ex:ilfil}, the inverse function is linear and hence can
    be represented by a matrix.  We denote it $A^{-1}$.  For the left inverse
    $L$ we have $L=A^{-1}$, since $L = L (A A^{-1}) = (LA) A^{-1} = A^{-1}$.

    Now suppose there exists an $R \in \matset{n}{n}$ such that $A R = I$.
    Then, $A$ is the left inverse of $R$ and, by the previous argument, $R$ is
    invertible with $A = R^{-1}$.  Pre-multiplying both sides by $R$ gives $RA=I$,
    so $R$ is also a left inverse of $A$.  As we now know, this means that $A$
    is nonsingular and $R=A^{-1}$.

    In summary, if either (i) or (ii) holds, then $A$ is nonsingular, with
    left and right inverses both equal to $A^{-1}$.
\end{proof}

Consider the linear system $A x = b$, where $A \in \matset{n}{n}$ and $b \in
\RR^n$.  We seek a solution $x \in \RR^n$.   The next theorem provides
conditions under which such an $x$ is always uniquely identified.

\begin{theorem}\label{t:nncase}
    The following statements are equivalent:
    \begin{enumerate}
        \item For each $b \in \RR^n$, the equation $A x = b$
            has a unique solution.
        \item The columns of $A$ are linearly independent.
        \item The columns of $A$ form a basis of $\RR^n$.
        \item $\kernel A = \{0\}$.
        \item $\rank A = n$.
        \item $\det A$ is nonzero.
        \item $A$ is nonsingular.
    \end{enumerate}
\end{theorem}

The symbol $\det A$ represents the \navy{determinant}\index{Determinant} of $A$. For
the definition, see, for example, \cite{jan} or \cite{cohen2021linear}.  We note only that,
for $n \times n$ matrices $A$ and $B$,
\begin{itemize}\label{item:detprop}
    \item $\det(A B) = \det(A) \det(B)$,
    \item $\det(\alpha A) = \alpha^n \det A$ for all $\alpha \in \RR$, and
    \item $\det(A^{-1}) = (\det A)^{-1}$ whenever $A$ is nonsingular.
\end{itemize}

Most of the equivalences in Theorem~\ref{t:nncase} can be established from the
results presented above. The key idea is that, under these equivalent
conditions, the columns of $A$ form a basis of $\RR^n$, and hence any $b \in
\RR^n$ can be expressed uniquely as a linear combination of elements of these
columns.  In other words, there exists a unique $x \in \RR^n$ with $Ax=b$.
The remaining points are just equivalent ways of saying that the columns of
$A$ form a basis of $\RR^n$.

\begin{Exercise}
    Let $A$ and $B$ be conformable in the sense that $AB$ is well defined.  Show that
    $\rank(AB) \leq \rank A$, with equality if and only if $B$ is
    nonsingular.  (Hint: Use Theorem~\ref{t:dimorder}.)
\end{Exercise}

\subsubsection{Orthogonality}\label{sss:ortho}

We recall that vectors $u,v$ in $\RR^n$
are called \navy{orthogonal}\index{Orthogonal} and we write \navy{$u \perp v$}
if $\inner{u,  v} = 0$.  For a linear subspace $L$ of $\RR^n$, we call $u \in
\RR^n$ \navy{orthogonal to $L$}  and write \navy{$u \perp L$} whenever $u
\perp v$ for all $v \in L$.  

\begin{Exercise}
    The \navy{orthogonal
    complement}\index{Orthogonal complement} of linear subspace $L$ is defined as
    $L^{\perp} := \setntn{v \in \RR^n}{v \perp L}$.
    Show that $L^\perp$ is a linear subspace of $\RR^n$. 
\end{Exercise}

\begin{Answer}
    Fix $u, v \in L^{\perp}$ and $\alpha, \beta \in \RR$.
    If $z \in L$, then
    \begin{equation*}
        \inner{\alpha u + \beta v, z }
        = \alpha \inner{ u, z} + \beta \inner{v, z}
         = \alpha \times 0  + \beta \times 0 = 0
    \end{equation*}
    Hence $\alpha u + \beta v \in L^{\perp}$, as was to be shown.
\end{Answer}

\begin{Exercise}
    Prove: for any linear subspace $L \subset \RR^n$, we have $L \cap L^{\perp}
    = \{0\}$.
\end{Exercise}

A set of vectors $O \subset \RR^n$ is called an \navy{orthogonal
set}\index{Orthogonal set} if any two distinct elements of $O$ are orthogonal.
For any orthogonal set $O$, the
\navy{Pythagorean law}\index{Pythagorean law} 
\begin{equation*}
    \left\| \sum_{u \in O} u \right\|^2 
    = \sum_{u \in O} \| u \|^2 
\end{equation*}
always holds.

\begin{Exercise}
    Prove: Orthogonality implies linear independence in the sense that
    \begin{equation*}
        O \subset \RR^n 
        \text{ is orthogonal and $0 \notin O$ }
        \; \implies \;
        \text{$O$ is linearly independent}.   
    \end{equation*}
\end{Exercise}

An orthogonal set $O \subset \RR^n$ is called an \navy{orthonormal
set}\index{Orthonormal!set} if $\| u \| = 1$ for all $u \in O$.  If $L$ is a
linear subspace of $\RR^n$, $O$ is orthonormal in $L$ and $\Span O = L$, then
$O$ is called an \navy{orthonormal basis}\index{Orthonormal basis} of $L$.
For example, the canonical basis $\{e_1, \ldots, e_n\}$ forms an orthonormal
basis of $\RR^n$.

\begin{Exercise}
    Explain why an orthonormal basis $O$ of a subspace $L$ is, in fact, a basis of
    $L$, in the sense of the definition in \S\ref{sss:bvd}.
\end{Exercise}

\begin{Answer}
    If $O$ is an orthonormal basis of $L$, then, by definition, $O$ spans $L$.
    In addition, the elements of $O$ are independent because they are
    orthogonal and nonzero.
\end{Answer}

A matrix $M$ is called an \navy{orthonormal matrix}\index{Orthonormal!matrix} if 
$M \in \matset{n}{n}$ for some $n \in \NN$ and, in addition, the columns of
$M$ form an orthonormal set in $\RR^n$.  These kinds of matrices will be
important to us when we analyze singular value decomposition.  Notice that, 
\begin{itemize}
    \item by definition, every orthonormal matrix is square, and
    \item the $n$ columns of an orthonormal matrix $M$ in $\RR^n$ form a basis of
        $\RR^n$, since they are nonzero and orthogonal.
\end{itemize}

The second point tells us that $M$ is nonsingular.  The next lemma summarizes
important properties of orthonormal matrices.

\begin{lemma}\label{l:porthms}
    Fix $M \in \matset{n}{n}$ and let $I$ be the $n \times n$ identity.  The
    following statements are equivalent:
    \begin{enumerate}
        \item $M$ is an orthonormal matrix.
        \item $M^\top M = I$.
        \item $M M^\top = I$.
        \item $M^\top = M^{-1}$.
    \end{enumerate}
\end{lemma}

The equivalence of (ii) and (iii) is quite striking.  It tells us that a
square matrix with orthonormal columns automatically has orthonormal rows.

\begin{proof}
    Clearly (i) and (ii) are equivalent, since they are two ways of writing
    the same thing.  Equivalence of (ii)--(iv) follows from
    Theorem~\ref{t:lrilc} on page~\pageref{t:lrilc}.
\end{proof}

\subsection{Convexity and Concavity}

Convexity and concavity are structures of enormous significance in economics and
finance, in terms of both computation and theory.  In this section we note the key
definitions and provide exercises that help build familiarity.

\subsubsection{Convexity and Polyhedra}\label{sss:cc}

Convexity plays a central role in optimization and fixed point theory.  As
usual, a subset $C$ of $\RR^n$ is called \navy{convex}\index{Convex set} if
\begin{equation*}
    u, v \in C \text{ and } \lambda \in [0, 1] 
    \; \implies \;
    \lambda u + (1-\lambda) v \in C.     
\end{equation*}

\begin{Exercise}
    Show that
    \begin{enumerate}
        \item the unit simplex in $\RR^n$ is a convex subset of $\RR^n$.
        \item For all $a \in \RR^n$ and $\epsilon > 0$, the
            sphere $B = \setntn{x \in \RR^n}{\|x - a\| < \epsilon}$ is convex.
        \item The intersection of an arbitrary number of convex sets in
            $\RR^n$ is again convex.
    \end{enumerate}
\end{Exercise}

\begin{Answer}
    For part (i), let $S=\{ x \in \RR^n_+: \1^{\top} x \leq 1 \}$ be the unit simplex,
    where $\1$ is a column vector of ones.
    Fix $u, v \in S$ and $\lambda \in [0, 1]$. We have
    \begin{equation*}
        \lambda u + (1 - \lambda) v \geq 0
    \end{equation*}
    and
    \begin{equation*}
        \1^{\top} (\lambda u + (1 - \lambda) v)
        = \lambda \1^{\top} u + (1 - \lambda) \1^{\top} v \\
        \leq \lambda + (1 - \lambda) = 1
    \end{equation*}
    Hence $S$ is a convex set.

    For part (ii), fix $u, v \in B$, $\lambda \in [0, 1]$ and $a \in \RR^n, \epsilon>0$. We have
    \begin{equation*}
        \| \lambda u + (1 - \lambda) v - a \| = \| \lambda (u - a) + (1 - \lambda) (v - a) \| \\
        \leq \lambda \| u - a \| + (1 - \lambda) \| v - a \| \\
        \leq \lambda \epsilon + (1 - \lambda) \epsilon = \epsilon
    \end{equation*}
    Hence $B$ is convex.

    Regarding part (iii), let $(S_i)$ be a convex set for $i=1, 2, \cdots, n$.
    Fix $u, v \in \cap^n_{i=1} S_i$ and $\lambda \in [0, 1]$. We have
    \begin{equation*}
        z = \lambda u + (1 - \lambda) v \in S_i \quad \text{ for all } i
    \end{equation*}
    This implies $z \in \cap^n_{i=1} S_i$ and hence $\cap^n_{i=1} S_i$ is convex.
\end{Answer}

In economic optimization problems, the convex sets within which we seek
extrema are usually polyhedra.  A \navy{polyhedron} in $\RR^n$ is a set of the form
\begin{equation}\label{eq:polyh}
    P = \setntn{x \in \RR^n}{A x \leq b}
    \quad \text{ for some $A \in \matset{k}{n}$ and $b \in \RR^k$}.
\end{equation}
Equivalently, $P$ is a polyhedron in $\RR^n$ if there exist scalars $b_1, \ldots, b_k$
and $n$-dimensional vectors $a_1, \ldots, a_k$ such that $x \in P$ if and only if
$a_i^\top x \leq b_i$ for $i=1, \ldots, k$.

\begin{example}
    A budget set of the form $B = \setntn{x \in \RR^n_+}{p^\top x \leq
    m}$ where $m \geq 0$ and $p \in \RR^n_+$ is a polyhedron.
    Here $p$ is a vector of prices and $x$ is a consumption bundle.
    Indeed, $B$ can be expressed as the set of all $x \in \RR^n$ such that
    $p^\top x \leq m$ and, for the nonnegativity constraints,
    $-\delta_i^\top x \leq 0$ for $i=1, \ldots, n$,
    where $\delta_i$ is the $i$-th canonical basis vector in $\RR^n$.
    This meets the definition of a polyhedron.
\end{example}

Given $b \in \RR$ and nonzero $c \in \RR^n$, 
\begin{itemize}
    \item $H_0 := \setntn{x \in \RR^n}{x^\top c = b}$ is called a
        \navy{hyperplane}\index{Hyperplane} in $\RR^n$, while
    \item $H_1 := \setntn{x \in \RR^n}{x^\top c \leq b}$ is called a
        \navy{halfspace}\index{Halfspace} in $\RR^n$.
\end{itemize}

Note our convention.  In \S\ref{sss:linhy} we defined \emph{linear}
hyperplanes, which correspond to $H_0$ when $b=0$.  Thus linear hyperplanes
are a special kind of hyperplane.  You will be able to confirm that the hyperplane
$H_0$ is a linear subspace of $\RR^n$ if and only if $b=0$.

\begin{Exercise}
    Show that $H_0$ and $H_1$ are both convex in $\RR^n$.
\end{Exercise}

\begin{Answer}
    Fix $u, v \in H_0$, $\lambda \in [0, 1]$, $b \in \RR$ and nonzero $c \in \RR^n$. Then
    \begin{equation*}
        (\lambda u + (1 - \lambda) v)^{\top} c 
        = \lambda u^{\top} c + (1 - \lambda) v^{\top} c 
        = \lambda b + (1 - \lambda) b = b
    \end{equation*}
    Thus $H_0$ is convex. The proof that $H_1$ is convex is similar.
\end{Answer}

It is immediate from the definition that $P \subset \RR^n$ is a polyhedron if
and only if $P$ is the intersection of $k$ halfspaces in $\RR^n$ for some $k
\in \NN$.   Figure~\ref{f:polyhedron1} helps illustrate this idea.

\begin{figure}
    \centering
    \scalebox{0.66}{\includegraphics{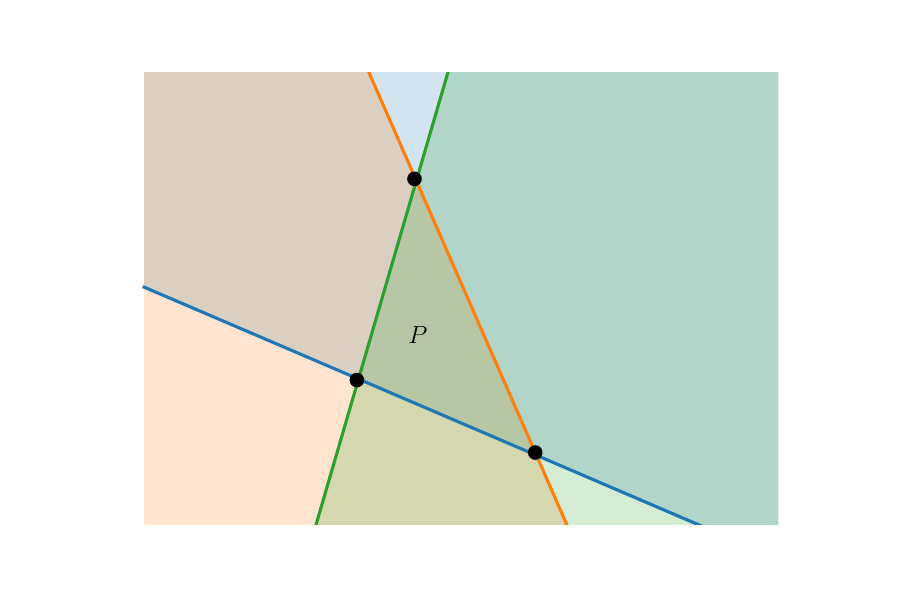}}
    \caption{\label{f:polyhedron1} A polyhedron $P$ represented as intersecting halfspaces}
\end{figure}

An \navy{extreme point}\index{Extreme point} of a polyhedron $P$ is a point $p
\in P$ that cannot be realized as the convex combination of other points in
$P$.  In other words, we cannot find two points $x, y \in P$ that are distinct
from $p$ and satisfy $\lambda x + (1-\lambda) y = p$ for some $\lambda \in [0,
1]$. The extreme points of $P$ in Figure~\ref{f:polyhedron1} are represented
  as black dots.

\begin{Exercise}\label{ex:polycon}
    Show that every polyhedron in $\RR^n$ is convex.
\end{Exercise}

A \navy{cone}\index{Cone} in $\RR^n$ is a set $C \subset \RR^n$ such that $x
\in C$ implies $\alpha x \in C$ for all $\alpha > 0$.

\begin{Exercise}\label{ex:concone}
    Show that (i) the intersection of any two cones is again a cone and (ii)
    a cone $C \subset \RR^n$ is convex (i.e., $C$ is a \navy{convex
    cone}\index{Convex cone}) if and only if it is closed under addition (i.e,
    $x, y \in C$ implies $x+y \in C$).
\end{Exercise}

\begin{Answer}
    Regarding part (i), let $C$ and $D$ be two cones. Fix $x \in C \cap D$ and
    $\alpha > 0$.  Since $x \in C$ and $x \in D$, and since both are cones, we
    have $\alpha x \in C \cap D$.  Hence $C \cap D$ is a cone.

    Regarding part (ii), suppose first that $C$ is a convex cone and fix $x, y \in
    C$.  Since $C$ is convex, it must be that $(1/2)(x+y) \in C$.  Since $C$ is a
    cone we can scale by $2$ without leaving $C$.  Hence $x+y \in C$.
    Suppose next that $C$ is a cone and closed under addition.  Fix $\alpha \in
    (0, 1)$ and $x,y \in C$.  Since $C$ is a cone, we have $\alpha x \in C$ and
    $(1-\alpha) y \in C$.  Since $C$ is closed under addition, it follows that 
    $\alpha x + (1-\alpha)y \in C$.  Hence $C$ is convex.
\end{Answer}

\begin{Exercise}
    The \navy{positive cone}\index{Positive cone} of $\RR^n$ is the set
    \begin{equation*}
        \RR^n_+ := \setntn{x \in \RR^n}{x \geq 0}.
    \end{equation*}
    Show that, for the partially ordered set $(\RR^n, \leq)$, the positive
    cone is an increasing subset of $\RR^n$.  Show in addition that $\RR^n_+$
    is a convex cone.
\end{Exercise}

\subsubsection{Convex and Concave Functions}\label{sss:conconfun}

A function $g$ from a convex subset $C$ of $\RR^n$
to $\RR^k$ is called \navy{convex}\index{Convex function} if 
\begin{equation*}
    g(\lambda u + (1-\lambda) v) \leq \lambda g(u) + (1-\lambda) g(v)
    \quad \text{whenever } u,v \in C \text{ and }
    0\leq \lambda \leq 1,   
\end{equation*}
and \navy{concave}\index{Concave function} if $-g$ is convex.
Concavity of $g$ is obviously equivalent to 
\begin{equation*}
    g(\lambda u + (1-\lambda) v) \geq \lambda g(u) + (1-\lambda) g(v)
    \quad \text{whenever } u,v \in C \text{ and }
    0\leq \lambda \leq 1.   
\end{equation*}

When $k=1$, the function $g$ is called, respectively, \navy{strictly convex}
or \navy{strictly concave} if, in addition, the inequalities become strict
whenever $u,v$ are distinct and $0 < \lambda < 1$.  

\begin{Exercise}
    Given an example of a function $g$ from $\RR^n$ to $\RR^n$ that is both
    convex and concave.
\end{Exercise}

\begin{Answer}
    Every linear function is both convex and concave.
\end{Answer}

These properties of functions are closely related to convexity of sets.
For example, the same function $g$ is convex if and only if its
\navy{epigraph}\index{Epigraph} 
\begin{equation*}
    \epi(g) := \setntn{(x, g(x)) \in \RR^{n+1}}{x \in C}
\end{equation*}
is convex.

\begin{Exercise}\label{ex:jendiscrete}
    If $C$ is a convex subset of $\RR^n$ and $g \colon \RR^n \to \RR$ is
    convex, then \navy{Jensen's inequality}\index{Jensen's inequality} states that,
    for any vectors $\{x_1, \ldots, x_k\} \subset C$ and 
    weights $\{\lambda_1, \ldots, \lambda_k\} \subset \RR$ with $0 \leq \lambda_i
    \leq 1$ and $\sum_i \lambda_i = 1$, we have
    \begin{equation*}
        g \left( \sum_{i=1}^k x_i \lambda_i \right)
        \leq 
        \sum_{i=1}^k  g \left( x_i \right) \lambda_i .
    \end{equation*}
    If $g$ is concave then the reverse inequality holds.  Prove Jensen's
    inequality for the case where $g$ is convex when $k=3$.
\end{Exercise}

In the next exercise, if $\{y_i\}$ is a finite collection of vectors in
$\RR^k$, then $\max_i y_i$ is the vector in $\RR^k$ obtained by taking the
maximum pointwise.  The minimum $\min_i y_i$ is defined in a similar way.

\begin{Exercise}\label{ex:prescon}
    Fix $m \in \NN$ and let $\{g_i\}_{i \in \natset{m}}$ be a collection of
    $\RR^k$-valued functions defined on a convex subset $C$ of $\RR^n$.
    Show that, 
    \begin{enumerate}
        \item if $g_i$ is convex for
            every $i$ in $\natset{m}$, then $g$ defined at each $x$ in $C$ by $g(x) :=
            \max_{i \in \natset{m}} g_i(x)$ is also convex on $C$ and
        \item if $g_i$ is concave for
            every $i$ in $\natset{m}$, then $g$ defined at each $x$ in $C$ by
            $g(x) := \min_{i \in \natset{m}} g_i(x)$ is also concave on $C$.
    \end{enumerate}
\end{Exercise}

\begin{Answer}
    It suffices to prove these results in the scalar case $n=k=1$, since the
    vector results are just element-by-element extensions of the scalar case.
    We focus on the setting where each $g_i$ is convex, as the concave case is
    similar.

    In general, for scalars $\{a_i\}_{i=1}^m$ and $\{b_i\}_{i=1}^m$, we always have 
    \begin{equation*}
        \max_i \{a_i + b_i\}
        \leq \max_i \{ a_i \} + \max_i \{b_i\}.
    \end{equation*}
    Hence, in the scalar case, when all $g_i$s are convex, we have, for any $\lambda \in
    [0,1]$ and $x, y \in C$,
    \begin{align*}
        \max_i \{ g_i(\lambda x + (1-\lambda) y)\}
        \leq \max_i \{ \lambda g_i(x) + (1-\lambda) g_i(y)\}
        \leq \lambda \max_i g_i(x) + (1-\lambda) \max_i g_i(y)
    \end{align*}
    This proves the scalar case  (and, by extension, the vector case).
\end{Answer}

\begin{Exercise}\label{ex:sumsconvex}
    Let $f$ and $g$ map $C$ to $\RR$, where $C \subset \RR^n$ is convex and so are $f$
    and $g$.  Show that (i) $f+g$ is convex and (ii) $f+g$ is strictly convex
    if, in addition, either $f$ or $g$ is strictly convex.
\end{Exercise}

Important examples of concave and convex scalar-valued functions include
certain kinds of quadratic forms.  For the following you should recall that a
symmetric $n \times n$ matrix $A$ is called 
\begin{itemize}
    \item \navy{positive semidefinite}\index{Positive semidefinite} if
        $x^{\top} A x \geq 0$ for any $x$ in $\RR^n$,
    \item \navy{positive definite}\index{Positive definite} if $x^{\top} A x >
        0$ for any nonzero $x$ in $\RR^n$,
    \item \navy{negative semidefinite}\index{Negative semidefinite} if
        $x^{\top} A x \leq 0$ for any $x$ in $\RR^n$, and 
    \item \navy{negative definite}\index{Negative definite} if $x^{\top} A x <
        0$ for any nonzero $x$ in $\RR^n$.
\end{itemize}
It's important to remember (but easy to forget) that symmetry is part of the
definition of these properties.  You probably remember that 
\begin{itemize}
    \item $A$ is positive definite (resp., semidefinite) if and only if all
        its eigenvalues are strictly positive (resp., nonnegative) and
    \item $A$ is positive definite (resp., semidefinite) $\implies$ its
        determinant is strictly positive (resp., nonnegative).
\end{itemize}
The second statement follows from the first, since the determinant of a matrix
equals the product of its eigenvalues.

The \navy{quadratic form} $q(x) = x^{\top} A x$ is 
\begin{itemize}
    \item convex if and only if $A$ is positive semidefinite,
    \item strictly convex if and only if $A$ is positive definite,
    \item concave if and only if $A$ is negative semidefinite, and
    \item strictly concave if and only if $A$ is negative definite.
\end{itemize}
See, for example, \cite{simon1994mathematics}.

\begin{Exercise}\label{ex:xpxpd}
    Let $X$ be $n \times k$.  Prove the following:
    \begin{enumerate}
        \item $X^\top X$ is positive semidefinite.
        \item If, in addition, $X$ has full column rank, then $X^\top X$ is positive definite.
    \end{enumerate}
\end{Exercise}

\subsection{Optimization}\label{ss:convexopt}

In this section we review some key topics in optimization theory.

\subsubsection{Definitions and Existence}\label{ss:optde}

A number $m$ contained in a subset $A$ of $\RR$ is called the
\navy{maximum}\index{Maximum} of $A$ and we write $m = \max A$ if 
$a \leq m$ for every $a \in A$.  It is called the
\navy{minimum}\index{Minimum} of $A$, written as $m=\min A$, if $m \in A$ and $a \geq m$ for every $a \in A$.  

Given an arbitrary set $D$ and a function $f \colon D \to \RR$, define
\begin{equation*}
    \max_{x \in D} f(x) := \max \setntn{f(x)}{x \in D}
    \quad \text{and} \quad
    \min_{x \in D} f(x) := \min \setntn{f(x)}{x \in D}.
\end{equation*}
A point $x^* \in D$ is called
\begin{itemize}
    \item a \navy{maximizer}\index{Maximizer} of $f$ on $D$ if 
        $x^* \in D$ and $f(x^*) \geq f(x)$ for all $x \in D$, and
    \item a \navy{minimizer}\index{Minimizer} of $f$ on $D$ if 
        $x^* \in D$ and $f(x^*) \leq f(x)$ for all $x \in D$. 
\end{itemize}
Equivalently, $x^* \in D$ is a maximizer of $f$ on $D$ if $f(x^*) = \max_{x
\in D} f(x)$, and a minimizer if   $f(x^*) = \min_{x \in D} f(x)$.
We define
\begin{equation*}
    \argmax_{x \in D} f(x) 
    := \setntn{x^* \in D}{f(x^*) \geq f(x) \text{ for all $x \in D$}}.
\end{equation*}
The set $\argmin_{x \in D} f(x)$ is defined analogously.

\begin{example}
    If $f(x) = x^2$ and $D=[-1, 1]$, then 
    $\argmax_{x \in D} f(x) = \{-1, 1\}$, while 
    $\argmin_{x \in D} f(x) = \{0\}$.  In the second case, where the solution
    set is a singleton, that is, a set with exactly one element, we write $\argmin_{x \in D} f(x) = 0$ as well.
\end{example}

\begin{Exercise}\label{eq:idmm}
    Let $f \colon D \to A \subset \RR$ be any given function.  Prove the
    following:
    \begin{enumerate}
        \item If $g \colon A \to \RR$ is a strictly increasing function, then
            $x^*$ is a maximizer of $f$ on $D$ if and only if $x^*$ is a
            maximizer of $g \circ f$ on $A$.
        \item If $g \colon A \to \RR$ is a strictly decreasing function, then
            $x^*$ is a maximizer of $f$ on $D$ if and only if $x^*$ is a
            minimizer of $g \circ f$ on $A$.
    \end{enumerate}
\end{Exercise}

\begin{Answer}
    Regarding part (i), let $f$ and $g$ be as stated.  Since $g$ is strictly
    increasing, we have
    \begin{equation*}
        f(x^*) \geq f(x), \;\; \forall \, x \in D
        \quad \iff \quad
        g(f(x^*)) \geq g(f(x)), \;\; \forall \, x \in D.
    \end{equation*}
    This proves the equivalence claimed in part (i).  (Note why the strictly
    increasing property of $g$ cannot be weakened here.)   The proof of (ii) is
    similar, with the second equality reversed.
\end{Answer}

One important special case of part (ii) is that $x^* \in D$ is a maximizer of
$f$ on $D$ if and only if $x^*$ is a minimizer of $-f$ on $D$.  Hence, any
maximization problem can be converted into a minimization problem and
vice-versa.

\subsubsection{Convexity and Extrema}

\S~\ref{ss:optde} discussed existence of optimizers.  In this section we
consider uniqueness.  The key observation is that, for convex and concave
functions, local optimizers are global optimizers. 

If $C \subset \RR^n$ and $f$ is a real-valued function on
$C$, then $u^* \in C$ is a \navy{local minimizer}\index{Local minimizer} of
$f$ on $C$ if there exists an open set $G$ in $C$ such that $u^* \in G$ and
$f(u^*) \leq f(u)$ whenever $u \in G$.  A \navy{local maximizer}\index{Local
maximizer} is defined analogously.

The next exercise highlights one of the most important facts concerning
computation of solutions to optimization problems in real-world applications.
It is especially valuable in high-dimensional settings, where optimization
problems can be very challenging.

\begin{Exercise}\label{ex:locglo}
    Show that, if $C \subset \RR^n$ is convex, $f$ is convex and $u^*$ is a
    local minimizer of $f$ on $C$, then $u^*$ is a minimizer of $f$ on $C$.
\end{Exercise}

\begin{Answer}
    Let $f$ and $C$ be as described and let 
    $u^*$ be a local minimizer.  Suppose, contrary to the claim in the
    exercise, that there exists a point $u$ in $C$ such that $f(u) < f(u^*)$.
    Then, by convexity, for each $\lambda$ in $[0, 1]$, we have
    \begin{equation*}
        f(\lambda u + (1-\lambda) u^*) \leq \lambda f(u) + (1-\lambda) f(u^*)
        < f(u^*).
    \end{equation*}
    Taking $\lambda \to 0$, we can find a point 
    $v := \lambda u + (1-\lambda) u^*$ arbitrarily close to $u^*$ such that
    $f(v) < f(u^*)$.  This contradicts the definition of a local minimizer.
\end{Answer}

Similarly, if $f$ is concave on $C$, then any local maximizer is a global
maximizer.

\subsubsection{Multivariate Quadratic Objectives}\label{sss:mqo}

\S~\ref{sss:quadgame} treated a one-dimensional quadratic optimization
problem.  Next we treat $n$-dimensional problems of the same type.

\begin{lemma}\label{l:quadopt}
    If $H \in \matset{n}{n}$ is positive definite, then, for any $b \in \RR$
    and $a \in \RR^n$,     
    \begin{equation*}
        u^* := H^{-1} a
        \text{ is the unique minimizer of }
        q(u) := u^\top H u - 2 u^\top a + b
        \text{ in } \RR^n.
    \end{equation*}
\end{lemma}

\begin{Exercise}
    Prove that any local minimizer of $q$ in Lemma~\ref{l:quadopt} is also a
    global minimizer.  (Hint: Use Exercise~\ref{ex:sumsconvex}.)
\end{Exercise}

\begin{Exercise}\label{ex:tffmc}
    Complete the proof of Lemma~\ref{l:quadopt} by showing that $u^*$ is a
    local minimizer.  It suffices to show that the derivative of $q$ at $u^*$
    is zero.  Use the following facts from matrix calculus:
    \begin{equation*}
        a \in \RR^n \text{ and } H \in \matset{n}{n}
        \quad \implies \quad
        \frac{\diff}{\diff u} u^\top a = a
        \; \text{ and } \;
        \frac{\diff}{\diff u} u^\top H u = (H + H^\top)u.
    \end{equation*}
\end{Exercise}

\subsection{Lagrangian Duality}\label{ss:ld}

If you have studied undergraduate economics, you will have used Lagrangian
methods to solve constrained optimization problems.  The objective of this
section is to supply insight on when and why the method works, as well as to
highlight the connection between an original constrained problem and a
so-called ``dual problem.'' This connection yields deep insights in economics,
finance, statistical learning, artificial intelligence and many other fields.

\subsubsection{Theory}\label{sss:theo}

Let $E$ be a subset of $\RR^n$ and let $f$ map $E$ to $\RR$.  We aim to solve
\begin{equation}\label{eq:lagran}
    \min_{x \in E} f(x)  
    \; \text{ subject to } \;
    g(x) = 0  \text{ and } h(x) \leq 0.
\end{equation}
Here $g$ maps $E$ to $\RR^m$ and $h$ maps $E$ to $\RR^k$.  For example, the
zero symbol in the statement $g(x) = 0$ is a vector of zeros, and the meaning
is that $g_i(x) = 0$ in $\RR$ for $i$ in $1, \ldots, m$.  The interpretation of the
second constraint is analogous.

Any $x \in E$ that satisfies the constraints in \eqref{eq:lagran} is called
\navy{feasible}.  Let $F(g, h)$ be the set of all feasible $x \in E$.  A
feasible point that attains the minimum in \eqref{eq:lagran} is called
\navy{optimal} for---or a \navy{solution} to---the optimization problem.

Our first claim is that the constrained optimization
problem~\eqref{eq:lagran} is equivalent to the \emph{un}constrained problem
\begin{equation}\label{eq:lagran2}
    P =
    \min_{x \in E} \,
        \max_{\theta \in \Theta}
        L(x, \theta)
\end{equation}
where $\theta := (\lambda, \mu)$ and
\begin{equation}\label{eq:lagdef}
    L(x, \theta) 
    = L(x, \lambda, \mu) 
    := f(x)  + \sum_{i=1}^m \lambda_i g_i(x)  + \sum_{i=1}^k \mu_i h_i(x) .
\end{equation}
Here $\theta \in \Theta$ combines the vectors of \navy{Lagrange
multipliers}\index{Lagrange multiplier} $\lambda \in \RR^m$ and $\mu \in
\RR^k_+$, where $\Theta := \RR^m \times \RR^k_+$.
By \emph{equivalent to,} we mean that
\begin{enumerate}
    \item $\hat f(x) := \max_{\theta \in \Theta} L(x, \theta)$ satisfies
        $\hat f = f$ on $F(g, h)$,
    \item $\hat f = +\infty$ on the complement of $F(g, h)$ and
    \item together $\hat f$ and $f$ satisfy
        \begin{equation}\label{eq:maffp}
            P 
            := \min_{x \in E} \hat f(x)
            = \min_{x \in F(g, h)} \hat f(x)
            = \min_{x \in F(g, h)} f(x).
        \end{equation}
\end{enumerate}

The first equality in \eqref{eq:maffp} is by definition.  The second two
follow from (ii) and (i) respectively.  Hence we need only verify (i)--(ii).

\begin{Exercise}
    Show that (i)--(ii) both hold.  You can assume that extrema exist.
\end{Exercise}

\begin{Answer}
    Regarding the first claim, suppose that $g(x) = 0$ and $h(x) \leq 0$, so that
    $x \in F(g, h)$.  In this case, 
    \begin{equation}\label{eq:lagcs}
        \hat f(x) 
        = \max_{\theta \in \Theta} L(x, \theta)
        = f(x) 
        + \max_{\lambda \in \RR^m} \sum_i \lambda_i g_i(x) 
        + \max_{\mu \in \RR^k_+} \sum_i \mu_i h_i(x) 
        = f(x).
    \end{equation}
    Hence $\hat f = f$ on $F(g, h)$ and claim (i) is verified.

    Regarding claim (ii), suppose that $g_i(x)$ deviates from zero for some $i$.
    Then $\max_{\lambda \in \RR^m} \sum_i \lambda_i g_i(x) $ equals $+\infty$,
    so $\hat f(x) = +\infty$.  In addition, if $h_i(x) > 0$ for some $i$, then
    $\max_{\mu \in \RR^k_+} \sum_i \mu_i h_i(x)  = +\infty$, so, once again
    $\hat f(x) = +\infty$.  We have confirmed that $\hat f = +\infty$ whenever $x
    \notin F(g, h)$.   
\end{Answer}

The function $L$ in problem \eqref{eq:lagdef} is usually called the
\navy{Lagrangian}\index{Lagrangian}.
The $P$ in \eqref{eq:lagran2} stands for \navy{primal}\index{Primal}.
So far we have shown that the original constrained problem and the
primal problem are the same.  The next step is to 
pair the primal problem with its \navy{dual}\index{Dual} problem, which
is obtained by reversing the order of min and max in the primal:
\begin{equation}\label{eq:lagdual}
    D =
    \max_{\theta \in \Theta} \;
    \min_{x \in E} \; 
    L(x, \theta).
\end{equation}
The dual problem has two attractive features.   One is that, when minimizing
with respect to $x$, we do not need to concern ourselves with the constraints
on $x$ in the original formulation~\eqref{eq:lagran}.  The second is that,
since $L(x, \theta)$ is linear in $\theta$, and since we are minimizing a
family of these functions with respect to $x$, the minimizer is a concave
function (see page~\pageref{ex:prescon}).  Concavity assists maximization,
which is the next step in solving $D$.

\begin{Exercise}\label{ex:opriexdu}
    Show that $D \leq P$ always holds. 
\end{Exercise}

\begin{Answer}
    Pick any $x' \in E$ and $\theta \in \Theta$.  We have $L(x', \theta) 
    \geq \min_x L(x, \theta)$, so
    \begin{equation*}
        \max_\theta L(x', \theta) \geq \max_\theta \min_x L(x, \theta).
    \end{equation*}
    \begin{equation*}
        \fore
        \min_x \max_\theta L(x, \theta) \geq \max_\theta \min_x L(x, \theta).
    \end{equation*}
\end{Answer}

The result $D \leq P$ in Exercise~\ref{ex:opriexdu} is called \navy{weak
duality}\index{Weak duality}.  If $P = D$,  then \navy{strong
duality}\index{Strong duality} is said to hold.  Unlike weak duality, strong
duality requires conditions on the primitives.

Evidently strong duality holds if and only if it is valid to reverse the order
of the min and max operators in the definition of the primal (or the dual).
Results of this kind are called \emph{min-max theorems}.  Such theorems hold
at ``saddle points'' of the function $L$.  The next section explains.

\subsubsection{Saddle Points and Duality}\label{sss:saddu}

We seek necessary and sufficient conditions for strong duality, which will
lead to a characterization of minimizers for the original
constrained problem.  To do so, we again take $\theta := (\lambda, \mu)$ and
$\Theta := \RR^m \times \RR^k_+$, so that 
\begin{equation}\label{eq:pandd}
    P = \min_{x \in E} \max_{\theta \in \Theta} L(x, \theta)
    \quad \text{and} \quad
    D = \max_{\theta \in \Theta} \min_{x \in E} L(x, \theta).
\end{equation}
A pair $(x^*, \theta^*)$ in $E \times \Theta$ is called a \navy{saddle
point}\index{Saddle point} of $L$ if 
\begin{equation}\label{eq:saddle}
    L(x^*, \theta) 
    \leq L(x^*, \theta^*) 
    \leq L(x, \theta^*) 
    \quad \text{for all} \quad
    (x, \theta) \in E \times \Theta.
\end{equation}
A depiction of a
saddle point of a given bivariate function $(x, \theta) \mapsto L(x, \theta) \in
\RR$ is given in Figure~\ref{f:saddle_1}.  The left hand side of the top panel is a 3D
visualization, and the right hand side is a contour plot of the same function.
The saddle point is at the center.

\begin{figure}
   \centering
   \scalebox{0.65}{\includegraphics[trim = 0mm 10mm 0mm 0mm, clip]{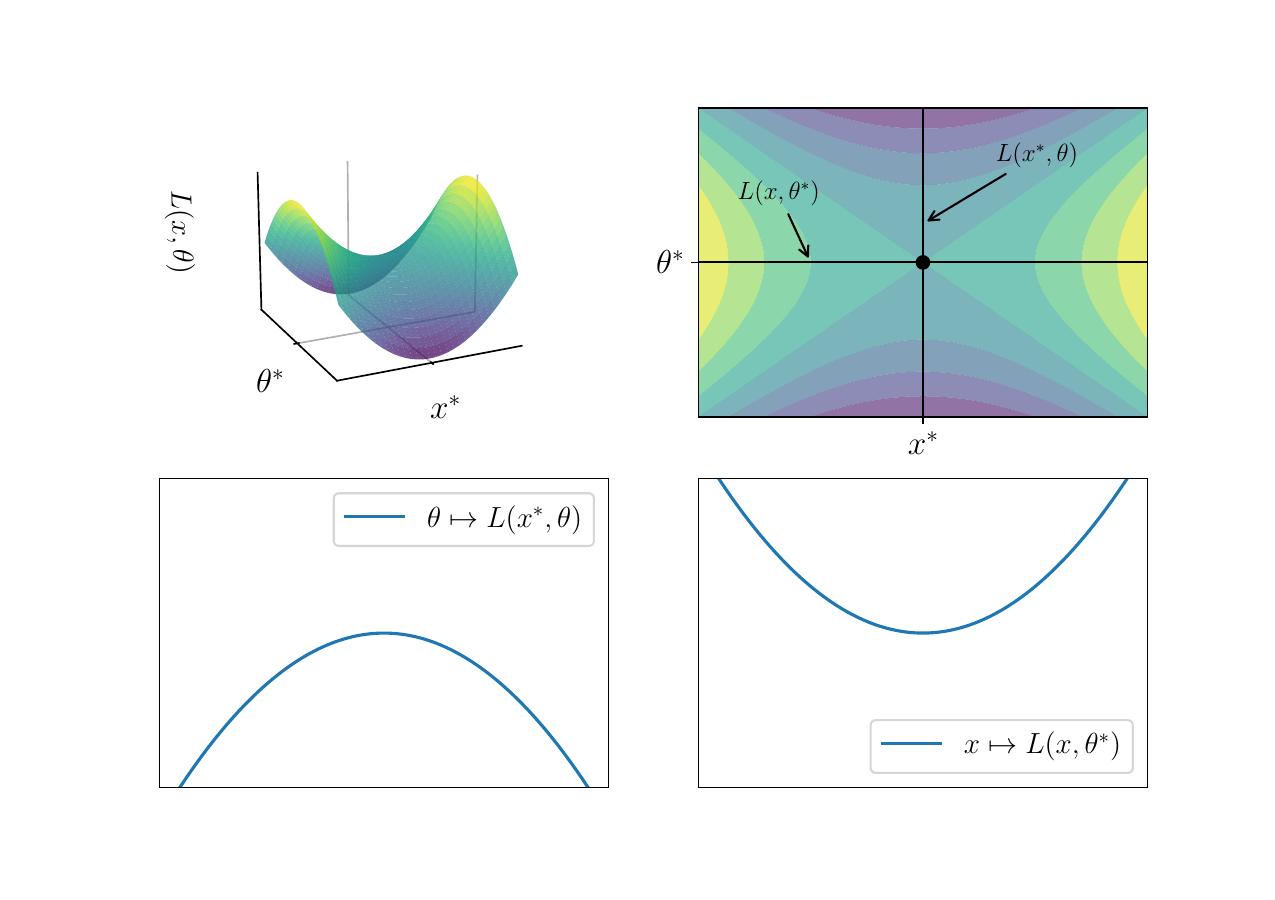}}
   \caption{\label{f:saddle_1} A saddle point $(x^*, \theta^*)$ of the function $L$}
\end{figure}

When the extrema in \eqref{eq:pandd} exist, we have the following result:

\begin{theorem}\label{t:saddle}
    If $L$ has a saddle point $(x^*, \theta^*)$ in $E \times \Theta$,
    then strong duality holds.  Moreover, $P=D = L(x^*, \theta^*)$ and $x^*$
    solves the constrained optimization problem~\eqref{eq:lagran}.
\end{theorem}

\begin{proof}
    Let $(x^*, \theta^*)$ be a saddle point $(x^*, \theta^*)$ of $L$ in $E \times \Theta$.
    We have, for all $(x, \theta) \in E \times \Theta$, 
    \begin{equation*}
        P 
        \leq \max_\theta L(x^*, \theta)
        \leq L(x^*, \theta^*) 
        \leq \min_x L(x, \theta^*),
    \end{equation*}
    where the first inequality is by definition and the second two are
    by the saddle point property.  The last term is clearly dominated by $D$, so we
    conclude that $P \leq L(x^*, \theta^*) \leq D$.  But, by weak duality, we
    also have $D \leq P$, so $P = L(x^*, \theta^*) = D$.  

    Finally, to confirm that $x^*$ solves the original constrained problem,
    suppose to the contrary that there exists an $x_0 \in E$ that satisfies
    the constraints and yields $f(x_0) < f(x^*)$.  Since the constraints are
    satisfied at both $x_0$ and $x^*$, we can apply~\eqref{eq:lagcs} on page~\pageref{eq:lagcs} to
    obtain $\max_{\theta \in \Theta} L(x_0, \theta) = f(x_0)$
    and $\max_{\theta \in \Theta} L(x^*, \theta) = f(x^*)$.  Hence the second inequality 
    of the saddle point condition implies
    \begin{equation*}
        f(x^*)
        = \max_{\theta \in \Theta} L(x^*, \theta)
            \leq \max_{\theta \in \Theta} L(x_0, \theta)
            = f(x_0).
    \end{equation*}
    This contradicts the hypothesis that $f(x_0) < f(x^*)$.
\end{proof}

Theorem~\ref{t:saddle} tells us that to solve the constrained optimization
problem and establish strong duality, we need only obtain a saddle point of
the Lagrangian.

\subsubsection{Karush, Kuhn and Tucker}\label{sss:kkt}

For well behaved problems, saddle points can be identified via well known
first-order conditions, called the
\navy{Karush--Kuhn--Tucker}\index{Karush--Kuhn--Tucker} (KKT) conditions.  To
state them, we return to the original problem \eqref{eq:lagran} and write
$\theta$ explicitly as $(\lambda, \mu)$, so that $(x^*, \lambda^*, \mu^*) \in
E \times \RR^m \times \RR^k_+$ satisfies the saddle point condition if
\begin{equation*}
    L(x^*, \lambda, \mu) 
        \leq L(x^*, \lambda^*, \mu^*) 
        \leq L(x, \lambda^*, \mu^*) 
        \quad \text{for all} \quad
        x \in E, \lambda \in \RR^m \text{ and } \mu \in \RR^k_+.
\end{equation*}
The \navy{KKT conditions} are met by $(x^*, \lambda^*, \mu^*) \in E \times
\RR^m \times \RR^k_+$ if $x^* \in F(g,h)$,
\begin{equation}\label{eq:swrtx}
    \nabla f(x^*) + \sum_{i=1}^m \lambda^*_i \nabla g_i(x^*)  + \sum_{i=1}^k
    \mu^*_i \nabla h_i(x^*) 
    = 0 
    \quad \text{and}
\end{equation}
\begin{equation}\label{eq:compslack}
    \mu^*_i h_i(x^*) = 0
    \quad \text{for all } i \in \natset{k}.
\end{equation}

Here we are requiring that $f, g, h$ are all differentiable at $x^*$ and, 
for a given function $q \colon E \to \RR$, we use $\nabla$ to represent the
vector of partial derivatives.   Equation \eqref{eq:swrtx} requires that the
derivative of the Lagrangian with respect to $x$ is zero at $x^*$.  The second
condition is called the \navy{complementary slackness}\index{Complementary
slackness} condition.

The KKT conditions are nothing more than saddle point conditions.
Condition~\eqref{eq:swrtx} is the first order condition for an interior
minimizer of the Lagrangian with respect to $x$, at the point $(x^*,
\theta^*)$, which tries to identify the local minimum visualized in the lower
right panel of Figure~\ref{f:saddle_1}.  The complementary slackness condition
is also needed because if $\mu^*_i h_i(x^*) = 0$ fails at some $i$, then
$\mu^*_i h_i(x^*) < 0$ must hold, in which case we could strictly increase the
Lagrangian by shifting $\mu^*_i$ to $0$.  This violates the saddle
point property.

If enough regularity conditions hold, then the KKT conditions exactly identify
extremal points.  Here is one example, which is proved in Chapter~8 of
\cite{matousek2007understanding}.

\begin{theorem}[Karush--Kuhn--Tucker]\label{t:kkt}
    If $E$ is open, $f$ is continuously differentiable and convex, and both
    $g$ and $h$ are affine functions, then $x^*$ minimizes $f$ on $F(g, h)$ if
    and only if there exists a pair $(\lambda^*, \mu^*) \in \RR^m \times
    \RR^k_+$ such that the KKT conditions hold.
\end{theorem}

\section{Solutions}\label{c:apsol} 

\shipoutAnswer

%% file: tikz/infsup.tex
\begin{tikzpicture}

\draw[->] (0,0) -- (8,0) node[anchor=north] { };

\draw[->] (0,0) -- (0,6) node[anchor=east] {};

\draw [black, domain=2:4] plot ({2 * tan(deg(\x)) + 5}, 1.5 * \x - 1.5);
\draw [black, domain=2:4] plot ({2 * tan(deg(-\x)) + 3}, 1.5 * \x - 1.5);

\draw [blue, 
       line width=5pt, 
       domain=2.7:4,
       line cap=round, 
       opacity=0.2] 
       plot ({2 * tan(deg(\x)) + 5}, 1.5 * \x - 1.5);
       
\draw [blue, 
		line width=5pt, 
		domain=2.7:4,
		line cap=round, 
		opacity=0.2] 
		plot ({2 * tan(deg(-\x)) + 3}, 1.5 * \x - 1.5);
       
\draw [red, 
		line width=5pt, 
		domain=2:2.66,
		line cap=round, 
		opacity=0.2] 
		plot ({2 * tan(deg(-\x)) + 3}, 1.5 * \x - 1.5);
		
\draw [red, 
		line width=5pt, 
		domain=2:2.65,
		line cap=round, 
		opacity=0.2] 
		plot ({2 * tan(deg(\x)) + 5}, 1.5 * \x - 1.5);

\draw(7.7,4.8) node[anchor=north] {$f$}
      (7.8,1.6) node[anchor=north] {$g$};
      
\draw [<-, Blue] (2.5,4)-- +(10pt,10pt) node[right, text=blue, opacity=0.8] {$f \vee g$};
\draw [<-, Red] (2.5,1.7)-- +(10pt,-10pt) node[right, text=Red, opacity=0.8] {$f \wedge g$};

\end{tikzpicture}